\documentclass{article}
\usepackage{hyperref}
\usepackage{amssymb,amsmath,amscd,dsfont, color}
\usepackage{amsthm,mathrsfs}
\usepackage[all,cmtip]{xy}

\usepackage{makeidx}         
\usepackage{multicol}        
\newcommand{\scheiding}{ \dotfill } 
\newcommand{\vulop}{ \dotfill } 
\makeindex

\newcommand{\pmat}[1]{\begin{pmatrix} #1 \end{pmatrix}}

\newcommand{\nin}{\noindent}

\newcommand{\cs}{\mathrm{cs}}

\newcommand{\QI}{I}

\newcommand{\cone}{\mathscr{C}}
\newcommand{\scri}{\mathscr{I}}

\newcommand{\ol}[1]{\ensuremath{\overline{#1}}}
\newcommand{\lra}[1]{\langle #1 \rangle}

\newcommand{\cZ}{\ensuremath{\mathcal{Z}}}
\newcommand{\cW}{\ensuremath{\mathcal{W}}}

\newcommand{\cE}{\ensuremath{\mathcal{E}}}
\newcommand{\cD}{\ensuremath{\mathcal{D}}}
\newcommand{\cF}{\ensuremath{\mathcal{F}}}
\newcommand{\cP}{\ensuremath{\mathcal{P}}}
\newcommand{\cA}{\ensuremath{\mathcal{A}}}
\newcommand{\cU}{\ensuremath{\mathcal{U}}}
\newcommand{\cB}{\ensuremath{\mathcal{B}}}
\newcommand{\cH}{\ensuremath{\mathcal{H}}}
\newcommand{\cM}{\ensuremath{\mathcal{M}}}
\newcommand{\cN}{\ensuremath{\mathcal{N}}}
\newcommand{\cC}{\ensuremath{\mathcal{C}}}
\newcommand{\cS}{\ensuremath{\mathcal{S}}}

\newcommand{\cL}{\ensuremath{\mathcal{L}}}
\newcommand{\cV}{\ensuremath{\mathcal{V}}}
\newcommand{\cG}{\ensuremath{\mathcal{G}}}
\newcommand{\cK}{\ensuremath{\mathcal{K}}}
\newcommand{\fz}{\ensuremath{\mathfrak{z}}}

\newcommand{\fsl}{\ensuremath{\mathfrak{sl}}}
\newcommand{\ft}{\ensuremath{\mathfrak{t}}}
\newcommand{\fg}{\ensuremath{\mathfrak{g}}}

\newcommand{\su}{\ensuremath{\mathfrak{su}}}
\newcommand{\hg}{\ensuremath{\hat{\mathfrak{g}}}}

\newcommand{\fk}{\ensuremath{\mathfrak{k}}}

\newcommand{\fK}{\ensuremath{\mathfrak{K}}}
\newcommand{\fZ}{\ensuremath{\mathfrak{Z}}}
\newcommand{\hfK}{\ensuremath{\widehat{\mathfrak{K}}}}
\newcommand{\fh}{\ensuremath{\mathfrak{h}}}
\newcommand{\fp}{\ensuremath{\mathfrak{p}}}

\newcommand{\gau}{\ensuremath{\mathfrak{gau}}}

\newcommand{\Ad}{\ensuremath{\operatorname{Ad}}}
\newcommand{\U}{\ensuremath{\operatorname{U}}}
\newcommand{\ad}{\ensuremath{\operatorname{ad}}}
\newcommand{\Gau}{\ensuremath{\operatorname{Gau}}}
\newcommand{\Aut}{\ensuremath{\operatorname{Aut}}}

\newcommand{\Fix}{\ensuremath{\operatorname{Fix}}}
\newcommand{\ev}{\ensuremath{\operatorname{ev}}}
\newcommand{\supp}{\ensuremath{\operatorname{supp}}}

\newcommand{\GL}{\ensuremath{\operatorname{GL}}}
\newcommand{\PSL}{\ensuremath{\operatorname{PSL}}}
\newcommand{\SL}{\ensuremath{\operatorname{SL}}}
\newcommand{\Aff}{\ensuremath{\operatorname{Aff}}}

\newcommand{\SO}{\ensuremath{\operatorname{SO}}}

\newcommand{\half}{\ensuremath{{\textstyle \frac{1}{2}}}}
\newcommand{\shalf}{\ensuremath{{\textstyle \frac{1}{2}}}}
\newcommand{\R}{\ensuremath{\mathbb{R}}}
\newcommand{\Q}{\ensuremath{\mathbb{Q}}}
\newcommand{\C}{\ensuremath{\mathbb{C}}}

\newcommand{\Z}{\ensuremath{\mathbb{Z}}}
\newcommand{\N}{\ensuremath{\mathbb{N}}}
\newcommand{\T}{\ensuremath{\mathbb{T}}}
\newcommand{\bv}{{\bf{v}}}
\newcommand{\hbv}{{\widehat{\bf{v}}}}

\newcommand{\bd}{{\bf{d}}}

\newcommand{\bP}{{\mathbb P}}

\newcommand{\bS}{{\mathbb S}}
\newcommand{\bE}{\ensuremath{\mathbb{E}}}
\newcommand{\bV}{\ensuremath{\mathbb{V}}}
\newcommand{\one}{\ensuremath{\mathbf{1}}}

\newcommand{\Exp}{\ensuremath{\operatorname{Exp}}}

\newcommand{\id}{\ensuremath{\operatorname{id}}}

\newcommand{\tr}{\ensuremath{\operatorname{tr}}}
\renewcommand{\tilde}{\widetilde}
\newcommand{\into}{\hookrightarrow}
\def\onto{\to\mskip-14mu\to}
\newcommand{\Spann}{\mathop{{\rm span}}\nolimits}

\newcommand{\Diff}{\mathop{{\rm Diff}}\nolimits}
\newcommand{\im}{\mathop{{\rm im}}\nolimits}
\newcommand{\trile}{\trianglelefteq}
\newcommand{\supeq}{\supseteq}

\newcommand{\derat}[1]{\frac{d}{dt} \hbox{\vrule width0.5pt
                height 5mm depth 3mm${{}\atop{{}\atop{\scriptstyle t=#1}}}$}}
\newcommand{\subeq}{\subseteq}
\newcommand{\g}{{\mathfrak g}}
\newcommand{\PU}{\mathop{\rm PU{}}\nolimits}
\newcommand{\OO}{\mathop{\rm O{}}\nolimits}

\newcommand{\der}{\mathop{\rm der}\nolimits}
\newcommand{\Spec}{\mathop{\rm Spec}\nolimits}
\newcommand{\SU}{\mathop{\rm SU}\nolimits}
\newcommand{\End}{\mathop{\rm End}\nolimits}
\newcommand{\per}{\mathop{\rm per}\nolimits}
\newcommand{\la}{\langle}
\newcommand{\ra}{\rangle}
\newcommand{\Rarrow}{\Rightarrow}

\newcommand{\ssssarr}{\hbox to 15pt{\rightarrowfill}}
\newcommand{\sssarr}{\hbox to 20pt{\rightarrowfill}}
\newcommand{\ssarr}{\hbox to 30pt{\rightarrowfill}}
\newcommand{\sarr}{\hbox to 40pt{\rightarrowfill}}
\newcommand{\arr}{\hbox to 60pt{\rightarrowfill}}
\newcommand{\larr}{\hbox to 60pt{\leftarrowfill}}
\newcommand{\Arr}{\hbox to 80pt{\rightarrowfill}}

\newcommand{\ssmapright}[1]{\smash{\mathop{\ssarr}\limits^{#1}}}

\theoremstyle{plain}
\newtheorem{Theorem}{Theorem}[section]
\newtheorem*{ELTheorem*}{Localization Theorem}
\newtheorem{Lemma}[Theorem]{Lemma}
\newtheorem{Proposition}[Theorem]{Proposition}

\newtheorem{Corollary}[Theorem]{Corollary}

\theoremstyle{definition}
\newtheorem{Definition}[Theorem]{Definition}

\newtheorem{Remark}[Theorem]{Remark}
\newtheorem{Example}[Theorem]{Example}

\renewcommand{\:}{\colon}
\newcommand{\res}{\vert}
\newcommand{\Lie}{\mathop{\bf L{}}\nolimits}
\renewcommand{\hat}{\widehat}
\renewcommand{\phi}{\varphi}
\newcommand{\dd}{{\tt d}}
\newcommand{\eps}{\varepsilon}

\newcommand\oline{\overline}

\begin{document}

\title{Positive energy representations of gauge groups I:\\
Localization}
\author{Bas Janssens 
and
Karl-Hermann Neeb
}

\maketitle

%
%
%
%

\begin{abstract}
This is the first in a series of papers on projective positive energy representations
of gauge groups. Let $\Xi \rightarrow M$ be a principal fiber bundle, and let
$\Gamma_{c}(M,\mathrm{Ad}(\Xi))$ be the group of compactly supported (local) gauge transformations.
If $P$ is a group of `space--time symmetries' acting on ${\Xi\rightarrow M}$, then
a projective unitary representation of $\Gamma_{c}(M,\mathrm{Ad}(\Xi))\rtimes P$
is of \emph{positive energy} if
every `timelike generator' $p_0 \in \fp$ gives rise to a Hamiltonian $H(p_0)$ whose spectrum
is bounded from below.

Our main result shows that in the absence of fixed points for the cone of timelike
generators, the
projective positive energy representations of the connected component $\Gamma_{c}(M,\Ad(\Xi))_0$
come from $1$-dimensional {$P$-orbits}.
For compact $M$ this yields a complete classification of the projective positive energy representations
in terms of lowest weight representations of affine Kac--Moody algebras. For noncompact $M$,
it yields a classification under further restrictions on the space of ground states.

In the second part of this series we consider larger groups of gauge transformations, which contain also
global transformations. The present results are used to localize the positive energy representations at
(conformal) infinity.
\end{abstract}


\setcounter{tocdepth}{2}
\tableofcontents

\section{Introduction}
\label{sec:1}

This is the first in a series of papers where we analyze the projective positive energy representations of gauge groups.

Our main motivation is the Wigner--Mackey classification \cite{Wigner1939, Ma49} of
projective unitary representations of the Poincar\'e group. Every irreducible such representation is labeled by
an ${\SO^\uparrow(1,d-1)}$-orbit in momentum space $\R^d$, together with an irreducible unitary representation
of the corresponding little group.
It is called a \emph{positive energy}
representation if for every 1-parameter group of timelike translations, the corresponding
Hamilton operator is bounded from below.
This
excludes the tachyonic orbits, leaving the
positive mass hyperboloids $p_{\mu}p^{\mu} = m^2$, 
the positive light cone $p_{\mu}p^{\mu} =0$, $p_0 >0$, 
and the origin $p=0$.
The corresponding little groups yield an intrinsic description of spin (for the positive mass hyperboloids)
and helicity (for the positive light cone).

In this series of papers we aim to extend this picture with an infinite dimensional group $\cG$ of gauge transformations, placing internal symmetries  and space-time symmetries on the same footing.

\subsection{Outline of Part I and II of this series}
For a gauge theory with structure group $K$, the fields over the space-time manifold $M$
are associated to a principal $K$-bundle $\Xi \rightarrow M$.
We consider the equivariant setting, where the group $P$ of space-time symmetries acts by automorphisms on
$\Xi \rightarrow M$, and the Lie algebra $\fp$ of $P$
contains a distinguished cone $\cone \subseteq \fp$  of `timelike generators'.
For Minkowski space, this is of course the Poincar\'e group $P$ with the cone $\cone$ of timelike translations.

The relevant group $\cG$ of gauge transformations depends on the context. It
always contains the group
$\cG_{c} := \Gamma_{c}(M,\Ad(\Xi))$
of compactly supported vertical automorphisms of $\Xi \rightarrow M$, and it is
this group that we will focus on in Part~I of this series.
In Part~II we consider also global gauge transformations. The group
$\cG$ is then larger than $\cG_{c}$, but it may be smaller than the group 
$\Gamma(M,\Ad(\Xi))$
of all vertical automorphisms due to boundary conditions at infinity.


A projective unitary representation of $\cG \rtimes P$ assigns to every timelike generator $p_0 \in \cone$
a one-parameter group of projective unitary transformations,
and hence a self-adjoint Hamiltonian $H(p_0)$ that is well defined up to a constant.
%
\begin{quote}
\emph{Our main objective is to study the projective unitary representations of $\cG\rtimes P$ that are of positive energy,
in the sense that the Hamiltonians $H(p_0)$ are bounded from below.}
\end{quote}
Perhaps surprisingly, this places rather stringent restrictions on the representation theory of $\cG$, leading to a complete classification in favorable cases.

\subsubsection{Outline of Part I}\label{sec:OutlinePartOne}
In the first part of this series, we focus on the group
$\cG_{c} := \Gamma_{c}(M,\Ad(\Xi))$
of \emph{compactly supported} gauge transformations. 
%
Our main result concerns the case where $M$ has no
fixed points for the cone $\cone$ of timelike
generators, and $K$ is a {$1$-connected}, semisimple Lie group.

\begin{ELTheorem*}
For every projective positive energy representation $(\overline{\rho},\cH)$ of the identity component
$\Gamma_c(M,\Ad(\Xi))_{0}$,
 there exists a 1-dimensional, $P$-equivariantly embedded
submanifold $S \subseteq M$ and
a positive energy representation $\ol{\rho}_{S}$
of $\Gamma_{c}(S,\Ad(\Xi))$ such that the following diagram commutes,
\begin{center}
$ $
\xymatrix{
\Gamma_{c}(M,\Ad(\Xi))_{0} \ar[d]_{r_{S}} \ar[r]^{\quad\ol{\rho}} & \mathrm{PU}(\cH)\\
\Gamma_{c}(S,\Ad(\Xi)), \ar[ru]_{\ol{\rho}_{S}}
}
\end{center}
where the vertical arrow denotes restriction to $S$.
\end{ELTheorem*}

%

%
%

This effectively reduces the classification of projective positive energy representations to the 1-dimensional case.
If $M$ is compact, then $S = \bigcup_{j=1}^{k}S_j$ is a finite union of circles. Correspondingly, the group
\begin{equation} \Gamma(S,\Ad(\Xi)) \cong \prod_{j=1}^{k}\Gamma(S_j, \Ad(\Xi))
\end{equation}
is a finite product of twisted loop groups if $K$ is compact, yielding
a complete classification 
in terms of tensor products of highest weight representations for the corresponding
affine Kac--Moody algebras \cite{Ka85, PS86, GW84, TL99}.

To some extent these results generalize to the case of noncompact manifolds~$M$,
where $S$ can then have infinitely many connected components.
We are able to classify the projective positive energy representations under the additional assumption
that they admit a cyclic ground state vector which is unique up to scalar.
These \emph{vacuum representations}
are classified in terms of infinite
tensor products of vacuum representations of affine Kac--Moody algebras.
In particular, every such representation is of type I.
Without the vacuum condition, the classification is considerably more involved.
We study in detail the case where all connected components of $S$ are circles.
Under a geometric `spectral gap' condition, we reduce the classification of projective positive energy representations
to the representation theory of UHF $C^*$-algebras, yielding a rich source of representations of type II and III.

\subsubsection{Outline of Part II}

In the second part of this series, we consider the case where $\cG$ contains \emph{global} as well as compactly supported gauge transformations.
To study the projective positive energy representations, we use the exact sequence
\[
 1 \rightarrow \cG_{c} \rightarrow \cG \rightarrow \cG/\cG_{c} \rightarrow 1.
\]
By the results from Part I on the positive energy representations of $\cG_{c}$, the problem essentially reduces
to the group $\cG/\cG_{c}$ of gauge transformations `at infinity'.

Needless to say, the resulting representation theory is very sensitive to the boundary
conditions at infinity. We focus on the situation where $M$ is an
\emph{asymptotically simple} space-time in the sense of Penrose \cite{Pe65,GKP72,HE73,PR86},
and $\cG$ consists of gauge transformations that extend smoothly to the conformal boundary.
For the motivating example of the Poincar\'e group acting on $d$-dimensional Minkowski space,
we obtain the following detailed account of the projective positive energy representation theory.


\paragraph{Minkowski space in dimension $d>2$:}
In this setting we
show that
the projective positive energy representations of $\cG$
depend only on the \emph{1-jets}
of the gauge transformation at spacelike infinity $\iota_{0}$ and at past and future timelike infinity~$\iota_{\pm}$.
This reduces the problem to the classification of projective positive energy representations
of the (finite dimensional!) semidirect product
\begin{equation}
\big(\SO^\uparrow(1,d-1)\times K^3\big)\ltimes \big(\R^d \oplus (\fk^3 \otimes \R^{d*})\big),
\end{equation}
where $\SO(1,d-1)$ acts on $\R^d$ in the usual fashion, and $K$ acts on its Lie algebra~$\fk$
by the adjoint representation.
The three copies of $K$ encode the values of the gauge transformation at $\iota_0$ and $\iota_{\pm}$, whereas
the three copies of
the additive group $\fk \otimes \R^{d*}$ encode the derivatives.

In the special case where the derivatives act trivially, we recover a projective positive energy representations of the Poincar\'e group
$\R^d \rtimes \SO^\uparrow(1,d-1)$, together with 3 projective unitary representations of the structure group $K$.
More generally, by Mackey's Theorem of Imprimitivity, the
irreducible projective positive energy
representations are labeled by an
orbit of $\SO^\uparrow(1,d-1) \times K^3$ in $\R^d \oplus (\fk^3 \otimes \R^{d*})$ whose energy is bounded from below,
together with a projective unitary representation of the corresponding little group.
In general these little groups will not contain the three copies of $K$, giving rise to phenomena that are
reminiscent of spontaneous symmetry breaking.

\paragraph{Minkowski space in dimension $d=2$:}
In contrast to the higher dimensional case, the projective positive energy representations in $d=2$ do \emph{not} in general factor through a finite dimensional Lie group.
%
For simplicity, we consider the group $\cG$ of gauge transformations that extend smoothly to the
\emph{conformal compactification}
of 2-dimensional Minkowski space.
Here the three points
$\iota_0$ and $\iota_{\pm}$ at space- and timelike infinity are collapsed to a single point~$\QI$, and past and future null infinity
$\scri^{-}$ and $\scri^{+}$ are identified along lightlike geodesics (cf.~\cite{PR86}).
The boundary of this space is a union of two circles $\bS^1_{L}$ and $\bS^1_{R}$ (corresponding
to left and right moving modes)
that intersect transversally in a single point $\QI$.
We prove that the positive energy representations of $\cG$ depend only on the
\emph{values} of the gauge transformations at null infinity
$\scri = (\bS^1_{L} \cup \bS^1_{R})\setminus\{I\}$,
and on the \emph{2-jets} at the single point $\QI$.

At the Lie algebra level, the problem therefore reduces to classifying the projective positive energy representations of
the abelian extension
\begin{equation}
0 \rightarrow |\fk| \rightarrow \fg \rightarrow \fg_{\mathrm{eq}}\rightarrow 0.
\end{equation}
Here
the equalizer Lie algebra
\begin{equation}
\fg_{\mathrm{eq}} = \Big\{(\xi, \eta) \in \Gamma\big(\bS^1_{L},\mathrm{ad}(\Xi)\big) \times \Gamma\big(\bS^1_{R}, \mathrm{ad}(\Xi)\big)
\,;\, \xi(\QI) = \eta(\QI)\Big\}
\end{equation}
represents the values of the infinitesimal gauge transformations on the conformal boundary $\bS^1_{L} \cup \bS^1_{R}$,
the abelian Lie algebra $|\fk|$ with underlying vector space $\fk$ represents the mixed second derivatives at $\QI$,
and $\fg_{\mathrm{eq}}$ acts on $|\fk|$ by evaluating at $\QI$ and composing with the adjoint representation.

Even in the untwisted case, where $\Xi$ is the trivial $K$-bundle, the classification of the
projective positive energy representations is by no means trivial. This is because the positive energy condition
is \emph{not} with respect to rigid rotations of $\bS^1_{L/R}$, but with respect to
the translations of the real projective line $\bS_{L/R}^1 = \R \cup \{\QI\}$ fixing the point $\QI$ at infinity.

Under the restrictive additional condition that the projective unitary representations are of positive energy
with respect to rotations as well as translations, we obtain a classification in terms of highest weight representations of
the two untwisted affine Kac--Moody algebras corresponding to $\bS^1_{L}$ and $\bS^1_{R}$,
together with a projective positive energy representation of the finite dimensional Lie group of 2-jets of
gauge transformations at $\QI$.

Although the Kac--Moody representations are familiar from the construction of loop group nets in conformal field theory, the positive energy representations involving 2-jets (which are Poincar\'e invariant but not conformally invariant) appear to be a novel feature.

\subsection{Structure of the present paper}\label{sec:structure}

For a closed quantum system that is described by a Hilbert space $\cH$, any two states
that differ by a global phase are physically indistinguishable.
The state space of the system is therefore described by the projective Hilbert space $\bP(\cH)$.
By Wigner's Theorem, a connected Lie group $G$ acts on the projective Hilbert space $\bP(\cH)$
by projective unitary transformations, resulting in a projective unitary representation
$\ol{\rho} \colon G \rightarrow \PU(\cH)$.

\subsubsection{Positive energy representations}
Since we are interested in the group of compactly supported gauge transformations,
we need to work with \emph{infinite dimensional} Lie groups
 modeled on locally convex spaces, or \emph{locally convex Lie groups} for short.  \index{Lie group!locally convex \vulop}
 In \S\ref{SectionPPER} we recall and extend some recent results from
\cite{Ne11, JN19} that allow us to go back and forth between
smooth projective unitary
representations of a locally convex Lie group~$G$,
smooth unitary representations of a central Lie group
extension~$G^\sharp$,
and the derived representations of its Lie algebra~$\g^\sharp$.

In \S\ref{sec:3} we introduce projective \emph{positive energy representations}
in the context of a Lie group $P$ that acts smoothly on $G$ by automorphisms.
For a distinguished \emph{positive energy cone} $\cone\subseteq \fp$, we require that the
spectrum of the corresponding self-adjoint operators in the derived representation is bounded from below.
Since a representation is of positive energy for the cone $\cone$ if and only if it is of positive energy
for the 1-parameter subgroups generated by $\cone \subseteq \fp$, we can always reduce to the case $P = \R$,
where the nonnegative spectrum condition pertains to a single Hamilton operator $H$.
Using the Borchers--Arveson Theorem, we further reduce the classification 
to the \emph{minimal} representations, where
$H\geq 0$ is the smallest
possible Hamilton operator with nonnegative spectrum.

In \S\ref{sec:4} we then turn to our  subject proper,
namely the locally convex Lie group $\cG_{c}$
of compactly supported gauge transformations.
We consider the setting where $M$ is a manifold, $P$ is a Lie group
acting smoothly on~$M$,
and $\cK \rightarrow M$ is a bundle of 1-connected semisimple Lie groups that is equipped
with a lift of this action.
The group $\cG_{c} = \Gamma_{c}(M,\cK)$ of compactly supported sections
then carries a smooth action of $P$ by automorphisms, and we consider
the smooth projective unitary representations of the semidirect product $\Gamma_{c}(M,\cK)\rtimes P$.

The motivating example is of course the case where $\cK =\Ad(\Xi)$
is the adjoint bundle of
 a principal fiber bundle $\Xi\rightarrow M$, and
 $\Gamma_{c}(M,\Ad(\Xi))$ is the group of vertical automorphisms of $\Xi$ that
 are trivial outside a compact subset of $M$. The reason for the minor generalization to bundles of Lie groups
 is purely technical; the reduction to simple structure groups in \S\ref{vanseminaarsimpel}
 is somewhat easier in that setting.

\subsubsection{The Localization Theorem}

The main result in the present paper is the
 following localization result (a minor generalization of the one in \S\ref{sec:OutlinePartOne}), which essentially reduces the classification of projective
 positive energy representations to the $1$-dimensional case.

\begin{ELTheorem*}[Theorem \ref{Thm:equivarloc}]
Let $(\overline{\rho},\cH)$ be
a projective positive energy representation of
$\Gamma_c(M,\cK)\rtimes P$.
If the cone $\cone$ has no fixed points in $M$, then
 there exists a 1-dimensional, $P$-equivariantly embedded
submanifold $S \subseteq M$ such that on the connected component $\Gamma_{c}(M,\cK)_{0}$,
the projective representation
$\ol{\rho}$ factors through the restriction homomorphism
$r_{S} \colon \Gamma_{c}(M,\cK)_{0} \rightarrow \Gamma_{c}(S,\cK)$.
\end{ELTheorem*}

We sketch the proof in the special case that the structure group $K$ of $\cK$ is a compact simple Lie group.
The result for semisimple Lie groups is reduced to
the simple case in \S\ref{vanseminaarsimpel}, and to the compact simple case in \S\ref{sec:redcpt}.
We require $K$ to be 1-connected, but this is by no means essential;
results beyond 1-connected groups are discussed in~\S\ref{sec:StatementDiscussion} and \S\ref{sec:nonconnected}.

Further, we will assume without loss of generality that $P$ is the additive group $\R$ of real numbers.
The corresponding flow is then given by a nonvanishing vector field $\bv_{M}$ on $M$, which lifts to a vector field
$\bv$ on $\cK$.
We denote the corresponding derivation of the gauge algebra by $D \xi := L_{\mathbf{v}} \xi$.
The reduction from $P$ to $\R$ is carried out  in
\S\ref{sec:HigherDimSymmetry}
by considering the 1-parameter subgroups of $P$ that are generated by
elements of the positive energy cone $\cone\subseteq \fp$.

\paragraph{Step 1.} Let $\fK \rightarrow M$ be the bundle of Lie algebras derived from
$\cK \rightarrow M$.
Then every smooth projective unitary representation of
$\Gamma_{c}(M,\cK)\rtimes \R$ gives rise to an $\R$-invariant
2-cocycle $\omega$ on the compactly supported gauge algebra $\Gamma_{c}(M,\fK)$.
In \S\ref{GySsCoc}
we show that every such cocycle is cohomologous
to one of the form
\begin{equation}\label{eq:introlambda}
\omega(\xi,\eta) = \lambda(\kappa(\xi, d_{\nabla}\eta)) \qquad \text{for} \qquad  \xi, \eta \in \Gamma_{c}(M,\fK),
\end{equation}
where $\nabla$ is a Lie connection on $\fK\rightarrow M$, $\kappa$ is a positive definite
invariant bilinear form on the Lie algebra $\fk$ of $K$, and $\lambda \colon \Omega^1_{c}(M)\rightarrow \R$ is a closed current that is
invariant under the flow.

\paragraph{Step 2.}
The positive energy condition for $(\rho, \cH)$ gives rise to a
Cauchy--Schwarz inequality
for the derived Lie algebra representation $\dd\rho$. In \S\ref{subsec:3.2} we show that
if $[\xi,D\xi] = 0$ and $\omega(\xi,D) = 0$, then
\begin{equation}\label{eq:introCS}
\langle\psi, i\dd\rho(\xi)\psi\rangle^2 \leq 2\omega(\xi,D\xi)
\langle\psi, H\psi\rangle
\end{equation}
for every smooth unit vector $\psi$.
Moreover, $\omega(\xi, D\xi)$ is nonnegative. In~\S\ref{PEcocycles} this is used
to show that
the closed current $\lambda$ from \eqref{eq:introlambda}
takes the form
\begin{equation}\label{eq:introCocyclemeasure}
\lambda(\alpha) = \int_{M}(i_{\bv_{M}}\alpha)(x) d\mu(x)
\end{equation}
for a flow-invariant regular Borel measure $\mu$ on $M$. In terms of this measure,
the Cauchy--Schwartz inequality
\eqref{eq:introCS} becomes
\begin{equation}\label{eq:CSforGaugeIntro}
\langle \psi, i\dd\rho(\xi)\psi\rangle^2 \leq 2\langle\psi, H \psi\rangle \|L_{\bv} \xi\|^2_{\mu}.
\end{equation}
In other words, the expectation value of the unbounded operator $i\dd\rho(\xi)$ is controlled in terms of the
energy $\langle\psi, H \psi\rangle$,
and the $L^2$-norm of the derivative of $\xi$ with respect to the measure $\mu$.
In fact, a small but important refinement allows one to control the expectation of
$i\dd\rho(L_{\bv} \xi)$ in terms of the same data.

\paragraph{Step 3.}
In \S\ref{sec:6} we use
the Cauchy--Schwarz estimate \eqref{eq:CSforGaugeIntro} and its refinement to
show that
\begin{equation}\label{eq:introEstH}
\pm i\dd\rho(\xi) \leq \|\xi\|_{\nu}\one + \|\xi\|_{\mu}H
\end{equation}
as unbounded operators. The measure $\nu$ is
absolutely continuous with respect to $\mu$, with a
density that is upper semi-continuous and invariant under the flow.
From a technical point of view, this is the heart of the proof.
It allows us to extend $\dd\rho$ to a positive energy representation of
the Banach--Lie algebra $H^{2}_{\partial}(M,\cK)$ of sections that are twice differentiable
in the direction of the flow, but only $\nu$-\emph{measurable} in the direction
perpendicular to the flow.

\paragraph{Step 4.} The final steps of the proof are carried out in \S\ref{sec:locthmsec}.
Every point in $M$ admits a flow box $U_0 \times I \simeq U \subseteq M$,
where the flow fixes all points in $U_0$ and acts by translation
on the interval $I \subseteq \R$ for small times.
Accordingly, the flow-invariant measure on $U$
decomposes as $\mu = \mu_0\otimes dt$.
Since the sections in $H^{2}_{\partial}(U,\cK) \subseteq H^{2}_{\partial}(M,\cK)$ need only be measurable in the direction perpendicular to the flow, we can continuously embed $C^{\infty}_{c}(I,\fk)$ as a Lie subalgebra of $H^{2}_{\partial}(U,\cK)$
by multiplying with an indicator function $\chi_{E}$ for a Borel subset $E\subseteq U_0$ of finite measure.
This yields a projective unitary representation of
$C^{\infty}_{c}(I,\fk)$ with central charge $2\pi \mu_0(E)$.

\paragraph{Step 5.}
Since the dense set of analytic vectors for $H$ is analytic for the extension of $\dd\rho$ to
$H^{2}_{\partial}(M,\cK)$, the projective unitary representation of $C^{\infty}_{c}(I,\fk)$ extends to
the 1-connected Lie group $G$ that integrates $C^{\infty}_{c}(I,\fk)$.
This gives rise to a smooth central $\T$-extension $G^{\sharp} \rightarrow G$.
For every smooth map $\sigma \colon \bS^2 \rightarrow G$, the pullback $\sigma^*G^{\sharp} \rightarrow \bS^2$
is a principal circle bundle, and integrality of the corresponding Chern class implies that $2\pi \mu_0(E) \in \N_0$.
Since this holds for every Borel set, we conclude that $\mu_0$ is a locally finite sum of point measures, and hence that
$\mu = \mu_0 \otimes dt$ is concentrated on a closed embedded submanifold $S_{U}\subseteq U$ of dimension 1.
Since the argument is local, the measure $\mu$ is concentrated on a closed, embedded, 1-dimensional
submanifold $S \subseteq M$.
Using \eqref{eq:introEstH}, one shows that $\dd\rho$ vanishes on the ideal of sections that vanish $\mu$-almost everywhere.
This proves the theorem at the Lie algebra level. The result at the group level follows because $\Gamma_{c}(S,\cK)$ is 1-connected.

\subsubsection{Classification of positive energy representations}

For manifolds with a fixed point free $\R$-action, the
Localization Theorem effectively reduces the projective positive energy representation theory
to the 1-dimensional setting.

\paragraph{Compact manifolds} For \emph{compact} manifolds $M$, we show in \S\ref{sec:8} that
 this leads to a full classification.
Indeed, since $S = \bigcup_{j=1}^{k}S_{j}$ is a finite union of periodic orbits $S_{j}$, the group
$\Gamma(S,\cK)$ is a finite product of \emph{twisted loop groups}
$\Gamma(S_j, \cK)$.
The projective positive energy representations of twisted loop groups are classified in~\S\ref{subsec:7.4},
using the rich structure and representation {theory} of affine Kac--Moody Lie algebras \cite{Ka85},
combined with the method of holomorphic
induction for Fr\'echet--Lie groups developed in \cite{Ne13,Ne14a}.
%
%
%

This leads to a full classification of projective positive energy representations
of $\Gamma(M,\cK)$, which is detailed in \S\ref{subsec:8.3}.
Up to unitary equivalence, every irreducible projective positive energy representation is determined by:
\begin{itemize}
\item Finitely many periodic $\R$-orbits $S_j\subseteq M$, each equipped with a central charge $c_j\in \N$.
\item For every pair $(S_j,c_j)$, an anti-dominant integral weight $\lambda_j$ of the corresponding
affine Kac--Moody algebra with central charge $c_j \in \N$.
\end{itemize}
Moreover, every projective positive energy representation is a direct sum of irreducible ones.

\paragraph{Noncompact manifolds.}
For \emph{noncompact} manifolds $M$, the situation is somewhat more intricate.
Here $S$ is a union of countably many $\R$-orbits $S_j$,
each of which is diffeomorphic to either $\R$ or $\bS^1$.
These two cases are considered separately in \S\ref{sec:9}.

In \S\ref{sec:PosenNoncptComponent} we consider the case where $S$ consists of countably many lines.
Since the bundle $\cK$ trivializes over every line, the gauge group $\Gamma_{c}(S, \cK)$ is a weak direct product of
countably many copies of $C^{\infty}_{c}(\R,K)$.
In order to arrive at a (partial) classification, we
impose the additional condition that the projective positive energy representation admit
a cyclic ground state vector $\Omega \in \cH$ that is unique up to a scalar.
In Theorem~\ref{Thm:NoncompactClassification} we show that
these \emph{vacuum representations}
are classified up to unitary equivalence by a nonzero central charge
$c_j \in \N$ for every connected component $S_j\simeq \R$.
The proof proceeds by reducing to the (important) special case $M=\R$, where the classification is
essentially due to Tanimoto~\cite{Ta11}.


In \S\ref{subsec:9.2} we consider the case where $S$ consists of infinitely many circles.
Here we impose the much less restrictive
condition that $\cH$ is a \emph{ground state representation}. This means that $\cH$ is generated by the space of ground states, but we do not require these ground states to be unique.
We show that under an (essentially geometric) \emph{spectral gap} condition, every positive energy
representation is automatically a ground state representation.
Since $\Gamma_{c}(S,\cK)$ admits projective positive energy
representations of Type II and III, it is necessary to consider factor representations.
If all orbits in $M$ are periodic, we show that
the minimal, factorial ground state representations of $\Gamma_{c}(M,\cK)$ are classified
up to unitary equivalence by 3 pieces of data. The first two are the same as in the case of compact manifolds:
\begin{itemize}
\item Countably many periodic orbits $S_{j}\subseteq M$, equipped with a central charge $c_j \in \N$. 
\item For every pair $(S_j,c_j)$ an anti-dominant integral weight $\lambda_j$ of the corresponding
affine Kac--Moody algebra with central charge $c_j$.
\end{itemize}
The integral weight $\lambda_j$ gives rise to a unitary lowest weight representation
$\cH_{\lambda_{j}}$ of the corresponding affine Kac--Moody algebra.
Using the ground state projections $P_j$, we consider
the collection of finite tensor products of the compact operators
$\cK(\cH_{\lambda_j})$ as a directed system of $C^*$-algebras. Its direct limit
\begin{equation}\cB = \bigotimes_{j}\cK(\cH_{\lambda_j})
\end{equation}
has a distinguished ground state projection $P_{\infty} = \bigotimes_{j} P_j$.
The third datum needed to characterize a minimal factorial ground state representation is:
\begin{itemize}
\item  A factorial representations of $\cB$ that is generated by fixed points of the projection
$P_\infty$.
\end{itemize}
Since $P_{\infty}\cB P_{\infty}$ is a UHF $C^*$-algebra, this provides a rich supply of
representations of type II and III, in marked contrast with the compact case.

\subsection{Connection to the existing literature}


\paragraph{Abelian structure groups.}
If the structure Lie algebra $\fk$ is merely assumed to be reductive,
then it decomposes as a direct sum $\fk = \fz \oplus \fk'$, where
$\fz$ is abelian and the commutator algebra
$\fk$ is semisimple. Since this decomposition is invariant under
all automorphism, we obtain a corresponding decomposition
on the level of Lie algebra bundles
$\fK \cong \fZ \oplus \fK'$. Accordingly, the Lie algebra
$\g = \Gamma_c(M,\fK)$ decomposes as a direct sum
$\fz \oplus \fg'$ and this decomposition is orthogonal with
respect to any $2$-cocycle because $\g'$ is perfect.
Therefore the classification of the positive energy representations
basically reduces to the cases where $\fk$ is semisimple
and where $\fk$ is one-dimensional abelian. We refer to
Solecki's paper \cite{So14} for some interesting results concerning
groups of maps with values in the circle group, and to \cite{Sh17}
for related results pertaining to defects in conformal field theory.
G.~Segal's paper \cite{SeG81} contains a number of interesting
results on projective positive energy representations
of loop groups with values in a torus.

\paragraph{Integrating representations of infinite dimensional Lie groups.}
The technique to integrate representations of infinite dimensional
Lie algebras to groups by first verifying suitable
estimates has already been used by R.~Goodman and N.~Wallach in \cite{GW84}
to construct the irreducible unitary positive energy representations
of loop groups and diffeomorphism groups. Their technique has later been refined by
V.~Toledano--Laredo \cite{TL99} to larger classes of infinite
dimensional Lie algebras. Other results on integrating Lie algebra
representations can be found in \cite{JN19}.

\paragraph{Non-commutative distributions.}
In \cite{A-T93} an irreducible unitary representation of
$\cG_{c} = \Gamma_c(M,\Ad(\Xi))$ is called a
{\it non-commutative distribution}. 
In view of the Borchers--Arveson Theorem \cite{BR02}, an irreducible projective
positive energy representation of $\cG_{c} \rtimes_\alpha \R$
remains irreducible when restricted to $\cG_{c}$. In this sense we
contribute to the program outlined in \cite{A-T93}
by classifying those non-commutative distributions for
$M$ compact and $K$ compact semisimple for which
extensions to positive energy representations exist.
The problem to classify all smooth projective irreducible
unitary representations of
gauge groups is still wide open, although the classification of their central extensions by Janssens and Wockel (\cite{JW13}) is a key step towards this goal.
As our treatment in \S\ref{PEcocycles} and \cite{JN17} shows,
it can be used  to determine in specific situations which cocycles actually arise.

\paragraph{Derivative representations.}
One of the first references concerning unitary representations
of groups of smooth maps such as $C^\infty(\R,\SU(2,\C))$ is
\cite{GG68}, where the authors introduce the concept of a
{\it derivative representation} 
which depends only on the
derivatives up to some order $N$ in some point $t_0 \in \R$.
Unitary representations on continuous tensor products were developed in
\cite{S69,PS72,VGG73,VGG74,PS76}.
In addition to these representations,
there exist irreducible representations of mapping
groups defined most naturally on groups of Sobolev $H^1$-maps, the so-called
energy representations
(cf.\ \cite{Is76, Is96}, \cite{AH78}, \cite{A-T93}, \cite{AT94}, \cite{An10}, \cite{ADGV16}).


\paragraph{The case where $M$ is a torus.}
In \cite{To87} (see also \cite[\S5.4]{A-T93}) Torresani studies
projective unitary ``highest weight representations'' of
$C^\infty(\T^d,\fk)$, where $\fk$ is compact simple.
Besides the finite tensor products of so-called evaluation
representations
({\it elementary representations}) he
finds finite tensor products of evaluation
representations of $C^\infty(\T^d,\fk)
\cong C^\infty(\T^{d-1}, C^\infty(\T,\fk))$, where the representations
 of the target algebra $C^\infty(\T,\fk)$ are projective
highest weight representations
({\it semi-elementary representations}).
Our results in \S\ref{sec:8} reduce to this picture in the special case
of the circle action on a torus.

\paragraph{Norm continuous representations.}
In a previous paper \cite{JN15}, we considered the
related problem to classify norm continuous unitary representations
of the connected groups $\Gamma_c(M,\cK)_0$.
In this case the problem also reduces to the case where
$\fk$ is compact semisimple
and the representations are linear rather than projective.
For every irreducible representation $\rho$,
there exists an embedded $0$-dimensional submanifold $S$, i.e., a
locally finite subset, $S \subeq M$
such that $\rho$ factors through the restriction
map $\Gamma_c(M,\fK) \to \Gamma_c(S,\fK) \cong \fk^{(S)}$.
If $M$ is compact, it follows that $\rho$ is a
finite tensor product of irreducible representations obtained by composing
an irreducible representation of $\fk$ with the evaluation in a point
$s \in S$. In particular, it is finite-dimensional.
If $M$ is non-compact, then the bounded representation
theory of the LF-Lie algebra  $\Gamma_c(M,\fK)$
is ``wild'' in the sense that there exist
bounded factor representations of type II and III.
The main result in \cite{JN15}
is a complete reduction of the classification of bounded irreducible
representations to the classification of irreducible representations
of UHF $C^*$-algebras.

\paragraph{Type III representations from noncompact orbits.}
For noncompact $M$, a different source of representations comes from
the group $C^{\infty}_{c}(\R,K)$ corresponding to a single noncompact connected component of $S$.
Here representations of Type $\mathrm{III}_{1}$ were constructed in \cite{We06,FH05}.
Other results in this context have recently been obtained
in \cite{dVIT20}, where solitonic representations of conformal nets on
the circle are constructed from non-smooth diffeomorphisms. These
in turn provide positive energy representations of $C^\infty_c(\R,K)
\cong C^\infty_c(\T \setminus \{-1\},K)$ which do
not extend to loop group representations (\cite[Thm.~3.4, \S 4.2]{dVIT20}).
In particular, irreducible representations of this type are obtained.


\paragraph{Acknowledgment.}
B.J.~is supported by the NWO grant 639.032.734
``Cohomology and representation theory of infinite dimensional Lie groups''.
He would like to thank the FAU Erlangen-N\"urnberg, the Universiteit Utrecht, and the Max Planck Institute for Mathematics
in Bonn for their hospitality during various periods in which the work was carried out.
K.-H.N.~acknowledges support from
the Centre Interfacultaire Bernoulli (CIB)
and the  NSF (National Science Foundation)
for a research visit at the EPFL and support by DFG-grant NE 413/10-1.
We are most grateful to Yoh Tanimoto, Helge Gl\"ockner, Milan Niestijl,
Lukas Miaskiwskyi and Tobias Diez for illuminating
discussions on this subject and for pointing out several references.

\section{Projective representations of Lie groups}\label{SectionPPER}

In this section, we introduce Lie groups modeled on locally convex vector spaces,
or \emph{locally convex Lie groups} for short. This is a
generalization of the concept of a finite dimensional Lie group
that captures a wide range of interesting examples (cf.\ \cite{Ne06} for an overview),
including gauge groups, our main object of study.
We then summarize the central results from \cite{JN19}, which
allow us to go back and forth between smooth projective unitary
representations of a locally convex Lie group $G$ and
smooth unitary representations of a central Lie group extension $\widehat{G}$ of $G$.
On the identity component $G_0$, these are characterized by representations
of the corresponding Lie algebra $\hg$.


\subsection{Locally convex Lie groups}

Let $E$ and $F$ be locally convex spaces, $U
\subeq E$ open and $f \: U \to F$ a map. Then the
{\it derivative   of $f$ at $x$ in the direction $h$} is defined as
$$ \partial_{h}f(x) := \lim_{t \to 0} \frac{1}{t}(f(x+th) -f(x)) $$
whenever it exists. We set $D f(x)(h) :=  \partial_{h}f(x)$. The function $f$ is called
{\it differentiable at
  $x$} if $D f(x)(h)$ exists for all $h \in E$. It is called {\it
  continuously differentiable} if it is differentiable at all
points of $U$ and
$$ D f \: U \times E \to F, \quad (x,h) \mapsto D f(x)(h) $$
is a continuous map. Note that this implies that the maps
$D f(x)$ are linear (cf.\ \cite[Lemma~2.2.14]{GN}).
The map $f$ is called a {\it $C^k$-map},
$k \in \N \cup \{\infty\}$,
if it is continuous, the iterated directional derivatives
$$ D^{j}f(x)(h_1,\ldots, h_j)
:= (\partial_{h_j} \cdots \partial_{h_1}f)(x) $$
exist for all integers $1\leq j \leq k$, $x \in U$ and $h_1,\ldots, h_j \in E$,
and all maps $D^j f \: U \times E^j \to F$ are continuous.
As usual, $C^\infty$-maps are called {\it smooth}. \index{smooth!function \vulop}

Once the concept of a smooth function
between open subsets of locally convex spaces is established, it is clear how to define
a locally convex smooth manifold (cf.\ \cite{Ne06}, \cite{GN}).
\begin{Definition}
A {\it locally convex Lie group} \index{Lie group!locally convex \vulop}
$G$ is a group equipped with a
smooth manifold structure modeled on a locally convex space
for which the group multiplication and the
inversion are smooth maps.
Morphisms of locally convex Lie groups are smooth group homomorphisms.
\end{Definition}

We write $\one \in G$ for the identity element.
The Lie algebra $\g$ of $G$ is identified with
the tangent space $T_{\one}(G)$, and the Lie bracket is obtained by identification with the
Lie algebra of left invariant vector fields. It is a
\emph{locally convex Lie algebra}
in the following sense.

\begin{Definition}
A {\it locally convex Lie algebra} is a \index{Lie algebra!10@locally convex \vulop}
locally convex vector space $\fg$ with a continuous Lie bracket
$[\,\cdot\,,\,\cdot\,] \colon \fg \times \fg \rightarrow \fg$.
Morphisms of locally convex Lie algebras are continuous Lie algebra
homomorphisms.
\end{Definition}

\begin{Definition}\label{def:localexp}
A smooth map $\exp \: \g \to G$  is called an
{\it exponential function} \index{exponential function \scheiding $\exp$}
if each curve $\gamma_x(t) := \exp(tx)$ is a one-parameter group
with $\gamma_x'(0)= x$.
A Lie group $G$ is said to be
{\it locally exponential} \index{Lie group!locally exponential \vulop}
if it has an exponential function for which there is an open $0$-neighborhood
$U$ in $\g$ mapped diffeomorphically by $\exp$ onto an
open subset of $G$.
\end{Definition}



\subsection{Smooth representations}

Let $G$ be a locally convex Lie group with Lie algebra $\fg$ and
exponential function $\exp \: \g \to G$.
In the context of Lie theory, it is natural to study \emph{smooth}
(projective) unitary representations
on a complex Hilbert space \index{Hilbert space \scheiding $\cH$}$\cH$.
We take the scalar product on $\cH$ to be linear in the \emph{second} argument, and denote
the group of unitary operators
\index{unitary group \scheiding $\U(\cH)$}
by $\U(\cH)$.

\subsubsection{Unitary representations}

A \emph{unitary representation}\index{representation!10@unitary \scheiding $(\rho,\cH)$} $(\rho,\cH)$ of $G$
is a Hilbert space $\cH$ with a group homomorphism $\rho\colon G \rightarrow \U(\cH)$.
A \emph{unitary equivalence} between $(\rho,\cH)$ and $(\rho',\cH')$ is a
unitary transformation $U \colon \cH \rightarrow \cH'$ such that $U \circ \rho(g) = \rho(g)' \circ U$
for all $g\in G$.

\begin{Definition} (Continuous unitary representations)
A unitary representation $(\rho,\cH)$ is called \emph{continuous} if
for all $\psi \in \cH$, the \index{representation!30@continuous \vulop}
orbit map
$G \rightarrow \cH \colon g \mapsto \rho(g)\psi$ is continuous
(see \cite{Ne14b} for more details).
\end{Definition}

\begin{Definition}(Smooth unitary representations) \label{def:smooth-proj-rep}
We call $\psi \in \cH$ a {\it smooth vector} \index{smooth!vectors \scheiding $\cH^{\infty}$}
if the orbit map $g \mapsto \rho(g)\psi$ is smooth, and
write $\cH^\infty \subeq \cH$ for the subspace of smooth vectors.
We say that $\rho$ is \emph{smooth} if \index{representation!40@smooth \vulop}
$\cH^{\infty}$ is dense in $\cH$
(see \cite{JN19} for more details).

\end{Definition}

Every smooth representation is continuous.
A representation $\rho$ is called {\it bounded} \index{representation!bounded \vulop}
if $\rho \: G \to \U(\cH)$
is continuous with respect to the norm topology on $\U(\cH)$.
Boundedness implies continuity, but many
interesting continuous representations, including the (smooth!) positive energy representations that are the main focus of this paper, are unbounded.
For some recent results on the automatic smoothness of unbounded
unitary representations satisfying certain spectral
conditions such as semiboundedness (Definition~\ref{def:semibounded}),
we refer to \cite{Ze17}.

\begin{Definition} (Derived representation)
For any smooth unitary representation $(\rho,\cH)$, the
\emph{derived representation} \index{representation!25@derived \scheiding $(\dd\rho, \cH^{\infty})$}
$\dd\rho \: \fg \to \End(\cH^\infty)$
of the Lie algebra $\fg$ is defined by
\[
\dd\rho(\xi)\psi := \derat0 \rho(\exp t\xi)\psi.\]
\end{Definition}


\begin{Remark} {\rm(Selfadjoint generators)}
The closure of any
operator $\dd\rho(\xi)$ coincides with the infinitesimal generator of the
unitary one-parameter group $\rho(\exp t\xi)$. In particular, the operators
$i\cdot \dd\rho(\xi)$ are essentially selfadjoint by Stone's Theorem
(cf.\ \cite[\S VIII.4]{RS75}).
\end{Remark}


\subsubsection{Projective unitary representations}

Let $\cH$ be a Hilbert space.
The projective Hilbert space is denoted by
\index{projective Hilbert space \scheiding $\bP(\cH)$}$\bP(\cH)$, and
its elements are denoted
$[\psi] = \C \psi$ for nonzero $\psi \in \cH$.
We denote the projective unitary group by
\index{projective unitary group \scheiding $\PU(\cH)$}$\PU(\cH) := \U(\cH)/\T \one$, and write $\oline U$ for the image of $U \in \U(\cH)$ in $\PU(\cH)$.

A \emph{projective unitary representation} \index{representation!20@projective unitary \scheiding $(\ol{\rho}, \cH)$}
$(\ol\rho,\cH)$ of a locally convex Lie group $G$
is a complex Hilbert space $\cH$ with
a group homomorphism $\ol{\rho}\colon G \rightarrow \PU(\cH)$.
A \emph{unitary equivalence} between $(\ol\rho,\cH)$ and $(\ol\rho',\cH')$ is a
unitary transformation $U \colon \cH \rightarrow \cH'$ such that $\ol U \circ \ol\rho(g) = \ol\rho(g)' \circ \ol U$
for all $g\in G$.


A projective unitary representation yields
an action of $G$ on $\bP(\cH)$.
Since $\bP(\cH)$ is a Hilbert manifold, we can use this to define continuous and smooth projective representations.

\begin{Definition}
(Continuous projective unitary representations)
A projective unitary representation $(\ol\rho,\cH)$ is called
\emph{continuous} if \index{representation!50@continuous projective \vulop}
for all $[\psi] \in \bP(\cH)$, the
orbit map
$G \rightarrow \bP(\cH) \colon g \mapsto \ol\rho(g)[\psi]$ is continuous.
\end{Definition}

\begin{Definition}(Smooth projective unitary representations) A
ray $[\psi] \in \bP(\cH)$ is called \emph{smooth} if its orbit map
$g \mapsto \overline{\rho}(g)[\psi]$ is smooth, and we denote the set of
smooth rays by $\bP(\cH)^{\infty}$. \index{representation!60@smooth projective \vulop}
A projective unitary representation $(\ol\rho,\cH)$
is called \emph{smooth} if $\bP(\cH)^{\infty}$ is dense in $\cH$.
\end{Definition}


\subsection{Central extensions}

In this paper, we are primarily interested in smooth projective unitary representations
$\ol \rho \colon G \rightarrow \PU(\cH)$
of a locally convex Lie group $G$.
We call $\ol \rho$
\emph{linear} if it comes from
a smooth unitary representation $\rho \colon G \rightarrow \U(\cH)$.

Although not every smooth
projective unitary representation of $G$ is linear,
it
can always be viewed as a smooth linear
representation of a \emph{central extension} of $G$ by the
circle group\index{circle group \scheiding $\T$}
 $\T \cong \R/2\pi \Z$.

\begin{Definition} (Central group extensions)\label{GroepUitbreiding}
\index{central extension!10@of Lie groups \scheiding $G^{\sharp}$}
A \emph{central extension} of $G$ by $\T$
is an exact sequence
\[1 \rightarrow \T \rightarrow G^{\sharp} \rightarrow G \rightarrow 1\]
of locally convex Lie groups (the arrows are smooth group homomorphisms) such that the image of $\T$ is central in $G^{\sharp}$
and $G^{\sharp} \rightarrow G$ is a locally trivial principal $\T$-bundle.
An \emph{isomorphism} $\Phi \colon G^{\sharp} \rightarrow G^{\sharp}{}'$ of
central $\T$-extensions is an isomorphism of locally convex Lie groups
that induces the identity maps on $G $ and $\T$.
\end{Definition}

For a smooth projective unitary representation $(\ol\rho,\cH)$ of $G$,
the group
\begin{equation}
  \label{eq:gsharp}
 G^\sharp := \{ (g,U) \in G \times \U(\cH) \: \oline\rho(g) = \oline U \}
\end{equation}
is a central Lie group extension of $G$ by $\T$ (\cite[Thm.~4.3]{JN19}).
Its smooth unitary
representation
\[\rho \: G^\sharp \to \U(\cH), \quad (g,U) \mapsto U\]
reduces to $z \mapsto z\one$ on $\T$ and induces $\overline{\rho}$ on $G$.
Since the restriction of $\rho$ to the identity component $G^{\sharp}_{0}$ is determined by
the derived Lie algebra representation (\cite[Prop.~3.4]{JN19}),
it is worth while to take a closer look at
central extensions of locally convex Lie algebras.

\begin{Definition} {\rm(Central Lie algebra extensions)} \label{LAUitbreiding}
A \emph{central extension} of a locally convex Lie algebra $\fg$ by $\R$ is
an exact sequence \index{central extension!20@of Lie algebras \scheiding $\fg^{\sharp}$}
\begin{equation}\label{LARijtje}
0 \rightarrow \R \rightarrow \fg^{\sharp} \rightarrow \fg \rightarrow 0
\end{equation}
of locally convex Lie algebras (the arrows are continuous Lie algebra homomorphisms) such that the image of
$\R$ is central in $\fg^{\sharp}$.
An \emph{isomorphism} $\phi \colon \fg^{\sharp} \rightarrow \fg^{\sharp}{}'$ of
central extensions is an isomorphism of locally convex Lie algebras
that induces the identity maps on $\fg$ and $\R$.
\end{Definition}
The group extensions 
of Definition~\ref{GroepUitbreiding} give rise to
Lie algebra extensions in the sense of Definition~\ref{LAUitbreiding}.
In order to classify the latter, we introduce continuous Lie algebra cohomology.

\begin{Definition} \label{def:cohom}
\index{cohomology \scheiding $H^n(\fg, \R)$}
The \emph{continuous Lie algebra
cohomology} space $H^n(\fg,\R)$ of a locally convex Lie algebra $\fg$
is the cohomology of the complex
$C^{\bullet}(\fg,\R)$, where $C^n(\fg,\R)$ consists of the
continuous alternating linear
maps $\g^n \rightarrow \R$
with differential $\delta \colon C^{n}(\fg,\R) \rightarrow C^{n+1}(\fg,\R)$ defined by
\[\delta \omega(\xi_0,\ldots,\xi_{n}):= \sum_{0\leq i<j\leq n}
(-1)^{i+j} \omega([\xi_i,\xi_j],\xi_1,\ldots,\widehat{\xi}_i, \ldots,
\widehat{\xi}_j, \ldots, \xi_n)\,.\]
\end{Definition}

The second Lie algebra cohomology  $H^2(\fg,\R)$ classifies central extensions
up to isomorphism. The 2-cocycle $\omega \: \g^2 \rightarrow \R$
gives rise to the Lie algebra $\fg^{\sharp}_{\omega} :=  \R \oplus_{\omega} \g$
with the Lie bracket
\[
[(z,\xi), (z',\xi')] := \big(\omega(\xi,\xi'), [\xi,\xi']\big).
\]
Equipped
with the obvious maps
$\R \rightarrow \fg^{\sharp}_{\omega} \rightarrow \fg$, this
is a central extension of~$\fg$.
Every central extension is isomorphic to one of this form, and
two central extensions are equivalent if and only if the corresponding
cohomology classes $[\omega] \in H^2(\g,\R)$ coincide (\cite[Prop.~6.3]{JN19}).
%

The following theorem collects some of the main results
of our previous paper (\cite[Cor.~4.5, Thm.~7.3]{JN19}).
It allows us to go back and forth between smooth projective unitary
representations of $G$, smooth unitary representations of a central extension $G^{\sharp}$ of $G$,
and the corresponding representations of its Lie algebra~$\fg^{\sharp}$.

\begin{Theorem}{\rm(Projective $G$-representations and linear
$\fg^{\sharp}$ -representations)} \label{ThmProjRepLARep}
\begin{itemize}
\item[\rm(A)]Every smooth projective unitary representation $(\overline{\rho}, \cH)$
of $G$ gives rise to a central
extension $\T \rightarrow G^{\sharp} \rightarrow G$
of locally convex Lie groups, and a
smooth unitary representation $(\rho,\cH)$ of $G^{\sharp}$.
In turn, this gives rise to
the central extension $\R \rightarrow \fg^{\sharp} \rightarrow \fg$
of locally convex Lie algebras
and
the derived representation $\dd\rho \colon \fg^{\sharp} \rightarrow
\mathrm{End}(\cH^{\infty})$
of $\fg^{\sharp}$ by essentially skew-adjoint operators.

\item[\rm(B)]If $G$ is connected, then $(\ol\rho, \cH)$ and $(\ol\rho', \cH')$
are unitarily equivalent if and only if
the derived Lie algebra representations
$(\dd\rho, \cH^\infty)$
and $(\dd\rho', \cH'^\infty)$ are unitarily equivalent.
This means that there exists an isomorphism
$\phi \colon \fg^{\sharp} \rightarrow \fg^{\sharp}{}'$ of central extensions
and a unitary isomorphism $U \colon \cH \rightarrow \cH'$ such that
$U \cH^{\infty} \subseteq \cH^{\infty}{}'$ and
$\dd\rho'(\phi(\xi)) \circ U = U \circ \dd\rho(\xi)$ for
all $\xi \in \fg^{\sharp}$.
\end{itemize}
\end{Theorem}



\subsection{Integration of projective representations}
\label{subsec:7.1}

In this subsection we discuss the integrability of (projective) unitary
representations of Banach--Lie algebras, based on the existence of analytic
vectors. Here our main result is the Integrability Theorem
for projective representations of Banach--Lie groups
(Theorem~\ref{thm:7.4}) that we derive with the methods from~\cite{Ne11}.

\begin{Definition} Let $(\rho,\cD)$ be a representation of the topological
Lie algebra $\g$ on the pre-Hilbert space $\cD$. We say that
\begin{itemize}
\item[\rm(i)] $\rho$ is a {\it $*$-representation} \index{representation!27@$*$-representation \scheiding $(\rho, \cD)$}
if all operators
$\rho(x)$, $x \in\g$, are skew-symmetric.
\item[\rm(ii)] $\rho$ is {\it strongly continuous}
\index{representation!35@strongly continuous \vulop} if all the maps
$\g \to \cD, x \mapsto \rho(x)\xi$ are continuous.
\item[\rm(iii)] $\xi \in \cD$ is an {\it analytic vector}  \index{analytic!vectors \scheiding $\cD^{\omega}$}
if there exists a $0$-neigh\-bor\-hood $U \subeq \g$ such that
 $\sum_{n = 0}^\infty \frac{\|\rho(x)^n \xi\|}{n!} < \infty$
for every $x \in U$. The analytic vectors form a linear subspace~$\cD^\omega \subeq \cD$.
\end{itemize}
\end{Definition}

\begin{Remark} If $\g$ is a Banach--Lie algebra, then \cite[Prop.~4.10]{Ne11} implies
that $\xi \in \cD$ is an analytic vector if and only if it is an analytic
vector for all operators $\rho(x)$, $x \in\g$
in the sense that there exists an $s > 0$ such that
$\sum_{n = 0}^\infty \frac{s^n\|\rho(x)^n \xi\|}{n!} < \infty$.
\end{Remark}

We shall need the following lemma that is not spelled out explicitly in \cite{Ne11}:

\begin{Lemma} \label{lem:ana-invar}
Let $(\rho,\cD)$ be a strongly continuous $*$-representation of the
Banach--Lie algebra $\g$. Then
$\cD^\omega$ is a $\rho(\g)$-invariant subspace.
\end{Lemma}

\begin{proof} Following \cite[Def.~3.2]{Ne11}, we call a linear functional
$\beta \: U(\g) \to\C$ on the enveloping algebra of $\g$
{\it an analytic functional} \index{analytic!functional \vulop}
if all $n$-linear maps
\[ {\g^n \to \C,\quad (x_1, \ldots, x_n) \mapsto \beta(x_1 \cdots x_n)}\]
 are continuous and the series
$\sum_{n = 0}^\infty \frac{\beta(x^n)}{n!}$ converges for every $x$ in a $0$-neigh\-bor\-hood of $\g$.
According to \cite[Prop.~6.3]{Ne11}, a vector $\xi \in \cD$ is analytic
if and only if the functional $\beta_\xi(D) := \la \xi,\rho(D)\xi\ra$
is analytic, where $\rho \: U(\g) \to \End(\cD)$ denotes the
extension of $\rho$ to the enveloping algebra. For
$\xi \in \cD^\omega$ and $x \in \g$, the functional
\[  \beta_{\rho(x)\xi}(D) := \la \rho(x)\xi, \rho(D)\rho(x)\xi \ra
= \beta_\xi((-x)Dx) \]
is also analytic by \cite[Thm.~3.6]{Ne11}, so that
$\rho(x)\xi \in \cD^\omega$ by \cite[Prop.~6.3]{Ne11}. This shows that
$\rho(\g)\cD^\omega \subeq \cD^\omega$.
\end{proof}

To formulate the Integrability Theorem for projective representations,
we first give a precise definition of a projective $*$-representation of a
topological Lie algebra $\g$.

\begin{Definition} Suppose that $\omega \: \g^2 \to \R$ is a continuous $2$-cocycle
and that $\g^\sharp_\omega = \R \oplus_\omega \g$ is the corresponding central extension.
Then any $*$-representation
$(\rho^\sharp,\cD)$ with $\rho^\sharp(1,0) = i\one$ leads to a linear map
\[ \rho \: \g \to \End(\cD), \quad \rho(x) := \rho^\sharp(0,x) \]
satisfying
\[ [\rho(x), \rho(y)] = \rho([x,y]) + \omega(x,y)i\one.\]
We then call $(\rho,\cD)$ a
{\it projective $*$-representation with cocycle $\omega$}.
\end{Definition}
\index{representation!110@projective *-representation \vulop}

\begin{Theorem} \label{thm:7.4} {\rm(Integrability Theorem for projective representations)}
Let $G$ be a $1$-connected Banach--Lie group with Lie algebra $\g$, and let
$(\rho, \cD)$ be a projective strongly continuous
$*$-repre\-sen\-tation of $\g$ on the dense subspace $\cD$ of the Hilbert space~$\cH$.
If $\cD$ contains a dense subspace of analytic vectors, then there exists a
smooth projective unitary representation $\pi \: G \to \PU(\cH)$
on $\cH$ with the property
that $\pi(\exp x) = q(e^{\oline{\rho(x)}})$ for $x \in \g$,
where $q \: \U(\cH) \to \PU(\cH)$ denotes the quotient map.
\end{Theorem}

\begin{proof} We proceed as in the proof of \cite[Thm.~6.8]{Ne11}.
Using Lemma~\ref{lem:ana-invar}, we see that we may w.l.o.g.\ assume that
$\cD = \cD^\omega$, so that $\cD$ consists of analytic vectors.
According to Nelson's Theorem \cite{Nel59},
the operators $\rho(x)$, $x \in \g$, are essentially skew-adjoint,
so that their closures generate unitary one-parameter groups.
The same holds for the operators
$\hat\rho(t,x)$, $(t,x) \in \g^\sharp$.
This leads to a map
\[ \tilde\pi \: \g^\sharp \to \U(\cH),\quad x \mapsto e^{\oline{\hat\rho(t,x)}}
= e^{it} e^{\oline{\rho(x)}}.\]
From the proof of \cite[Thm.~6.8]{Ne11}, we immediately derive that
\begin{equation}
  \label{eq:mult2}
\tilde\pi((t,x)*(s,y)) = \tilde\pi(t,x)\tilde\pi(s,y)
\end{equation}
holds for $(t,x), (s,y)$ in some open $0$-neighborhood $U^\sharp \subeq \g^\sharp$.
This implies that
\begin{equation}
  \label{eq:mult3}
q(e^{\oline{\rho(x*y)}}) = q(e^{\oline{\rho(x)}}) q(e^{\oline{\rho(y)}})
\end{equation}
for $x,y$ in some open $0$-neighborhood $U \subeq \g$.
Now \cite[Ch.~3, \S6, Lemma 1.1]{Bou89} implies the existence of a
unique homomorphism $\pi \: G \to \PU(\cH)$ such that
$\pi(\exp x) = q(e^{\oline{\rho(x)}})$ holds for all elements $x$ in
some $0$-neighborhood of $\g$.

That $\pi$ is a smooth projective representation
(Definition~\ref{def:smooth-proj-rep})
follows from the analyticity of the orbit
maps $G \to \bP(\cH), g \mapsto \pi(g)[v]$ for $v \in \cD^\omega$, which
in turn follows from $\pi(\exp x)[v] = [e^{\oline{\rho(x)}}v]$.
\end{proof}

\subsection{Double extensions}\label{Sec:doubleextensions}

Suppose that $G$ is a locally convex Lie group and $\alpha \: \R \to \Aut(G)$ a homomorphism
defining a smooth $\R$-action on $G$.
Then the semidirect product $G \rtimes_{\alpha}\R$ is a Lie group
with Lie algebra $\fg \rtimes_{D} \R$, where
$D \in \der(\g)$ is the infinitesimal generator of the $\R$-action on
$\g$ induced by~$\alpha$.

If $\oline\rho \: G \rtimes_\alpha \R \to \PU(\cH)$
is a smooth projective unitary representation,
then Theorem~\ref{ThmProjRepLARep} yields a central extension
\[\T \rightarrow \hat G := (G\rtimes_{\alpha}\R)^{\sharp} \rightarrow G\rtimes_{\alpha}\R \]
with a smooth unitary representation $\rho$ of $\hat G$
on $\cH$ that induces~$\oline\rho$.
From Theorem~\ref{ThmProjRepLARep}, we see that the restriction of $\ol{\rho}$
to $(G \rtimes_{\alpha}\R)_{0}$
is determined up to unitary equivalence by the
derived representation $\dd\rho$
of the central extension $\hat\g= (\fg \rtimes_{D}\R)^{\sharp}$.
We identify $\g$ with the linear subspace ${\{0\} \times\g \times \{0\}}$
of $\widehat{\fg}$.
We write
\[ C := (1,0,0) \quad \mbox{ and } \quad D := (0,0,1) \]
for the central element and
derivation in $\R \oplus_\omega (\g\rtimes_{D}\R)$ respectively, so that
\begin{equation}
  \index{double extension \scheiding $\widehat{\fg}$}
  \label{eq:d-elt}
\hat\g =  \R C \oplus_\omega (\g \rtimes  \R D).
\end{equation}
We trust that using the same symbol for the derivation $D \in \mathrm{der}(\fg)$ and
the Lie algebra element $D \in \widehat{\fg}$ that implements it will
not lead to confusion. Note that
in the representation $\dd\rho$ of $\widehat{\fg}$, the central element $C$ acts
by~$i\one$. Writing $(z,x,t) = z C + x + t D$, the bracket in $\hat\g$ takes the form
\begin{align*}
& [z C + x + tD, z'C + x' + t'D] \\
&= \big(\omega(x,x') + t\omega(D,x') - t' \omega(D,x)\big)C
+ ([x,x']  + t Dx' - t' D x).
\end{align*}

\section{Positive energy representations}
\label{sec:3}

In this section we introduce positive energy representations
and some tools to handle them.
In \S\ref{subsec:3.1}, we give the precise definition
on both the linear and the projective level.
In \S\ref{sec:minimal}, we use the Borchers--Arveson Theorem to
reduce the classification of positive energy representations to the
so-called \emph{minimal} ones.
Finally,
in \S\ref{subsec:3.2} we describe the key tool of this paper
in a first general form: the Cauchy--Schwarz estimates for projective
positive energy representations. Here we will discuss them for general
groups, but they will be refined in the context
of gauge algebras in \S\ref{subsec:5.3} below.

\subsection{Positive energy representations}
\label{subsec:3.1}

Let $G$ be a locally convex Lie group with Lie algebra $\fg$ and let
$\alpha \colon \R \rightarrow \mathrm{Aut}(G)$ be a homomorphism defining
a smooth $\R$-action on $G$. Then it also induces a smooth action
$\alpha^\g$ on $\g$ and we write $D \in \mathrm{der}(\fg)$
for its infinitesimal generator
\begin{equation}\label{eq3.1DefinitieVanD}
Dx  := \derat0 \alpha^\g_t(x) \qquad \mbox{ for } \qquad x \in \g.
\end{equation}
In this section, we investigate smooth projective unitary representations of
$G$ that extend to projective  \emph{positive energy} representations
of $G \rtimes_\alpha \R$.

\begin{Definition} (Projective positive energy representations) \label{def:posenerdef}
A smooth, projective, unitary representation
$\oline\rho \: G \rtimes_\alpha \R \to \PU(\cH)$
is called a  \index{representation!85@projective positive energy \vulop}
\emph{positive energy representation} if one (hence any)
strongly continuous homomorphic lift $U \colon \R \rightarrow \U(\cH)$
of $\ol U \colon \R \rightarrow \PU(\cH), t \mapsto
\oline\rho(\one,t)$ has a generator
\[H := i \derat0 U_{t}\]
whose spectrum is bounded below. We then call $H$ a
\emph{Hamiltonian} \index{Hamiltonian!plain \scheiding $H$}
and note that $U_t = e^{-itH}$ holds in the sense of functional calculus.
\end{Definition}

\begin{Remark}
By adding a constant,
we can always choose
a Hamiltonian $H$ that satisfies
$\mathrm{Spec}(H) \subseteq [0,\infty)$.
\end{Remark}

We have seen in \S\ref{Sec:doubleextensions}
that every smooth  projective unitary representation
$\oline\rho$ of $G \rtimes_\alpha \R$
gives rise to a smooth linear representation $(\rho,\cH)$ of a
locally convex Lie group
$\widehat{G} = (G\rtimes_{\alpha}\R)^{\sharp}$,
a central $\T$-extension of
$G \rtimes_{\alpha} \R$ with Lie algebra
\[ \widehat{\fg} = \R \oplus_{\omega} (\fg \rtimes_{D} \R)
= \R C \oplus_\omega (\g \rtimes \R D) \]
as in \eqref{eq:d-elt}.

\begin{Definition} (Linear positive energy representations) \label{Posendef}
Let $\rho \colon \widehat{G} \rightarrow\U(\cH)$ be
a smooth unitary representation  of~$\hat{G}$.
Then $\rho$ gives rise to a derived representation $\dd \rho$ of $\hg$ on
the space $\cH^{\infty}$ of smooth vectors.
We call
\[ H:= i \dd \rho(D) \] the \emph{Hamiltonian} and
we say that $\rho$ is a \emph{positive energy representation}
if  \index{representation!80@positive energy \vulop}
\[ \dd\rho(C) = i\one \quad \mbox{ and if } \quad
\mathrm{Spec}(H) \subseteq [0,\infty).\]
\end{Definition}

\begin{Remark}\label{Rk:perfectcocycles} (a) If $\dd\rho(C) = i \one$ and $\mathrm{Spec}(H) \subseteq [E_{0},\infty)$ is
bounded below,
we can always replace $D$ by $D + E_0 C$ to obtain a positive Hamiltonian.
Note that this does not change the cocycle $\omega$ on $\g \rtimes_D \R$,
only the isomorphism
between $\widehat{\fg}$ and $(\fg \rtimes_{D}\R)^{\sharp}$.

(b) For a cocycle $\omega$ on $\g \rtimes_D \R$, the relation
\begin{equation}
  \label{eq:cocd}
\omega(D,[\xi,\eta]) = \omega(D\xi, \eta) + \omega(\xi,D\eta)
\end{equation}
shows that the linear functional $i_D \omega$ measures the non-invariance of
the restriction of $\omega$ to $\g \times \g$ under the derivation~$D$.
It also shows that if the Lie algebra $\g$ is perfect, then the linear functional
$i_D\omega \: \g \to \R$ is completely determined by 
the restriction of $\omega$ to $\fg \times \fg$.
\end{Remark}

\subsection{Equivariant positive energy representations}\label{sec:eqposener}

We will also need an equivariant version
of positive energy representations.
Let $P$ be a Lie group with Lie algebra $\fp$ and let
$\alpha \colon P \rightarrow \mathrm{Aut}(G)$ be a homomorphism defining
a smooth $P$-action on $G$.

\begin{Definition} (Equivariant projective positive energy representations) \label{def:posenerdefsym}
A~smooth, projective, unitary representation
$\oline\rho \: G \rtimes_\alpha P \to \PU(\cH)$
is called a \emph{positive energy representation with respect to $p \in\fp$}
if the projective representation
\[\ol{\rho}_{p} \colon G \rtimes_{\alpha \circ \exp_{p}} \R \rightarrow \PU(\cH)\]
defined by
$
\ol{\rho}_{p}(g,t) := \ol{\rho}(g,\exp(p t))
$
is of positive energy in the sense of Definition~\ref{def:posenerdef}.
The \emph{positive energy cone} \index{positive energy cone \scheiding $\cone$}
$\cone \subeq \fp$ is the set of all elements $p\in \fp$ for which $\ol{\rho}$ is a positive energy representation.
\end{Definition}

Note that $\cone$ is an $\Ad_{P}$-invariant cone.
In particular, the representation $\ol{\rho}$ is of positive energy with respect to
$p\in \fp$
if and only if it is of positive energy for all elements in
the cone generated by
the adjoint orbit $\Ad_{P}(p)\subseteq \fp$ of $p$.

The homomorphism $\alpha \colon P \rightarrow \mathrm{Aut}(G)$ can be twisted by
an inner automorphism $\mathrm{Ad}_{g_0}$,
yielding $\alpha' = \mathrm{Ad}_{g_0}\alpha \mathrm{Ad}_{g_0}^{-1}$.
Essentially, these inner twists do not affect the class
of equivariant projective positive energy representations.
\begin{Proposition}\label{Prop:twistthelift}
Let $(\ol\rho,\cH)$ be an equivariant projective positive energy representation of $G\rtimes_{\alpha} P$, and let
$\ol{U}_0 := \ol\rho(g_0)$. Then
\[\ol{\rho}{}'(g,p) := \ol{U}_0 \ol{\rho}(\mathrm{Ad}_{g_0}^{-1}(g),p )\ol{U}^{-1}_0\] is an
equivariant projective positive energy representation of $G\rtimes_{\alpha'} P$ with the same restriction to $G$, and with the same positive energy cone
$\cone \subseteq \fp$.
\end{Proposition}
\begin{proof}
To see that $\ol{\rho}{}'$ is a projective representation of $G\rtimes_{\alpha'} P$, one checks that
the following is a commutative diagram of group homomorphisms:
\begin{center}
$ $
\xymatrix{
G\rtimes_{\alpha} P \ar[r]^{\ol{\rho}}\ar[d]_{ (\mathrm{Ad}_{g_0}, \mathrm{Id}_{P})} &
\mathrm{PU}(\cH)\ar[d]^{\mathrm{Ad}_{U_0}}\\
G\rtimes_{\alpha'} \ar[r]^{\ol{\rho}'}P & \mathrm{PU}(\cH).
}
\end{center}
For the positive energy condition, note that any lift $t \mapsto V_t$ of $t \mapsto \ol{\rho}(\exp(tp))$ yields a lift
$t \mapsto U_0 V_t U_{0}^{-1}$ of $t \mapsto \ol{\rho}'(\exp(tp))$ whose generator has the same spectrum.
\end{proof}

\subsection{Minimal representations}\label{sec:minimal}

The following refinement of the Borchers--Arveson Theorem \cite{BR02}
will be used in the proof of Corollary~\ref{cor:borch} below.

\begin{Theorem} \label{thm:BAthm}
Let $\cH$ be a Hilbert space and let $\cM \subeq B(\cH)$ be a von Neumann
algebra. Further, let $(U_t)_{t \in \R}$ be a strongly continuous unitary one-parameter group on $\cH$ for which $\cM$ is invariant under conjugation with the operators $U_t$, so that
we obtain a one-parameter group $\alpha \: \R \to \Aut(\cM)$ by
$\alpha_t(M) := \Ad(U_t)M := U_t M U_t^*$ for $M \in \cM$.
If $U_t = e^{-itH}$ with $H \geq 0$, then the following assertions hold:
\begin{itemize}
\item[\rm(i)] There exists a strongly continuous unitary one-parameter group
$(V_t)_{t \in \R}$ in $\cM$ with $\Ad(V_t) = \alpha_t$ and
$V_t = e^{-itH_0}$ with $H_0 \geq 0$.
It is uniquely determined by requirement that it is minimal in the sense that, for
any other one-parameter group $(V_{t}')_{t\in \R}$ with these properties,
the central one-parameter group $V_t' V_{-t} = e^{-itZ}$ in $\cM$ satisfies $Z \geq 0$.
\item[\rm(ii)] If $V_T = \one$ and $\cF \subeq \cH$ is an $\cM$-invariant subspace,
then the subspace $\cF_0 := \{ \xi \in \cF \: H_0 \xi = 0\}$ is $\cM$-generating in $\cF$.
\item[\rm(iii)] If $\alpha_T = \id_\cM$, then $V_T =~\one$.
\end{itemize}
\end{Theorem}

\begin{proof} (i) This is the Borchers--Arveson Theorem
(\cite[Thm.~II.4.6]{Bo96}; see also \cite[Thm.~3.2.46]{BR02} and \cite{BGN20} for a
detailed discussion).

(ii) If $V_{T} = \one$, then $\Spec(H_0) \subseteq \frac{2\pi}{T} \Z$. In particular
$H_0$ is diagonalizable. If $\cF_0$ is not $\cM$-generating in $\cF$,
then $\cE := (\cM\cF_0)^\bot\cap \cF$ is a non-zero $\cM$-invariant subspace of $\cF$
with $\inf\Spec(H_0\res_{\cE}) \geq \frac{2\pi}{T}$.
As $\cH_0 := \ker H_0 \subeq \cE^\bot$, we also have $\cM\cH_0 \subeq \cE^\bot$. Since
$\cM\cH_0$ is invariant under $\cM$ and $\cM'$,
the orthogonal projection $Z$ onto  $\cH_1 := (\cM\cH_0)^\bot$ is central in $\cM$.
On this subspace we have $\inf\Spec(H_0\res_{\cH_1}) \geq \frac{2\pi}{T}$,
so that $H := H_0 - Z \frac{2\pi}{T} \geq 0$, contradicting minimality.

(iii) If $\alpha_{T} = \id_{\cM}$, then $V_{T}$ is 
contained in the center
$\cZ(\cM) = \cM \cap \cM'$ of $\cM$. As $\cZ(\cM)$ is a direct sum of
$L^\infty$-algebras, there exists a non-negative
$Z \geq 0$ in $\cZ(\cM)$ with $\Spec(Z) \subeq [0,\frac{2\pi}{T}]$ and
$V_T = e^{iTZ}$. Now $V_t' := e^{-it(H_0+Z)} = V_t e^{-itZ}$ also has a non-negative
generator $H_1 := H_0 + Z$ and satisfies $V'_T = \one$. In particular,
$\Spec(H_1) \subeq \frac{2\pi}{T}\Z$.

We claim that the minimality of $V$ implies that, for every $\eps > 0$, the central support
of the spectral projection $P := P^{H_0}[0,\eps]$ of $H_0$ in $\cM$ equals $\one$.
To see this, note that the central support $Q$ of $P$ is the orthogonal
projection onto the closed subspace generated by $\cM P\cH$. If this subspace
is proper, then the restriction $H_1$ of $H_0$ to $\cH_1 := (\one-Q)\cH$ satisfies
$H_1 \geq \eps\one$, so that $H' := H_0 - \eps(\one - Q) \geq 0$. The minimality of $H_0$
now yields $\one = Q$.

We now show that $\Spec(Z) \subeq \{0, \frac{2\pi}{T}\}$,
which implies that 
\[ V_T = V_T' e^{-iTZ}= V_T' = \one.\] 
Assume that this is not the case. Then there exists a non-zero spectral value
$0 < a < \frac{2\pi}{T}$ of $Z$. Let $\eps > 0$ be such that
$0 < a - 2\eps < a + 2\eps < \frac{2\pi}{T}$ and consider the spectral
projection $Q := P^Z([a-\eps,a+\eps])$ for $Z$, which is contained in
$\cZ(\cM)$.
Since the central support of $P^{H_0}[0,\eps]$ is $\one$, we have
$Q P^{H_0}([0,\eps]) \not=0$, so that
$\Spec(QH_0) \cap [0,\eps] \neq \emptyset$.
Since $\Spec(QZ) \subeq [a-\eps,a+\eps]$, this leads to
\[ \Spec((H_0 + Z)Q)\cap [a-\eps, a+ 2 \eps] \not= \emptyset.\]
This contradicts $\Spec((H_0 + Z)Q) = \Spec(H_1 Q) \subeq \Spec(H_1) \subeq \frac{2\pi}{T}\Z$.
\end{proof}

Using the Borchers--Arveson Theorem, every smooth positive energy representation
$(\rho, \cH)$ can be brought in the the following standard form.

\begin{Definition} (Minimal representations) \label{def:mini} A positive energy
representation $(\rho, \cH)$ of $\hat G$ is called
{\it minimal} if the 1-parameter group \index{representation!minimal \vulop}
$U_{t} = \rho(\exp(tD))$
is minimal with respect to  the von Neumann algebra $\rho(\widehat{G})''$.
\end{Definition}

\begin{Corollary} \label{cor:borch}   
Let $(\rho,\cH)$ be a positive energy representation  of
$\hat G$ and let $G^\sharp \subeq \hat G$ be the inverse image of the subgroup
$G$ of $G \rtimes_\alpha \R$, so that $\hat G \cong G^\sharp \rtimes \R$.
Then there exists a unitary $1$-parameter group $(W_t)_{t \in \R}$
in the commutant $\rho(\hat G)'$ such that
$\rho_0(g,t) := \rho(g,t) W_t^{-1}$ has the following properties:
\begin{itemize}
\item[\rm(i)] $\rho_0(\hat G)'' = \rho(G^\sharp)''$.
\item[\rm(ii)] If $\rho$ is  irreducible, then so is $\rho\res_{G^\sharp}$.
\item[\rm(iii)] If $\alpha_T = \id_G$, then $\rho_0(\one,T) = \one$
and, for every closed $\rho(\hat G)$-invariant subspace $\cF \subeq \cH$,
the subspace $\cF_0 := \{ \xi \in \cF \: H_0 \xi = 0\}$ is $\hat G$-generating in $\cF$.
\item[\rm(iv)] $\rho_0$ is a smooth positive energy representation.
\end{itemize}
\end{Corollary}

\begin{proof} (i) Theorem~\ref{thm:BAthm} implies that 
${U_t := \rho(\exp tD)}$ can be written as
${U_t = V_t W_t}$, where $(V_t)_{t \in \R}$ is a continuous unitary one-parameter group
in the von Neumann algebra $\cM := \rho(G^\sharp)''$ and $W_t \in \rho(G^\sharp)'$.

(ii) If $\rho$ is irreducible, then Schur's Lemma implies that $W_t \in \T\one$,
hence that the restriction $\rho\res_{G^\sharp}$ remains irreducible.

(iii) follows from Theorem~\ref{thm:BAthm}(iii) and (ii).

(iv) As $V_t = \rho_0(\one,t)$ has a positive generator,
$\rho_0$ also is a positive energy representation.
It remains to see that $\rho_0$ is smooth. Since $(W_t)_{t\in \R}$ lies in the
commutant $\rho(\hat G)'$, all its spectral subspaces are invariant under $\hat G$.
Therefore $\rho$ is a direct sum of subrepresentations for which $W$ is norm continuous.
We may therefore assume w.l.o.g.\ that $W$ is norm continuous.
Then we can consider $W$ as a smooth representation of
$\hat G$ and therefore $\rho_0(g,t) = \rho(g,t) W_{-t}$
is a smooth representation of $\hat G$.
\end{proof}


In view of the factorization $\rho(g,t) = \rho_0(g,t) W_t$,
we can adopt the point of view that we know all positive
energy representations if we know the minimal ones. On the level
of the irreducible representations, the only difference is a phase factor corresponding
to the minimal energy level. In general, the ambiguity consists in
unitary one-parameter groups of the commutant, and these can be classified
in terms of spectral measures.

\subsection{Cauchy--Schwarz estimates (general case)}
\label{subsec:3.2}

We show that the requirement that a representation
be of positive energy severely restricts the class of cocycles that may occur.

Let $\rho$ be a positive energy representation of $\widehat{G}$.
For a smooth unit vector $\psi \in \cH^{\infty}$ the expectation values
\[ \langle H \rangle_{\psi} := \langle \psi , H \psi \rangle
\quad \mbox{ and }\quad \langle i\dd\rho(\xi) \rangle_{\psi} :=
\langle \psi , i\dd\rho(\xi) \psi \rangle \]
of $H$ and  $\xi \in \fg$ are defined.
The following is a non-commutative adaptation of \cite[Thm.~2.8]{NZ13}.

\begin{Lemma}{\rm(Cauchy--Schwarz estimate)}\label{Lem:CSRaw}
Let $\rho$ be a positive energy representation of $\widehat{G}$.
Let $\xi \in \fg$ be such that
$[\xi,D\xi] = 0$. Then, for every unit vector
$\psi \in \cH^\infty$, we have
\begin{equation}\label{CSRaw}
\big(\langle  i\dd\rho(D\xi)\rangle_{\psi} + \omega(\xi,D)\big)^2 \leq
2 \omega(\xi,D\xi) \langle H \rangle_{\psi},
\end{equation}
and further $\omega(\xi,D\xi) \geq 0$.
\end{Lemma}

\begin{proof}
Since $H = i\dd\rho(D)$ has non-negative spectrum, the expectation value of the
energy in the state defined by $\exp(t\dd\rho(\xi))\psi$ is non-negative for all $t\in \R$;
\begin{equation}\label{guichelheil}
0 \leq \langle H \rangle_{\exp(t\dd\rho(\xi))\psi} =
\langle e^{-t \ad_{\dd\rho(\xi)}}H \rangle_{\psi}\,.
\end{equation}
Since $[\xi,D\xi]= 0$, the exponential series terminates at order $2$,
\begin{eqnarray}\label{paardebloem}
\exp(-t \,\ad_{\dd\rho(\xi)})(H)
&=&
i \dd\rho(e^{-t\ad_\xi}D)\notag \\
&=& i \dd\rho\big(D + t D \xi - t \omega(\xi, D)C - {\textstyle\frac{t^2}{2}} \omega(\xi,D\xi)C\big) \notag\\
&=&H + t\big(i\dd \rho(D\xi) + \omega(\xi,D)\big) + {\textstyle\frac{t^2}{2}}
\omega(\xi,D\xi),
\end{eqnarray}
so that substitution in (\ref{guichelheil}) yields the inequality
\[
0 \leq \langle H \rangle_{\psi} + t\big(\langle i\dd \rho(D\xi)\rangle_{\psi} +
\omega(\xi,D)\big) + {\textstyle\frac{t^2}{2}} \omega(\xi,D\xi)
\quad \mbox{ for } \quad t\in\R.\]
The proposition now follows from the simple observation that
$a t^2 + b t + c \geq 0$ for all $t \in \R$ is equivalent to
$0 \leq a, c$ and $b^2 \leq 4ac$.
\end{proof}

The Cauchy--Schwarz estimate will play an important role in
the rest of the paper. We will use it mainly in situations where $\omega(D,\fg) = \{0\}$,
so that the bilinear form $(\xi,\eta) \mapsto \omega(\xi, D\eta)$ is symmetric.
This is the case for gauge algebras (cf.\ Remark~\ref{omooieconnectie}),
but also more generally for locally convex Lie algebras with an admissible derivation
in the sense of
\cite[Def.~9.1, Prop.~9.10]{JN19}.

In \S\ref{PEcocycles} we use 
Lemma~\ref{Lem:CSRaw} to show that
$(\xi,\eta) \mapsto \omega(\xi, D\eta)$ is a positive
semidefinite form on the gauge algebra $\fg$, and that
%
%
%
%
every cocycle coming from a positive energy representation
can be represented by a \emph{measure} (Theorem~\ref{MeasureThm}).
In \S\ref{sec:6}, we make extensive use of
the bound on the expectation value
$\langle  i\dd\rho(D\xi)\rangle_{\psi}$
in terms of the average energy $\langle H \rangle_{\psi}$ afforded by Lemma~\ref{Lem:CSRaw}.
In fact, we shall need such bounds also for Lie
algebra elements which are not in the image of~$D$.
The following refinement of the Cauchy--Schwarz estimate was designed for this purpose.

We start out with a proposition on Lie algebras which are \emph{Mackey complete},
in the sense that\index{Mackey completeness \vulop}
every smooth curve $\zeta \: [0,1] \to \g$ has a weak
integral $\int_0^1 \zeta(t)\, dt$ in $\g$.
For a Mackey complete Lie algebra $\fg$, the operator
$
\int_0^1 e^{s \ad_y}\, ds
$
on $\g$ is denoted
$\frac{e^{\ad_y} - \one}{\ad_y}$.

\begin{Proposition}\label{prop:centadjoint}
Let $\hat\g = \R \oplus_\omega \g$ be a central extension of a Mackey complete
Lie algebra $\g$ of the Lie group $G$ with exponential function $\exp$. Then
the adjoint action $\Ad^{\hat\g}$ of $G$ on $\hat\g$ satisfies
\[ \Ad^{\hat\g}_{\exp y}(z,x) = \Big(z + \omega\Big(y, 
\frac{e^{\ad_y} - \one}{\ad_y}(x)
\Big),
e^{\ad_y} x\Big).\]
\end{Proposition}

\begin{proof} This is verified by solving the ODE
\[ \gamma'(t) = [(0,y), \gamma(t)] \quad \mbox{ with } \quad
\gamma(0) = (z,x).\]
Writing $\gamma(t) = (\alpha(t),e^{t \ad_y}(x))$, it leads to
$\alpha'(t) = \omega(y, e^{t \ad_y}x)$.
\end{proof}


\begin{Lemma}\label{4point} {\rm(Refined Cauchy--Schwarz estimate)}
Let $\fg$ be a Mackey complete Lie algebra, and let
$\rho$ be a positive energy representation of $\widehat{G}$.
Let $\xi,\eta \in \fg$ be such that
$[\xi,D\xi] = 0$ and $[\eta,D\eta] = 0$.
Then, for all $s\in \R$, we have
\begin{eqnarray*}
&&\left(
\left\langle i \dd\rho\big(e^{-s\ad_\eta}D\xi\big)\right\rangle_{\psi}
+\omega(\xi,D)
+\omega\left(\frac{e^{-s\ad_{\eta}} -\mathbf{1}}
{\ad_{\eta}}(D\xi) , \eta \right)
\right)^2 \\
& \leq &
2\omega(\xi,D\xi)
\Big(
\langle H\rangle_{\psi}
+ s\big(
\langle i\dd\rho(D\eta)\rangle_{\psi}
+ \omega(\eta, D)\big) + {\textstyle\frac{s^2}{2}} \omega(\eta,D\eta)
\Big)\,.
\end{eqnarray*}
In particular, if $\omega(\xi,D) = 0$, $\omega(\eta,D) = 0$
and $\omega(\ad^{n}_{\dd\rho(\eta)}(D\xi),\eta) = 0$ for all
$n\geq 0$, then
\[
\left\langle i \dd\rho\big(e^{-s\ad_\eta}D\xi\big)\right\rangle^{2}_{\psi}
\leq 2\omega(\xi,D\xi)
\Big(
\langle H\rangle_{\psi}
+ s
\langle i\dd\rho(D\eta)\rangle_{\psi}
+  {\textstyle\frac{s^2}{2}}\omega(\eta, D\eta)\Big)\,.
\]
\end{Lemma}

\begin{proof}
We write $W_{s,t} := \exp(t\dd\rho(\xi))\exp(s\dd\rho(\eta))$,
and exploit the fact that $H_{s,t} := W_{s,t}^{*} H W_{s,t}$ has non-negative spectrum.
Repeated use of (\ref{paardebloem}) on
\[H_{s,t} = \exp(-s\ad_{\dd\rho(\eta)})\left(
\exp(-t\ad_{\dd\rho(\xi)})(H)\right)\]
yields
\[
H_{s,t} = A_0(s) + A_1(s)t + A_2 t^2
\]
with
\begin{eqnarray*}
A_0(s) &=& H + s(i\dd\rho(D\eta) + \omega(\eta,D)\mathbf{1})
+ {\textstyle\frac{s^2}{2}} \omega(\eta, D\eta)\mathbf{1}\,,\\
A_1(s)& = & \omega(\xi,D)\mathbf{1} +
\exp(-s\ad_{\dd\rho(\eta)})(i\dd\rho(D\xi))\,,\\
A_2 & = &  {\textstyle\frac{1}{2}} \omega(\xi,D\xi)\mathbf{1}\,.
\end{eqnarray*}
With the preceding proposition, we obtain
for $\exp(-s\ad_{\dd\rho(\eta)})(i\dd\rho(D\xi))$ the expression
\[
i\dd\rho\big(e^{-s\ad_{\eta}}D\xi\big)
+\omega\left(\frac{e^{-s\ad_{\eta}} -\mathbf{1}}
{\ad_{\eta}}(D\xi), \eta\right)\one\,,
\]
and thus
\[
A_1(s) = \omega(\xi,D)\mathbf{1} +
i\dd\rho\big(e^{-s\ad_\eta}(D\xi)\big)
+\omega\left(\frac{e^{-s\ad_{\eta}} -\mathbf{1}}
{\ad_{\eta}}(D\xi),\eta\right)\mathbf{1}\,.
\]
Consider the expectation value $\langle H_{s,t} \rangle_{\psi}\geq 0$.
Setting
\[ \alpha_0(s) := \langle A_0(s) \rangle_{\psi},\quad
\alpha_1(s) := \langle A_1(s) \rangle_{\psi} \quad \mbox{ and } \quad
\alpha_2 := \langle A_2 \rangle_{\psi},\] we observe that
\[\langle H_{s,t} \rangle_{\psi} = \alpha_0(s) + \alpha_1(s)t + \alpha_2 t^2
\]
is a non-negative polynomial in $t$ of degree at most 2.
From this, we obtain the inequality
$\alpha_1(s)^2\leq 4 \alpha_2 \alpha_0(s)$.
This is the first inequality mentioned above, the second one
is a direct consequence.
\end{proof}


\section{Covariant extensions of gauge algebras}
\label{sec:4}

The results in the preceding section concerned the general
level of Lie groups of the form $G \rtimes_\alpha \R$. Now we turn
to the specifics of gauge groups.
After introducing
gauge groups and their Lie algebras in
\S\ref{SectionGGA}, we describe in
\S\ref{vanseminaarsimpel} a procedure that provides a
reduction from semisimple to simple structure Lie algebras,
at the expense of replacing $M$ by a finite covering manifold~$\hat M$.
In \S\ref{GySsCoc}, we recall the classification \cite{JN17}
of 2-cocycles for the extended gauge algebra $\fg \rtimes_{D} \R$.


\subsection{Gauge groups and gauge algebras }
\label{SectionGGA}

Let $\cK \rightarrow M$\index{bundle!10@of Lie groups \scheiding \cK} be a smooth bundle of Lie groups,
and let
$\fK \rightarrow M$ be the corresponding Lie algebra bundle\index{bundle!20@of Lie algebras \scheiding $\fK$}
with fibers $\fK_{x} = \mathrm{Lie}(\cK_{x})$.
If $M$ is connected, then the fibers $\cK_{x}$ of $\cK\rightarrow M$
are all isomorphic to a fixed structure group $K$,
and the fibers $\fK_{x}$ of $\fK$ are isomorphic to its Lie algebra $\fk = \mathrm{Lie}(K)$.

\begin{Definition}\label{def:gaugegroup} (Gauge group)
The \emph{gauge group}
\index{gauge group!20@compact support \scheiding $\Gamma_{c}(M, \cK)$}
is the group $\Gamma(M, \cK)$ of smooth sections
of $\cK \rightarrow M$, and
the \emph{compactly supported gauge group}
is the group $\Gamma_{c}(M,\cK)$ of smooth compactly supported sections.
\end{Definition}

\begin{Definition}(Gauge algebra)
The \emph{gauge algebra} \index{gauge algebra!10@smooth \scheiding $\Gamma(M, \fK)$}%
\index{gauge algebra!20@compact support \scheiding $\Gamma_{c}(M, \fK)$}
is the Fr\'echet--Lie algebra $\Gamma(M,\fK)$
of smooth sections of $\fK \rightarrow M$, equipped with the pointwise Lie bracket.
The \emph{compactly supported gauge algebra} $\Gamma_{c}(M,\fK)$
is the LF-Lie algebra of smooth compactly supported sections.
\end{Definition}

The compactly supported gauge group $\Gamma_{c}(M,\cK)$ is a locally convex Lie group,
whose Lie algebra is the compactly supported gauge algebra $\Gamma_{c}(M,\fK)$.
It is locally exponential, with
$\exp \colon \Gamma_{c}(M,\fK) \rightarrow \Gamma_{c}(M,\cK)$ given by
pointwise exponentiation \cite[Prop.~2.3]{JN17}.

\begin{Definition} In the following we write
$\tilde\Gamma_c(M,\cK)_0$ for the simply connected covering group of
the identity component $\Gamma_c(M,\cK)_0$
and
\[ q_\Gamma \:
\tilde\Gamma_c(M,\cK)_0 \to
\Gamma_c(M,\cK)_0 \]
for the covering map. Then $\tilde\Gamma_c(M,\cK)_0$ has
the same Lie algebra $\Gamma_c(M,\fK)$ as the gauge group $\Gamma_c(M,\cK)$, and its exponential
function $\Exp$ satisfies $q_\Gamma \circ \Exp = \exp$.
\end{Definition}

\subsubsection{Gauge groups from principal fiber bundles}

The motivating example of a gauge group is of course the group $\Gau(\Xi)$ of vertical automorphisms
of a principal $K$-bundle\index{bundle!40@principal \scheiding $\Xi$} $\pi \colon \Xi \rightarrow M$.

\begin{Definition}
A \emph{vertical automorphism} \index{vertical automorphism \vulop}
of a principal fiber bundle $\pi \colon \Xi \rightarrow M$
is a $K$-equivariant diffeomorphism ${\alpha \colon \Xi \rightarrow \Xi}$ such that
$\pi \circ \alpha = \pi$.
The group $\Gau(\Xi)$ of vertical automorphisms is called the
\emph{gauge group} of $\Xi$.
\end{Definition}

In order to interpret $\Gau(\Xi)$ as a gauge group in the sense of
Definition~\ref{def:gaugegroup}, define the bundle of groups
$\Ad(\Xi) \rightarrow M$\index{bundle!50@adjoint (groups) \scheiding $\Ad(\Xi)$} with typical fiber $K$ by
\[\Ad(\Xi) := {\Xi \times K/\sim}\,,\]
where the relation $\sim$ is given by
$(pk,h) \sim (p, khk^{-1})$ for $p \in \Xi$ and  $k, h \in K$.
We obtain an isomorphism
\[
\Gau(\Xi) \simeq \Gamma(M, \Ad(\Xi))
\]
by mapping the section $\sigma \in \Gamma(M, \Ad(\Xi))$ to the corresponding
vertical automorphism $\alpha_{\sigma} \in \mathrm{Gau}(\Xi)$, defined by
$
	\alpha_{\sigma}(p) = p \cdot k
$
if $\sigma(\pi(p))$ is the class of $(p,k)$ in $\Ad(\Xi) = \Xi\times K/\sim$.

The bundle of Lie algebras associated to $\Xi$ is
the \emph{adjoint bundle} ${\mathrm{ad}(\Xi) \rightarrow M}$,
defined as the quotient \index{bundle!60@adjoint (Lie algebra) \scheiding $\ad(\Xi)$}
\[\ad(\Xi) := {\Xi\times_{\Ad}\fk}\]
of $\Xi \times \fk$ modulo the relation $(pk,X) \sim (p, \mathrm{Ad}_{k}(X))$ for
$p \in \Xi$, $X \in \fk$ and $k \in K$.
Here $\mathrm{Ad}_{k} \in \mathrm{Aut}(\fk)$ is the Lie algebra automorphism
induced by the group automorphism $h \mapsto khk^{-1}$.

The \emph{compactly supported gauge group}
$\mathrm{Gau}_{c}(\Xi)\subseteq \Gau(\Xi)$
is the group of vertical bundle automorphisms of $\Xi$
that are trivial outside the preimage
of some compact subset of $M$.
Since it is isomorphic to $\Gamma_{c}(M,\Ad(\Xi))$, it is a locally convex Lie group
with Lie algebra $\gau_{c}(\Xi) = \Gamma_{c}(M,\ad(\Xi))$.


\begin{Remark}\label{Rk:Tussenijk}
In applications to gauge theory on non-compact manifolds $M$, the relevant group $\cG$ of gauge transformations
may be smaller than $\mathrm{Gau}(\Xi)$ due to boundary conditions at infinity.
One expects $\cG$ to contain at least
$\mathrm{Gau}_{c}(\Xi)$, or perhaps even some larger Lie group of gauge transformations
specified by a decay condition at infinity
(cf.\ \cite{Wa10, Go04}). In Part II of this series of papers, we will focus on the case where $M = \R^{d}$
is Minkowski space, and $\cG \subset \Gamma(\R^d, \Ad(\Xi))$ is the group of gauge transformations
that extend continuously to the conformal completion of Minkowski space. If the extension of $\Xi$ to the conformal completion is trivial, then $\cG$ contains global as well as compactly supported gauge transformations.
\end{Remark}

\subsubsection{Gauge groups and space-time symmetries}

An \emph{automorphism} of
$\pi \colon \cK \rightarrow M$ is a pair $(\gamma, \gamma_{M}) \in \Diff(\cK)\times \Diff(M)$
with $\pi \circ \gamma = \gamma_{M}\circ \pi$, such that for each fiber $\cK_{x}$, the map
$\gamma|_{\cK_{x}} \colon \cK_{x} \rightarrow \cK_{\gamma_{M}(x)}$ is a group homomorphism.
Since $\gamma_{M}$ is determined by $\gamma$, we will omit it from the notation.
We denote the group of automorphisms of $\cK$ 
by $\Aut(\cK)$. 

\begin{Definition}(Geometric $\R$-actions)\label{def:gemetrictype}
\index{geometric $\R$-action \scheiding $\alpha$}
In the context of gauge groups, we will be interested in $\R$-actions
$\alpha \colon \R \rightarrow \Aut(\Gamma(M,\cK))$ which are of
\emph{geometric} type,
i.e., derived from a $1$-parameter group
$\gamma \colon \R \rightarrow \mathrm{Aut}(\cK)$
by
\begin{equation}\label{eq:DefinitieVanAlphat}
\alpha_{t}(\sigma) := \gamma_{-t} \circ \sigma \circ \gamma_{M, t}.
\end{equation}
The $\R$-action on $\Gamma(M,\cK)$ preserves the subgroup
$\Gamma_c(M,\cK)_0$ on which it defines a smooth action. Moreover,
it lifts to a smooth action on the simply connected covering group
$\tilde\Gamma_c(M,\cK)_0$ (cf.\ \cite[Thm.~VI.3]{MN03}).
\end{Definition}

\begin{Remark}
If $\cK$ is of the form $\Ad(\Xi)$ for a principal fiber bundle $\Xi\rightarrow M$,
then a 1-parameter group of automorphisms of $\Xi$ induces a 1-parameter
group of automorphisms of $\cK$.
\end{Remark}

The 1-parameter group $\alpha \colon \R \rightarrow \Aut(\Gamma(M,\cK))$ of group automorphisms
differentiates to a 1-parameter group $\alpha^\g \colon \R \rightarrow \Aut(\Gamma(M,\fK))$
of Lie algebra automorphisms given by
\begin{equation}\label{LAaut1par}
\alpha^\g_{t}(\xi) = \frac{d}{d\eps}\Big|_{\eps = 0} \gamma_{-t}\circ e^{\varepsilon \xi}
\circ \gamma_{M,t}\,.
\end{equation}
The corresponding derivation $D := \frac{d}{dt}\big\vert_{t= 0} \alpha^\g_{t}$ of $\Gamma(M,\fK)$ can be described
in terms of the infinitesimal generator of $\gamma$,
\begin{equation}
  \label{eq:def-bv}
\bv := \frac{d}{dt}\Big\vert_{t = 0} \gamma_{-t} \in \cV(\cK).
\end{equation}
We identify $\xi \in \Gamma(M,\fK)$ with the vertical, fiberwise
left invariant vector field $\Xi_{\xi} \in \cV(\cK)$
defined by $\Xi_{\xi}(k_{x}) = \frac{d}{d\eps}\big|_{\eps = 0} k_{x}e^{\varepsilon \xi(x)}$.
Using the equality $[\bv, \Xi_{\xi}] = \Xi_{D(\xi)}$, we write\index{geometric derivation \scheiding $D$}
\begin{equation}
D (\xi) = L_{\bv}\xi\,.
\end{equation}
For $\fg = \Gamma_{c}(M,\fK)$, the Lie algebra $\fg \rtimes_{D}\R$ then has bracket
\begin{equation}\label{eq:smurfenliedje}
[\xi \oplus t, \xi' \oplus t'] = \Big( [\xi,\xi'] + (tL_{\bv}\xi' - t' L_{\bv}\xi)\Big) \oplus 0\,.
\end{equation}

\begin{Remark}\label{remarkalgebroids}
Alternatively, we can consider $\gamma \colon \R \rightarrow \Aut(\cK)$
as a smooth 1-parameter group of bisections of the gauge groupoid
$\mathcal{G}(\cK) \rightrightarrows M$,
the Lie groupoid whose objects are points $x,y\in M$, and whose morphisms
are Lie group isomorphisms $\cK_{x} \rightarrow \cK_{y}$.
%
It gives rise to a smooth 1-parameter family $\dot{\gamma}$ of bisections of the Lie groupoid
$\mathcal{G}(\fK) \rightrightarrows M$, whose morphisms from $x$ to $y$ are Lie algebra
isomorphisms $\fK_{x} \rightarrow \fK_{y}$.
Its generator
$\bv = -\frac{d}{dt}\big|_{t=0} \dot{\gamma}$
is thus a section of its Lie algebroid $\mathfrak{a}(\fK) \rightarrow M$,
called the \emph{Atiyah algebroid}.
A section $\xi \in \Gamma(M,\fK)$ can be considered as an element of
$\Gamma(M,\mathfrak{der}(\fK)) \subseteq \Gamma(M,\mathfrak{a}(\fK))$, and we
interpret $L_{\bv}\xi$ as the commutator $[\bv,\xi]$ in $\Gamma(M,\mathfrak{a}(\fK))$.
We will need this picture in \S\ref{vanseminaarsimpel}, where
the bundle of Lie groups is not available.
\end{Remark}

\subsection{Reduction to simple structure algebras}
\label{vanseminaarsimpel}

In this paper, we consider
gauge algebras with a \emph{semisimple} structure algebra~$\fk$.
The following theorem shows that, without further loss of generality,
we may restrict attention to the case where $\fk$ is \emph{simple}.

\begin{Theorem} \label{reductienaarsimpel}
{\rm(Reduction from semisimple to simple structure algebras)}
If $\fK \rightarrow M$ is a smooth locally trivial bundle of Lie algebras
with semisimple fibers,
then there exists a
finite cover $\widehat{M} \rightarrow M$ and a smooth locally trivial bundle of Lie
algebras $\widehat{\fK} \rightarrow \widehat{M}$ with simple fibers
such that there exist isomorphisms
$\Gamma(M,\fK) \simeq \Gamma(\widehat{M},\widehat{\fK})$ and
$\Gamma_{c}(M,\fK) \simeq \Gamma_{c}(\widehat{M},\widehat{\fK})$
of locally convex Lie algebras.
\end{Theorem}
This is proven in \cite[Thm.\ 3.1]{JN17}. In brief, one uses local trivializations
of $\fK \rightarrow M$ to give a manifold structure to
\[
\widehat{M} := \bigcup_{x\in M} \mathrm{Spec}(\fK_{x})\,,
\]
where $\mathrm{Spec}(\fK_{x})$ is the finite
set  of maximal ideals $I_x  \subset \fK_{x}$.
The bundle of Lie algebras 
is then defined by\index{bundle!30@of simple Lie algebras \scheiding $\widehat{\fK}$}
\[
\widehat{\fK} := \bigcup_{I_x\in \widehat{M}} \fK_{x}/I_{x}\,,
\]
and one shows that the natural projection $\pi \colon \widehat{\fK} \rightarrow \widehat{M}$
is a locally trivial vector bundle.
Note that the finite cover $\widehat{M} \rightarrow M$ is not necessarily connected, and
that the isomorphism classes of the fibers of $\widehat{\fK} \rightarrow \widehat{M}$ are not
necessarily the same over different connected components of $\widehat{M}$.

\begin{Remark}\label{remark:overdederivatie}
Since a smooth 1-parameter family of automorphisms of $\fK \rightarrow M$ acts naturally
on the maximal ideals, we obtain a smooth action on the Lie algebra bundle
$\widehat{\fK} \rightarrow \widehat{M}$.
We denote  the corresponding section of the Atiyah algebroid
$\mathfrak{a}(\hfK) \rightarrow \widehat{M}$
by $\widehat{\bv} \in \Gamma(\widehat{M},\mathfrak{a}(\hfK))$, and we denote the corresponding
vector field on $\widehat{M}$
by $\bv_{\hat M} := \pi_*\widehat{\bv}$.
Since $\widehat{\fK}$ has simple fibers, the Atiyah algebroid $\mathfrak{a}(\hfK)$ fits in the exact sequence
\begin{equation}
\hfK \rightarrow \mathfrak{a}(\hfK) \rightarrow T\widehat{M}\,,
\end{equation}
where the first map is given by the pointwise adjoint action, and the second by the anchor.
Note that
the action on $\widehat{M}$ is locally free or periodic if and only if the action on $M$ is.
In that case, the period on $\widehat{M}$ is a multiple of the period on $M$.

\end{Remark}

In many situations, the connected components of $\widehat{M}$ are diffeomorphic to~$M$.
However, nontrivial covers $\widehat{M} \rightarrow M$ do occur naturally, for example in
connection to non-orientable 4-manifolds.

\begin{Example}\label{ex:simple}
If the fibers of $\fK \rightarrow M$ are simple, then $\widehat{M} = M$.
\end{Example}
\begin{Example}\label{ex:trivial}
If $\fK = M \times \fk$ is trivial, then $\hat M = M \times \Spec(\fk)$ and all connected
components of $\hat M$ are diffeomorphic to~$M$.
\end{Example}

\begin{Example}\label{ex:different}
Suppose that $M$ is connected, and that the typical fiber $\fk$ of $\fK \rightarrow M$
is a semisimple Lie algebra with $r$ simple ideals
that are mutually non-isomorphic. Then
$\widehat{M} = \bigsqcup_{i=1}^{r} M$ is a disjoint union of copies of~$M$.
\end{Example}

\begin{Example}(Frame bundles of 4-manifolds)\label{Ex:MhatOrientable} 
Let $M$ be a 4-dimensional Riemannian manifold.
Let $\Xi := {\rm OF}(M)$ be the principal $\OO(4,\R)$-bundle of orthogonal frames, and let
$\fK = \mathrm{ad}(\Xi)$.
Then $K = \OO(4,\R)$ and $\fk = \mathfrak{so}(4,\R)$ is isomorphic to
$\mathfrak{su}_L(2,\C) \oplus \mathfrak{su}_R(2,\C)$.
The group $\pi_0(K)$ is of order 2, the non-trivial element acting by
conjugation with $T = \mathrm{diag}(-1,1,1,1)$.
Since this permutes the two simple ideals,
the manifold $\widehat{M}$ is the orientable double cover of $M$. This is
the disjoint union $\widehat{M} = M_{L} \sqcup M_{R}$ of two copies of $M$
if $M$ is orientable, and a connected twofold cover $\widehat{M} \rightarrow M$ if it is not.
\end{Example}

\subsection{Central extensions of gauge algebras}
\label{GySsCoc}

Let $\fg$ be
the compactly supported gauge algebra $\Gamma_{c}(M,\fK)$, where $\fK \rightarrow M$ is
a Lie algebra bundle with simple fibers.
In this section, we classify all possible central extensions of
$\fg \rtimes_{D} \R$. This amounts to
calculating the continuous second Lie algebra
cohomology $H^2(\fg \rtimes_{D} \R,\R)$ with trivial coefficients.
In \S\ref{PEcocycles}, we will characterize those cocycles coming
from a positive energy representation.

\subsubsection{Universal invariant symmetric bilinear forms}\label{invbil}

Let $\fk$ be a finite dimensional, simple real Lie algebra.
Then its automorphism group $\Aut(\fk)$ is a closed subgroup
of $\mathrm{GL}(\fk)$, hence a Lie group with Lie algebra $\der(\fk) \simeq \fk$.
Since $\fk$ acts trivially on the space\index{co-invariants \scheiding $V(\fk)$}
\[V(\fk) := S^2(\fk)/{(\fk \cdot S^2(\fk))}\,\]
of $\fk$-coinvariants of the twofold symmetric tensor power $S^2(\fk)$,
the $\Aut(\fk)$-representation on $V(\fk)$ factors through $\pi_0(\Aut(\fk))$.
The \emph{universal $\fk$-invariant symmetric bilinear form}
is defined by \index{invariant bilinear form!universal \scheiding $\kappa$}
\[ \kappa \colon \fk \times \fk \rightarrow V(\fk), \quad
\kappa(x,y) := [x\otimes_{s}y] = \frac{1}{2}[x \otimes y + y \otimes x].\]
We associate to
$\lambda \in V(\fk)^*$
the $\R$-valued, $\der(\fk)$-invariant, symmetric, bilinear form
$\kappa_{\lambda} := \lambda \circ \kappa$.
This correspondence
is a bijection between $V(\fk)^*$ and the space of $\der(\fk)$-invariant symmetric bilinear forms
on $\fk$.

Since $\fk$ is simple, we have $V(\fk) \simeq \C$ if $\fk$
admits a complex structure, and $V(\fk) \simeq \R$ if it does not
(cf.~\cite[App.~B]{NW09}).
In the latter case, $\fk$ is called
{\it absolutely simple}. \index{Lie algebra!20@absolutely simple \vulop}
The universal invariant symmetric bilinear form can be identified with
the Killing form of the real Lie algebra $\fk$ if $V(\fk) \simeq \R$
and with the Killing form of the underlying complex Lie algebra if $V(\fk) \simeq \C$.
In particular, in the important special case that $\fk$ is a compact simple Lie algebra,
a universal invariant bilinear form $\kappa \colon \fk \times \fk \rightarrow V(\fk)$
is the negative definite Killing form given by
$\tr(\ad x \ad y)$. However, in the following,
we shall always use the normalized invariant positive definite
symmetric bilinear form $\kappa$ that satisfies
\begin{equation}
  \label{eq:kappa-normal}
\kappa(i\alpha^\vee, i\alpha^\vee)=2
\end{equation}
for the coroots $\alpha^\vee$
corresponding to long roots in the root decomposition of $\fk_\C$
(cf.\ \cite{PS86, Ne01} and Appendix~\ref{app:1}).
\index{invariant bilinear form!normalized \scheiding $\kappa$}

\subsubsection{The flat bundle $\bV = V(\fK)$} \label{flatbdl}

If $\fK \rightarrow M$ is a bundle of Lie algebras with simple fibers, then we denote by
\index{bundle!70@of co-invariants \scheiding $\bV$}$\bV \rightarrow M$
the vector bundle with fibers $\bV_{x} = V(\fK_{x})$.
It carries a canonical flat
connection $\dd$, defined by 
\[ \dd\kappa(\xi,\eta) := \kappa(d_{\nabla}\xi,\eta) + \kappa(\xi,d_{\nabla}\eta)
\quad \mbox{ for } \quad \xi, \eta \in \Gamma(M,\fK),\]
where $\nabla$ is a \emph{Lie connection} \index{Lie connection \scheiding $\nabla$}
on $\fK$, meaning that
$d_{\nabla} [\xi,\eta] = [d_{\nabla} \xi , \eta] + [\xi, d_{\nabla} \eta]$ for all $\xi,\eta \in \Gamma(M,\fK)$.
Since the fibers are assumed to be simple,
any two Lie connections differ by a $\fK$-valued 1-form, so that the preceding
definition is independent of the choice of $\nabla$ (cf.\ \cite{JW13}).


Let $\fk_i$ be the fiber of $\fK$ over a connected component $M_i$ of $M$.
If $\fk_i$ is absolutely simple (hence in particular when $\fk$ is compact),
we have $V(\fk_i) \simeq \R$, and $\pi_0(\mathrm{Aut}(\fk))$ acts trivially on $V(\fk_i)$.
In this case, $\bV \rightarrow M_i$ is simply the trivial line bundle $M_i \times \R \rightarrow M_i$.

If $\fk_i$ possesses a complex structure, then $V(\fk_i) \simeq \C$, and
$\alpha \in \Aut(\fk_i)$ flips
the complex structure on $\C$ if and only if it flips the complex structure on $\fk$.
In this case, $\bV \rightarrow M_i$ is a vector bundle of real rank 2.

\begin{Remark}
In the context of positive energy representations,
we will see in Theorem~\ref{red2cpt} below
that $\fk$ must be compact, so that $\bV \rightarrow M$ is the trivial real line bundle.
Although we need to consider the \emph{a priori} possibility of nontrivial bundles, then,
it will become clear in the course of our analysis that they will not
give rise to positive energy representations.
\end{Remark}

\subsubsection{Classification of central extensions}\label{subsec:2cocyc}
We define 2-cocycles\index{cocycle \scheiding $\omega_{\lambda,\nabla}$} $\omega_{\lambda,\!\nabla}$ on $\fg \rtimes_{D} \R$ whose
classes span the cohomology group $H^2(\fg \rtimes_{D} \R,\R)$.
They depend on a \emph{$\bV$-valued
$1$-current}\index{current \scheiding $\lambda$} $\lambda \in \Omega^1_{c}(M,\bV)'$, and on a
\emph{Lie connection} $\nabla$ on $\fK$.
A $1$-current $\lambda \in \Omega^1_{c}(M,\bV)'$ is said to be
\begin{itemize}
\item[\rm(L1)]  {\it closed}
if $\lambda(\dd C^\infty_c(M,\bV)) = 0$, and
\item[\rm(L2)]  {\it $\bv_M$-invariant} if
$\lambda(L_{\bv_M}\Omega^1_{c}(M,\bV)) = \{0\}$.
\end{itemize}
Given a closed, $\bv_M$-invariant current $\lambda \in \Omega^1_{c}(M,\bV)'$,
we define the
2-cocycle $\omega_{\lambda,\!\nabla}$ on $\g \rtimes_{D}\R$
by skew-symmetry and the equations
\begin{eqnarray}
\omega_{\lambda,\!\nabla}(\xi,\eta) &=& \lambda(\kappa(\xi, d_{\nabla}\eta)),\label{cdef1}\\
\omega_{\lambda,\!\nabla}(D,\xi) &=& \lambda(\kappa(L_{\bv}\nabla,\xi))\,, \label{cdef2}
\end{eqnarray}
where we write $\xi$ for $(\xi,0) \in \fg \rtimes_{D} \R$ and $D$
for $(0,1) \in \fg \rtimes_{D} \R$ as in \eqref{eq:d-elt}.
We define the $\der(\fK)$-valued 1-form
$L_{\bv}\nabla \in \Omega^1(M,\der(\fK))$ by
\begin{equation} \label{eq:defcurv}
(L_{\bv}\nabla)_{w}(\xi) = L_{\bv} (d_{\nabla}\xi)_{w} - \nabla_{w}L_{\bv}\xi =
L_{\bv}(\nabla_{w}\xi) - \nabla_{w}L_{\bv}\xi - \nabla_{[\bv_M,w]} \xi
\end{equation}
for all $w \in \cV(M)$, $\xi \in \Gamma(M,\fK)$. Since the fibers of $\fK \rightarrow M$ are simple,
all derivations are inner, so we can identify $L_{\bv}\nabla$
with an element of $\Omega^1(M,\fK)$.
Using the formul\ae{}
\begin{eqnarray}
\dd\kappa(\xi,\eta) &=& \kappa(d_{\nabla} \xi, \eta) + \kappa(\xi, d_{\nabla} \eta),\label{fijneformule1} \\
L_{\bv_M}\kappa(\xi,\eta) &=&
\kappa(L_\bv\xi,\eta) + \kappa(\xi, L_\bv\eta),\label{fijneformule2}\\
{}L_\bv(d_{\nabla} \xi) - d_{\nabla} L_\bv \xi &=& [L_{\bv}\nabla,\xi],\label{fijneformule3}
\end{eqnarray}
it is not difficult to check that $\omega_{\lambda,\!\nabla}$ is a cocycle.
Skew-symmetry follows from (\ref{fijneformule1}) and (L1).
The vanishing of
$\delta\omega_{\lambda,\!\nabla}$ on $\fg$ follows from
(\ref{fijneformule1}), the derivation property of $\nabla$ and
invariance of $\kappa$. Finally, $i_{D}\delta\omega_{\lambda,\!\nabla} = 0$
follows from skew-symmetry, (\ref{fijneformule3}),
(\ref{fijneformule2}), (L2) and the invariance of $\kappa$.

Note that the class $[\omega_{\lambda,\!\nabla}]$ in $H^2(\fg\rtimes_{D}\R,\R)$
depends only on $\lambda$, not on $\nabla$. Indeed,
two connection $1$-forms $\nabla$ and $\nabla'$ differ by $A \in \Omega^1(M,\der(\fK))$.
Using $\der(\fK) \simeq \fK$, we find
\[ \omega_{\lambda,\!\nabla'} - \omega_{\lambda,\!\nabla} = \delta \chi_{A}
\quad \mbox{ with } \quad
\chi_{A}(\xi  \oplus t) := \lambda(\kappa(A,\xi)).\]

According to the following theorem, every continuous Lie algebra $2$-cocycle on
$\fg \rtimes_{D} \R$ is cohomologous to one of the type
$\omega_{\lambda,\!\nabla}$ as defined in \eqref{cdef1} and \eqref{cdef2}.

\begin{Theorem} {\rm(Central extensions of extended gauge algebras)}\label{EqIjkcykel}
Let $\cK\rightarrow M$ be a bundle of Lie groups with simple fibers, equipped with a
1-parameter group of automorphisms with generator $\bv \in \cV(\cK)$. Let
$\fg = \Gamma_c(M,\fK)$ be the compactly supported gauge algebra, and let $\fg \rtimes_{D}\R$
be the Lie algebra \eqref{eq:smurfenliedje}.
Then the map $\lambda \mapsto [\omega_{\lambda,\!\nabla}]$ induces an isomorphism
\[\Big(\Omega^1_{c}(M,\bV) / \big(\dd\Omega^0_{c}(M,\bV) + L_{\bv_M}\Omega^1_{c}(M,\bV)\big)\Big)'
\stackrel{\sim}{\longrightarrow} H^2(\fg\rtimes_{D}\R,\R) \]
between the space of closed, $\bv_M$-invariant $\bV$-valued currents and
$H^2(\g \rtimes_D \R,\R)$.
\end{Theorem}

This is proven in \cite[Thm.~5.3]{JN17}.
The proof relies heavily on the description of $H^2(\fg,\R)$ provided
in \cite[Prop.~1.1]{JW13}.
\begin{Remark}{\rm(Temporal gauge)}\label{Rk:curvatureisnice}
If the Lie connection $\nabla$ on $\fK$ can be chosen so as to make $\bv \in \cV(\cK)$
horizontal, $\nabla_{\bv_M} \xi = L_{\bv}\xi$ for all $\xi \in \Gamma(M,\fK)$,
then equation \eqref{eq:defcurv} shows that
$L_{\bv}\nabla = i_{\bv_M} R$, where $R$ is the curvature of $\nabla$.
%
For such connections, (\ref{cdef2}) is equivalent to
\begin{equation}
  \label{eq:curv}
\omega_{\lambda,\!\nabla}(D, \xi ) = \lambda(\kappa(i_{\bv_M}R,\xi)).
\end{equation}
\end{Remark}

%
%
%

\section{Cocycles for positive energy representations}
\label{PEcocycles}

Having classified all the possible 2-cocycles on $\Gamma_{c}(M,\fK)\rtimes \R$, we now address
the restrictions that are imposed on these cocycles by the Cauchy--Schwarz estimates from \S\ref{subsec:3.2}.

In \S\ref{subsec:5.1} we derive
a local normal form of the cocycle $\omega$ in a flow box
around a point $m \in M$, where
$\bv_M \in \cV(M)$ does not vanish.
In \S\ref{subsec:5.2}, we use this to derive a
global normal form for $\omega$, provided that
$\bv_M$ is nowhere vanishing.
It turns out that $\omega$ is characterized by a \emph{measure} $\mu$ on the covering space $\widehat{M}$.
In \S\ref{subsec:5.3}, we plug this information back into the Cauchy--Schwarz estimate.
This yields the basic estimates needed for the continuity results in \S\ref{sec:6}.

The setting of this section is as follows. As before,
$\pi \colon \cK \rightarrow M$ is a bundle of Lie groups with semisimple fibers, and $\fK \rightarrow M$
is the corresponding bundle of Lie algebras.
We consider positive energy representations of $\widehat{G}$, where
$G =\Gamma_{c}(M,\cK)$
is the compactly supported gauge group with Lie algebra $\fg = \Gamma_{c}(M,\fK)$.
In fact, we will work mainly at the Lie algebra level, so our results continue to hold for the slightly more
general case that $G = \widetilde{\Gamma}_{c}(M,\cK)_{0}$ is
the simply connected cover of the identity component.
Using  \S\ref{vanseminaarsimpel}, we identify $\fg = \Gamma_{c}(M,\fK)$
with $\fg = \Gamma_{c}(\widehat{M},\hfK)$, where $\hfK \rightarrow \widehat{M}$ is a Lie algebra bundle
with simple fibers over a covering space $\widehat{M}$ of $M$.
We assume that the 1-parameter group of automorphisms is of geometric type
in the sense of Definition~\ref{def:gemetrictype}. The analogs
of the generators $\bv_M \in \cV(M)$
and $\bv \in \Gamma(M,\mathfrak{a}(\fK))$ for $\hfK$ are denoted by $\pi_*\hbv \in \cV(\widehat{M})$
and $\hbv \in \Gamma(\widehat{M},\mathfrak{a}(\hfK))$.

\subsection{Local gauge algebras}
\label{subsec:5.1}

The following simple lemma will be used extensively throughout the rest of the paper.
It gives a normal form for the pair $(\Gamma_c(M,\fK),\bv)$ in the neighborhood
of a point $m \in M$ where the vector field $\bv_M$ does not vanish.

\begin{Definition}[Good Flowbox] \label{def:goodflowbox}
A \emph{good flowbox} \index{flowbox \scheiding $I \times U_{0}$}
is a $\bv$-equivariant, local trivialization
$(I \times U_0)\times K \rightarrow \cK$ of $\cK$
over an open neighborhood $U \subseteq M$
that is equivariantly diffeomorphic to $I \times U_{0}$. Here $I \subseteq \R$
is a bounded open interval, and $U_{0} \subseteq \R^{n-1}$ is open.
Note that for $n=1$, we may take $U_0 = \{0\}$.
\end{Definition}

In particular, we have coordinates
$t:= x_0$ for $I$ and $\vec{x} := (x_1, \ldots, x_{n-1})$ for $U_0$
such that $\bv_M \in \cV(U)$ corresponds to
$\partial_{t} \in \cV(I\times U_0)$.

\begin{Lemma}\label{Uedele}
For any point $m\in M$ with $\bv_M(m)\neq 0$,
there is a good flowbox $U \simeq I \times U_0$ containing $m$.
Under the trivialization
$U \times \fk \rightarrow \fK|_{U}$, the induced isomorphism
$C^{\infty}_{c}(U,\fk) \simeq \Gamma_{c}(U,\fK)$
yields an inclusion
\[ I_{U} \colon C^{\infty}_{c}(U,\fk) \rtimes_{\partial_{t}} \R \hookrightarrow
\Gamma_{c}(M,\fK) \rtimes_{D} \R.\]
\end{Lemma}

\begin{proof}
Since $\bv_M(m)\neq 0$, we can find a neighborhood
$U \subseteq M$ of $m$ and local coordinates $t, x_1, \ldots, x_{n-1}$
such that the vector field $\bv_M$ on $U$ is of the form $\partial_{t}$.
We may assume that $U \simeq I \times U_{0}$
where $U_0 \subseteq \R^{n-1}$ corresponds to $t = 0$ and $I \subseteq \R$
corresponds to $\vec{x} = 0$.
We choose $U_0$ sufficiently small for there to exist a trivialization
$\Phi \colon U_0 \times K \rightarrow \cK|_{U_0}$, which we then extend to a trivialization
$U \times K \simeq \cK|_{U}$
over $U$ by $(t,x, k) \mapsto \gamma(-t) \Phi(x,k)$.
As $\frac{d}{dt}|_{t=0}\gamma(-t) = \bv$, the vector field $\bv
\in \cV(\cK)$ is horizontal in this trivialization.
\end{proof}

We consider $\fg_{U} := C^{\infty}_{c}(U,\fk)$ as a subalgebra of
$\fg = \Gamma_{c}(M,\fK)$
and wish to study the restriction $\dd\rho_{U}$
of the representation $\dd\rho$ to the subalgebra
\[ \hg_{U} := \R \oplus_{\omega} (\fg_{U}\rtimes_{\partial_{t}} \R).\]
Note that the subalgebra $\hg_{U}\hookrightarrow \hg$ does not correspond
to a Lie subgroup of $\widehat{G}$ unless $U$ is $\gamma$-invariant,
so we cannot work at the level of Lie groups.

If $A \in \Omega^1(U,\fk)$ is the local connection $1$-form corresponding to
the Lie connection $\nabla$,
then up to coboundaries,
by \eqref{cdef1} and \eqref{cdef2}
the restriction $\omega_{U}$ of $\omega$ to $\fg_U \rtimes_{\partial_t}\R$
takes the form
\begin{eqnarray}
\omega_{U}(fX,gY) &\simeq& \lambda_{U}(\kappa(fX,\dd g \cdot Y + g [A,Y]))\label{Ucocykel1}\\
\omega_{U}(\partial_{t},fX) &\simeq& \lambda_{U}(\kappa(\partial_{t}A, fX)),
 \label{Ucocykel2}
\end{eqnarray}
for some $\lambda_{U} \in \Omega^1_{c}(U,V(\fk))'$, where
$f,g \in C^{\infty}_{c}(U,\R)$ and $X,Y \in \fk$.

\begin{Proposition}\label{Splijtstof}
Let $m\in M$ be a point with $\bv_M(m)\neq 0$ and let
$U \simeq I \times U_{0}$ be a good flowbox
{\rm(cf.\ Definition~\ref{def:goodflowbox})}.
%
Let $\iota \colon U_0 \hookrightarrow M$ be the corresponding inclusion.
Then
the map
$\Omega^1_{c}(U,V(\fk)) \rightarrow \Gamma_{c}(U_0, \iota^*T^*M\otimes V(\fk)),
\beta \mapsto \oline\beta$, defined by the integration
\[
\ol{\beta}(x_1,\ldots,x_{n-1}) :=
\int_{-\infty}^{\infty} \beta(t,x_{1},\ldots,x_{n-1}) dt\,,
\]
yields a split exact sequence
\[ 0 \rightarrow
L_{\partial_t} \Omega^1_{c}(U,V(\fk)) \hookrightarrow
\Omega^1_{c}(U,V(\fk)) \rightarrow
\Gamma_{c}(U_0,\iota^*T^*M\otimes V(\fk)) \rightarrow 0
\]
of locally convex spaces.
%
In particular,
$\lambda_{U} \colon \Omega^{1}_{c}(U,V(\fk)) \rightarrow \R$
factors through a continuous linear map
$\ol{\lambda}_{U_0} \colon \Gamma_{c}(U_0,\iota^*T^*M\otimes V(\fk) ) \rightarrow \R$.
\end{Proposition}
\begin{proof}
The second statement follows from the first because
\[ \lambda_{U}(L_{\partial_t} \Omega^1_{c}(U,V(\fk))) = \{0\} \]
by Theorem~\ref{EqIjkcykel}.
The kernel of $\beta \mapsto \ol{\beta}$
is precisely $L_{\partial_t} \Omega^1_{c}(U,V(\fk))$ by the Fundamental Theorem of
calculus.
A bump function $\phi \in C^{\infty}_{c}(I,\R)$ of integral 1
yields the required continuous right inverse
$\Gamma_{c}(U_0,\iota^*T^*M\otimes V(\fk)) \rightarrow \Omega^1_{c}(U,V(\fk))$
for the integration map by sending $\vec{x} \mapsto \beta(\vec{x})$ to
$(t,\vec{x}) \mapsto \phi(t)\beta(x_1,\ldots x_{n-1})$.
\end{proof}

For $X=Y$ 
we obtain with \eqref{Ucocykel1} the relation
\[ \omega(\partial_{t}fX,fX) = \lambda_{U}\big(\partial_{t}f\cdot \dd f\cdot \kappa(X,X)\big).\]
Unlike (\ref{Ucocykel1}), which holds only modulo coboundaries,
this equation is exact because $(\partial_{t}f)X$
and $fX$ commute.
Lemma \ref{Lem:CSRaw} (the Cauchy--Schwarz estimate) then yields
\begin{equation}\label{neushoorn}
-\lambda_{U}\big(\partial_{t}f\cdot \dd f\cdot \kappa(X,X)\big)  \geq 0\,.
\end{equation}
This allows us to characterize $\lambda_{U}$ as follows.

\begin{Proposition}\label{taart}
Let $m\in M$ be a point with $\bv_M(m)\neq 0$.
Then there exists an open neighborhood $U\subseteq M$ of $m$
such that, for each $X\in \fk$, there exists
a unique $\bv_M$-invariant
positive locally finite regular Borel measure
$\mu_{U,X}$ on $U$ such that the functional
$\lambda_{U,X} \in \Omega^1_{c}(U,\R)'$
defined by $\lambda_{U,X}(\beta) := -\lambda_{U}\big(\beta \cdot \kappa(X,X)\big)$
takes the form
\[\lambda_{U,X}(\beta) = \int_{U}(i_{\bv_M}\beta) d\mu_{U,X}(m)\,.\]
\end{Proposition}

\begin{proof}
Introduce coordinates $x_0 := t$ and $\vec{x} := (x_1, \ldots, x_{n-1})$ on $U\simeq I \times U_0$
as in
Definition~\ref{def:goodflowbox}.
Define $\lambda_{U,i} \in C^{\infty}_{c}(U,\R)'$, $i = 0,\ldots, n-1$, by
$\lambda_{U,i}(f) := \lambda_{U,X}(f\dd x_i)$ and let
$\lambda_i \in C^{\infty}_{c}(U_0,\R)'$ be the corresponding
distribution on $U_0$ (cf.\ Proposition~\ref{Splijtstof}), so $\lambda_{U,i}(f) = \lambda_{i}(\overline{f})$
with $\overline{f}(\vec{x}) := \int_{I}f(t,\vec{x}) dt$.
Then $\lambda_{U,X}(f \dd g) = \sum_{i=0}^{n-1}\lambda_{i}(\ol{f\partial_{i}g})$
for all $f,g \in C^{\infty}_{c}(U,\R)$.
Equation (\ref{neushoorn}) then yields
\begin{equation}\label{positive}
\lambda_0\left(\ol{(\partial_{t}f)^2}\right) +
\sum_{i=1}^{n-1}\lambda_i\left(\ol{\partial_t f \partial_{i}f}\right)\geq 0.
\end{equation}

First, we show that $\lambda_0(h^2) \geq 0$
for any $h$ in $C^{\infty}_c(U_0,\R)$.
Note that every element $B$ of $C^{\infty}_{c}(I,\R)$ satisfies
$\int_I B\partial_{t}B dt = 0$.
We choose $B \neq 0$, normalize it by $\int_{I}(\partial_{t}B)^2 dt = 1$ and
define $f(t,\vec{x}) := B(t)h(\vec{x})$. We then have
$\ol{(\partial_{t}f)^2} = h^2$ and
$\ol{\partial_{t}f \partial_{i}f} = h\partial_{i}h \int_IB\partial_{t}B dt = 0$
for $i \geq 1$. Therefore,
(\ref{positive}) yields $\lambda_0(h^2)\geq 0$ as required.

Since $\lambda_0$ extends\footnote{%
For every compact $S \subseteq U_{0}$,
there exists a $\phi\in C^{\infty}_{c}(U_0,\R)$ with $\phi|_{S}>1$.
With $L_{S} = \lambda_0(\phi^2)$, it then follows
from the inequality
$\lambda_0(\|f\|_{\infty}\phi^2 \pm f) \geq 0$ that
$|\lambda_{0}(f)| \leq L_{S}\|f\|_{\infty}$ for all $f$ with support in $S$.}
to a positive linear functional on $C_{c}(U_{0},\R)$,
Riesz' Representation Theorem \cite[Thms.~2.14, 2.18]{Ru87} yields
a unique locally finite regular Borel measure $\mu_{0}$
on $U_0$ such that $\lambda_0(f) = \int_{U_0}f\, d\mu_0(x)$.
This implies that $\lambda_{U,0}(f) = \int_{U}f(u)\, d\mu_{U,X}(u)$,
with $\mu_{U,X}$ the product
of $\mu_0$ with the Lebesgue measure on $I$.


To finish the proof, we now prove that $\lambda_i=0$ for $i>0$.
It suffices to show that $\lambda_{i}(h^2) = 0$ for all
$h \in C^{\infty}_c(U_0,\R)$.
Choose $B_C, B_S \in C^{\infty}_{c}(I,\R)$
so that $\int_I B_S(t)B'_C(t) dt = 1$,
choose $C, S \in C^{\infty}_{c}(U_0,\R)$ so that
$C(x) = \cos(\sum_{i=1}^{n} k_ix^i)$
and $S(x) =  \sin(\sum_{i=1}^{n} k_ix^i)$
for $x \in \supp(h)$, $k_i \in \Z$, and set
$$f(t,\vec{x}) := h(\vec{x})\Big(B_C(t) C(\vec{x}) + B_S (t) S(\vec{x})\Big)\,.$$
Then with $E := \int_I \big(|B'_C(t)|  + |B'_S(t)|\big)^2dt$, we have
$$
0 \leq \ol{(\partial_t f)^2 } = h^2(\vec{x})
\int_I \big(B'_C(t) C(\vec{x}) + B'_S(t)S(\vec{x})\big)^2dt
\leq E h^2(x)\,.
$$
Making repeated use of
$\int_I F(t,\vec{x})\partial_tF(t,\vec{x})dt = 0$
and $\int_I B'_C B_S + B'_S B_C dt = 0$,
we find, for $i= 1,\ldots, n-1$,
$$
\ol{\partial_t f \partial_i f} = k_i h^2\,.
$$
Equation~(\ref{positive}) then yields
\begin{equation}\label{kwaaddaglicht}
\lambda_0\left(\ol{(\partial_t f)^2}\right) + \sum_{i=1}^{n-1} k_i\lambda_i(h^2) \geq
0 \quad \mbox{ for all } \quad k_i \in \Z,
\end{equation} where the function $f$ depends on the $k_i$.
As $\lambda_0(\overline{(\partial_{t} f)^2}) \leq E \lambda_0(h^2)$,
the non-negative term
$\lambda_0(\overline{(\partial_{t} f)^2})$ is bounded by a number that does not depend on $k_i$.
It therefore follows from inequality (\ref{kwaaddaglicht}) that
$\lambda_i(h^2) = 0$ for all $i>0$, as was to be proven.
\end{proof}

\subsection{Infinitesimally free \texorpdfstring{$\R$-actions}{R-actions}}
\label{subsec:5.2}

In \S\ref{vanseminaarsimpel}, we saw that
$\Gamma_{c}(M,\fK)$
is isomorphic to the gauge algebra $\Gamma_{c}(\widehat{M},\widehat{\fK})$, where
$\widehat{\fK} \rightarrow \widehat{M}$ is a Lie algebra bundle with
\emph{simple} fibers
over a cover $\widehat{M} \rightarrow M$.
The decomposition $\widehat{M} = \bigsqcup_{i=1}^{r} \widehat{M_i}$
in connected components therefore gives rise to a direct sum decomposition
\begin{equation}\label{eq:decompsum}
\Gamma_{c}(M,\fK) = \bigoplus_{i=1}^{r} \Gamma_{c}(\widehat{M}_{i}, \widehat{\fK})\,,
\end{equation}
where $\widehat{\fK} \rightarrow \widehat{M}_i$ is a Lie algebra bundle with simple
fibers isomorphic to $\fk_i$.

\subsubsection{Reduction to compact simple structure algebras}

If $\bv_M$ is non-vanishing, then
we can restrict attention to the terms
in~\eqref{eq:decompsum}
where~$\fk_i$ is a compact simple Lie algebra.

\begin{Corollary}\label{Cer:acorro}
Suppose that $\fk_{i}$ is not compact, and
let $m\in \widehat{M}_i$ be a point with $\pi_*\widehat{\bv}_{m}\neq 0$.
Let $U \subseteq \widehat{M}_i$
be as in {\rm Proposition~\ref{taart}} and let $\lambda_{U} \in \Omega^{1}_{c}(U,V(\fk_i))'$
be as in {\rm \eqref{Ucocykel1}} and {\rm\eqref{Ucocykel2}}.
Then
$\lambda_{U}\colon \Omega^{1}_{c}(U,V(\fk)) \rightarrow \R$
is zero.
Consequently,
$\omega_{U}$ is cohomologous to zero on
$\Gamma_{c}(U,\widehat{\fK})$.
\end{Corollary}
\begin{proof}
It suffices to show that $\mu_{U,X} = 0$ for all $X \in \fk_i$.
If $X,Y \in \fk_i$ with $\kappa(X,X) = -\kappa(Y,Y)$, then
$\mu_{U,X} = - \mu_{U,Y}$ implies $\mu_{U,X} = \mu_{U,Y} = 0$.
If $\fk_{i}$ is a complex Lie algebra,
i.e., $V(\fk_i)\simeq \C$
(cf.~\S\ref{invbil}), then the previous argument with $Y = iX$
yields $\mu_{U,X} = 0$ for all $X\in \fk_{i}$.
If $V(\fk_{i}) \simeq \R$, then $\fk_i$ is noncompact if and only if
$\{\kappa(X,X)\,;\, X\in \fk_{i}\} = \R$. Therefore, the same reasoning applies.
\end{proof}

\begin{Corollary}
If $\rho$ is a positive energy representation of $\widehat{G}$
and $\bv_M$ has no zeros,
then $\omega$ is cohomologous to a cocycle that vanishes on the subalgebras
$\Gamma_{c}(\widehat{M}_i, \widehat{\fK})$, where $\fk_i$ is noncompact.
\end{Corollary}

\begin{proof}
By Theorem~\ref{EqIjkcykel} applied to $\Gamma_{c}(\widehat{M},\widehat{\fK})$, the class $[\omega] \in H^2(\fg\rtimes_{D}\R,\R)$
is uniquely determined by a 
$\bV$-valued current
$\lambda \colon \Omega_{c}(\widehat{M},\bV)\rightarrow \R$.
Since $\bv_M$ is everywhere nonzero, the same holds for $\pi_*\widehat{\bv}$.
If $\fk_i$ is noncompact, by Corollary~\ref{Cer:acorro},
$\widehat{M}_i$ can be covered
with open sets $U_{ij}$
such that $\lambda$ vanishes on  $\Omega_{c}(U_{ij},\bV)$.
As every element of $\Omega_{c}(\widehat{M}_i,\bV)$
can be written as a finite sum of elements of
$\Omega_{c}(U_{ij},\bV)$, the current $\lambda$
vanishes on $\Omega_{c}(\widehat{M}_i,\bV)$.
\end{proof}

\subsubsection{Reduction of currents to measures}
\label{subsec:5.2redcurmeas}


Let $\rho \colon \widehat{G} \rightarrow\U(\cH)$
be a positive energy representation,
where $G = \tilde\Gamma_{c}(M,\cK)_{0}$ is the simply connected
Lie group with Lie algebra $\g=\Gamma_c(M,\fK)$, which covers
the identity component of the compactly supported gauge group.
This gives rise to a Lie algebra cocycle $\omega$ on $\fg\rtimes_{D}\R$.
Using the results of \S\ref{vanseminaarsimpel}, we identify the gauge Lie algebra
$\fg = \Gamma_{c}(M,\fK)$ with $\fg = \Gamma_{c}(\widehat{M},\widehat{\fK})$,
where $\widehat{\fK} \rightarrow \widehat{M}$ is a Lie algebra bundle with
simple fibers.
The cocycle $\omega$ can then be represented by a \emph{measure} on $\widehat{M}$.

\begin{Theorem}\label{MeasureThm}
\index{measure associated to cocycle \scheiding $\mu$}
Suppose that $\bv_M$ has no zeros, and that $\omega$ is a 2-cocycle
on $\g\rtimes_D \R$ induced by a  positive energy representation
$\rho \colon \widehat{G} \rightarrow\U(\cH)$.
Then there exists a positive, regular, locally finite Borel measure $\mu$
on $\widehat{M}$ invariant under the flow $\gamma_{\hat M}$ on $\hat M$ induced
by $\gamma_{\cK}$, such that
$\omega$ is cohomologous to the
2-cocycle $\omega_{\mu,\nabla}$, given by
\begin{eqnarray}
\omega_{\mu,\!\nabla}(\xi,\eta) &=& -\int_{\widehat{M}}
\kappa(\xi,\nabla_{\widehat{\bv_M}}\eta)\, d\mu(m)\,,\label{freecocycle1}\\
\omega_{\mu,\!\nabla}(D,\xi) &=& -
\int_{\widehat{M}} \kappa(i_{\widehat{\bv_M}}(L_{\widehat{\bv}}\nabla),\xi)\, d\mu(m)
\quad \mbox{ for } \quad \xi,\eta \in \Gamma_{c}(\widehat{M},\widehat{\fK}).
\label{freecocycle2}
\end{eqnarray}
The support of $\mu$ is contained in the union of the connected components
$\widehat{M}_i$ where the fibers of $\widehat{\fK}$
are compact simple Lie algebras.
In \eqref{freecocycle1} and \eqref{freecocycle2}, we identify
$\kappa$ with
the positive definite
invariant bilinear form normalized as in~\eqref{eq:kappa-normal}.
\end{Theorem}


\begin{proof}
As $\bv_M$ is nowhere zero, we can cover $\widehat{M}$ by good flowboxes
$U \subseteq \widehat{M}$ in the sense of Definition~\ref{def:goodflowbox}.
In the corresponding local trivialization $\Gamma_{c}(U,\hfK) \simeq C^{\infty}_{c}(U,\fk)$
(cf.\ Lemma~\ref{Uedele}),
we may assume that $\fk$ is compact by Corollary~\ref{Cer:acorro}.
We normalize $\kappa$ as in~\eqref{eq:kappa-normal}
and define $\mu_{U}$ as $\mu_{U,X}$ for any $X\in \fk$ with $\kappa(X,X) = 1$.
If $U$ and $U'$ are two such open sets, then the measures
$\mu_{U}$ and $\mu_{U'}$ from Proposition~\ref{taart} coincide
on the intersection $U \cap U'$, as both measures
are uniquely determined by the cocycle $\omega$.
The measures $\mu_{U}$ thus splice together
to form a positive regular locally finite Borel measure
on~$\widehat{M}$.
Equations (\ref{freecocycle1}) and (\ref{freecocycle2}) then follow
immediately from
(\ref{cdef1}), (\ref{cdef2}) in \S\ref{subsec:2cocyc},
and \eqref{Ucocykel1}, \eqref{Ucocykel2}.
\end{proof}

\begin{Remark}\label{omooieconnectie}
As the cohomology class $[\omega_{\lambda,\!\nabla}]$ is independent of the choice
 of the Lie connection,
we are free to choose $\nabla$ so that $\widehat{\bv}$ is horizontal.
In that case, we have $i_{\widehat{\bv_M}}(L_{\widehat{\bv}}\nabla) = 0$ and
$L_{\widehat{\bv}}\xi = \nabla_{\pi_*\widehat{\bv}}\xi$
(cf.~Remark~\ref{Rk:curvatureisnice}).
Equation \eqref{freecocycle2} then becomes
\[\omega_{\mu,\!\nabla}(D,\xi) = 0\,.\]
\end{Remark}

From Examples \ref{ex:simple}--\ref{Ex:MhatOrientable} in \S\ref{vanseminaarsimpel},
we obtain the following.

\begin{Example}
If $\fK \rightarrow M$ has simple fibers, then
$\widehat{M} = M$. 
The class $[\omega_{\mu,\!\nabla}]$ then corresponds to a
measure $\mu$ on $M$.
It vanishes
on the connected components of $M$ where the fibers of $\fK \rightarrow M$ are noncompact.
\end{Example}

\begin{Example}
Suppose that $M$ is connected, and that
the typical fiber $\fk = \bigoplus_{i=1}^{r} \fk_i$ is the direct sum of $r$ mutually
non-isomorphic simple ideals $\fk_i$.
Then $\widehat{M}$ is the disjoint union of $r$ copies of $M$. 
The class $[\omega_{\mu,\nabla}]$ is then given by $r$ measures $\mu_{i}$ on $M$, one for each
simple ideal.
The same holds if $\fK = M \times \fk$ is trivial, and the $\fk_{i}$ are not necessarily non-isomorphic.
\end{Example}


\begin{Example} (Frame bundles of 4-manifolds)
(cf.~Examples~\ref{Ex:MhatOrientable}). 
Suppose that $M$ is a Riemannian 4-manifold, and
$\fK = \ad({\rm OF}(M))$ is the adjoint bundle of its
orthogonal frame bundle.
%
If $M$ is orientable,
then $\omega_{\mu} = \omega_{\mu_L} + \omega_{\mu_{R}}$ is the
sum of two cocycles with measures $\mu_L$ and $\mu_{R}$
on $M$ corresponding to the simple factors $\mathfrak{su}_{L}(2,\C)$
and $\mathfrak{su}_{R}(2,\C)$ of $\mathfrak{so}(4,\R)$.
If $M$ is not orientable, then $\omega_{\mu}$ is described by a single
measure $\mu$ on the orientable cover $\widehat{M} \rightarrow M$.
\end{Example}

\subsection{Cauchy--Schwarz estimates revisited}
\label{subsec:5.3}

Using the explicit form of the cocycles determined in Theorem \ref{MeasureThm},
we revisit the Cauchy--Schwarz estimates of \S\ref{subsec:3.2}.
In this subsection, we assume that $\fK \rightarrow M$ has semisimple fibers, and that the vector field
$\bv_M$ on $M$ is nowhere vanishing.
As before, we identify $\Gamma_{c}(M,\fK)$ with $\Gamma_{c}(\widehat{M},\hfK)$, where $\hfK \rightarrow \widehat{M}$
has simple fibers.

Define the positive semidefinite symmetric bilinear form
on $\fg = \Gamma_{c}(\widehat{M},\widehat{\fK})$ by
\begin{equation}
  \label{eq:bilform}
\langle \xi,\eta \rangle_{\mu} := \int_{\widehat{M}}
\kappa (\xi,\eta)\, d\mu(m).
\end{equation}
Using Theorem~\ref{MeasureThm} and Remark~\ref{omooieconnectie},
we may replace $\omega$ by $\omega_{\mu,\!\nabla}$
for a Lie connection $\nabla$ on $\hfK$ that makes $\widehat{\bv}$ horizontal.
In that case, we have $i_{\widehat{\bv_M}}(L_{\widehat{\bv}}\nabla) = 0$ and
$L_{\widehat{\bv}}\xi = \nabla_{\pi_*\widehat{\bv}}\xi$
(cf.~Remark~\ref{Rk:curvatureisnice}).
We may thus assume, without loss of generality, that the cocycle
associated to a positive energy representation takes the form
\begin{equation}\label{eq:omega}
\omega(\xi,\eta) = -\langle\xi,L_{\widehat{\bv}}\eta\rangle_{\mu}
=  \langle L_{\widehat{\bv}} \xi,\eta\rangle_{\mu},\qquad \omega(D,\xi) = 0.
\end{equation}
The Cauchy--Schwarz estimate (Lemma \ref{Lem:CSRaw}) can now be reformulated as follows.

\begin{Lemma} {\rm(Cauchy--Schwarz Estimate)}\label{CSgauge}
Let $\rho$ be a positive energy representation of $\widehat{G}$, where
$G = \tilde\Gamma_c(M,\cK)_{0}$ is the simply connected gauge group.
If the vector field $\bv_M$ on $M$ has no zeros, then,
after replacing the linear lift $\dd\rho \: \g \to \End(\cH^\infty)$
of the projective representation
 $\oline{\dd\rho}$ of $\g$ by $\dd\rho + i \chi \one$ for some continuous
 linear functional $\chi \colon \fg \rightarrow \R$,
%
we have
\begin{equation}\label{eq:CSestuary}
\langle i \dd\rho(L_{\widehat{\bv}}\xi)\rangle_{\psi}^2 \leq 2 \langle H\rangle_{\psi} \|L_{\widehat{\bv}}\xi\|^2_{\mu} \quad \mbox{ for all } \quad \xi \in \fg \quad \mbox{ with } \quad
[L_{\widehat{\bv}}\xi,\xi] = 0
\end{equation}
and
every unit vector $\psi \in \cH^{\infty}$.
\end{Lemma}

\begin{proof}
First we observe that the passage from $\omega$ to an equivalent cocycle
corresponds to replacing the subspace
$\g \subeq \hat\g$ by the subspace
$\chi(\xi)C + \xi$, $\xi \in \g$, where $\chi \: \g \to \R$ is a continuous
linear functional. For the representation $\dd\rho$ this changes
the value of $\dd\rho(\xi)$ by adding $i\chi(\xi)$, so that we can achieve
a cocycle of the form \eqref{eq:omega} by Theorem~\ref{MeasureThm}. Now we apply
 Lemma~\ref{Lem:CSRaw} with $i_{D}\omega_{\mu,\!\nabla} = 0$ and
$\omega_{\mu,\!\nabla}(\xi,D\xi) = \|L_{\widehat{\bv}}\xi\|^2_{\mu}$.
\end{proof}

In the same vein, the refined Cauchy--Schwarz estimate, Lemma~\ref{4point}, can be reformulated as follows.
\begin{Lemma}\label{4pointGauge}
Under the assumptions of {\rm Lemma~\ref{CSgauge}}, we have
\begin{equation}\label{waterhoen}
\begin{split}
\left( \Big\langle i\dd\rho\big(e^{-s\ad_\eta}(L_{\widehat{\bv}}\xi)\big)\Big\rangle_{\psi}
-\left\langle \frac{e^{-s\ad_{\eta}} - \one}{\ad_{\eta}} (L_{\widehat{\bv}}\xi),
L_{\widehat{\bv}}\eta \right\rangle_{\mu}
\right)^{2} \leq \\
 2\|L_{\widehat{\bv}}\xi\|^2_{\mu} \Big(\langle H \rangle_{\psi} +
s\langle i\dd\rho(L_{\widehat{\bv}}\eta)\rangle_{\psi} +
{\textstyle \frac{s^2}{2}} \|L_{\widehat{\bv}}\eta\|^{2}_{\mu} \Big)
\end{split}
\end{equation}
for all $s\in \R$, and for all $\xi,\eta \in \Gamma_{c}(\widehat{M},\hfK)$ such that $[\xi,L_{\widehat{\bv}}\xi] = 0$
and $[\eta, L_{\widehat{\bv}}\eta] = 0$.
\end{Lemma}


\section{Continuity properties}
\label{sec:6}


Having determined which cocycles are compatible
with the Cauchy--Schwarz estimates, we now turn to the classification
of positive energy representations for the central extensions
that correspond to these cocycles.
This will be achieved in \S\ref{sec:locthmsec}, using
continuity and extension results developed in the present section.

In this section, we assume that the flow $\bv_M$ is nowhere vanishing.
Further, we assume that
the fibers of $\fK \rightarrow M$
are \emph{simple} Lie algebras.
This entails no loss of generality compared to semisimple fibers,
as one can switch to the Lie algebra bundle
$\hfK \rightarrow \widehat{M}$ in that case
by the results in \S\ref{vanseminaarsimpel}.

In~\S\ref{sec:redcpt}, we use the Cauchy--Schwarz estimate~\ref{CSgauge}
to further reduce the problem to the case where $\fK$ has
\emph{compact} simple fibers.
In \S\ref{sec:nummertweevanzes}, we use the refined Cauchy--Schwarz estimate
\ref{4pointGauge} to bound $i\dd\rho(\xi)$ in terms of the Hamilton operator $H$, the $L^2$-norm
$\|\xi\|_{\mu}$ with respect to the measure $\mu$ of Theorem~\ref{MeasureThm}, and the $L^2$-norm $\|\xi\|_{B\mu}$ with respect to
the product of $\mu$ with a suitable upper semi-continuous function
$B \colon M \rightarrow \R^+$.
In \S\ref{sec:nummer3van6}, we interpret these estimates as a continuity property,
and use this to define an extension of $\dd\rho$ to a space
$H^1_{B\mu}(M,\fK)$ of sections that are differentiable in the direction of the orbits,  but merely measurable in the transversal direction.
In \S\ref{sec:SobolevLie}, we construct a subspace $H^1_{\partial}(M,\fK)$ of
bounded sections that is closed under the pointwise Lie bracket.
Finally, in \S\ref{sec:context5van6}, we show that by extending to
$H^1_{B\mu}(M,\fK)$ and restricting to $H^1_{\partial}(M,\fK)$, one obtains
a representation of the latter Lie algebra that is compatible with the Hamiltonian $H$. On a subalgebra $H^2_{\partial}(M,\fK)$ of sections that are twice differentiable in the orbit direction, we then show that there is a dense set of uniformly
analytic vectors. In \S\ref{sec:locthmsec}, this will be needed in order to integrate the Lie algebra representation to the group level.

%
%
%

\subsection{Reduction to compact simple structure algebras}\label{sec:redcpt}

As a direct consequence of Lemma~\ref{CSgauge}, we see that $\dd\rho(L_{\bv}\xi)$ vanishes
for all $\xi \in \fg$ with $[\xi,L_{\bv}\xi] = 0$ and $\|L_{\bv}\xi\|_{\mu} = 0$.
We use this to show that every positive energy representation factors through
a gauge algebra with \emph{compact} structure algebra.

\begin{Proposition}\label{DgDg}
For $\fg = \Gamma_{c}(M,\fK)$ with $\bv_M$ everywhere nonzero,
we have $\fg = D\fg + [D\fg,D\fg]$. 
Considered as subsets of $\widehat{\fg} = \R \oplus_\omega
(\fg \rtimes_{D}\R)$,
with $\omega$ as in {\rm Theorem~\ref{MeasureThm}},
we have
$\R \oplus_\omega \fg= D\fg + [D\fg,D\fg]$.
\end{Proposition}

\begin{proof}
By a partition of unity argument, it suffices to prove this for
$\fg = C^{\infty}_{c}(U,\fk)$, where $U = I \times U_0$ is a good flowbox
(cf.\ Definition~\ref{def:goodflowbox}) and
$D\xi = \frac{d}{dt}\xi$ (cf.\ Lemma~\ref{Uedele}).
%
If $f\in C_{c}^{\infty}(U,\fk)$ and $X\in \fk$, then
$fX$ lies in $D\fg\subseteq \fg$
if and only if 
$f_0(x) := \int_{-\infty}^{\infty} f(t,x)dt$ is zero in $C^{\infty}_{c}(U_0,\R)$.
Fix
$\zeta \in C^{\infty}_{c}(I,\R)$ with $\int_{-\infty}^{\infty}\zeta(t)dt = 0$
and $\int_{-\infty}^{\infty}\zeta^2(t)dt = 1$.
Then $fX = (f - \zeta^2 f_0)X + \zeta^2 f_0 X$ with
$(f - \zeta^2 f_0)X \in D\fg$.
To show that
$\zeta^2 f_0 X \in [D\fg, D\fg]$, choose
$\chi \in C^{\infty}_{c}(U_0,\fk)$ such that $\chi|_{\mathrm{supp}(f_0)} = 1$, and choose
$Y_i, Z_i \in \fk$ such that
$X = \sum_{i=1}^{r}[Y_i,Z_i]$.
Since  $\sum_{i=1}^{r}[\zeta f_0 Y_{i},\zeta \chi Z_i] = \zeta^2 f_0 X$ with $\zeta f_0Y_{i},\zeta \chi Z_i \in D\fg$, we have
\begin{equation}\label{eq:grondslurf}
fX = (f - \zeta^2 f_0)X + \sum_{i=1}^{r}[\zeta f_0 Y_{i},\zeta \chi Z_i] \in D\fg + [D\fg,D\fg]\,.
\end{equation}
This holds for the Lie bracket in $\fg$ as well as for the Lie bracket
in $\widehat{\fg}$.
Since $\int_{-\infty}^{\infty} \zeta \frac{d}{dt}\zeta dt = 0$, we find
$\omega(\zeta f_0Y_{i},\zeta \chi Z_i)=0$.
This shows that $\fg = D\fg + [D\fg ,D\fg]$ in $\fg$ and also
$\fg \subseteq D\fg + [D\fg ,D\fg]$ in $\widehat{\fg}$.
Since $\omega$ is not identically zero on $D\fg \times D\fg$, the
subspace $D\g + [D\g,D\g]$ of $\hat\g$ cannot be proper
and thus $\R C \subeq D\g + [D\fg,D\fg]$. This shows that
$\hat\g = \R C +  \fg = D\fg + [D\fg ,D\fg]$.
\end{proof}

\begin{Theorem} {\rm(Reduction to compact structure algebra)} \label{red2cpt}
Let $M_i \subseteq M$ be a connected component
such that the (simple) fibers of $\fK|_{M_i}$ are not compact.
Suppose that $\bv_M$ is non-vanishing on $M_i$.
Then, after twisting by a functional $\chi \in \Gamma_{c}(M_{i},\fK)'$
if necessary,
every positive energy representation $\dd\rho$ of $\Gamma_{c}(M,\fK)$
vanishes on the ideal $\Gamma_{c}(M_i,\fK)$.
\end{Theorem}

\begin{proof}
By a partition of unity argument, it suffices to
consider the
restriction of $\dd\rho$ to $C^{\infty}_{c}(U,\fk)$ for a good flowbox $U\subseteq M_i$
(cf.\ Definition~\ref{def:goodflowbox}).
Every $\xi\in DC^{\infty}_{c}(U,\fk)$ is a finite sum of elements
of the form $f' X$, with $f \in C^{\infty}_{c}(U,\R)$ and $X\in \fk$.
Since $\fk$ is noncompact,
$\mu$ vanishes on $M_i$
by Theorem~\ref{MeasureThm}. Since $\|f'X\|_{\mu} = 0$ and
$[fX_i,f'X_i] = 0$, it follows from
Lemma~\ref{CSgauge} that,
after twisting
by $\chi$ so as to change $\omega$ to $\omega_{\mu,\nabla}$, we have
$\dd\rho(f'X_i) =0$.
Since ${\dd\rho(DC_{c}^{\infty}(U,\fk)) = \{0\}}$,
Proposition~\ref{DgDg} yields $\dd\rho(C_{c}^{\infty}(U,\fk)) = \{0\}$.
Thus $\dd\rho(\Gamma_{c}(M_i,\fK)) = \{0\}$, as required.
\end{proof}

This shows that we can restrict attention to bundles $\fK \rightarrow M$ with \emph{compact}
simple fibers.
In conjunction with
Proposition~\ref{DgDg}, Lemma~\ref{CSgauge} can also be used to prove the following.
\begin{Corollary}\label{zondagskind}
If $\fg = \Gamma_{c}(M,\fK)$, where $\fK \rightarrow M$ has compact simple fibers,
then, after twisting by $\chi \in \Gamma_{c}(M,\fK)'$
if necessary, every positive energy representation $\dd\rho$ of $\widehat{\fg}$
vanishes on the ideal
\[I_{\mu} := \{\xi \in \fg \,;\, \mu(\{x\in M\,;\, \xi(x)\neq 0\}) = 0\}
\]
of sections that vanish $\mu$-almost everywhere.
\end{Corollary}

\begin{proof}
By a partition of unity argument, we may assume that
$\fg = C^{\infty}_{c}(U,\fk)$, with $U \subseteq M$ a good flowbox
(Definition~\ref{def:goodflowbox}).
Let $\xi \in I_{\mu}$ and consider the subset
$\mathcal{O}_{\xi} := \{x\in M\,;\, \xi(x)\neq 0\}$, which is the ``open support''
of $\xi$.
Since $\xi$ is a linear combination of terms $fX$ with smaller or equal open
support, we may assume that $\xi = fX$ for $f\in C^{\infty}_{c}(U,\R)$ and
$X \in \fk$.
If $fX \in D\fg$, then $fX \in I_{\mu}$ implies $\|fX\|_{\mu} = 0$ and hence
$\dd\rho(fX) = 0$ by Lemma \ref{CSgauge}.
Decompose $fX$ as in equation~\eqref{eq:grondslurf},
\[ fX = (f-\zeta^2 f_0)X + \sum_{i=1}^{r}[\zeta f_0 Y_i,\zeta \chi Z_i]\,.\]
As  $\mathcal{O}_\xi$ is open and $\mu = dt \otimes \mu_0$, we have
$\mu(\mathcal{O}_\xi) = 0$ if and only if
$\mu_0\big(p(\mathcal{O}_\xi)\big)$ vanishes, where $p \colon U \rightarrow U_0$
is the projection on the orbit space.
Now ${(f - \zeta^2f_0)X}$ and $\zeta f_0 Y_i$ are in $D\fg$ and vanish outside
$p^{-1}p(\mathcal{O}_{\xi})$, so that their images under $\dd\rho$ vanish.
Indeed, as these are both of the form $L_{\bv}\eta$ with $\|L_{\bv}\eta\|_{\mu} = 0$
and $[L_{\bv}\eta,\eta] = 0$, this follows from Lemma~\ref{CSgauge}.
It follows that $\dd\rho(fX) = 0$, as required.
\end{proof}

\subsection{Extending the estimates from \texorpdfstring{$D\fg$ to $\fg$}{Dg to g}}\label{sec:nummertweevanzes}

To see that $\dd\rho$ factors through a linear map on
$\fg/I_{\mu}$, we used the Cauchy--Schwarz estimate of Lemma~\ref{CSgauge}.
Using the \emph{refined} Cauchy--Schwarz estimate of Lemma~\ref{4pointGauge},
we then extend $\dd\rho$ to a linear map on
$\overline{\fg/I_{\mu}}$, the $L^2$-completion of $\fg/I_{\mu}$ with respect to
the measure $\mu$.

Note that an extension to the subspace $\overline{D\fg/I_{\mu}}\subseteq \overline{\fg/I_{\mu}}$ can already be achieved
using the `ordinary' Cauchy--Schwarz estimate of Lemma~\ref{CSgauge}.
Indeed, for $\xi \in D\fg$, one can use \eqref{eq:CSestuary} to show that
$\dd\rho(\xi)$ satisfies an operator inequality of the form
\begin{equation}\label{eq:goal}
\pm i \dd\rho(\xi) \leq \|\xi\|_{\mu}(\alpha\one + \beta H)
\end{equation}
for certain constants $\alpha,\beta > 0$.
With this, one can prove that the linear map 
$\dd\rho \colon D\fg/I_{\mu} \rightarrow \mathrm{End}(\cH^{\infty})$
is weakly continuous with respect to the $L^2$-topology on $D\fg/I_{\mu}$, and
that it extends to the $L^2$-completion $\overline{D\fg/I_{\mu}}$.

In order to extend $\dd\rho$ to the full space $\overline{\fg/I_{\mu}}$,
however, we will need an analogue of \eqref{eq:goal} that holds
not just for $\xi\in D\fg$, but
for all $\xi \in \fg$.
This is Proposition~\ref{Prop:rafelfluitje},
which we prove using
the refined Cauchy--Schwarz
estimate of Lemma~\ref{4pointGauge}.

\subsubsection{The local gauge algebra with fibers $\fk = \mathfrak{su}(2,\C)$}
\label{subsubsec:su2}
First, we restrict our attention to the
compact structure algebra $\fk = \mathfrak{su}(2,\C)$.
We will later derive the general case from this example.
Let $\kappa(a,b) = -\mathrm{tr}(ab)$
be the invariant bilinear form on $\fk$, normalized so that
elements $x$ with 
\[ \Spec(\ad x) = \{0,\pm  2i\} \quad \mbox{  satisfy } \quad \kappa(x,x) = 2.\]

Further, let $U' \subset U$ be a \emph{good pair of flowboxes}
in the sense of the following definition.
We write $U \Subset V$ if the closure of $U$ is contained in an open subset of $V$.

\begin{Definition}[Good pair of flowboxes]\label{def:goodpair}
Let $U' \simeq I' \times U'_0$ and $U \simeq I \times U_0$ be good flowboxes in the sense of
Definition~\ref{def:goodflowbox}, and let $U'\subset U$. We call $U'\subset U$ a
\emph{good pair of flowboxes} if $I'\Subset I$ and $U'_{0} \Subset U_0$.
\end{Definition}
\begin{Remark}
Note that $U'_{0} = U_{0} = \{0\}$ is allowed!
Unless specified otherwise,
we assume that $I'= (-T'/2,T'/2)$ and $I = (-T/2,T/2)$
with $0 < T'< T< \infty$.
\end{Remark}

\begin{Remark}
Recall that $M$ is equipped with a flow-invariant measure $\mu$, which takes the form
$dt \otimes \mu_0$ on $I \times U_0$. To a good pair of flowboxes, we can therefore assign
the number
\[{\textstyle
\frac{T}{T-T'}\frac{\mu_0(U_0)}{T'} = \frac{\mu(U)}{T'(T-T')}\,,
}
\]
which will play a significant role throughout this section. If this is not too large,
we think of the flowboxes as `sufficiently quadratic'.
\end{Remark}

In view of Lemma~\ref{Uedele}, we restrict attention
to the case where the Lie algebra is $\fg = C^{\infty}_{c}(I \times U_{0}, \fk)$, and
$\bv = \partial_{t}$.
For $z \in C^{\infty}_{c}(U',\C)$, we define $\xi(z) \in \fg$ by
\begin{equation}\label{eq:xiofz}
\xi(z)(t,u) := \begin{pmatrix}
             0& z(t,u)\\
		-\overline{z}(t,u) & 0
             \end{pmatrix}
\end{equation}
and note that $[\xi, \frac{\partial}{\partial t} \xi] =0$.
We also consider the element $\eta \in \fg$ defined by
\begin{equation}\label{kabouterspel}
\eta(t,u) := \chi(u)\begin{pmatrix}
             i\tau(t)& 0\\
		0 & -i\tau(t)
             \end{pmatrix}\,,
\end{equation}
where
$\tau \in C^{\infty}_{c}(I,\R)$ and $\chi \in C^{\infty}_{c}(U_0,\R)$ are such
such that $\tau(t) = t$ for $t \in I'$ and $\chi(u) = 1$ for $u \in U'_{0}$.
It also satisfies $[\tau, \frac{\partial}{\partial t} \tau] =0$.
Thus $\chi(u)\tau(t) = t$ on $U'$, hence in particular on the support of
every $z \in C^{\infty}_{c}(U',\C)$.

On $C^{\infty}_{c}(\R,\C)$, we define the usual scalar product
\[\langle f,g \rangle_{dt} := \int_{-\infty}^{\infty}\overline{f(t)}g(t) dt\]
and the Fourier transform
\[\widehat{f}(k) := \int_{-\infty}^{\infty}f(t)e^{-ikt} dt, \quad
k \in \R.\]
For $z\in C^{\infty}_{c}(U,\C)$, we will denote by $\widehat{z}(k,u)$
the `parallel' Fourier transform, i.e.\
the Fourier transform of $t \mapsto z(t,u)$ evaluated at $k$.

We can choose $\tau$ such that $\|\tau'\|_{dt}^2$
is arbitrarily close to $\frac{TT'}{T-T'}$, and we can
choose $0 \leq \chi \leq 1$ so that $\|\chi\|^2_{\mu_0} \leq \mu_0(U_0)$.
Thus $\|\chi\tau'\|^2_{\mu}$ can be chosen arbitrarily close
to $\frac{TT'}{T-T'}\mu_0(U_0)$. Therefore, for every
$\varepsilon > 0$, there exists an $\eta \in C^{\infty}_{c}(U,\C)$
as in \eqref{kabouterspel} satisfying
\begin{equation}\label{blijetuup}
\|L_{\bv}\eta\|^2_{\mu} = 2{\textstyle \frac{TT'}{T-T'}}\mu_0(U_0)
+ \varepsilon\,.
\end{equation}
For $\eta$ as in (\ref{kabouterspel}), Lemma~\ref{4pointGauge} yields the following estimate.
\begin{Proposition}\label{schatje}
Let $z \in C^{\infty}_{c}(U',\C)$, and let $k\in \R$ be such
that $\widehat{z}(k,u)= 0$ for all $u\in U'_{0}$. Then we have
\[
\big\langle i\dd\rho(\xi(z))\big\rangle_{\psi}^2 \leq
4 \|z\|^2_{\mu} \Big(\langle H\rangle_{\psi} - \half k \langle i\dd\rho(L_{\bv}\eta)\rangle_{\psi}
+ {\textstyle \frac{1}{8}} k^2 \|L_{\bv}\eta\|^2_{\mu}\Big)\,.
\]
\end{Proposition}
\begin{proof}
Since $[\xi(z),\xi(z)'] = 0$ and $[\eta,\eta'] = 0$, we may
apply Lemma~\ref{4pointGauge}. First, we
evaluate the left hand side of
inequality~(\ref{waterhoen}).
Since $\tau(t)\chi(u) = t$ on $\mathrm{supp}(\xi(z))$, we have
$\ad_{\eta}(\xi(z)) = \xi(2tiz)$. Since $L_{\bv}\xi(z) = \xi(z')$, we have
\[e^{-s\ad_\eta}(L_{\bv}\xi(z)) = \xi(z'e^{-2its})\,.\]
As $\kappa(\xi(z),L_{\bv}\eta) = 0$ for all $z\in C^{\infty}_{c}(U',\C)$, we have
\[\left\langle
\frac{e^{-s\ad_{\eta}} - \mathbf{1}}{\ad_{\eta}}(L_{\bv}\xi(z)),
L_{\bv}\eta\right\rangle_{\mu}
= \left \langle
\xi\left(\frac{e^{-2tis} - 1}{2it}z'\right),L_{\bv}\eta\right\rangle_{\mu} = 0\,.\]
On the right hand side of
inequality~(\ref{waterhoen}),
we have $\|L_{\bv}\xi(z)\|^{2}_{\mu} = 2\|z'\|^{2}_{\mu}$.
We thus obtain
\[
\Big\langle i\dd\rho\big(\xi(z'e^{-2ist})\big)\Big\rangle_{\psi}^2 \leq
4 \|z'\|^2_{\mu} \Big(\langle H\rangle_{\psi} + s \langle i\dd\rho(L_{\bv}\eta)\rangle_{\psi}
+ {\textstyle \frac{s^2}{2}} \|L_{\bv}\eta\|^2_{\mu}\Big)
\]
for all $s\in \R$ and $z \in C^{\infty}_{c}(U',\C)$.
Note that $w \in C^{\infty}_{c}(U',\C)$ is of the form $w = z'e^{-2ist}$
for some $z \in C^{\infty}_{c}(U,\C)$ if and only if the parallel Fourier transform
$\widehat{w}(k,u)$
vanishes for $k = -2s$.
Since in that case $\|w\|^2_{\mu} = \|z'\|^{2}_{\mu}$, the proposition follows.
\end{proof}
We thus obtain a $1$-parameter family of inequalities indexed by~$k\in \R$,
the case $k=0$ reducing to the Cauchy--Schwarz estimate because
${\widehat{z}(0,u) = 0}$ is equivalent to
$\xi(z) \in D\fg$.
The idea of the following proposition is
to lift the requirement that the Fourier transform vanish
by showing that every ${z \in C^{\infty}_{c}(U',\C)}$ can be written,
in a controlled way, as the
sum of two functions whose parallel Fourier transform vanishes
for some $k\in \R$.
%

\begin{Proposition}\label{aenb}
There exist $a, b \in \R$ such that, for all $z \in C^{\infty}_{c}(U',\C)$
with $U' = I' \times U'_0$ contained in $U = I \times U_0$,
we have
\begin{equation}\label{smeerkees}
\big\langle i\dd\rho(\xi(z))\big\rangle_{\psi}^2 \leq
(a + b \langle H \rangle_{\psi}) \|\xi(z)\|^2_{\mu}\,
\end{equation}
for constants $a$ and $b$ that depend on the interval lengths
$T = |I|$ and $T' = |I'|$ and on $\mu_{0}(U_0)$, but not on
$z$ or~$\psi$.
\end{Proposition}

\begin{proof}
Let $k$ be an arbitrary real number not equal to zero, and
choose a function $\zeta \in C^{\infty}_{c}(I',\C)$ with
$\widehat{\zeta}(0) \neq 0$ and $\widehat{\zeta}(k) = 0$.
(Such functions certainly exist.
 For instance, one can choose
$\zeta(t) = \alpha'(t) e^{ikt}$
for some $\alpha \in C^\infty_c(I',\R)$ with $\hat\zeta(0) = \hat{\alpha'}(-k)
= -ik \hat{\alpha}(-k) \not=0$.)
If we split $z$ into $z = z_0 + z_k$ with
\[z_k(t,u) := \widehat{z}(0,u)\widehat{\zeta}(0)^{-1}\zeta(t)
\quad \mbox{ and } \quad z_0 := z - z_{k},\]
 then
$\widehat{z}_{0}(0,u) = 0$ and $\widehat{z}_{k}(k,u) = 0$.
We apply Proposition~\ref{schatje} separately to both terms
on the right hand side  of
\begin{equation*}
|\langle i\dd\rho(\xi(z)) \rangle_{\psi}|
\leq
|\langle i\dd\rho(\xi(z_{0})) \rangle_{\psi}|
+
|\langle i\dd\rho(\xi(z_{k})) \rangle_{\psi}|
\end{equation*}
to obtain
\begin{equation}\label{fruup}
|\langle i\dd\rho(\xi(z)) \rangle_{\psi}|
\leq
2 \|z_0\|_{\mu}
\sqrt{\langle H \rangle_{\psi}}
+
2 \|z_k\|_{\mu}
\sqrt{\langle H \rangle_{\psi} +
{\textstyle \frac{k^2}{4} \frac{TT'}{T-T'}} \mu_0(U_0)}\,.
\end{equation}
Indeed, the term
$k\langle i \dd\rho(L_{\bv}\eta)\rangle_{\psi}$ can be assumed non-positive
by switching $k$ with $-k$ and $\zeta$ with $\overline{\zeta}$ if necessary.
The term $\|L_{\bv}\eta\|^2_{\mu}$ is then estimated by
(\ref{blijetuup}), and we take the limit $\varepsilon \downarrow 0$.

Since $|\widehat{z}(0,u)|^2 \leq T'\|z(\,\cdot\,, u)\|^2_{dt}$, we have
$\|\widehat{z}(0,\,\cdot\,)\|^2_{\mu_0} \leq T'\|z\|^2_{\mu}$.
It follows that
$\|z_{k}\|_{\mu}$ can be estimated in terms of $\|z\|_{\mu}$ as
\[
\|z_{k}\|_{\mu} = \|\widehat{z}(\,\cdot\,, 0)\|_{\mu_0}
\left\|\widehat{\zeta}(0)^{-1}\zeta\right\|_{dt}
\leq \sqrt{T'} \left\|\widehat{\zeta}(0)^{-1}\zeta\right\|_{dt} \|z\|_{\mu}\,.
\]
Similarly, $\|z_{0}\|_{\mu}$ can be estimated in terms of
$\|z\|_{\mu}$ by means of
\[ \|z_{0}\|_{\mu} \leq \|z\|_{\mu} + \|z_{k}\|_{\mu} \]
and the above estimate on $\|z_{k}\|_{\mu}$.
Substituting this into
(\ref{fruup}), we derive the estimate
\begin{equation}\label{troonpragger}
\langle i\dd\rho(\xi(z)) \rangle_{\psi}^2 \leq
4 \|z\|^2_{\mu}\left(1 + 2 \sqrt{T'} \|\widehat{\zeta}(0)^{-1}\zeta\|_{dt} \right)^2
\Big(
\langle H\rangle_{\psi} + {\textstyle \frac{k^2}{4} \frac{TT'}{T-T'}} \mu_0(U_0)
\Big)\,.
\end{equation}
Since $\|\xi(z)\|^2_{\mu} = 2\|z\|^2_{\mu}$, equation~(\ref{troonpragger})
is equivalent to (\ref{smeerkees})
with the constants
\begin{eqnarray}
a &:=& 2 \Big({\textstyle \frac{k^2}{4} \frac{TT'}{T-T'}} \mu_0(U_0)\Big)
\left(1 + 2 \sqrt{T'} \|\widehat{\zeta}(0)^{-1}\zeta\|_{dt} \right)^2\label{aaa}\\
b &:=& 2 \left(1 + 2 \sqrt{T'} \|\widehat{\zeta}(0)^{-1}\zeta\|_{dt} \right)^2\,.\label{beee}
\end{eqnarray}
This completes the proof.
\end{proof}

For $\xi(z)$ of the form \eqref{eq:xiofz} in a gauge algebra $\fg = C^{\infty}_{c}(U',\fk)$ with
$\fk = \su(2,\C)$,
we can now prove an operator inequality
of the form \eqref{eq:goal}.
\begin{Proposition}\label{unboundedtwee}
There exist constants $a,b \in \R$, depending on $T$, $T'$ and $\mu_{0}(U_0)$,
such that for all $\alpha, \beta$ with $\alpha^2 \geq a$ and $2\alpha\beta \geq b$, we have
\begin{equation}\label{brombaard}
\pm i\dd\rho(\xi(z)) \leq\|\xi(z)\|_{\mu} (\alpha\one +\beta H)
\quad \mbox{ for } \quad
z \in C^\infty_c(U',\C)
\end{equation}
as an inequality of unbounded operators on $\cH$ with domain containing
$\cH^{\infty}$. 
\end{Proposition}

\begin{proof}
Note that the inequality
\eqref{brombaard} is equivalent to
\[
\langle \psi, i\dd\rho(\xi(z)) \psi\rangle^2
\leq\|\xi(z)\|^{2}_{\mu} \langle \psi, (\alpha\one + \beta H) \psi\rangle^2
\quad \mbox{ for all } \quad \psi \in \cH^{\infty}.\]
As $\beta^2\langle\psi, H\psi \rangle^2 \geq 0$,
this follows from
Proposition~\ref{aenb}
under the above conditions on $\alpha$ and $\beta$.
%
\end{proof}

\begin{Remark}
The estimate \eqref{brombaard} is rather crude for large energies, in the sense that one expects
$\dd\rho(\xi) \sim \sqrt{H}$, not $\dd\rho(\xi) \sim H$.
\end{Remark}

It will be convenient to
gain more control over the constants $a$ and $b$
in Proposition \ref{aenb}, and
the constants $\alpha,\beta$ in Proposition \ref{unboundedtwee}.
For this, we need to
remove the dependence on $\zeta$ in \eqref{aaa} and \eqref{beee}.

\begin{Proposition}\label{prop:kieseend}
The constants $a$ and $b$ in {\rm Proposition~\ref{aenb}} can be chosen as
\begin{eqnarray}
a &=&
{\textstyle \frac{T}{T-T'}} \left(\frac{\mu_0(U_0)}{T'}\right) \nu^2\, b,
\quad\text{with}
\label{aaaeeen}\\
b &=& 2 \Bigg(1 + \frac{2}{\sqrt{1- (\sin(\nu)/\nu)^2}} \Bigg)^2\,.
\label{beeetweee}
\end{eqnarray}
Here 
$\nu >0 $ can be chosen at will.
\end{Proposition}

\begin{Remark}\label{rk:choosenu}
It will be convenient to choose  $\nu = \pi$.
Then $b$ attains its minimal value
$b=18$, and $a = 18\pi^2 \frac{T}{T-T'}\frac{\mu_0(U_0)}{T'}$.
\end{Remark}

\begin{proof}
In \eqref{aaa} and \eqref{beee},
we need to minimize the expression
$\sqrt{T'} \|\widehat{\zeta}(0)^{-1}\zeta\|_{dt}$
over all $\zeta \in C^{\infty}_c(I',\C)$ with $\widehat{\zeta}(k) = 0$ and $\widehat{\zeta}(0)\neq 0$,
where $k \in \R^\times$ is arbitrary.
Since $\widehat{\zeta}(k) = \lra{e^{ikt},\zeta}_{dt}$ and
$\widehat{\zeta}(0) = \lra{1,\zeta}_{dt}$, this amounts to maximizing
\[
	F(\zeta) := \big(\sqrt{T'} \|\widehat{\zeta}(0)^{-1}\zeta\|_{dt}\big)^{-1} =
	\frac{|\lra{1,\zeta}_{dt}|}{\|1\|_{dt} \|\zeta\|_{dt}}\,.
\]
Since $F$ is continuous on $L^2(I') \setminus \{0\}$, and
$C^{\infty}_c(I',\C)$ is dense in $L^2(I')$, $F(\zeta_{\mathrm{max}})$ is maximal on
the projection $\zeta_{\mathrm{max}}$ of $1$ on the orthogonal complement
of the function $e^{ikt} \in L^2(I')$.
This is essentially a two-dimensional problem in the space spanned by
$e_0 := \frac{1}{\sqrt{T'}}1$ and $e_k := \frac{1}{\sqrt{T'}}e^{ikt}$,
with
\begin{equation}\label{snorhoofd}
\textstyle
\lra{e_0,e_0} = \lra{e_k,e_k} = 1 \quad\text{and}\quad
\lra{e_0,e_k} = \frac{\sin(kT'/2)}{kT'/2}\,.
\end{equation}
It follows that $\zeta_{\mathrm{max}} =
e_0 - \lra{e_k,e_0}e_k$, and
$F(\zeta_{\mathrm{max}}) =
\sqrt{1-|\lra{e_0,e_k}|^2}$.
Using \eqref{snorhoofd}, we find
\begin{equation}\label{eq:flathoofd}
F(\zeta_{\mathrm{max}}) = \sqrt{1-\left(\textstyle\frac{\sin(kT'/2)}{kT'/2}\right)^2}\,.
\end{equation}
Equations \eqref{aaaeeen} and \eqref{beeetweee}
are now obtained from \eqref{aaa} and \eqref{beee}
with $k = 2 \nu /T'$
by substituting the maximal value \eqref{eq:flathoofd}
for
$F(\zeta) = \big(\sqrt{T'} \|\widehat{\zeta}(0)^{-1}\zeta\|_{dt})^{-1}$.
\end{proof}

\subsubsection{The local gauge algebra with compact simple fibers}

We now extend
Proposition~\ref{aenb}
to the case
where $\fk$ is an arbitrary compact simple Lie algebra.
With $I' \times U'_{0}$ and $I \times U_0$
a good pair of flowboxes  (cf. Definition~\ref{def:goodpair}),
we consider
$\fg_{U'} := C^{\infty}_{c}(I' \times U_0',\fk)$
and $\fg_{U} := C^{\infty}_{c}(I \times U_0,\fk)$
as subalgebras of $\fg = \Gamma_{c}(M,\fK)$.

\begin{Lemma}\label{nucompact}
Let $\dd\rho$ be a positive energy representation of $\hat\g$, and let $\eta >0$.
Then we have
\begin{equation}\label{zebravistwee}
\pm i\dd\rho(\xi) \leq \|\xi\|_{\mu} \,\big(K(\eta)\one +  \eta H\big)
\quad \mbox{ for all } \quad \xi \in \fg_{U'},
\end{equation}
where $K(\eta)$ is a constant independent of $\xi$. More precisely,
\begin{equation}\label{eq:waaitrug}
K(\eta) = \mathrm{max}\Big(9d_{\fk}/\eta,3\pi\sqrt{2d_{\fk}
{\textstyle \frac{T}{T-T'} \frac{\mu_0(U_0)}{T'} }}\,\Big)\,,
\end{equation}
where $d_{\fk}$ is the dimension of $\fk$. 
\end{Lemma}

\begin{proof} Using the root decomposition of $\fk_\C$ with respect to
the complexification $\ft_\C$ of a maximal abelian subalgebra
$\ft \subeq \fk$, one obtains a basis $(X_1, \ldots, X_{d_{\fk}})$ of $\fk$
with $\kappa(X_i,X_j) = 2\delta_{ij}$,
where $-\kappa$ is the Killing form of $\fk$ and such that every $X_j$
is contained in some $\su(2,\C)$-triple in $\fk$ (\cite[Prop.~6.45]{HM98}).
Every $\xi \in \fg_{U'}$ can then be written as
$\xi = \sum_{i} f_iX_i$ for $f_i\in C^{\infty}_{c}(U'_0\times I',\R)$.
Since every $X_i$ is contained in an $\mathfrak{su}(2,\C)$-triple,
we can apply Proposition~\ref{unboundedtwee}
to $f_i X_{i} \in \fg_{U'}$ with $z = f_i$. We obtain
$\pm i \dd\rho(fX_i) \leq \|fX_i\|_{\mu}(\alpha\one + \beta H)$,
and thus
\[\pm i \dd\rho(\xi) \leq \Big(\sum_{i=1}^{d_{\fk}}\|fX_i\|_{\mu}\Big)(\alpha\one + \beta H)\,.\]
As the different terms $f_iX_i$ are orthogonal, we have
$
 \sum_{i=1}^{d_{\fk}}\|f_iX_i\|_\mu \leq \sqrt{d_{\fk}} \|\xi\|_{\mu}
$,
and we obtain
\begin{equation} \label{eq:transvet}
\pm i \dd\rho(\xi) \leq \|\xi\|_{\mu}(\sqrt{d_{\fk}}\alpha\,\one + \sqrt{d_{\fk}}\beta H)\,.
\end{equation}
By Proposition~\ref{unboundedtwee},
we are allowed to choose any
$\alpha$ and $\beta$ with
$\alpha^2\geq a$ and $2\alpha\beta \geq b$.
Following Remark~\ref{rk:choosenu}, we take
$a = 18 \pi^2 \lambda (\mu_0(U_0)/T')$ and $b=18$.
The inequality \eqref{eq:transvet} therefore holds for any value of
$\beta >0$ if we set
\begin{equation}\label{eq:aapteek}
\alpha = \mathrm{max}\Big(9/\beta,3\pi\sqrt{2{\textstyle \frac{T}{T-T'} \frac{\mu_0(U_0)}{T'}} }\,\Big)\,.
\end{equation}
Inequality \eqref{zebravistwee} now follows from
\eqref{eq:transvet} with $\beta = \eta/\sqrt{d_{\fk}}$
and $K(\eta) = \sqrt{d_{\fk}}\alpha$.
\end{proof}

\begin{Proposition}\label{prop:explsemibound}
For all $\xi \in \fg_{U'}$ and $t >0$, the spectrum of
$tH \pm i \dd\rho(\xi)$ is bounded below.
More precisely,
\begin{equation}\label{zondagsmugtwee}
- \mathrm{max}\Big( 9d_{\fk}\|\xi\|^2_{\mu}/t,3\pi\|\xi\|_{\mu}\sqrt{2d_{\fk} {\textstyle \frac{T}{T-T'} \frac{\mu_0(U_0)}{T'} }}\,\Big) \leq \inf\big(\Spec(t H \pm i\dd\rho(\xi))\big)\,.
\end{equation}
\end{Proposition}
\begin{proof}
If $\|\xi\|_{\mu} = 0$, then $\dd\rho(\xi) = 0$ by Corollary~\ref{zondagskind}.
In that case, (\ref{zondagsmugtwee}) simply follows from the fact that
$H$ has non-negative spectrum.
If $\|\xi\|_{\mu} \neq 0$, we apply Lemma~\ref{nucompact} with
$\eta = t/\|\xi\|_{\mu}$.
\end{proof}

\subsubsection{Global estimates and the bounding function}

We need to derive estimates of the type \eqref{zebravistwee}
globally, on the full Lie algebra $\Gamma_{c}(M,\fK)$ rather than merely on
$C^{\infty}_{c}(U,\fk)$.
In this section, we show how to do this for compact as well as noncompact manifolds
$M$, under the assumption that $\bv_M$ is nowhere vanishing.

For compact manifolds, we will derive an estimate of the form \eqref{zebravistwee},
albeit with a larger constant $K(\eta)$.
%
In the noncompact case, however, the expression $\|\xi\|_{\mu} K(\eta)$ in
\eqref{zebravistwee} will have to be replaced by $\|\xi\|_{B_{\eps}\mu}$,
where $B_{\eps} \colon M \rightarrow \R^{+}$ is a suitable
upper semi-continuous
function on $M$ that is invariant under the flow,
and $\|\xi\|_{B_{\eps}\mu}$
is the $L^2$-norm of $\xi \in \Gamma_{c}(M,\fK)$ with respect to the measure
$B_{\eps} \mu$,
\begin{equation}\label{eq:Binp}
\|\xi\|^{2}_{B_{\eps}\mu} = \lra{\xi,\xi}_{B_{\eps}\mu}, \qquad
\lra{\xi,\eta}_{B_{\eps}\mu} = \int_{M} \kappa(\xi,\eta) B_{\eps}(m)d\mu(m)\,.
\end{equation}
In this setting, we will prove that
\begin{equation}
  \label{eq:cc}
\pm i\dd\rho(\xi) \leq  \|\xi\|_{B_{\eps}\mu}\one +  \eps\|\xi\|_{\mu} H \quad \mbox{ for all } \quad
\xi \in \Gamma_{c}(M,\fK).
\end{equation}

Note that, since $\bv_M$ is nowhere vanishing on $M$,
every $m\in M$ is contained in a good pair of flow boxes
in the sense of Definition~\ref{def:goodpair}.
\begin{Definition}\label{def:defb}
For $m\in M$, define
$b(m)$ as the infimum of the set of numbers
$\frac{T}{T-T'}\frac{\mu_0(U_{0})}{T'}$, ranging over all good pairs of flowboxes
$U' \Subset U$ containing~$m$.
\end{Definition}

\begin{Proposition}\label{Prop:rafelfluitje}
The function $b \colon M \rightarrow \R^{+}$ is invariant under the flow
$(\gamma_{M,t})_{t\in \R}$. Further, it is
upper semi-continuous, hence in particular measurable.
\end{Proposition}
\begin{proof}
The invariance under the flow follows from the fact that
$U' \subset U$ is a good pair of flow boxes around $m$ if and only if
$\gamma_{M,t}(U') \subset \gamma_{M,t}(U)$ is a good pair of flow boxes
around $\gamma_{M,t}(m)$.
For the upper semi-continuity, note that 
for every $\eps >0$, there is a good pair of flowboxes
$U'\subset U$ around $m$ such that
$\frac{T}{T-T'}\frac{\mu_0(U_0)}{T'} < b(m) + \eps$.
For every $m'$ in the open
neighborhood $U'$ of $m$, we thus have  $b(m') \leq b(m) + \eps$.
\end{proof}

\begin{Theorem}\label{thm:lapje}
Let $\dd\rho$ be a positive energy representation of $\widehat{\fg}$, and let $\eps >0$.
Then we have
\begin{equation}\label{eq:laatstesnok}
\pm i \dd\rho(\xi) \leq \|\xi\|_{B_{\eps}\mu}\one + \eps \|\xi\|_{\mu} H
\quad \mbox{  for every } \quad \xi \in \Gamma_{c}(M,\fK).
\end{equation}\label{eq:lillekerd}
Here $B_{\eps} \colon M \rightarrow \R^+$\index{upper semi-continuous density \scheiding $B_{\eps}$} is the upper semi-continuous function defined by
\begin{equation}\label{eq:defBE}
B_{\eps}(m) := \max\Big(
81d^2_{\fk}(d_{M}+1)^4)/\eps^2,
18 \pi^2d_{\fk}(d_{M}+1)^2b(m)
\Big)\,,
\end{equation}
with $b(m)$ as in {\rm Definition~\ref{def:defb}}. It is invariant under the flow
$(\gamma_{M,t})_{t\in \R}$.
\end{Theorem}
\begin{proof}
Let $d$ be a Riemannian metric on $M$ for which $M$ is complete,
so that closed bounded subsets of $M$ are compact by the Hopf--Rinow Theorem.
Let $V \subseteq M$ be the compact support of
$\xi \in \Gamma_{c}(M,\fK)$.
Since $b$ is upper semi-continuous, the functions $\beta_{n} \colon M \rightarrow \R^{+}$
defined by
\[\beta_{n}(m) := \sup \{b(m')\,;\, d(m,m') \leq 1/n \}\]
constitute a decreasing sequence converging pointwise to $b$ as $n \rightarrow \infty$.
We now show that the functions $\beta_n$ are upper semi-continuous.
To see this, note that, for every $m_0 \in M$ and every $\epsilon>0$,
there exists an $n \in \N$ with
$b(m) < b(m_0) + \varepsilon/2$ for $m$ in the closed ball
$\overline{W}_{2/n}(m_0)$ with radius $2/n$ around $m_0$.
Since this ball is compact, it contains finitely many $m_i$
such that it is covered by  open neighborhoods $\mathcal{O}_i$ of $m_i$
such that $b(m) \leq b(m_i) + \varepsilon/2$ for all $m\in \mathcal{O}_i$.
If $d(m,m_0)< \frac{1}{n}$, then
$W_{1/n}(m) \subseteq \bigcup_{i}\mathcal{O}_{i}$, so that
$\beta_n(m) < \beta_n(m_0) + \varepsilon$.
In particular, $\beta_n$ is measurable, and bounded on the compact set $V$.
%

For every $n\in \N$, choose
a cover of $V$ by finitely many open balls $W_{r_i}(m_i)$ of radius
$r_i \leq 1/n$ around $m_i$, with the
property that $W_{r_i}(m_i) \subseteq U' \Subset U$ for a good pair $U'\Subset U$ of flow boxes
with $\frac{T}{T-T'}\frac{\mu_0 (U_{0})}{T} \leq b(m_{i}) + 1/n$.
Since $b(m_i) \leq \beta_{n}(m)$ for all $m\in W_{r_i}(m_i)$, it follows that
\begin{equation}\label{eq:flapmuts}
\textstyle{\frac{T}{T-T'}\frac{\mu_0 (U_{0})}{T}} \leq \beta_{n}(m) + 1/n
\qquad
\text{for all}
\qquad
m\in W_{r_i}(m_i)\,.
\end{equation}
By the Brouwer--Lebesgue Paving
Principle \cite[Thm.~V1]{Hurewicz41}, there exists a finite subcover $(W_{j})_{j\in J}$
with the property that every point $m \in V$
is contained in at most $d_{M} + 1$ sets.

Let $\phi_j$ be a partition of unity with respect to $(W_{j})_{j\in J}$.
By Lemma~\ref{nucompact}, applied to $\eta := \eps/(d_{M} +1)$, we obtain
$\pm i \dd\rho(\phi_j \xi) \leq \|\phi_j \xi\|_{\mu} (K_j(\eta)\one + \eta H)$, where $K_j(\eta)$ is given by
\eqref{eq:waaitrug} for a good pair of flowboxes $U' \Subset U$ containing~$W_j$.
From \eqref{eq:waaitrug} and \eqref{eq:flapmuts}, we find that
\[
B_{n,\eta}(m) := \max\Big( (9d_{\fk}/\eta)^2, 18 \pi^2 d_{\fk}(\beta_{n}(m) + 1/n) \Big) \geq
K_j(\eta)^2
\quad
\text{for all}
\quad
m\in W_{j}\,.
\]
As $\|\phi_{j}\xi\|_{\mu} K(\eta) \leq \|\phi_j \xi\|_{B_{n,\eta} \mu}$, we have
$\pm i \dd\rho(\phi_j\xi) \leq \|\phi_j\xi\|_{B_{n,\eta} \mu} \one + \eta \|\phi_j \xi\|_{\mu} H$ for all $j\in J$,
and thus
\begin{equation}
\pm i \dd\rho(\xi) \leq \Big(\sum_{j\in J} \|\phi_j\xi\|_{B_{n,\eta}\mu}\Big) \one +
\eta\,\Big(\sum_{j\in J} \|\phi_j\xi\|_{\mu}\Big) H\,.
\end{equation}
Since $\|(\phi_j \xi)(m)\|_{\kappa} \leq \|\xi(m)\|_{\kappa}$, and since
at most $d_{M}+1$ of the values $\phi_j(m)$ are nonzero,
it follows that
\[
\sum_{j\in J} \|\phi_j \xi\|_{\mu} \leq (d_{M}+1)\|\xi\|_{\mu}
\quad\text{and}\quad
\sum_{j\in J} \|\phi_j \xi\|_{B_{n,\eta}\mu} \leq (d_{M}+1)\|\xi\|_{B_{n,\eta}\mu}\,,
\]
so that
\begin{equation}\label{eq:plukjehaar}
\pm i \dd\rho(\xi) \leq (d_{M} +1) (\|\xi\|_{B_{n,\eta}\mu}\one +
 \eta \|\xi\|_{\mu} H)\,.
\end{equation}
To obtain \eqref{eq:laatstesnok} from \eqref{eq:plukjehaar},
recall that $\eta := \eps/(d_{M} +1)$.
The second term on the right is thus
$(d_{M}+1)\eta \|\xi\|_{\mu} = \eps \|\xi\|_{\mu}$, as required.
For the first term, note that
$\beta_n + 1/n$ is a bounded, decreasing sequence converging pointwise
to $b$ on $V$.
The bounded, decreasing sequence
$(d_{M}+1)^2B_{n,\eta}(m)$ thus converges to
$B_{\eps}(m)$ in \eqref{eq:laatstesnok}, where $\eps = (d_{M}+1)\eta$.
By the Dominated Convergence Theorem, we find that,
for $n \rightarrow \infty$, the squared norm
$((d_M + 1)\|\xi\|_{B_{n,\eta}})^2$
approaches
\[
 \int_{V}\|\xi\|^2_{\kappa} (d_M+1)^2B_{n,\eta}\, d\mu(m)
\rightarrow \int_{V}\|\xi\|^2_{\kappa}B_{\eps}\, d\mu(m)
= \|\xi\|_{B_{\eps}\mu}^2\,.
\]
Since \eqref{eq:plukjehaar} holds for every $n\in \N$, the proposition follows.
\end{proof}

Note that if the function $b \colon M \rightarrow \R^+$ of
Definition~\ref{def:defb} is bounded, then
so is $B_{\eps}$. If we define
$K(\eps)^2 := \|B_{\eps}\|_{\infty}$,
then we recover
the inequality
\begin{equation}\label{eq:cptbound}
\pm i \dd\rho(\xi) \leq \|\xi\|_{\mu}(K(\eps)\one + \eps H)\,,
\end{equation}
since $\|\xi\|_{B_{\eps}\mu} \leq K(\eps)\|\xi\|_{\mu}$.
This happens in particular if $M$ is compact because the
upper semi-continuous function $B_{\eps}$ is then automatically bounded.

\begin{Corollary}
Suppose that $M$ is compact and~$\bv_M$ is nowhere vanishing on~$M$.
Then, for every $\eps >0$, there exists a constant $K(\eps)>0$ such that
\eqref{eq:cptbound} holds for all $\xi \in \Gamma(M,\fK)$.
\end{Corollary}
Another important situation in which $B_{\eps}$ is bounded is for
product manifolds of the form $M = \R \times \Sigma$.

\begin{Corollary}\label{CorBEforR}
Suppose that $M \simeq \R \times \Sigma$ with
$\bv_M = \frac{\partial}{\partial t}$.
Then the inequality \eqref{eq:cptbound} holds for the compactly supported gauge algebra
$\fg = \Gamma_{c}(M,\fK)$,
with constant $K(\eps) = 9d_{\fk}(d_{M}+1)^2/\eps$
depending on $M$ and $\fK$ only through the dimension.
\end{Corollary}
\begin{proof}
For $(t,x) \in \R \times \Sigma$, choose $U'_{0} \Subset U_0 \subseteq \Sigma$
with $U_0 \subseteq \Sigma$ relatively compact, and
$x \in U'_{0}$. For $T'$ sufficiently large, $(t,x)$
is contained in the good pair of flowboxes $U'= U'_0 \times (-T'/2,T'/2)$,
and $U = U_0 \times (-T/2,T/2)$ for $T = 2T'$.
Since $\frac{T}{T - T'} \frac{\mu_0(U_0)}{T'} = 2\mu_0(U_0)/T'$ approaches
$0$
for $T'\rightarrow \infty$, it follows that $b(t,x) =0$.
In particular, $B_{\eps}(m) = 81d^2_{\fk}(d_{M}+1)^4/\eps^2$ is constant,
and the result follows.
\end{proof}


\subsection{Extending representations to Sobolev spaces}\label{sec:nummer3van6}

In this section,
we extend the map $\dd\rho$ to the Hilbert completion
$L^2_{B\mu}(M,\fK)$\index{gauge algebra!30@quadratic \scheiding $L^2_{B\mu}(M,\fK)$}
of $\fg/I_{\mu}$
with respect to the inner product \eqref{eq:Binp}
corresponding to $B_{\eps} \mu$.

Note that since $\|\xi\|_{\mu}$ is dominated by a multiple of
$\|\xi\|_{B_{\eps}\mu}$, the inner product $\lra{\xi,\eta}_{\mu}$
is continuous on $L^2_{B\mu}(M,\fK)$.
As the difference between $\|\xi\|_{B_{\eps}\mu}$ and
$\|\xi\|_{B_{\widetilde{\eps}}\mu}$ for $\eps, \widetilde{\eps} >0$ is a multiple of
$\|\xi\|_{\mu}$, the space $L^2_{B\mu}(M,\fK)$ and its topology are
independent of $\eps$. (This is why we omit $\varepsilon$ from the notation in
$L^2_{B\mu}(M,\cK)$.)

\subsubsection{The completion $L^2_{B\mu}(M,\fK)$ in $L^2$-norm}
We use Theorem~\ref{thm:lapje} to
extend $\dd\rho$ from $\fg$ to $L^2_{B\mu}(M,\fg)$.
Define\index{Hamiltonian!shifted \scheiding $\Gamma_{\xi}$, $A$}
\[
\Gamma_{\xi} = \|\xi\|_{B_{\eps}\mu}\one + \eps \|\xi\|_{\mu}H\,,
\]
and note that its domain $\cD(\Gamma_{\xi})$ is contained in the domain $\cD(H)$ of $H$.
With this notation, \eqref{eq:laatstesnok} becomes
\begin{equation}\label{eq:ineqgamma}
0 \leq \Gamma_{\xi} \pm i\dd\rho(\xi)\,,
\end{equation}
as an inequality of unbounded operators on $\cH^{\infty}$.
Further, define
\begin{equation}
  \label{eq:defA}
A := \one + H \quad \mbox{ with } \quad \cD(A) = \cD(H).
\end{equation}
%


\begin{Proposition}\label{zwak}
Let $0<\eps\leq 1$.
There exists a map $r$ from $L^2_{B\mu}(M,\fK)$ to
the unbounded, skew-symmetric operators on $\cH$ such that
$\cD(r(\xi))$ contains $\cD(H)$ for all $\xi \in L^2_{B\mu}(M,\fK)$,
$r(\xi)|_{\cH^{\infty}} = \dd\rho(\xi)$ for all $\xi \in \fg$, and,
for all $\psi \in \cD(H)$,
the functional
\[ L^{2}_{B\mu}(M,\fK)\rightarrow \C\quad \mbox{  defined by } \quad
\xi \mapsto \langle r(\xi) \rangle_{\psi} \]
is continuous.
%
Furthermore, there exists a continuous map
\[ \lambda \colon L^2_{B\mu}(M,\fK) \rightarrow B(\cH) \]
into the bounded operators such that
$\|\lambda(\xi)\| \leq \|\xi\|_{B_{\eps}\mu}$, $\lambda(\xi)$ is skew-hermitian,
$\lambda(\xi)$ leaves the domain of $A^{1/2}$ invariant, and
\[r(\xi) = A^{1/2} \lambda(\xi)A^{1/2},\]
as an equality of unbounded operators on $\cD(H)$.
%
\end{Proposition}
\begin{proof}
Let $\xi_n$ be a sequence in $\g/I_{\mu}$
with $\|\xi - \xi_n\|_{B_{\eps}\mu} \rightarrow 0$,
and hence also $\|\xi - \xi_n\|_{\mu} \rightarrow 0$.
Without loss of generality, we assume that
$\|\xi - \xi_{n}\|_{B_{\eps}\mu} \leq  \frac{1}{2}$
and  $\eps \|\xi - \xi_{n}\|_{\mu} \leq \frac{1}{2}$ for all $n$, so that
\begin{equation}\label{eq:schathalf}
\Gamma_{\xi} - \Gamma_{\xi_n} + A \geq {\textstyle \frac{1}{2}}A\,.
\end{equation}
Define the sesquilinear forms
\[
B^{\pm}_n \colon \cH^{\infty} \times \cH^{\infty} \rightarrow \C,
\qquad
B^{\pm}_{n}(\psi,\chi):= \langle \psi,
\big((\Gamma_{\xi} + A) \pm i\dd\rho(\xi_{n})\big)\chi\rangle\,.\]
The forms $B^{\pm}_n$ are positive definite;
combining \eqref{eq:schathalf} with
inequality \eqref{eq:ineqgamma} applied to $\xi_n$, we find
\begin{equation} \label{eq:bn-esti}
B^{\pm}_n(\psi,\psi) \geq \lra{\psi, (\Gamma_{\xi} - \Gamma_{\xi_n}+ A) \psi}
\geq {\textstyle\frac{1}{2}}\lra{\psi, A\psi} \geq {\textstyle \frac{1}{2}}\|\psi\|^2.
\end{equation}
By \eqref{eq:ineqgamma} and the convergence of $\xi_n$, we find that
$B^{+}_{n}(\psi,\psi)$ is a Cauchy sequence
for every $\psi \in \cH^{\infty}$,
\[ |B^+_n(\psi,\psi) - B^+_m(\psi,\psi)|
= |\lra{\psi, i\dd\rho(\xi_n - \xi_m)\psi}|
\leq \lra{\psi, \Gamma_{\xi_n - \xi_m} \psi} \rightarrow 0\,.
\]
It follows that
$B^+(\psi,\chi) := \lim_{n\rightarrow \infty} B^+_n(\psi,\chi)$
defines a
positive definite, sesquilinear form $\cH^{\infty}\times \cH^{\infty} \rightarrow \C$.
Here we use that the estimate \eqref{eq:bn-esti} is independent of~$n$.
The same argument applies to $B^-(\psi,\chi) := \lim_{n\rightarrow \infty} B^-_n(\psi,\chi)$.
Note that
\begin{equation}\label{eq:vlotschrift}
 {\textstyle\frac{1}{2}}\langle\psi,A\psi\rangle
\leq B^{\pm}(\psi,\psi)\leq
\lra{\psi, (2\Gamma_{\xi} + A)\psi}
\leq c_{\xi} \lra{\psi, A \psi}
\end{equation}
for some $c_{\xi} >0$.
The forms $B^{\pm}$ therefore extend uniquely to closed, sesquilinear forms
$\overline{B}{}^{\pm} \colon \cD(A^{1/2}) \times \cD(A^{1/2}) \rightarrow \C$.
In turn, the forms $\overline{B}{}^{\pm}$ define a Friedrichs extension;
a closed, possibly unbounded positive operator $b^{\pm}(\xi) \colon \cD(H) \rightarrow \cH$,
such that
$\overline{B}{}^{\pm}(\psi,\chi) = \langle\psi,b^{\pm}(\xi)\chi\rangle$ for all
$\psi,\chi \in \cD(H)$
(cf.~\cite[App.~I.A.2]{GJ87}). Set
\[ r(\xi) := \frac{1}{2i}(b^+(\xi) - b^-(\xi)).\]
Since $b^+(\xi)$ and $b^-(\xi)$ are selfadjoint, $r(\xi)$ is skew-symmetric.
If $\xi \in \fg$, then
$\lra{\psi,r(\xi)\chi} = \lra{\psi,\dd\rho(\xi)\chi}$
for all $\psi, \chi \in \cH^{\infty}$,
so $r(\xi)$ is an extension of $\dd\rho(\xi)$.

Define
\[ \lambda(\xi) \colon D(A^{1/2}) \rightarrow  D(A^{1/2}), \qquad
	\lambda(\xi) := A^{-1/2}r(\xi) A^{-1/2}\,.
\]
Then, for $\psi, \chi \in D(A^{1/2})$, we have $A^{-1/2}\psi, A^{-1/2}\chi \in \cD(H)$.
Therefore,
\begin{equation}\label{eq:pootje}
\lra{\psi, \lambda(\xi) \chi} =
-\lra{\lambda(\xi)\psi,\chi} =
{\textstyle \frac{1}{2i}}(\overline{B}{}^+ - \overline{B}{}^{-})(A^{-1/2}\psi, A^{-1/2}\chi)\,.
\end{equation}
By \eqref{eq:vlotschrift} and Cauchy--Schwarz, we have
\[ |\ol{B}{}^{\pm}(\psi,\chi)| \leq c_{\xi}
\|A^{1/2}\psi\|\|A^{1/2}\chi\|, \]
so $|\lra{\psi, \lambda(\xi) \chi}| \leq c_{\xi}\|\psi\|\|\chi\|$
by \eqref{eq:pootje}. Therefore,
$\lambda(\xi)$ extends to a hermitian operator on $\cH$.
As such, the operator norm $\|\lambda(\xi)\|$ is the supremum of $|\lra{\psi,\lambda(\xi)\psi}|$
over all $\psi$ in the unit sphere of $\cH$.
For $\psi \in A^{1/2}\cH^{\infty}$,
\eqref{eq:ineqgamma}
yields
\begin{equation}\label{eq:druipsteensmurf}
|\lra{\psi, \lambda(\xi) \psi}| = \lim_{n \rightarrow \infty}
|\lra{A^{-1/2}\psi, \dd\rho(\xi_n)A^{-1/2}\psi}|
\leq \lra{\psi, A^{-1/2}\Gamma_{\xi} A^{-1/2} \psi}\,.
\end{equation}

We claim that
\begin{equation} \label{eq:a-est}
  A^{-1/2}\Gamma_{\xi}{A}^{-1/2} \leq \|\xi\|_{B_{\eps}\mu}
  \quad \mbox{ for } \quad 0<\eps<1.
\end{equation}
In fact, since
$\Gamma_{\xi}$ and $A$ commute, this is equivalent to
$\Gamma_{\xi} \leq A \|\xi\|_{B\mu}$, which in turn is equivalent to
\[ \|\xi\|_{B_\epsilon\mu} \one + \eps \|\xi\|_{\mu} H \leq \|\xi\|_{B_\epsilon\mu} (\one + H) \]
and this to $\eps \|\xi\|_{\mu} \leq \|\xi\|_{B_\epsilon\mu}$, which,
for $\eps < 1$, follows from the estimate ${B_{\eps}\geq 81 d_\fk^2(d_M +1)^4/\eps^2 > 1}$.
With \eqref{eq:a-est}, we find
\begin{equation}
|\lra{\psi, \lambda(\xi) \psi}|
\leq \|\xi\|_{B_{\eps}\mu}\|\psi\|^2
\quad \text{for} \quad
\psi \in A^{1/2}\cH^{\infty}
\,.
\end{equation}
To prove that $\|\lambda(\xi)\| \leq \|\xi\|_{B_{\eps}\mu}$, it therefore suffices to
show that $A^{1/2}\cH^{\infty}$ is dense in $\cH$.
First, we show that $A \cH^{\infty}$ is dense in $\cH$.
Since
$\exp(it A) = e^{it}\exp(itH)$ leaves the space
$\cH^{\infty}$ of smooth vectors invariant, the restriction
$A_0$ of $A$ to $\cH^{\infty}$ is essentially self-adjoint
\cite[\S VIII.4]{RS75}.
Suppose that
$\psi \perp A_0 \cH^{\infty}$. Then $\psi \in \cD(A^{*}_{0}) = \cD(A)$,
and $A_0^{*} \psi = A \psi = 0$.
Since $A$ is injective, $\psi = 0$
and $A \cH^{\infty}$ is dense in $\cH$.
Applying the contraction $A^{-1/2}$, we find that
$A^{1/2} \cH^{\infty}$ is dense in $A^{-1/2}\cH$.
Since $A^{-1/2}\cH = \cD(A^{1/2})$ is dense in $\cH$, we conclude that
$A^{1/2} \cH^{\infty}$ is dense in $\cH$, as required.
\end{proof}

For $s\in \R$, denote by
$\cH_{s}$ the Hilbert space completion of $\mathcal{D}(A^{s})$
with respect to the inner product
\begin{equation}
\lra{\psi,\chi}_{s} = \lra{A^{s}\psi,A^{s}\chi}\,.
\end{equation}
Denote the corresponding norm by $\|\psi\|_{s} = \|A^{s}\psi\|$,
and denote the norm of a continuous
operator $A \colon \cH_{s} \rightarrow \cH_{t}$ by $\|A\|_{s,t}$.
As $r(\xi) = A^{1/2}\lambda(\xi)A^{1/2}$ with
$\|\lambda(\xi)\| \leq \|\xi\|_{B_{\eps}\mu}$, the operator
$r(\xi) \colon \cD(A) \rightarrow \cH$ extends to a bounded operator
$r(\xi) \colon \cH_{1/2} \rightarrow \cH_{-1/2}$,
with
\begin{equation}\label{eq:grondstof}
	\|r(\xi) \psi\|_{-1/2} \leq  \|\xi\|_{B_{\eps}\mu} \|\psi\|_{1/2}\,.
\end{equation}
We thus have $\|r(\xi)\|_{1/2,-1/2} \leq \|\xi\|_{B_{\eps}\mu}$.


\subsubsection{The completion in Sobolev norm}

Note that convergence of $\xi_n$ to $\xi$ in $L^2_{B\mu}(M,\fK)$ only implies
\emph{weak} operator convergence of $r(\xi_n)$ to $r(\xi)$, as operators on the pre-Hilbert space
$\cD(H)$. In this section, we define a subspace $H^1_{B\mu}(M,\fK)$
of $L^2_{B\mu}(M,\fK)$ where convergence to $\xi$ implies
\emph{strong} convergence to $r(\xi)$.

\begin{Definition}[Parallel Sobolov spaces]
For $k\geq 0$, the
\emph{parallel Sobolev norm} $q_k$\index{norm!10@Sobolev \scheiding $q_{k}$} is defined by
\begin{eqnarray}
q_k(\xi) := \sum_{n=0}^{k} \|\xi\|_{n}\,,
\quad\text{where}\quad
\|\xi\|_{n} := \|D^n \xi\|_{B_{\eps}\mu}\,.
\end{eqnarray}
The \emph{parallel Sobolev space}\index{gauge algebra!50@Sobolev \scheiding $H^k_{B\mu}(M,\fK)$} $H^k_{B\mu}(M,\fK) \subseteq
L^2_{B\mu}(M,\fK)$ is the Banach completion of $\fg/I_{\mu}$
with respect to the norm $q_k$.
\end{Definition}

\begin{Proposition}\label{sterk}
Let $r$ be as in {\rm Proposition~\ref{zwak}}.
If $\xi \in H^k_{B\mu}(M,\fK)$, then $r(\xi)$
maps $\cD(H^{k+1})$ into $\cD(H^k)$. For $k = 1$, we have
\begin{equation}
[H,r(\xi)] = ir(D\xi) \label{ah}
\end{equation}
as an equality of unbounded
operators on $\cD(H^2)$.
Furthermore, if $\xi \in H^{k}_{B\mu}(M,\fK)$, then $r(\xi)$
extends to a continuous operator $\cH_{k+1/2} \rightarrow \cH_{k-1/2}$ with
\begin{equation}\label{eq:stapjes}
\|r(\xi) \psi\|_{k-1/2} \leq \sum_{j=0}^{k} \binom{k}{j} \|\xi\|_{j}\|\psi\|_{k-j+1/2}\,.
\end{equation}
Finally, for all $\xi \in H^1_{B\mu}(M,\fK)$, the
skew-symmetric operator $r(\xi)$ is essentially skew-adjoint.
\end{Proposition}
\begin{proof}
We prove that for $\xi \in H^k_{B\mu}(M,\fK)$, $r(\xi)$ maps $\cD(H^{k+1})$ into $\cD(H^{k})$.
We proceed by induction on $k$.
Since $H^0_{B\mu}(M,\fK) = L^2_{B\mu}(M,\fK)$,
 the case $k=0$ follows from Proposition~\ref{zwak}.
Suppose that the statement holds for all $\xi \in H^k_{B\mu}(M,\fK)$.
%
For $\xi \in H^{k+1}_{B\mu}(M,\fK)$ and
$\psi \in \cD(H^{k+2})$, we show that $r(\xi) \psi \in \cD(H^{k+1})$.
Since $H^{k+1}$ is selfadjoint, it suffices to show that
$\chi \mapsto \langle r(\xi)\psi, H^{k+1}\chi \rangle$
is a continuous, linear functional on $\cH^{\infty}$ with respect to the
subspace topology induced by the inclusion in $\cH$.

Let
$\xi_{n} \in \g/I_{\mu}$ be a sequence such that
$\xi_{n} \rightarrow \xi$ and $D\xi_n \rightarrow D\xi$ in $L^{2}_{B\mu}(M,\fK)$.
Since $Hr(\xi_n) = r(\xi_{n})H + ir(D\xi_n)$ on $\cH^{\infty}$, we have
\begin{eqnarray}
\langle r(\xi)\psi, H^{k+1}\chi \rangle &=&
-\lim_{n\rightarrow \infty} \langle\psi, r(\xi_n)H^{k+1}\chi \rangle \nonumber\\
&=&
-\lim_{n\rightarrow \infty} \langle H\psi, r(\xi_n)H^{k}\chi \rangle
+
\lim_{n\rightarrow \infty} \langle \psi, ir(D\xi_n)H^{k}\chi \rangle \nonumber\\
&=&
\langle r(\xi) H\psi +i r(D\xi)\psi, H^{k}\chi \rangle\,.
\label{eq:stroomput}
\end{eqnarray}
As $\psi \in \cD(H^{k+2})$, both $H\psi$ and $\psi$ are in $\cD(H^{k+1})$.
Since $\xi \in H^{k+1}_{B\mu}(M,\fK)$, we have
$\xi, D\xi \in H^{k}_{B\mu}(M,\fK)$, so that
$r(\xi) H\psi + i r(D\xi)\psi \in \cD(H^{k})$ by the induction hypothesis.
From~\eqref{eq:stroomput}, we thus find that
\[
\langle r(\xi)\psi, H^{k+1}\chi \rangle = \langle H^{k}\big(r(\xi) H\psi +i r(D\xi)\psi\big), \chi \rangle\,,
\]
which is manifestly continuous in $\chi$. We conclude that $r(\xi)$ maps $\cD(H^{k+2})$
to $\cD(H^{k+1})$. Moreover, for $k=0$, we find that
$Hr(\xi) - r(\xi)H - ir(D\xi)$ vanishes
on $\cD(H^2)$.

The inequality \eqref{eq:stapjes} is proven in a similar fashion.
Assume by induction that \eqref{eq:stapjes} holds for all
$\xi \in H^k_{B\mu}(M,\fK)$
and $\psi \in \cH_{k+1/2}$, the case $k=0$
being \eqref{eq:grondstof}.
We recall that $\|\psi\|_{s} = \|A^s \psi\|$ with $A = \one + H$
(see \eqref{eq:defA}).
For $\xi \in H^{k+1}_{B\mu}(M,\fK)$ and $\psi \in \cH_{k+3/2}$,
we use $A r(\xi) \psi = r(\xi) A \psi + ir(D\xi) \psi$ to see that
\[
\|r(\xi) \psi\|_{k + 1/2} = \|A r(\xi) \psi\|_{k-1/2} 
\leq
\|r(\xi) A \psi\|_{k-1/2} + \|r(D\xi)\|_{k-1/2}.
\]
By the induction hypothesis with $\|A \psi\|_{k-j+1/2} = \|\psi\|_{(k+1)-j+1/2}$
(for the first term)
and $\|D\xi\|_{j} = \|\xi\|_{j+1}$ (for the second), we find that $\|r(\xi) \psi\|_{k + 1/2}$ is bounded by
\[
\sum_{j=0}^{k} \binom{k}{j} \|\xi\|_{j}\|\psi\|_{(k+1)-j+1/2}
+
\sum_{j=0}^{k} \binom{k}{j} \|\xi\|_{j+1}\|\psi\|_{k -j+1/2} \,.
\]
Since $\binom{k}{j} + \binom{k}{j-1} = \binom{k+1}{j}$, we have \\
\[
\|r(\xi) \psi\|_{k + 1/2} \leq
\sum_{j=0}^{k+1} \binom{k+1}{j} \|\xi\|_{j}\|\psi\|_{k+1 - j + 1/2}\,,
\]
as required.

Finally, if $\xi \in H^1_{B\mu}(M,\fK)$, then $\xi, D\xi \in L^2_{B\mu}(M,\fK)$.
By \eqref{eq:grondstof}, the operators
$r(\xi)$ and $[A, r(\xi)] = ir(D\xi)$
from $\cD(H)$ to  $\cH$ extend continuously to bounded operators
$\cH_{1/2} \rightarrow \cH_{-1/2}$.
It then follows from a result of Nelson (\cite[Prop.~2]{Nelson1972})
that $r(\xi)$ is essentially skew-adjoint.
\end{proof}

If we estimate $\|\xi\|_{j} \leq q_k(\xi)$
and $\|\psi\|_{k-j+1/2} \leq \|\psi\|_{k+1/2}$
in \eqref{eq:stapjes}, we find that
$r(\xi) \colon \cH_{k+1/2} \rightarrow \cH_{k-1/2}$ satisfies
\begin{equation}\label{eq:brombaard}
\|r(\xi) \psi\|_{k-1/2} \leq 2^k  q_k(\xi) \,\|\psi\|_{k+1/2}\,,
\end{equation}
so the linear map $H^k_{B\mu}(M,\fK) \times \cH_{k+1/2} \rightarrow \cH_{k-1/2}$
defined by $(\xi,\psi) \mapsto r(\xi)\psi$ is jointly continuous.
For $k=1$, we find from \eqref{eq:stapjes}
the slightly stronger estimate
\begin{equation}\label{zwakdik}
\|r(\xi) \psi\| \leq \|r(\xi)\psi\|_{1/2} \leq q_1(\xi) \|A^{3/2} \psi\|\,.
\end{equation}
In particular, convergence of $\xi_n$ to $\xi$ in $H^1_{B\mu}(M,\fK)$ implies
\emph{strong}
convergence of $r(\xi_n)$ to $r(\xi)$ on $\cD(A^{3/2})$.


\subsection{Sobolev--Lie algebras}\label{sec:SobolevLie}
Having established that the positive energy representation $\dd\rho$
extends to a continuous map $r$ on $H^k_{B\mu}(M,\fK)$, we would like
to determine whether $r$ gives rise to a Lie algebra representation.
Since the spaces $H^k_{B\mu}(M,\fK)$ do not inherit
the Lie algebra structure from $\fg/I_{\mu}$, we
introduce two spaces of \emph{bounded} Sobolev sections
of $\fK\rightarrow M$, both equipped with the pointwise Lie bracket.

For an open subset $N \subseteq M$, we define the Lie algebra $H^k_{b}(N,\fK)$
of bounded parallel Sobolev sections, and a certain subalgebra $H^k_{\partial}(N,\fK)$
of sections that vanish to order $k$ at the boundary
of the 1-point compactification of $N$.
As before, the underlying measure is the restriction to $N$ of the flow-invariant measure
$B_{\eps}\mu$ on $M$.
For convenience of notation, we will denote this measure by
\begin{equation}
\nu = B_{\eps}\mu\,.
\end{equation}

\subsubsection{The Lie algebra $L^2_{b}(N,\fK)$ of bounded $L^2$-sections}

Let $N$ be an open subset of $M$, and let
$\xi$ be a measurable section of $\fK \rightarrow N$.
Then
$\|\xi\|_{\kappa} = \sqrt{\kappa(\xi,\xi)}$ is a measurable function on $N$.
We define $\|\xi\|_{\infty}$ to be the essential supremum of $\|\xi\|_{\kappa}$
with respect to $\nu$, and we define $L^{\infty}(N,\fK)$ to be the Lie algebra
of equivalence classes of essentially bounded, measurable sections of
$\fK \rightarrow N$.
This is a Banach--Lie algebra with respect to the norm $\|\xi\|_{\infty}$, and the Lie bracket coming
from the pointwise bracket of sections. We define
$L^2_{b}(N,\fK)$\index{gauge algebra!40@bounded quadratic \scheiding $L^2_{b}(N,\fK)$} to be the space of equivalence
classes of sections which are both essentially bounded and square integrable
with respect to $\nu$.
Since both $L^2(N,\fK)$ and
$L^{\infty}(N,\fK)$ are complete, it follows that $L^2_{b}(N,\fK)$ is a Banach space
with respect to the norm $\|\xi\|_{\nu} + \|\xi\|_{\infty}$.

Let $c_{\fk}$ be a constant such that
\begin{equation}\label{eq:pruut}
\|[X,Y]\|_{\kappa} \leq c_{\fk}\|X\|_{\kappa}\|Y\|_{\kappa}
\end{equation}
for all $X,Y \in \fk$. Then we find
\begin{eqnarray}\label{eq:dera}
\|[\xi,\eta]\|_{\nu} &\leq& c_{\fk} \|\xi\|_{\infty}\|\eta\|_{\nu},\label{eq:infmu}\\
\|[\xi,\eta]\|_{\infty} &\leq& c_{\fk} \|\xi\|_{\infty}\|\eta\|_{\infty}\,.\label{eq:infinf}
\end{eqnarray}
It follows that the Lie bracket
$[\,\cdot\,,\,\cdot\,] \colon L_{b}^{2}(N,\fK) \times L_{b}^{2}(N,\fK) \rightarrow
L^{\infty}(N,\fK)$
takes values in $L^2_{b}(N,\fK)$ and is continuous with respect to the
norm $p_0(\xi) := \|\xi\|_{\nu} + \|\xi\|_{\infty}$.
In particular, $L^2_{b}(N,\fK)$ is a Banach--Lie algebra, and the inclusion
$L^2_{b}(N,\fK) \hookrightarrow L^{\infty}(N,\fK)$ is a continuous morphism of
Banach--Lie algebras.
If $N \subseteq N'$, then $L^2_{b}(N,\fK)$ is a subalgebra of
$L^2_b(N',\fK)$ in the natural fashion.

\subsubsection{The `parallel' Sobolev--Lie algebras $H^k_{b}(N,\fK)$}

Recall from Definition~\ref{def:gemetrictype} that a one-parameter group
$(\gamma_t)_{t\in \R}$ of automorphisms of $\fK\rightarrow M$ gives rise to a
one-parameter group
$(\alpha_t)_{t\in \R}$ of automorphisms of $\fg = \Gamma_{c}(M,\fK)$.
In the exact same way, we obtain a one-parameter group of automorphisms of
$L^2_{b}(M,\fK)$.

Indeed, since the Killing form is invariant under automorphisms, we have
$\|\alpha_t (\xi)\|_{\kappa}= \|\xi\|_{\kappa}\circ \gamma_{M,t}$,
so that in particular $\|\alpha_t(\xi)\|_{\infty} = \|\xi\|_{\infty}$.
Further, since the measure $\nu = B_{\eps}\mu$ is invariant under the flow
$\gamma_{M,t}$ (Theorem~\ref{MeasureThm}),
 we find $\|\alpha_t(\xi)\|_{\nu} = \|\xi\|_{\nu}$.

Since $\alpha_t$ is a one-parameter group of unitary transformations of the Hilbert space
$L^2_{\nu}(M,\fK)$, it is generated by a skew-adjoint operator $D$.
We define $H^1_{\nu}(N,\fK)$ to be the intersection of its domain with
$L^2_{\nu}(N,\fK)$, and we define $H^1_{b}(N,\fK)$ to be the space of all
$\xi \in H^1_{\nu}(N,\fK)$,
where both $\|\xi\|_{\infty}$ and $\|D\xi\|_{\infty}$ are finite.
In other words, $H^1_{b}(N,\fK)$ is the space of equivalence classes of
essentially bounded,
square integrable sections $\xi$ of $\fK \rightarrow N$ such that the $L^2$-limit
\begin{equation}
  \label{eq:d-def}
D(\xi) := L^2\text{-}\lim_{t \rightarrow 0} {\textstyle \frac{1}{t}}(\alpha_t(\xi) - \xi)
\end{equation}
exists, and $\|D(\xi)\|_{\infty}$ is finite.

\begin{Proposition}\label{prop:heen}
The space $H^1_{b}(N,\fK)$ is a Lie subalgebra of $L^2_{b}(N,\fK)$, and the generator
$D \colon H^1_{b}(N,\fK) \rightarrow L^2_{b}(N,\fK)$ satisfies
\begin{equation}\label{eq:derprop}
D([\xi,\eta]) = [D(\xi),\eta] + [\xi,D(\eta)]
\quad \text{for all} \quad \xi,\eta \in H^1_{b}(N,\fK)
\,.
\end{equation}
\end{Proposition}
\begin{proof}
Let $\xi,\eta \in H^1_{b}(N,\fK)$, and
denote by $L^2$-$\lim$ the limit with respect to the norm $\| \xi \|_{\nu}$.
First, we show that $[\xi,\eta]$ is in the domain
of $D$.
\begin{align*}
D([\xi,\eta]) &=
L^2\text{-}\lim_{t\rightarrow 0} {\textstyle \frac{1}{t}}(\alpha_t([\xi,\eta])-[\xi,\eta])\\
&=
L^2\text{-}\lim_{t\rightarrow 0} \,[D\xi, \alpha_t(\eta)]
+
L^2\text{-}\lim_{t\rightarrow 0} \,[{\textstyle \frac{1}{t}}(\alpha_t(\xi) - \xi) - D(\xi), \alpha_{t}(\eta)]
\\
&\quad+L^2\text{-}\lim_{t\rightarrow 0} \,[\xi, {\textstyle \frac{1}{t}}(\alpha_t(\eta) - \eta)]
 = [D\xi,\eta] + [\xi,D\eta]\,.
\end{align*}
In the last step, we used the inequality \eqref{eq:infmu} three times.
Since $\|D\xi\|_{\infty}$ is bounded and
$L^2$-$\lim_{t\rightarrow 0}\alpha_t(\eta) = \eta$,
it follows from \eqref{eq:infmu} that the first term is given by
$L^2\text{-}\lim_{t\rightarrow 0} \,[D\xi, \alpha_t(\eta)] = [D\xi,\eta]$.
Similarly, since
$\|\xi\|_{\infty}$ is bounded and
$L^2$-$\lim_{t\rightarrow 0}{\textstyle \frac{1}{t}}(\alpha_t(\eta) - \eta) = D(\eta)$,
the third term equals $[\xi,D(\eta)]$.
To see that the second term is zero, note that
$\|\alpha_t(\eta)\|_{\infty} = \|\eta\|_{\infty}$. It then follows from \eqref{eq:infmu}
and the fact that
$L^2$-$\lim_{t\rightarrow 0} {\textstyle \frac{1}{t}}(\alpha_t(\xi) - \xi) - D(\xi) = 0$.

This shows not only that $[\xi,\eta]$ is in the domain of $D$, but also that
\eqref{eq:derprop} holds.
By \eqref{eq:infinf}, it follows that
$\|D([\xi,\eta])\|_{\infty} \leq c_{\fk}(\|D\xi\|_{\infty}\|\eta\|_{\infty} + \|\xi\|_{\infty}\|D\eta\|_{\infty})$
is finite, so that $[\xi,\eta] \in H^1_{b}(N,\fK)$.
\end{proof}

This allows us to define parallel Sobolev--Lie algebras\index{gauge algebra!60@bounded Sobolev \scheiding $H^k_{b}(M,\fK)$} of order $k\in \N$.
we set $H^0_{b}(N,\fK) := L^2_b(N,\fK)$, and define $H^1_{b}(N,\fK)$ as above.
For ${k\geq 2}$, we define $H^k_{b}(N,\fK)$ as
$H^{k-1}_{b}(N,\fK)\cap D^{-1}(H^{k-1}_{b}(N,\fK))$.
In other words, $\xi$ is in $H^k_{b}(N,\fK)$ if both $\xi$ and $D\xi$
are in $H^{k-1}_{b}(N,\fK)$.

\begin{Proposition}
The space $H^k_{b}(N,\fK)$ is a Lie subalgebra of $H^{k-1}_{b}(N,\fK)$.
\end{Proposition}
\begin{proof}
The proof is by induction on $k$, where $k=1$ is Proposition
\ref{prop:heen}. 
If $\xi,\eta \in H^k_{b}(N,\fK)$, then
$\xi, D(\xi), \eta, D(\eta) \in H^{k-1}_{b}(N,\fK)$.
Since $H^{k-1}_{b}(N,\fK)$ is a Lie algebra, it follows that
$D([\xi,\eta]) = [D(\xi),\eta] + [\xi,D(\eta)]$ is in $H^{k-1}_{b}(N,\fK)$.
Thus $[\xi,\eta] \in H^k_{b}(N,\fK)$, as required.
\end{proof}

On $H^k_{b}(N,\fK)$, we define for every
$n\in \{0,\ldots,k\}$ the derived norms 
\[
\|\xi\|_{n,\infty} := \|D^n\xi\|_{\infty} \quad\text{and}\quad\|\xi\|_{n} := \|D^n\xi\|_{\nu}\,.
\]
The parallel $C^k$-norm $q_{C^k}$ and the parallel Sobolev norm $q_k$
are defined by
\begin{equation}\label{parsobtwee}
q_{C^k}(\xi) := \sum_{n=0}^{k} \|\xi\|_{n,\infty}
\quad\text{and}\quad
q_{k}(\xi) := \sum_{n=0}^{k} \|\xi\|_{n}\,,
\end{equation}
respectively.
%
We equip $H^{k}_{b}(N,\fK)$ with the topology derived from
the combined norm\index{norm!20@bounded Sobolev \scheiding $p_k$}
\begin{equation}\label{eq:normbank}
p_k(\xi) := \sum_{n=0}^{k} \|\xi\|_{n,\infty} + \|\xi\|_{n}\,.
\end{equation}
Note that for $\xi \in H^{k}_{b}(N,\xi)$, we have
$p_{k-1}(\xi) \leq p_{k}(\xi)$, but also ${p_{k-1}(D(\xi)) \leq p_{k}(\xi)}$.
It follows that both the inclusion
$\iota \colon H^{k+1}_{b}(N,\fK) \hookrightarrow H^{k}_{b}(N,\fK)$
and the derivative $D \colon H^{k+1}_{b}(N,\fK) \rightarrow H^{k}_{b}(N,\fK)$
are continuous.

\begin{Proposition}\label{prop:isbanalg}
For every $k\geq 0$,
$H^k_{b}(N,\fK)$ is a Banach--Lie algebra
with respect to the norm $p_k$.
The Lie bracket is separately continuous
with respect to the Sobolev norm $q_{k}$.
%
\end{Proposition}
\begin{proof}

By the derivation property and \eqref{eq:pruut}, we have
\[
\|D^n ([\xi,\eta])\|_{\kappa} \leq c_{\fk}
\sum_{j=0}^{n}\binom{n}{j}\|D^j\xi\|_{\kappa} \|D^{n-j}\eta\|_{\kappa}\,.
\]
Since $\|[\xi,\eta]\|_{n} = \|D^n([\xi,\eta])\|_{\nu}$
and $\|[\xi,\eta]\|_{n,\infty} = \|D^n([\xi,\eta])\|_{\infty}$
, it follows that
\begin{align}
\|[\xi,\eta]\|_{n} &\leq  c_{\fk} \sum_{j=0}^{n}
\binom{n}{j}\|\xi\|_{j,\infty} \|\eta\|_{n-j},
\label{eq:haakneus}
\\
\|[\xi,\eta]\|_{n,\infty} &\leq
c_{\fk}
\sum_{j=0}^{n}\binom{n}{j}\|\xi\|_{j,\infty} \|\eta\|_{n-j,\infty}\,.
\end{align}
Taking $n=k$ and estimating the binomial coefficients by $2^k$,
it follows that
\begin{align}
q_{k}([\xi,\eta]) &\leq 2^{k}c_{\fk} q_{C^k}(\xi) q_{k}(\eta),\label{eq:smosbrood}\\
q_{C^k}([\xi,\eta]) &\leq 2^{k}c_{\fk}q_{C^k}(\xi) q_{C^k}(\eta).
\end{align}
This shows that the Lie bracket is continuous
for the norm $p_k$, and separately continuous for the Sobolev norm
$q_{k}$.

%


To show that $H^{k}_{b}(N,\fK)$ is complete, we note that
$H^0_{b}(N,\fK) = L^2_{b}(N,\fK)$ is a Banach space, and proceed by
induction on $k$. Let $\xi_n \in H^{k}_{b}(N,\fK)$ with $p_k(\xi_n-\xi_m) \rightarrow 0$.
Then $p_{k-1}(\xi_n - \xi_m) \rightarrow 0$ and
$p_{k-1}(D(\xi_n)-D(\xi_m)) \rightarrow 0$, so by induction, there exist
$\xi,\Xi \in H^{k-1}_{b}(N,\fK)$ with $p_{k-1}(\xi - \xi_n) \rightarrow 0$
and $p_{k-1}(\Xi -D(\xi_n)) \rightarrow 0$.
Since $D \colon H^1_{\nu}(M,\fK) \rightarrow L^2(M,\fK)$ is the generator
of a strongly continuous 1-parameter group of unitary operators, Stone's Theorem
implies that it is selfadjoint, and hence in particular closed.
It follows that $\xi$ lies in the domain of $D$, and $D(\xi) = \Xi$ lies in
$H^{k-1}_{b}(N,\fK)$.
Thus $\xi \in H^{k}_{b}(N,\fK)$, and
$p_k(\xi - \xi_n) \leq p_{k-1}(\xi - \xi_n) + p_{k-1}(D(\xi) - D(\xi_{n})) \rightarrow 0$.
\end{proof}

We denote by $H^{\infty}_{b}(N,\fK)$ the Fr\'echet--Lie algebra
arising from the inverse limit of the Banach--Lie algebras
$H^k_{b}(N,\fK)$ with respect to the natural inclusions
$\iota \colon H^{k+1}_{b}(N,\fK) \hookrightarrow H^{k}_{b}(N,\fK)$.
The derivative $D\colon H^{\infty}_{b}(N,\fK) \rightarrow H^{\infty}_{b}(N,\fK)$
is a continuous derivation,
giving rise to the Fr\'echet--Lie algebra $H^{\infty}_{b}(N,\fK) \rtimes \R D$.

\subsubsection{Boundary conditions and the Lie algebras $H^k_{\partial}(N,\fK)$}

Let $H^1_{\partial}(N,\fK)$ be the closure of $\Gamma_{c}(N,\fK)$
in $H^1_{b}(N,\fK)$ with respect to the parallel Sobolev norm
$q_1(\xi) = \|\xi\|_{\nu} + \|\xi\|_{1,\nu}$.

\begin{Proposition}\label{prop:delLie1}
The space $H^1_{\partial}(N,\fK)$ is a closed Lie subalgebra of $H^1_{b}(N,\fK)$.
In particular,
it is a Banach--Lie algebra with respect to the subspace topology, induced by the
norm $p_1(\xi)$ of \eqref{eq:normbank}.
\end{Proposition}
\begin{proof}
Since $H^1_{\partial}(N,\fK)$ is by definition closed with respect to
the Sobolev norm $q_1(\xi)$, it is a forteriori closed with respect to the larger
norm $p_1(\xi)$ that defines the Banach space topology on $H^1_{b}(N,\fK)$.
As $H^1_{\partial}(N,\fK)$ is a closed subspace of a Banach space, it is a
Banach space itself.

It remains to show that $H^1_{\partial}(N,\fK)$ is closed under the Lie bracket.
For every $\xi \in H^1_{b}(N,\fK)$, the linear operator
$\mathrm{ad}_{\xi} \colon H^1_{b}(N,\fK) \rightarrow H^1_{b}(N,\fK)$
is continuous with respect to the norm $q_1(\xi)$, as
$q_1(\mathrm{ad}_{\xi}(\eta)) \leq 2c_{\fk}q_{C^1}(\xi)q_1(\eta)$
by~\eqref{eq:smosbrood}.
If $\xi \in \Gamma_{c}(N,\fK)$, then $\mathrm{ad}(\xi)$ maps
$\Gamma_{c}(N,\fK)$ to
$\Gamma_{c}(N,\fK)$. As $\mathrm{ad}_{\xi}$ is continuous for the norm $q_1$, it also maps $H^1_{\partial}(N,\fK)$ to $H^1_{\partial}(N,\fK)$.
It follows that, for all $\eta \in H^1_{\partial}(N,\fK)$, $\mathrm{ad}_{\eta}$ maps
$\Gamma_{c}(N,\fK)$ to $H^1_{\partial}(N,\fK)$. By continuity with respect to
$q_1$, it therefore maps
$H^1_{\partial}(N,\fK)$ to $H^1_{\partial}(N,\fK)$, and we conclude that
$H^1_{\partial}(N,\fK)$ is closed under the Lie bracket.
\end{proof}

For $k\geq 2$, we define $H^k_{\partial}(N,\fK)$\index{gauge algebra!70@$1/\sqrt{t}$ decay \scheiding $H^k_{\partial}(N,\fK)$} as the space of all $\xi \in H^k_{b}(N,\fK)$
such that both $\xi$ and $D(\xi)$ lie in $H^{k-1}_{\partial}(N,\fK)$.
\begin{Proposition}
The space $H^k_{\partial}(N,\fK)$ is a closed Lie subalgebra of $H^k_{b}(N,\fK)$.
In particular, it is a Banach--Lie algebra with respect to the subspace topology, induced by the
norm $p_k(\xi)$ of \eqref{eq:normbank}.
\end{Proposition}
\begin{proof}

We proceed by induction on $k$, the case $k=1$ being Proposition~\ref{prop:delLie1}.
Recall that both the inclusion $\iota \colon H^{k}_{b}(N,\fK) \hookrightarrow H^{k-1}_{b}(N,\fK)$
and the derivative $D \colon H^{k}_{b}(N,\fK) \rightarrow H^{k-1}_{b}(N,\fK)$
are continuous. Since
\[
H^{k}_{\partial}(N,\fK) =
\iota^{-1}(H^{k-1}_{\partial}(N,\fK)) \cap D^{-1}(H^{k-1}_{\partial}(N,\fK))\]
is the
intersection of two closed subspaces, it is a closed subspace of $H^k_{b}(N,\fK)$ itself.
To show that it is closed under the Lie bracket, suppose that 
${\xi,\eta \in H^k_{\partial}(N,\fK)}$,
so that $\xi,\eta,D\xi,D\eta \in H^{k-1}_{\partial}(N,\fK)$.
As $H^{k-1}_{\partial}(N,\fK)$ is a Lie algebra, it follows that
$[\xi,\eta]$ and $D([\xi,\eta]) = [D(\xi),\eta] + [\xi,D(\eta)]$ are both in
$H^{k-1}_{\partial}(N,\fK)$. From this, one sees that also
$[\xi,\eta] \in H^{k}_{\partial}(N,\fK)$.
\end{proof}

Note that the 2-cocycle $\omega(\xi,\eta) = \lra{D\xi,\eta}_{\mu}$ on $\fg$
is continuous for the Sobolev norm $q_1(\xi)$ and hence
extends uniquely to $H^k_{\partial}(N,\fK)$. This defines a continuous
central extension of $H^k_{\partial}(N,\fK)$,
\[\R C \oplus_{\omega} H^k_{\partial}(N,\fK)\,.\]

Define the Fr\'echet--Lie algebra $H^{\infty}_{\partial}(N,\fK)$ as the inverse limit
of the Banach--Lie algebras $H^{k}_{\partial}(N,\fK)$
under the natural inclusions $H^{k}_{\partial}(N,\fK) \rightarrow
H^{k-1}_{\partial}(N,\fK)$.
Since $D \colon H^{\infty}_{\partial}(N,\fK) \rightarrow H^{\infty}_{\partial}(N,\fK)$
is a continuous derivation, we obtain the double extension of Fr\'echet--Lie algebras
\[ (\R C \oplus_{\omega} H^{\infty}_{\partial}(N,\fK)) \rtimes \R D\,.\]

\subsubsection{Intervals and blocks}

Suppose that $N \simeq \Sigma \times I$, where $I \subseteq \R$ is an open, not necessarily finite interval with the Lebesgue measure $dt$,
and $\Sigma$ is a $(d_{M}-1)$-dimensional manifold with locally finite measure $\nu_0$.
The bundle $\fK\res_N \simeq N \times \fk$ is trivial, and the translation
by $t'$ sends $(x,t)$ to $(x,t-t')$ wherever it is defined.
In this cartesian product situation, it will be useful to
separate the variables in $\Sigma$ from those in $I$.

Define a 
bilinear map
\[ T \colon L^2_{b}(\Sigma,\R)\times L^2_{b}(I,\fk) \rightarrow L^2_{b}(N,\fk),
\qquad T(f, \xi)(x,t) = f(x)\xi(t).\]
It is continuous since
$\|f\xi\|_{\nu} = \|f\|_{\nu_0}\|\xi\|_{dt}$
and $\|f\xi\|_{\infty} = \|f\|_{\infty}\|\xi\|_{\infty}$.
\begin{Proposition}\label{prop:sietsiet}
The product $T(f, \xi) = f\xi$ defines a continuous bilinear map
\[ T \colon L^2_{b}(\Sigma,\R) \times H^{k}_{\partial}(I,\fk) \rightarrow
H^{k}_{\partial}(N,\fk)\,.
\]
\end{Proposition}
\begin{proof}
Since $\|f\xi\|_{\nu} = \|f\|_{\nu_0}\|\xi\|_{dt}$, and since time translation acts only on $\xi$,
it follows that $f\xi \in \mathcal{D}(D)$ if and only if $\xi \in \mathcal{D}(D)$,
and $D(f\xi) = fD(\xi)$.
From this, one derives that $T$ maps
$L^2_{b}(\Sigma,\R)\times H^{k}_{b}(I,\fk)$ to $H^{k}_{b}(N,\fk)$,
with ${\|f\xi\|_{n} = \|f\|_{\nu_0}\|\xi\|_{n}}$ and
$\|f\xi\|_{n,\infty} = \|f\|_{\infty}\|\xi\|_{n,\infty}$.

Suppose that $\xi \in H^1_{\partial}(I,\fk)$, so that there
exists a sequence $\xi_n \in C^{\infty}_{c}(I,\fk)$
with $\|\xi - \xi_{n}\|_{dt} \rightarrow 0$ and $\|D(\xi) - D(\xi_n)\|_{dt} \rightarrow 0$.
For every $f\in L^2_{b}(\Sigma,\R)$, it is possible to find a sequence
$f_n \in C^{\infty}_{c}(\Sigma,\R)$
with $\|f-f_{n}\|_{\nu_0} \rightarrow 0$.
Then
\[
\|f\xi - f_n\xi_n\|_{\nu} \leq \|f-f_n\|_{\nu_0}\|\xi\|_{dt} +
\|f_n\|_{\nu_0}\|\xi - \xi_n\|_{dt} \rightarrow 0\,.
\]
Similarly, one finds that $\|D(f\xi) - D(f_n\xi_n)\|_{\nu} = \|fD(\xi) - f_n D(\xi_n)\|
\rightarrow 0$.
It follows that $T$ maps $L^2_{b}(\Sigma,\R) \times H^1_{\partial}(I,\fk)$
to $H^1_{\partial}(N,\fk)$. From $D(f\xi) = fD(\xi)$, one then finds that it maps
$L^2_{b}(\Sigma,\R) \times H^k_{\partial}(I,\fk)$
to $H^k_{\partial}(N,\fk)$.
\end{proof}

In Lemma~\ref{lem:integralitymeasure}, we will need the above result in the following form.
\begin{Corollary}\label{cor:sietsiet}
Let $E \subseteq \Sigma$ be a subset of finite measure, and let
$\chi_{E}$
be the corresponding indicator function.
Then the map
$\iota_{E} \colon H^k_{\partial}(I,\fk) \rightarrow H^k_{\partial}(N,\fK)$
defined by $\iota_{E}(\xi)(x,t) = \chi_{E}(x)\xi(t)$ is a continuous Lie algebra homomorphism.
\end{Corollary}

\subsection{The continuous extension theorem}\label{sec:context5van6}

It follows from Proposition \ref{zwak} that the Lie algebra representation
$\dd\rho$ extends from $\fg = \Gamma_{c}(M,\fK)$ to $L^2_{B\mu}(M,\fK)$.
In the following theorem, we show that
this extension yields a Lie algebra
representation of $\R C \oplus_{\omega}H^1_{\partial}(M,\fK)$,
which is compatible with the derivation
$D \colon H^1_{\partial}(M,\fK) \rightarrow L^2_{b}(M,\fK)$.

\begin{Theorem}{\rm(Continuous extension)}\label{Shylock}
Let $\rho$ be a positive energy representation of $\hat G$
with derived representation $\dd\rho$, and let $N\subseteq M$ be an open subset.
\begin{itemize}
\item[\rm(a)]
There exists a linear map $r$ from $L^2_{b}(N,\fK)$ to the unbounded,
skew-sym\-metric operators on $\cH$ with domain $\cD(H)$
such that $r(\xi)\psi$ coincides with $\dd\rho(\xi)\psi$ 
for all $\xi \in \Gamma_{c}(N,\fK)$ and $\psi \in \cH^{\infty}$.
\item[\rm(b)]
This defines a representation of the Banach--Lie algebra
$\R C \oplus_{\omega}H^1_{\partial}(N,\fK)$ by essentially skew-adjoint
operators.
For $\xi,\eta \in H^1_{\partial}(N,\fK)$, the operators $r(\xi)$ and $r(\eta)$
map $\cD(H^2)$ to $\cD(H)$.
On $\cD(H^2)$, we have the commutator relation
\begin{equation}
  \label{eq:r-bracket}
[r(\xi), r(\eta)] = r([\xi,\eta]) + i\omega(\xi,\eta)\one, 
\end{equation}
where $\omega(\xi,\eta) = \lra{D\xi,\eta}_{\mu}$.
\item[\rm(c)] If $\xi \in H^1_{\partial}(N,\fK)$, then
$D\xi \in L^2_{b}(N,\fK)$ and
\begin{equation}
[\dd\rho(D),r(\xi)] = r(D\xi)\,.
\end{equation}
In particular, we obtain a positive energy representation of
the Fr\'echet--Lie algebra
$\big(\R C \oplus_{\omega} H^{\infty}_{\partial}(N,\fK)\big) \rtimes \R D\,.$
\end{itemize}
%
%
%
\end{Theorem}

\begin{proof}
The derived Lie algebra representation $\dd\rho$ is defined on the Lie algebra
$\widehat{\fg} = (\R C \oplus_{\omega} \fg) \rtimes \R D$.
By Proposition~\ref{zwak}, we obtain an extension
$r$ of $\dd\rho$ to $L^2_{B\mu}(M,\fK)$, hence in particular to
$L^2_{b}(N,\fK)$.
From Proposition~\ref{sterk}, we see that $r(\xi)$
is essentially skew-adjoint for $\xi$ in the smaller space
$H^1_{B\mu}(M,\fK) \subseteq L^2_{B\mu}(M,\fK)$, and that
$[\dd\rho(D),r(\xi)] = r(\xi')$ for all $\xi \in H^1_{B\mu}(M,\fK)$, hence in particular
for $\xi \in H^1_{\partial}(N,\fK) \subseteq H^1_{B\mu}(M,\fK)$.

By Cauchy--Schwarz and the inequality \eqref{zwakdik}, 
we have
\begin{equation}\label{dubbeldik}
|\langle r(\xi)\psi, r(\eta)\chi\rangle|
\leq
\|A^{3/2}\psi\|\|A^{3/2}\chi\| q_1(\xi)q_1(\chi)
\end{equation}
for all $\psi, \chi \in \cD(H^2)$ and $\xi,\eta \in H^{1}_{B\mu}(M,\fg)$, where
$A := \one + H$ and $q_1$ is the parallel Sobolev norm of \eqref{parsobtwee}.
Further, by Proposition~\ref{sterk},
 the products $r(\xi)r(\eta)$ and $r(\eta)r(\xi)$
are well defined on $\cD(H^2)$. Since
\[ \langle \psi, [r(\xi),r(\eta)]\chi\rangle = -\langle r(\xi)\psi, r(\eta)\chi\rangle
+ \langle r(\eta)\psi, r(\xi)\chi\rangle,\]
it follows that the bilinear form
\[ H^{1}_{B\mu}(M,\fK) \times H^{1}_{B\mu}(M,\fK) \rightarrow \C,
\qquad (\xi,\eta) \mapsto \langle \psi, [r(\xi),r(\eta)]\chi\rangle \]
is continuous with respect to the parallel Sobolev norm $q_1$.
In particular, its restriction to $H^1_{\partial}(N,\fK) \subseteq H^1_{B\mu}(M,\fK)$
is continuous with respect to $q_1$.


Similarly, using Cauchy--Schwarz and
(\ref{zwakdik}),  
we find for $\xi,\eta \in H^1_{\partial}(N,\fK)$ that
\[ |\langle\chi, r([\xi,\eta])\psi\rangle| \leq \|\chi\|\|A^{3/2}\psi\|
\,q_1([\xi,\eta])\,.\]
Since
the Lie bracket on
$H^1_{\partial}(N,\fK)$ is
\emph{separately} continuous for the norm $q_1$ by
Proposition~\ref{prop:isbanalg}, it follows that the
bilinear form
$H^1_{\partial}(N,\fK)\times H^1_{\partial}(N,\fK) \rightarrow \C$ defined by
$(\xi,\eta) \mapsto \langle\chi, r([\xi,\eta])\psi\rangle$ is \emph{separately} continuous
with respect to $q_1$.
%

As the cocycle $\omega(\xi,\eta) = \lra{D\xi,\eta}_{\mu}$
extends to a bilinear map
on $H^1_{\partial}(N,\fK)$ that is continuous with respect to $q_1$,
the bilinear form
\[
(\xi,\eta) \mapsto \left\langle \chi, \big([r(\xi),r(\eta)]
- r([\xi,\eta]) - i \omega(\xi,\eta)\big)\psi \right\rangle
\]
is \emph{separately} continuous for the $q_1$-topology.
Since it vanishes on the dense subset $\Gamma_{c}(N,\fK) \subseteq H^1_{\partial}(N,\fK)$,
it is identically zero.
It follows that \[[r(\xi),r(\eta)]\psi = r([\xi,\eta])\psi + i\omega(\xi,\eta)\psi\]
for all $\psi \in \cD(H^2)$. The operator $r([\xi,\eta]) + i\omega(\xi,\eta)\one$ with domain containing $\cD(H)$
is thus an essentially skew-adjoint extension of the operator $[r(\xi),r(\eta)]$ with domain~$\cD(H^2)$.
\end{proof}

\subsubsection{Semibounded representations}

The concept of a semibounded representation, introduced in
\cite{Ne09, Ne10a}, is much stronger than that of a positive energy condition.
As results in \cite{NSZ17} show,
it provides enough regularity
to lead to a
sufficient supply of $C^*$-algebraic tools to decompose representations as
direct integrals.

\begin{Definition} (Semibounded representations)
\label{def:semibounded}
We call a smooth representation
$(\rho, \cH)$ of a locally convex Lie group $G$
{\it semibounded}  \index{representation!semibounded \vulop}
if the function
\begin{equation}\label{eq:defsr}
s_\rho \: \g \to \R \cup \{ \infty\}, \quad
s_\rho(x)
:= \sup\big(\Spec(i\dd\rho(x))\big)
\end{equation}
is bounded on a neighborhood of some point $x_0 \in \g$.
Then the set $W_\rho$ of all such points $x_0$
is an open $\Ad(G)$-invariant convex cone in~$\g$.
\end{Definition}
For Lie groups $G$ which are locally exponential
or whose Lie algebra $\g$ is barreled,\begin{footnote}
{These are the locally convex spaces for which the Uniform Boundedness
Principle holds. All Fr\'echet spaces and
locally convex direct limits of Fr\'echet spaces are barreled, which
includes in particular LF spaces of test functions on non-compact manifolds.}
\end{footnote}a semibounded representation is bounded if and only if $W_\rho = \g$
(\cite[Thm.~3.1, Prop.~3.5]{Ne09}).
The positive energy representation $r$
of $H^1_{\partial}(N,\fK)$ fulfills the following semiboundedness condition.

%

\begin{Proposition} \label{Prop:propsembo}
Let $\xi \in L^2_{B\mu}(M,\fK)$, and let $t>0$.
Then
\begin{equation}
-\|\xi\|_{B_1\mu} - 9\|\xi\|^2_{\mu}d_{\fk}(d_M +1)^2/t  \leq
\inf\Big(
\mathrm{Spec}\big(ir(tD \oplus \xi)\big)
\Big)\,.
\end{equation}
In particular, the spectrum of
$tH \pm ir(\xi)$ is bounded below for every $t>0$, and this bound is uniform on an open neighborhood of $D$ in $L^2_{B\mu}(M,\fK) \rtimes \R D$.
\end{Proposition}
\begin{proof}
Using Proposition~\ref{zwak}, one finds that the map
$L^2_{B\mu}(M,\fK) \rightarrow \C$
defined by
 $\xi \mapsto \lra{\Gamma_{\xi} \pm i r(\xi)}_{\psi}$
is continuous for every $\psi \in \cD(H)$, and every $\eps >0$.
It is non-negative on the dense subspace
$\Gamma_{c}(M,\fK)$ by Theorem~\ref{thm:lapje}, and hence on all of $L^2_{B\mu}(M,\fK)$ by continuity.
If $\|\xi\|_{\mu} = 0$, then $r(\xi) = 0$, and the proposition holds trivially.
If $\|\xi\|_{\mu} \neq 0$ and $\eps := t/\|\xi\|_{\mu}$, then
$\Gamma_{\xi} = tH + \|\xi\|_{B_{\eps}\mu}\one$  satisfies
$0 \leq \lra{\Gamma_{\xi} \pm i r(\xi)}_{\psi}$, and thus
\[
-\|\xi\|_{B_{\eps}\mu} \|\psi\|^2 \leq   \lra{\psi, tH \pm ir(\xi),\psi}\,.
\]
Since $\|\xi\|_{B_{\eps}\mu} \leq \|\xi\|_{B_1\mu} + 9\|\xi\|_{\mu}d_{\fk}(d_M +1)^2/\eps$,
the result follows by substituting $\eps = t/\|\xi\|_{\mu}$.
\end{proof}

\begin{Corollary}\label{cor:6.34}
The positive energy representation
$\dd\rho$ of the Lie algebra \linebreak[4]
$\big(\R C \oplus_{\omega} \Gamma_{c}(M,\fK)\big) \rtimes \R D$
is semibounded and the cone $W_\rho$ contains the open half space
$\big(\R C \oplus_{\omega} \Gamma_{c}(M,\fK)\big) - \R_+ D$.
\end{Corollary}

\begin{proof} This follows from
Proposition~\ref{Prop:propsembo} because $\dd\rho$
comes from a group representation,
the central element $C$ is represented by a constant, and the
inclusion $\Gamma_{c}(M,\fK) \hookrightarrow L^2_{B\mu}(M,\fK)$
is continuous.
\end{proof}

\subsubsection{Analytic vectors}

A vector $\psi$ in a Banach space $\mathfrak{X}$
is called \emph{analytic}\index{analytic!vectors \scheiding $\cD^{\omega}$} for an unbounded operator
$A$ on $\mathfrak{X}$  if $\psi \in \bigcap_{n\in \N} \cD(A^n)$, and the series
$\sum_{n=0}^{\infty} \frac{s^n}{n!}\|A^n \psi\|$
has positive radius of convergence $R_A>0$.

\begin{Lemma}  \label{lem:6.26}
Let $\xi\in H^2_{\partial}(N,\fK)$, and consider $H$ and $r(\xi)$
as unbounded operators on $\cH$.
If $\psi \in \cH$ is an analytic vector for $H$, then it is also analytic for $r(\xi)$.
If $\psi$ has radius of convergence $R_H$ for $H$, then
the exponential series $\exp(r(\xi))\psi = \sum_{n=0}^{\infty} \frac{1}{n!}r(\xi)^n \psi$
is absolutely convergent on the ball 
defined by
\begin{equation}\label{eq:bahbah}
p_2(\xi) < - {\textstyle \frac{1}{2c_{\fk}}}\log\left(
1 - {\textstyle\frac{(2c_{\fk})^2}{(c_{\fk}+1)^2}}
\Big(
1-\exp\Big(
-{\textstyle \frac{(c_{\fk}+1)^2}{2c_{\fk}}R_{H}}
\Big)
\Big)
\right)\,.
\end{equation}
%
\end{Lemma}

\begin{proof}
We apply \cite[Thm.~1]{Nel59} to
$r(\xi)$ and $A = \one + H$, considered as unbounded operators on the Banach
space~$\cH_{1/2}$.
For $\xi \in H^1_{B\mu}(M,\fK)$ and $\psi \in \cD(H^2) \subseteq \cH_{1/2}$,
the inequality~\eqref{zwakdik} 
yields
\begin{align}\label{eq:poedeltje}
\|r(\xi) \psi\|_{1/2} \leq q_1(\xi) \|A \psi\|_{1/2}.
\end{align}
By \eqref{ah},
we have $\mathrm{ad}_{r(\xi)} A = -ir(D\xi)$.
If $\xi \in H^2_{\partial}(N,\fK)$, then by definition, both
$\xi$ and $D\xi$ are in $H^1_{\partial}(N,\fK)$. It follows that also
$\mathrm{ad}^{n-1}_{\xi}(D\xi) \in H^1_{\partial}(N,\fK)$ for $n\geq 1$.
By \eqref{eq:r-bracket} and induction, we find
\begin{equation}\label{eq:weetsjon}
\mathrm{ad}^{n}_{r(\xi)} (A) =
-i \mathrm{ad}^{n-1}_{r(\xi)} (r(D\xi))
= -i r(\mathrm{ad}^{n-1}_{\xi}(D\xi)) +
\omega(\xi, \mathrm{ad}^{n-2}_{\xi}(D\xi) ) \one
\end{equation}
as an equality of unbounded operators from $\cD(H^2)$ to $\cH_{1/2}$.
From
\eqref{eq:dera} and \eqref{eq:smosbrood},
we infer that
\begin{eqnarray}\label{eq:vlotruit}
\|\mathrm{ad}^{n}_{\xi}(D\xi) \|_{B_{\eps}\mu} &\leq& (c_{\fk} \|\xi\|_{\infty})^n\|D\xi\|_{B_{\eps}\mu},\\
q_1\big(\mathrm{ad}^{n}_{\xi}(D\xi) \big) &\leq&
(2c_{\fk} q_{C^1}(\xi))^n q_{1}(D\xi)\,.\label{eq:vlotruit2}
\end{eqnarray}
Next we estimate $\|\mathrm{ad}^{n}_{r(\xi)} (A)\psi\|_{1/2}$.
Applying~\eqref{eq:weetsjon} and noting that
\[ |\omega(\xi,\eta)| = |\lra{D\xi,\eta}_{\mu}|\leq \|D\xi\|_\mu \|\eta\|_\mu
\quad \mbox{ and }\quad \|D\xi\|_{\mu} \leq \|D\xi\|_{B_{\eps}\mu},\] the second
term on the right hand side of \eqref{eq:weetsjon} satisfies
\begin{equation}\label{eq:woeiwaai1}
\|\omega(\xi, \mathrm{ad}^{n-2}_{\xi}(D\xi) ) \psi \|_{1/2} \leq
 (c_{\fk} \|\xi\|_{\infty})^{n-2}\|D\xi\|^2_{B_{\eps}\mu} \|\psi\|_{1/2}\,.
\end{equation}
Applying \eqref{eq:poedeltje} and \eqref{eq:vlotruit2} to the first term
on the right hand side of \eqref{eq:weetsjon}, we find
\begin{equation}\label{eq:woeiwaai2}
\|r(\mathrm{ad}^{n-1}_{\xi}(D\xi)) \psi\|_{1/2}
\leq (2 c_{\fk} q_{C^1}(\xi))^{n-1} q_1(D\xi) \|A \psi\|_{1/2}\,.
\end{equation}
Combining \eqref{eq:woeiwaai1} and \eqref{eq:woeiwaai2}
with
\eqref{eq:weetsjon},
and using that
$\|\psi\|_{1/2} \leq \|A \psi\|_{1/2}$,
 we find
\begin{align}\label{eq:ndec}
\|\mathrm{ad}^{n}_{r(\xi)} (A) \psi\|_{1/2}
\leq c_{n} \|A \psi\|_{1/2} \,,
\qquad \text{with} \qquad\qquad \nonumber \\
c_n =
q_1(D\xi)
\Big(\|D\xi\|_{B_{\eps}\mu} + (2c_{\fk} q_{C^1}(\xi)) \Big)
 (2c_{\fk} q_{C^1}(\xi))^{n-2}\,.
\end{align}

Since the series $\nu(s) := \sum_{n=1}^{\infty} \frac{c_n}{n!}s^n$ has positive
radius of convergence, we may now fix some
$t_0 > 0$ with $\nu(t_0) < 1$ and assume that $0 \leq s,t \leq t_0$.
Applying \cite[Thm.~1]{Nel59} to $\cH_{1/2}$ guarantees that
for $\varpi(s) := \int_{0}^{s}(1-\nu(t))^{-1}dt$, we have
\[
\sum_{n=0}^{\infty} \frac{s^n}{n!}\| r(\xi)^n \psi\|_{1/2}
\leq
\sum_{n=0}^{\infty} \frac{(c\cdot\varpi(s))^n}{n!}\|A^n \psi\|_{1/2}\,,
\quad \mbox{ with } \quad c := q_1(\xi) \] 
 as in \eqref{eq:poedeltje}. Since $\|r(\xi)^n\psi\| \leq \|r(\xi)^n \psi\|_{1/2}$ and
$\|A^n\psi\|_{1/2} \leq \|A^{n+1}\psi\|$, this yields
\begin{equation}\label{eq:schatexp}
\sum_{n=0}^{\infty} \frac{s^n}{n!}\| r(\xi)^n \psi\|
\leq
\sum_{n=0}^{\infty} \frac{(c\cdot\varpi(s))^n}{n!}\| A^{n+1} \psi\|\,.
\end{equation}
To get an explicit estimate on the radius of convergence,
note that all norms  of (derivatives of) $\xi$ occurring in
\eqref{eq:ndec}
are dominated by $p_2(\xi)$ (cf.~\eqref{eq:normbank}).
The estimate $c_n \leq a b^{n}$ with
$a := (1+2c_{\fk})/(2c_{\fk})^2$
and $b:= 2c_{\fk} p_2(\xi)$
yields ${\nu(s) \leq a(e^{bs} -1)}$.
Accordingly, $\nu(t_0) < 1$ is ensured if
\[  b t_0
< \log\Big(1 + \frac{1}{a}\Big)
= -\log\Big(1 - \frac{1}{1+a}\Big).\]
In particular, $s = 1$ is allowed if
$p_2(\xi) < \frac{1}{2c_\fk}\log\big(1 + \frac{1}{a}\big)$.
Substituting this in $\varpi(s) = \int_{0}^s(1-\nu(t))^{-1}dt$ and integrating, we obtain
\begin{equation}\label{eq:kruutje}
\varpi(s) \leq - {\textstyle \frac{1}{(1+a)b}} \log\left(
(1+a)e^{-bs} -a
\right)\,.
\end{equation}

If $\psi$ is an analytic vector for $H$, it is analytic for $A = \one + H$
with the same radius of convergence $R_H$.
The right hand hand side of \eqref{eq:schatexp} therefore
converges absolutely if $c\cdot \varpi(s)  < R_{H}$, where $c = q_1(\xi)$.
Since
$q_1(\xi) \leq p_2(\xi)$, we find
$\frac{c}{b} \leq \frac{1}{2c_{\fk}}$, and hence
$\frac{c}{(1+a)b} \leq 2c_{\fk}/(c_{\fk} +1)^2$.
Substituting this in \eqref{eq:kruutje}, we find that
$c\cdot \varpi(s) \leq R_{H}$ if
\[
bs \leq -\log\left(
1 - {\textstyle \frac{1}{1+a}}\Big(
1- \exp\Big(-{\textstyle \frac{(c_{\fk}+1)^2}{2c_{\fk}}} R_{H}\Big)
\Big)
\right) <
-\log\Big(1 - \frac{1}{1+a}\Big).\]
Putting $s=1$, and substituting $a$ and $b$ in the above equation,
we find that \eqref{eq:schatexp} converges if $p_2(\xi)$
satisfies \eqref{eq:bahbah}.
\end{proof}

\section{The Localization Theorem}\label{sec:locthmsec}

In this section, we use the continuity and analyticity results
from \S\ref{sec:6} to prove a \emph{Localization Theorem}.
Our main result reduces the classification of positive energy representations
of the identity component $\Gamma_c(M,\cK)_0$
to the case where the base manifold $M$ is one-dimensional.
We start in the setting of a fixed point free $\R$-action on the manifold $M$,
and extend this to more general
Lie group actions in \S\ref{sec:HigherDimSymmetry}.

\subsection{Statement and discussion of the theorem}\label{sec:StatementDiscussion}

\begin{Theorem} {\rm(Localization Theorem)} \label{thm:7.11} 
Let $\pi \: \cK \rightarrow M$ be a Lie group bundle with
$1$-connected, semisimple fibers. Let
$\gamma_\cK \: \R \to \Aut(\cK)$ be a
homomorphism that defines a smooth action on~$\cK$, and induces
a fixed-point free flow $\gamma_M$ on~$M$.
Then, for every projective positive energy representation
\[\oline\rho \: \Gamma_c(M,\cK)_0 \to \PU(\cH)\]
of the connected gauge group $\Gamma_c(M,\cK)_0$, there exists
a one-dimensional, closed, embedded, flow-invariant submanifold
$S \subeq M$ such that $\ol{\rho}$
factors through a 
projective positive energy
representation $\ol{\rho}_{S}$ of the connected Lie group $\Gamma_c(S,\cK)$. The diagram
\begin{center}
$ $
\xymatrix{
\Gamma_{c}(M,\cK)_{0} \ar[d]^{r_{S}} \ar[r]^{\ol{\rho}} & \mathrm{PU}(\cH)\\
\Gamma_{c}(S,\cK) \ar[ru]_{\ol{\rho}_{S}}
}
\end{center}
commutes,
where $ r_S\: \Gamma_c(M,\cK)_0 \to \Gamma_c(S,\cK)$ is the
restriction homomorphism.
\end{Theorem}

\begin{Remark}
It is convenient to define $\Gamma_{c}(\emptyset,\cK) := \{\one\}$,
so that the above theorem holds for the trivial representation
with $S = \emptyset$.
\end{Remark}

\begin{Remark}{\rm (Localization for the simply connected cover)}
In fact, we will prove a slightly stronger result: every
projective positive energy representation
\[\oline\rho \: \tilde\Gamma_c(M,\cK)_{0} \to \PU(\cH)\]
of the simply connected cover of $\Gamma_{c}(M,\cK)_{0}$
factors through the composition
$\tilde r_S := {r_S \circ q_{\Gamma}}$
of the covering map $q_{\Gamma} \colon \tilde\Gamma_c(M,\cK)_{0} \rightarrow \Gamma_c(M,\cK)$
with the restriction ${r_S\: \Gamma_c(M,\cK)_0 \to \Gamma_c(S,\cK)}$.
This strengthening of Theorem~\ref{thm:7.11} is needed in Part II, where we handle localization for
gauge groups on manifolds with boundary.
\end{Remark}

Note that $M$ is not required to be compact or connected, and that
the fibers of $\cK\rightarrow M$ are not required to be compact.
The result for noncompact $M$ is a major feature, which we will use
extensively later on (see \S\ref{sec:9} and Part II). Allowing noncompact fibers, however, is not a big step.
Indeed, noncompact simple fibers result in trivial representations
by Theorem~\ref{red2cpt}, so we already know that the theorem
holds with $S=\emptyset$ in that case.
Before proceeding with the proof in \S\ref{sec:kleinbewijs},
we show that the assumption of 1-connectedness of the fibers
is not essential.

\begin{Remark}[Non-simply connected fibers] \label{rem:7.6}
Suppose that the typical fibers of $\cK \rightarrow M$ are connected, but not necessarily simply connected.
Let $K_i$ be the typical fiber over the connected component $M_i$ of~$M$,
and
let $\widetilde{K}_i$ be its universal 1-connected cover.
The kernel $\pi_1(K_i)$ of the covering map $\widetilde{K}_{i} \twoheadrightarrow K_i$
is a finite, central subgroup, yielding a central extension
\begin{equation}\label{eq:exactgroupsingle}
\pi_1(K_i) \hookrightarrow
\widetilde{K}_i \twoheadrightarrow K_i\,.
\end{equation}
%
The natural
inclusion $\Aut(K_i) \into \Aut(\tilde{K}_{i})$, obtained by the canonical lift of
automorphisms, yields a Lie group bundle $\tilde{\cK}_{i} \to M_{i}$ with
fiber $\tilde K_i$
over each $M_i$, and hence a Lie group bundle $\tilde \cK \to M$ over $M$.
 It comes with a
natural bundle map $\tilde{\cK} \to \cK$ over the identity of $M$,
which restricts to the universal covering map on every fiber.
The kernel $\cZ \subseteq \tilde \cK$ of this map is a bundle of discrete,
abelian groups, whose fibers over $M_i$ are isomorphic to $\pi_1(K_i)$.
Analogous to \eqref{eq:exactgroupsingle}, we thus obtain an exact sequence
of Lie group bundles
\begin{equation}
\cZ \hookrightarrow \tilde \cK \twoheadrightarrow \cK\,.
\end{equation}

The 1-parameter group $\gamma_{\cK} \colon \R \rightarrow \Aut(\cK)$ lifts to
$\gamma_{\widetilde{\cK}} \colon \R \rightarrow \Aut(\widetilde{\cK})$
with the same infinitesimal generator $\bv \in \Gamma(M,\mathfrak{a}(\fK))$
(cf. Remark~\ref{remarkalgebroids}).
As every smooth section
$\xi \in \Gamma_c(M,\cK)_0$ lifts to a section of $\tilde\cK$
because the natural map
$\Gamma_c(M,\tilde\cK) \to \Gamma_c(M,\cK)$ is a covering morphism of Lie groups,
the projection $\tilde\cK \to \cK$
yields a surjective Lie group homomorphism, and hence an exact sequence
\begin{equation}\label{eq:exactgroupbundle}
\Gamma_{c}(M,\cZ) \hookrightarrow
\Gamma_c(M,\tilde\cK) \rightarrow 
\Gamma_c(M,\cK).
\end{equation}
Since the fibers of $\cZ$ are discrete, the group $\Gamma_{c}(M_i,\cZ_i)$
of compactly supported sections of $\cZ_i \rightarrow M_i$ is trivial
if $M_i$ is noncompact.
If $M_i$ is compact, $\Gamma_{c}(M_i,\cZ_i)$ can be identified with
$\pi_1(K_i)^{\pi_1(M_i)}$, the fixed point subgroup of $\pi_1(K_i)$
under the monodromy action $\pi_1(M_i) \rightarrow \Aut(\pi_1(K_i))$.
We thus obtain an isomorphism
\begin{equation}\label{eq:formulapione}
\Gamma_{c}(M,\cZ) \simeq  \prod_{i\in I}{}^{'} \pi_1(K_i)^{\pi_1(M_i)}
\end{equation}
of discrete groups where $\prod_{i\in I}' \pi_1(K_i)^{\pi_1(M_i)}$ denotes
the weak direct product of the finite abelian groups $\pi_1(K_i)^{\pi_1(M_i)}$
(all tuples with finitely many non-zero entries),
running over all $i$ for which the
connected component $M_i$ is compact.
In particular, it follows from \eqref{eq:exactgroupbundle} and \eqref{eq:formulapione} that
projective positive energy representations of
$\Gamma_c(M,\cK)_0$ correspond to
projective positive energy representations of $\Gamma_c(M,\widetilde{\cK})_0$
that are trivial on
$Z_{[M]} := \Gamma_{c}(M,\cZ) \cap \Gamma_{c}(M,\tilde{\cK})_{0}$.
\end{Remark}

Note that the embedding $S \hookrightarrow M$
yields a `diagonal' morphism $Z_{[M]} \rightarrow Z_{[S]}$.
The term `diagonal' is justified by the special case
where $\cK$ is a trivial bundle over a compact, connected manifold $M$.
Then the embedded 1-dimensional submanifold
$\emptyset \neq S \subseteq M$ is the disjoint
union of $N$ circles, and $Z_{[M]} \simeq \pi_1(K)$ can literally
be identified with the diagonal subgroup of $Z_{[S]} \simeq \pi_1(K)^N$.

Combining Theorem~\ref{thm:7.11} with
Remark~\ref{rem:7.6}, we obtain a Localization Theorem
for bundles whose fibers are not necessarily simply connected.

\begin{Corollary}\label{cor:locnonsimply}{\rm (Localization Theorem for non-simply connected fibers)}
Suppose that the fibers of $\cK \rightarrow M$ are connected, but not
necessarily simply connected. Then
$\ol \rho$ arises by factorization from
a projective positive energy
representation of
$\Gamma_{c}(S,\tilde{\cK})$ that is trivial on
the image of  $Z_{[M]}$ in $Z_{[S]}$.
\end{Corollary}

\begin{Remark}{\rm(Abelian groups)} In the Localization Theorem~\ref{thm:7.11} we have
assumed that the fiber Lie group $K$ is semisimple. We now explain why this is crucial
and that there is no localization for abelian target groups, so that the
Localization Theorem does not extend to bundles with general compact fiber Lie algebras.
To this end, let $K = (\fk,+)$ be a finite-dimensional real vector space
and fix a positive definite scalar product $\kappa$ on $\fk$.
Further, let $M$ be a smooth manifold and consider the Lie group $G := \g
:= C^\infty_c(M,\fk)$,
which can be identified with the group of
compactly supported sections of the trivial bundle $\cK = M \times K$.
We also fix a smooth flow $\gamma_M \:  \R \to \Diff(M)$, its generator
$\bv_M \in \cV(M)$, and a $\gamma_M$-invariant positive Radon measure $\mu$ on $M$.
Then
\[ \kappa_\g(\xi,\eta) := \int_M \kappa(\xi,\eta)\, d\mu \]
defines a positive semidefinite scalar product on $\g$ which is invariant
under the $\R$-action on $\g$ given by $\alpha_t\xi := \xi \circ \gamma_M(t)$,
whose infinitesimal generator is $D\xi = \cL_{\bv_M} \xi$.
Then
\[  \omega(\xi,\eta) := \kappa_\g(D\xi,\eta)
= \int_M \kappa(\cL_{\bv_M}\xi,\eta)\, d\mu \]
is an $\R$-invariant skew-symmetric form on the abelian Lie algebra $\g$, hence a
$2$-cocycle. Combining Theorems~3.2 and 5.9 in \cite{NZ13}, it now follows that
all these cocycles
are obtained from projective positive energy representations of the
groups $G \rtimes_\alpha \R$. This shows that, for abelian fibers, no restrictions
on the measure $\mu$ exist.
\end{Remark}

\begin{Example} We consider the Lie algebra
$\g = C^\infty(\T^d,\fk)$, $\fk$ compact simple and
$\alpha_t(\xi) = \xi \circ \gamma_{t}$, where
\[ \gamma_t(z_1, \ldots, z_d) = (e^{2\pi it\theta_1}z_1, \ldots,
e^{2\pi it \theta_{d-1} }z_{d-1}, e^{2\pi it}z_d).\]
This means that $\bv_M$ is the invariant vector field on the Lie group
$\T^d \cong \R^d/\Z^d$ with exponential function
\[ \exp(x_1,\ldots, x_d) = (e^{2\pi i x_1}, \ldots, e^{2\pi i x_d}) \]
whose value in $\one$ is given by $x = (\theta_1, \ldots, \theta_{d-1},1)$.
This action has a closed orbit if and only if
the one-parameter group $A := \exp(\R x)$ is closed, which is equivalent
to $\theta_j \in \Q$ for all $j$.

If this condition is satisfied, then $A \cong \T$
and the $\alpha$-orbits are the $A$-cosets in the group $\T^d$.
This situation is also studied by Torresani in \cite{To87}.
If this condition is not satisfied, then the Localization Theorem
implies that there are no non-trivial projective positive energy
representations. \end{Example}

\begin{Remark} \label{rem:7.7}
The Localization Theorem also yields partial information
for flows with fixed points, and for manifolds with boundary.

(a)
If the vector field $\bv_M$ has zeros, then
\[ M^\times := \{ x  \in M \: \bv_M(x) \not=0\} \]
is an open flow-invariant submanifold of $M$ and the
Localization Theorem applies to the bundle $\fK\res_{M^\times}$.
In this context, this theorem does not provide a complete reduction
to the one-dimensional case for two reasons.
One is that the representations of $\Gamma_{c}(M^{\times},\fK)$
do not uniquely determine those of  $\Gamma_{c}(M,\fK)$
and the other is that the $1$-dimensional
submanifold $S$ of $M^\times$ need not be closed in $M$,
so that the extendability of the representation
of $\Gamma_c(M^\times,\fK)$ to the Lie algebra
$\Gamma_c(M,\fK)$ provides ``boundary conditions at infinity'' for
the corresponding representations of $\Gamma_c(S,\fK)$.
We will further explore these boundary conditions in future work.

(b) Similarly, if $\ol{M}$ is a manifold with boundary, then both its interior
$M = \ol{M} \setminus\partial M$ and its boundary $\partial M$ are invariant under the flow.
In Part~II of this series of papers, we apply the Localization Theorem to
$M$ and $\partial M$ separately, and combine the information to obtain classification
results for positive energy representations of the gauge group $\Gamma(\ol{M}, \cK)$.
The main challenge here is that although
every projective unitary representation of $\Gamma(\ol{M},\cK)$ automatically restricts to
$\Gamma_{c}(M,\cK)$, we heavily rely on the positive energy condition to obtain a representation of
$\Gamma_{c}(\partial M, \cK)$.
\end{Remark}

\begin{Example} \label{ex:R2circles} A typical example of a flow with fixed points
is the 2-sphere $M = \bS^2$, where
\begin{equation}\label{eq:circleactions2}
\gamma_{M,t}(x,y,z) = \begin{pmatrix} \cos(t) & \sin(t) & 0\\
						-\sin(t) & \cos(t) & 0\\
						0&0&1
		 	\end{pmatrix} \begin{pmatrix}x\\y\\z\end{pmatrix}
\end{equation}
is the rotation around the $z$-axis with unit angular velocity, and
$P = \bS^2 \times \fk$ is the trivial bundle.
The lift of the infinitesimal action is then given by
\begin{equation}\label{eq:liftsphere}
	\bv(x,y,z) = (y\partial_x - x\partial_y) + A(x,y,z),
\end{equation}
where the first part is the horizontal lift of the infinitesimal action corresponding to \eqref{eq:circleactions2},
and the second part is the vertical vector field corresponding to a smooth function $A \colon \bS^2 \rightarrow \fk$.

Then $M^\times = \bS^2 \setminus \{(0,0,\pm 1)\}$, and the integral curves on $\bS^2$
are precisely the circles of latitude.
Therefore $S$ is either compact and a finite union of circles, or it is
non-compact and an infinite union of circles. More precisely,
\[
S = \big\{ (x,y,z) \in \bS^2 \: z \in J\big\},
\]
where $J \subset (-1,1)$ is a discrete set that has at most two accumulation points $\pm 1$,
corresponding to the two fixed points of the circle action.
We return to this example in \S\ref{sec:fixedpoint}.
\end{Example}

\subsection{Localization at the Lie algebra level}\label{sec:kleinbewijs}
The remainder of this section is devoted to the proof of Theorem~\ref{thm:7.11}.
We start by proving the statement at the level of Lie algebras.
This proceeds through several lemmas. In the first one,
relying heavily on Theorem~\ref{Shylock} and Lemma~\ref{lem:6.26},
we
derive integrality results
for the flow-invariant measure $\mu$ of \S\ref{subsec:5.2redcurmeas}.

\begin{Lemma}\label{lem:integralitymeasure}
Suppose that the fibers of $\fK \rightarrow M$ are simple Lie algebras.
Let $U \simeq U_0 \times I \subseteq M$ be a good flow box around $x\in M$
in the sense of {\rm Definition~\ref{def:goodflowbox}}, so that
the restriction of the invariant measure $\mu$
to $U \simeq U_0 \times I$ takes the form $\mu|_{U} = \mu_0 \otimes dt$.
Then, for every measurable subset $E \subseteq U_0$,
\[
\mu_0(E) \in \frac{1}{2\pi} \N_0\,.
\]
\end{Lemma}
\begin{proof}
We may assume without loss of generality that the fibers of $\fK$ over $U$
are compact, as $\mu_0(E)$ would otherwise be zero by Corollary~\ref{Cer:acorro}.
Let ${\chi_{E} \colon U_0 \rightarrow \{0,1\}}$ be the
indicator function of $E$.
Consider
the Lie algebra homomorphism
\[ \iota_{E} \colon \R C\oplus_\omega H^{2}_{\partial}(I,\fk) \rightarrow
\R C\oplus_\omega H^2_{\partial}(U,\fk), \qquad
z C \oplus \xi  \mapsto z C \oplus \chi_{E}\xi\]
whose continuity follows from Corollary~\ref{cor:sietsiet}.
If we pull back the representation $r$ of $\R C\oplus_{\omega} H^2_{\partial}(U,\fk)$
of Theorem~\ref{Shylock} along $\iota_E$, we obtain a projective
$*$-representation of the Banach--Lie algebra $\fh := H^2_\partial(I,\fk)$.
By Lemma~\ref{lem:6.26}, its space of analytic vectors is dense
in $\cH$.

Since $\fh$ consists of functions $I \to \fk$ and it contains
$C^\infty_c(I,\fk)$, the fact that $\fz(\fk) = \{0\}$ implies that the center
of $\fh$ is trivial.
As $\fh$ is a Banach--Lie algebra, it is in particular locally exponential,
so there exists a $1$-connected
Lie group $H$ with Lie algebra $\fh$ by \cite[Thm.~IV.3.8]{Ne06} (see \cite{GN} for a complete proof).

Now Theorem~\ref{thm:7.4} provides a smooth,
projective, unitary representation ${\pi \: H \to \PU(\cH)}$.
By Theorem~\ref{MeasureThm},
the corresponding Lie algebra cocycle is given by
\begin{align} \label{eq:omega-form}
\omega(\xi,\eta)
&= -\int_U \kappa(\chi_E \xi,\nabla_{\bv_M}(\chi_E\eta))\, d\mu
= -\int_{U_0 \times I} \kappa(\chi_E \xi,\chi_E \eta')\, d\mu_0\, dt  \notag \\
&= -\mu_0(E) \int_I \kappa(\xi,\eta')\, dt
= \mu_0(E) \int_I \kappa(\xi',\eta)\, dt\,.
\end{align}
Theorem~\ref{ThmProjRepLARep} now implies the existence of a
central Lie group extension $H^\sharp$ of $H$ by $\T\cong \R/2\pi \Z$
with Lie algebra $\fh^\sharp_\omega = \R C \oplus_\omega \fh$.


This in turn implies integrality conditions on the values of $\mu_0(E)$.
To see how these can be obtained, we associate to $\omega$ the corresponding
left invariant $2$-form $\Omega$ on $H$ with $\Omega_\one = \omega$. This form
defines a {\it period homomorphism} \index{period homomorphism \scheiding $\per_\omega$}
\[  \per_\omega \: \pi_2(H) \to \R, \quad
[\sigma] \mapsto \int_\sigma \Omega \]
(cf.\ \cite[Def.~5.8]{Ne02}), and \cite[Lemma~5.11]{Ne02} implies that
\[ \im(\per_\omega) \subeq  2\pi \Z.\]
Since the rescaling map
\[\gamma\colon C^\infty_c(I,\fk) \rightarrow C^{\infty}_{c}((-\pi,\pi),\fk),\quad
\gamma(\xi)(\theta) = \xi({\textstyle\frac{T}{2\pi}}\theta)
\]
from the interval $I = (-T/2, T/2)$ to the interval $(-\pi,\pi)$
is an isomorphism of Lie algebras,
the cocycle $\int_I \kappa(\xi',\eta)\, dt$ on $C^\infty_c(I,\fk)$
has the same period group as the cocycle $\int_{-\pi}^{\pi} \kappa(\xi',\eta)\, d\theta$ on
$C^\infty_c((-\pi,\pi),\fk)$.
In \cite[Lemma~V.11]{Ne04}, it was shown that this, in turn,
has the same period group as the cocycle
$\int_{-\pi}^{\pi}\kappa(\xi',\eta)\, d\theta$
on $C^\infty_c(\bS^1,\fk)$.
By \cite[Thm.~II.5]{Ne01}, the period group of
$\frac{1}{2\pi}\int_{-\pi}^{\pi}\kappa(\xi',\eta)d\theta$
is $2\pi \Z$, provided that $\kappa$ is normalized as in \eqref{eq:kappa-normal}.
Combining all this, we conclude that $\mu_0(E)  \in {\textstyle \frac{1}{2\pi}}\Z$.
%
\end{proof}

As the measure $2\pi \mu_0$ takes integral values, the following
proposition shows that it is automatically discrete.

\begin{Proposition} \label{prop:discretemeasure}
Let $\zeta$ be a locally finite, regular Borel measure on a locally compact
space $\Sigma$.
If $\zeta$ takes values in $\N_0 \cup \{\infty\}$, then
there exists a locally finite
subset $\Lambda \subseteq \Sigma$ and natural numbers $c_x = \zeta(\{x\})$
such that $\zeta = \sum_{x\in \Lambda} c_{x} \delta_{x}$.
\end{Proposition}

\begin{proof}
By regularity, $\zeta$ is determined by its values on compact subsets,
so it suffices to assume that $\Sigma$ is compact and to show that, in this case,
$\zeta$ is a finite sum of Dirac measures.

Let $\cF$ be the family of compact subsets of full measure.
For $F_1, F_2\in \cF$, we have $\mu(F_1\backslash F_2) = \mu(F_2\backslash F_1) = 0$,
so that $F_1 \cap F_2$ also has full measure. This shows that $\cF$ is closed under finite intersections.
We show that ${C := \bigcap_{F\in \cF} F}$ has full measure.
Let $V$ be an open set containing $C$.
Since the open complements $F^{c}$ cover the
compact set $V^{c}$, there exist finitely many $F_i$ such that
${F^{c}_{1}\cup \ldots \cup F^{c}_{k} \supseteq V^{c}}$, and hence
$F_{1}\cap \ldots \cap F_{k} \subseteq V$.
Since $\cF$ is closed under finite intersections, every open set $V$
containing $C$ has full measure. By regularity, we conclude that $C$ has full measure itself.

Pick $x \in C$. For any open neighborhood $U$ of $x$ in $C$, the minimality
of
$C$ implies that $\zeta(C\setminus U) < \zeta(C)$, so that $\zeta(U) > 0$.
Let $U$ be an open neighborhood of $x$ in $C$ for which $\zeta(U)$ is minimal;
here we use that the values of $\zeta$ are contained in $\N_0$.
For any smaller open neighborhood $V\subeq U$ of $x$ in $C$
we then have ${\zeta(V) = \zeta(U)}$ and therefore $\zeta(U \setminus V) = 0$.
This implies that $\zeta(K) = 0$ for any compact subset
$K \subeq U \setminus \{x\}$ and hence that $\zeta(U\setminus \{x\}) = 0$
by the regularity of $\zeta$. Now the minimality of $C$ entails
that $C = \{x\} \cup (C \setminus U)$. Since $x\in C$ was arbitrary,
it follows that $C$ is discrete, hence finite:
${C = \{x_1, \ldots, x_k\}}$. Accordingly, the restriction of
$\zeta$ to a compact subset is the finite sum 
${\zeta = \sum_{j = 1}^k \zeta(\{x_j\}) \delta_{x_j}}$ of
Dirac measures.
\end{proof}

Recall from Theorem~\ref{reductienaarsimpel} that the bundle $\fK \rightarrow M$ of semisimple Lie algebras gives rise to a bundle $\widehat{\fK} \rightarrow \widehat{M}$ of simple Lie algebras with $\Gamma_{c}(M,\fK) \simeq \Gamma_{c}(\widehat{M},\widehat{\fK})$. By Remark~\ref{remark:overdederivatie}, it inherits the 1-parameter group of automorphisms.

\begin{Lemma}\label{lemma:streepjes}
If the flow on $M$ has no fixed points, then the support $\widehat{S}$ of $\mu$ is a one-dimensional, flow-invariant, closed embedded submanifold of $\widehat{M}_{\mathrm{cpt}}$,
the part of~$\widehat{M}$ over which the fibers of $\widehat{\fK}$ are compact.
\end{Lemma}
\begin{proof}
Since the flow on $M$ has no fixed points,
the vector field $\bv_M$ on $M$ has no zeros.
As the same holds for its lift to $\widehat{M}$, every point
$x \in \widehat{M}$ is contained in a good flow box $U \cong U_0 \times I$
in the sense of Definition~\ref{def:goodflowbox}.
In any such flow box, the measure
$\mu$ is of the form $\mu_0 \otimes dt$, where
$\mu_0$ is a regular measure on~$U_0$.
From Lemma~\ref{lem:integralitymeasure} and Proposition~\ref{prop:discretemeasure},
we conclude that
$\mu_0$ has finite support in $U_0$, so that
$\widehat{S} \cap U \cong F \times I$, where $F \subeq U_0$ is a finite subset.
This implies that $\widehat{S}$ is a one-dimensional, closed embedded submanifold
invariant under the flow on~$\widehat{M}$.
The final statement follows from Theorem~\ref{red2cpt}.
\end{proof}

Combined with Corollary~\ref{zondagskind}, this shows that
Theorem~\ref{thm:7.11} holds at the level of Lie algebras.

\begin{Lemma}\label{lem:kleinbewijs}
There exists a
1-dimensional, closed, embedded, flow-invariant submanifold $S\subseteq M$
such that
the projective positive energy representation $\dd\rho$ of
$\Gamma_{c}(M,\fK)$
factors through the restriction map
$r_S^\fk \: \Gamma_c(M,\fK) \to \Gamma_c(S,\fK)$.
\end{Lemma}
\begin{proof}
Combining Lemma~\ref{lemma:streepjes} with Corollary~\ref{zondagskind}
and Theorem~\ref{red2cpt},
we conclude that the projective Lie algebra representation $\dd\rho$
of $\Gamma_{c}(\widehat{M},\widehat{\fK})$
vanishes on the ideal
\[ J_{\widehat{S}} := \{ \xi \in \Gamma_c(\widehat{M},\widehat{\fK}) \: \xi\res_{\widehat{S}} = 0\}.\]
It follows that the projective positive energy representation of $\Gamma_{c}(M,\fK)$
vanishes on
$J_S := \{ \xi \in \Gamma_c(M,\fK) \: \xi\res_S =0\}$, where
$S \subseteq M$ is the image of $\widehat{S}$ under the finite, $\R$-equivariant covering map $\widehat{M} \rightarrow M$.
Since $\widehat{S} \subseteq \widehat{M}$ is a
1-dimensional, closed, embedded, flow-invariant submanifold,
the same holds for $S\subseteq M$.
This implies that the projective representation
factors through the restriction map
$r_S^\fk \: \Gamma_c(M,\fK) \to \Gamma_c(S,\fK)$,
which is a quotient map of locally convex spaces.
\end{proof}

\subsection{Twisted loop groups}
\label{subsec:8.1}

Let $S$ be a one-dimensional, embedded, flow-invariant submanifold of $M$.
Then $S = \bigsqcup_{j\in J}S_j$ is the disjoint union of its connected components
$S_j$, which are either diffeomorphic
to $\R$ (for a non-periodic orbit), or to $\bS^1 \cong \R/\Z$
(for a periodic orbit).

Fix $j \in J$ and let $K = K_j$ denote the fiber of $\cK\res_{S_j}$.
If $S_j \cong \R$, then the bundle $\cK\res_{S_j}$ is trivial,
i.e., equivalent to $S_j \times K \cong \R \times K$. This trivialization
can be achieved $\R$-equivariantly, using an integral curve in the corresponding
frame bundle $\Aut(\cK) \to \R$, a principal bundle with fiber $\Aut(K)$.
The action of $\R$ on $\cK$ 
is then simply given by
\begin{equation}
  \label{eq:action-on-line}
\gamma_t(x,k) = (x+t, k) \quad \text{for}\quad t,x \in \R \quad\text{and}\quad k \in K.
\end{equation}

If $S_j \cong \bS^1$ is a periodic orbit,
then the universal covering map
$q_j \: \tilde S_j \to S_j$ can be identified with the
quotient map $\R \to \R/\Z$.
If the period of the orbit $S_j$ is $T$,
then we scale the $\R$-action on $\bS^1 = \R/\Z$ by $1/T$, yielding
\[\gamma_{\bS^1,t}([x]) = [x + t/T].\]
We have seen above that the pullback $q_j^*(\cK\res_{S_j})$
is equivariantly equivalent to the trivial bundle
$\R \times K$ on which $\R$ acts by translation in the first factor.
The action of the fundamental group $\pi_1(S_j) \cong \Z$
on $\R \times K$ is given by
bundle automorphisms that commute with the $\R$-action; there exists an
automorphism $\Phi \in \Aut(K)$ such that
\[ n.(x,k) = (x+ n, \Phi^{-n}(k)) \quad\mbox{ for all }\quad  n \in \Z.\]
Accordingly, we have an equivariant isomorphism
\[ \cK\res_{S_j} \cong (\R \times K)/\sim,\qquad\text{where}\qquad
(x,k) \sim (x+n, \Phi^{-n}(k))
\]  for all $
x\in\R, k \in K$ and $n \in \Z$.
We write the equivalence classes as $[x,k]$,
and we denote the $K$-bundle over $\bS^1 = \R/\Z$ obtained in this way by\index{bundle!80@of groups over $\bS^1$ \scheiding $\cK_{\Phi}$}
\[ \cK_\Phi := (\R \times K)/\sim, \quad\text{with}\quad
\cK_{\Phi} \to \R/\Z \quad\text{given by} \quad
 [x,k] \mapsto [x] = x + \Z\,.  \]
The $\R$-action is given in these terms by
\[ \gamma_t([x,k]) = [x+t/T, k].\]
Note that
\[ \gamma_T([x,k]) = [x+1, k] = [x, \Phi(k)],\]
so that $\Phi$ can be interpreted as a {\it holonomy}.  \index{holonomy \vulop}

Recall that, for two automorphisms
$\Phi,\Psi \in \Aut(K)$, the corresponding $K$-bundles
$\cK_\Phi$ and $\cK_\Psi$ are equivalent if and only if the
classes $[\Phi]$ and $[\Psi]$ are conjugate in the component group
$\pi_0(\Aut(K))$, and they are $\R$-equivariantly isomorphic if and only if
$\Phi$ and $\Psi$ are conjugate in $\Aut(K)$.

Indeed, any isomorphism
$\Gamma_{\Psi,\Phi} \: \cK_\Phi \to \cK_\Psi$
inducing the identity on the base is of the form
\[ \Gamma_{\Psi,\Phi}([x,k]) = [x, \zeta_x(k)], \]
where $\zeta \: \R \to \Aut(K)$ is smooth and satisfies
\begin{equation}
  \label{eq:zeta}
 \zeta_{x+1} = \Psi^{-1} \circ \zeta_{x} \circ \Phi \quad \mbox{ for all } \quad
x \in \R.
\end{equation}
Such a smooth curve $\zeta$ exists if and only if
$[\Phi]$ and $[\Psi]$ are conjugate in the finite group $\pi_0(\Aut(K))$.
In particular, the set of equivalence classes of group bundles with fiber $K$
over $\bS^1$
corresponds to
the set of conjugacy classes in the group  $\pi_0(\Aut(K))$,
which is finite for a semisimple compact Lie group~$K$.
This follows from the compactness of the group
$\Aut(K) \subeq \Aut(\tilde K) \cong \Aut(\fk)$
as a subgroup of $\GL(\fk)$ preserving the scalar product~$\kappa$.

The bundle isomorphism $\Gamma_{\Psi,\Phi}$ is $\R$-equivariant
if and only if the function $\zeta$ is constant. Accordingly, the two bundles
$\cK_\Phi$ and $\cK_\Psi$ are $\R$-equivariantly isomorphic if and only if
$\Phi$ and $\Psi$ are conjugate in $\Aut(K)$, so that
equivariant isomorphism classes of principal
$K$-bundles over $\bS^1$ correspond to 
conjugacy classes in the group $\Aut(K)$
(cf.\ \cite[\S4.4]{PS86}, \cite[\S9]{FHT11}).

The group $\Gamma_c(\R/\Z, \cK_\Phi)$
is isomorphic to the twisted loop group\index{loop group!10@twisted \scheiding $\cL_\Phi(K)$}
\begin{equation}
  \label{eq:twistloop1}
\cL_\Phi(K) := \{ \xi \in C^\infty(\R,K)\:
(\forall x \in \R)\ \xi(x + 1) = \Phi^{-1}(\xi(x))\}
\end{equation}
with Lie algebra\index{loop algebra!10@twisted \scheiding $\cL_\phi(\fk)$}
\begin{equation}
  \label{eq:twistloop1b}
\cL_{\phi}(\fk) := \{ \xi \in C^\infty(\R,\fk)\:
(\forall x \in \R)\ \xi(x + 1) = \phi^{-1}(\xi(x))\},
\end{equation}
where $\phi \in  \Aut(\fk)$ is the automorphism of $\fk$ induced by~$\Phi$.
The $\R$-action on $\cL_\phi(\fk)$ is given by
\[ \alpha_t(\xi)(x) =\xi(x+t/T) \quad \mbox{ and } \quad
D\xi = \frac{1}{T}\xi'.\]

In some situations it is convenient to use a slightly
different normalization for which $\Phi$ is of finite order,
but then the $\R$-action becomes more complicated.
If $K$ is compact, then $\Aut(K)$ is compact as well.
In this case, there exists a finite
subgroup $F \subeq \Aut(K)$ with $\Aut(K) = F \Aut(K)_0$
(\cite[Thm.~6.36]{HM98}) and we may choose $\Phi_0 \in [\Phi] \in \pi_0(\Aut(K))$
in such a way that $\Phi_0 \in F$.

If $\Gamma_{\Phi,\Phi_0} \: \cK_{\Phi_0} \to \cK_\Phi$ is a group bundle isomorphism
specified by the smooth curve $\zeta \: \R \to \Aut(K)$ satisfying
$\zeta_{x+1} = \Phi^{-1} \zeta_x \Phi_{0}$ for $x \in \R$ (see \eqref{eq:zeta}), then the
$\R$-action on $\Gamma_c(\R/\Z, \cK_{\Phi_0}) \cong \cL_{\Phi_0}(K)$ takes the form
\[ \tilde\alpha_t(\xi)(x) = \zeta^{-1}_x \zeta_{x+t/T} \,\xi(x+t/T)\quad \mbox{ for } \quad
\xi \in \cL_{\Phi_0}(K). \]
On the Lie algebra level we obtain the corresponding derivation given by \\ 
\[ \tilde D \xi = \frac{1}{T}\left(\xi' + \delta^l(\zeta)\xi\right),\]
where 
\[ \delta^l(\zeta) \: \R \to \Lie(\Aut(K)) = \der(\fk), \quad \delta^l(\zeta)_x
= \derat0 \zeta_x^{-1} \zeta_{t+x} \]
is the left logarithmic derivative of $\zeta$.
Identifying $\fk$ via the adjoint representation  with $\der(\fk)$,
we obtain a smooth curve $A \: \R \to \fk$ with
$\ad \circ A = \delta^l(\zeta)$ for which
\begin{equation}
  \label{eq:tilde-D}
\tilde D \xi = \frac{1}{T}\left( \xi' + [A,\xi] \right).
\end{equation}
Note that $A$ 
belongs to the twisted loop algebra $\cL_{\phi_0}(\fk)$; since
$\zeta_{x+1} = \Phi^{-1}\zeta_{x}\Phi_{0}$,
we have $\zeta_{x+1}^{-1}\zeta_{x+1+t} = \Phi_0^{-1}(\zeta_{x}^{-1}\zeta_{x+t})\Phi_0$, and hence
$\delta^{l}(\zeta)_{x+1} = \phi_0^{-1}\delta^l(\zeta)_{x}\phi_0$.
It follows that the curve $A$ satisfies $A_{x+1} = \phi^{-1}_{0}A_x$,
so that $A\in \cL_{\phi_0}(\fk)$.

\begin{Remark} \label{rem:7.13}
We denote by $\cL^{\sharp}_\Phi(K)_{c}$ the central $\T$-extension
of $\cL_\Phi(K)$ corresponding to the
Lie algebra cocycle
\[ \omega(\xi,\eta) = \frac{c}{2\pi} \int_0^1 \kappa(\xi',\eta)\, dt, \qquad
c \in \Z \]
with period group $2\pi c \,\Z$ (see the discussion in \S\ref{sec:kleinbewijs}).
If the central charge $c$ is~$1$, we omit the subscript and simply write $\cL^{\sharp}_\Phi(K)$.
Since the Lie algebra $\cL_\phi(\fk)$ of $\cL_\Phi(K)$
is perfect, \cite[Thm.~VI.3]{MN03} implies that
the $\R$-action $\alpha$ on $\cL_\Phi(K)$
lifts to a smooth $\R$-action $\alpha^\sharp$ on $\cL^{\sharp}_\Phi(K)_{c}$, and
we obtain a double extension of the form
\[ \hat\cL_\Phi(K)_{c} \cong \cL^{\sharp}_\Phi(K)_{c}\rtimes_{\alpha^\sharp} \R.\]
The $c$-fold cover $\T \twoheadrightarrow \T\colon z \mapsto z^c$ extends to a $c$-fold cover
$\cL^{\sharp}_{\Phi}(K) \twoheadrightarrow \cL^{\sharp}_{\Phi}(K)_{c}$, for which the following diagram commutes:
\begin{center}
$ $
\xymatrix{
\T \ar[d]_{z^c} \ar[r] & \cL_{\Phi}^{\sharp}(K)\ar[r]\ar[d] & \cL_{\Phi}(K)\ar[d]^{\id} \\
\T \ar[r] & \cL_{\Phi}^{\sharp}(K)_{c}\ar[r] & \cL_{\Phi}(K)\,.
}
\end{center}
Using this covering map, we can identify the representations of
$\cL_{\Phi}(K)_{c}$ with those representations of $\cL_{\Phi}(K)$
for which the roots $\{z\in \T\,;\, z^{c}=1\} \subseteq \T$ of order $c$ acts trivially.
\end{Remark}

\subsection{Localization at the group level}
\label{subsec:7.4b}

To obtain the localization result at the group level, we need the
following factorization lemma.

\begin{Lemma} \label{lem:7.11} Let $r \: G \to H$ be an open, surjective
morphism of locally exponential Lie groups, and let
$R \: G \to U$ be a continuous homomorphism of topological groups such that
\[ \Lie(\ker R) := \{ x \in \g \: \exp(\R x) \subeq \ker R\}
\supseteq \ker \Lie(r) = \Lie(\ker r).\]
Then $R$ factors through a continuous homomorphism
$\oline R \: G/(\ker r)_0 \to U$ and $r$ induces a covering morphism
$G/(\ker r)_0 \to H$ of Lie groups.
\end{Lemma}

\begin{proof} In view of \cite[Prop.~IV.3.4]{Ne06}
(see \cite{GN} for a complete proof), ${N := \ker r}$ is a closed,
locally exponential Lie subgroup of $G$. In particular, its identity component
$N_0$ is open in $N$, so that the isomorphism $G/N \to H$ of locally exponential Lie groups
leads to a covering morphism $G/N_0 \to H$ (\cite[Thm.~IV.3.5]{Ne06}).
For every $x \in \Lie(N)$,
we have $\exp (\R x)  \subeq \ker R$, so that
$N_0 = \la \exp \Lie(N)\ra \subeq \ker R$. Therefore $R$ factors through $G/N_0$.
\end{proof}

\begin{Lemma}\label{lem:1-connectedfibers}
Let $S\subseteq M$ be a closed, 1-dimensional submanifold
and suppose that the fibers of $\cK|_{S} \rightarrow S$
are 1-connected, semisimple Lie groups.
Then $\Gamma_{c}(S,\cK)$ is $1$-connected.
\end{Lemma}

For $S \cong \R/T\Z \cong \bS^1$, it follows in particular that, for a $1$-connected Lie group
$K$ and an automorphism $\Phi \in \Aut(K)$, the twisted loop group\index{loop group!20@$T$-periodic \scheiding $\cL_\Phi^{T}(K)$}.
\begin{equation}\label{eq:twistloop}
 \cL_{\Phi}^T(K) :=
\{ \xi \in C^\infty(\R,K)\:
(\forall t \in \R)\ \xi(t + T) = \Phi^{-1}(\xi(t))\}
\end{equation}
is $1$-connected.

\begin{proof} If $S$ has connected components $(S_j)_{j \in J}$ with typical fiber $K_j$
of $\cK\res_{S_j}$, then
\begin{equation}
  \label{eq:gauge-dirsum}
\Gamma_c(S,\cK) \cong {\textstyle \prod_{j \in J}^{'}} \Gamma_c(S_j,\cK).
\end{equation}
(We refer to \cite[Prop.~7.3]{Gl03} for a discussion of weak direct products of Lie groups.)

If $S_{j} \simeq \bS^1$, then $\Gamma_{c}(S_j,\cK)$
is isomorphic to the twisted loop group
$\cL_{\Phi_{j}}^T(K_j)$, where $\Phi_{j}$ is an automorphism of $K_j$.
Since $\pi_{0}(K_{j})$, $\pi_1(K_j)$ vanish, $\pi_2(K_j)$ vanishes as well.
\begin{footnote} {Since
$K_{j}$ is homotopy equivalent to a maximal compact subgroup,
this follows from Cartan's theorem \cite[Thm.~3.7]{Mim95}.}
\end{footnote}
The long exact sequence
of homotopy groups corresponding to the Serre fibration
$\ev_{0} \colon \cL_{\Phi_{j}}^T(K_j) \rightarrow K_j$ thus yields an
isomorphism between the homotopy groups $\pi_0$ and $\pi_1$
of $\cL_{\Phi_{j}}^T(K_j)$ and
$\cL_{\Phi_{j}}^T(K_j)_{*} := \ker(\ev_0)$.
Since the inclusion $\cL_{\Phi_{j}}^T(K_j)_{*} \hookrightarrow
\cL_{\Phi_{j}}^T(K_j)_{*,\mathrm{ct}}$ into the
group of continuous, based, twisted
loops is a homotopy equivalence by \cite[Cor.~3.4]{NW09},
and since
$\pi_m(\cL_{\Phi_{j}}^T(K_j)_{*,\mathrm{ct}}) \simeq \pi_m(\Omega K_j) \simeq \pi_{m+1}(K_j)$ for $m\in \N_0$
(cf.\ \cite[p.~391]{NW09}), we conclude that $\cL_{\Phi_{j}}^T(K_j)$ is 1-connected.

If $S_j \simeq \R$, then
$\Gamma_{c}(S_{j},\cK)\simeq C^\infty_c(\R,K_j)$
is $1$-connected by \cite[Thm.~A.10]{Ne04}.
From \cite[Prop.~3.3]{Gl08}, we then conclude that the
locally exponential Lie group
\eqref{eq:gauge-dirsum}
is $1$-connected. 
\end{proof}

With these topological considerations out of the way, we now complete the proof of the Localization Theorem.

\begin{proof}[{\bf Proof of Theorem~\ref{thm:7.11}.}] \label{proof:Prrofmainresult}
In Lemma~\ref{lem:kleinbewijs}, we showed that
the projective positive energy representation $\dd\rho$ of $\Gamma_{c}(M,\fK)$
factors through the restriction map
$r_S^\fk \: \Gamma_c(M,\fK) \to \Gamma_c(S,\fK)$, so
it remains to prove the corresponding factorization on the group level.
For this, apply Lemma~\ref{lem:7.11} to the locally exponential Lie groups
$G = \tilde\Gamma_{c}(M,\cK)_{0}$ and $H = \Gamma_{c}(S,\cK)$
(which are both $1$-connected by Lemma~\ref{lem:1-connectedfibers}),
and the topological group $U = \PU(\cH)$.
The homomorphism $r$ is the homomorphism
$\tilde r_S \: \tilde\Gamma_c(M,\cK)_{0} \to \Gamma_c(S,\cK)$,
induced by the restriction
$r_S \: \Gamma_c(M,\cK)_0 \to \Gamma_c(S,\cK)$, and $R$
is the projective representation
$\ol \rho \: \tilde\Gamma_c(M,\cK)_0 \to \PU(\cH)$.
We conclude that $\ol \rho$ factors through a projective positive energy
representation of the $1$-connected Lie group $\Gamma_c(S,\cK)$.

Since every representation of $\Gamma_c(M,\cK)_0$ defines by pullback a
representation of its simply connected covering, the assertion also follows
for representations of this group.
This concludes the proof of the theorem.
\end{proof}

\subsection{Localization for equivariant representations}
\label{sec:HigherDimSymmetry}

In this subsection we 
extend the Localization Theorem~\ref{thm:7.11}
to the \emph{equivariant} setting, where
the action of $\R$ on $M$ is replaced by
a smooth action of a Lie group $P$ on~$M$.
The positive energy condition (cf.\ \S\ref{sec:eqposener}) then refers not to an $\R$-action, but to the \emph{positive energy cone}
$\cone \subseteq \fp$ inside the Lie algebra $\fp$ of $P$.

Let $M$ be a manifold, let $P$ be a Lie group
acting smoothly on~$M$,
and let $\cK \rightarrow M$ be a bundle of $1$-connected, semisimple Lie groups
that is equipped with a lift of this action.
We denote the $P$-action on~$M$ by
$\gamma_{M} \colon P \rightarrow \Diff(M)$, its lift to $\cK$ by
$\gamma \: P \to \Aut(\cK)$, and the corresponding action on the compactly supported gauge group by
$\alpha \: P \to \Aut(\Gamma_c(M,\cK))$.
On the infinitesimal level, the action of $P$ on $M$ gives rise to the
action $\bv_M \colon \fp \rightarrow \cV(M), p \mapsto \bv_M^p$
of the Lie algebra $\fp := \Lie(P)$.

Let $(\overline{\rho},\cH)$ be
a smooth,
projective, positive energy representation of the semidirect product
${\Gamma_c(M,\cK) \rtimes_{\alpha} P}$
(cf.\ Definition~\ref{def:posenerdefsym}), with positive energy cone
$\cone \subeq \fp$.

\begin{Definition}
The \emph{fixed point set}\index{fixed point set \scheiding $\Sigma$}
$\Sigma \subseteq M$ of the positive energy cone $\cone \subseteq \fp$
(a closed convex invariant cone in $\fp$) is defined as
\[
\Sigma := \{m\in M\: (\forall p \in \cone)\ \bv_{M}^{p}(m) = 0\}\,.
\]
\end{Definition}
Since the positive energy cone $\cone$ is $\Ad_{P}$-invariant, its fixed point set
$\Sigma$ is a closed, $P$-invariant subset of $M$.
In the following we first consider the fixed-point-free scenario $\Sigma = \emptyset$, and return to the general case in Part II
(\cite{JN20}).


\begin{Definition}
Let $\ol{\rho}$ be a smooth, projective, unitary representation
of $\Gamma_c(M,\cK)$.
The \emph{support} of $\ol{\rho}$, denoted $\supp(\ol{\rho})$,
is defined as the complement of the union of all open subsets
$U \subeq M$ for which the kernel of $\ol{\rho}$
contains the normal subgroup $\Gamma_c(U,\cK)$.
Similarly, the {\it support} of $\dd\rho$ \index{support of representation \vulop}
is the complement of the union of all open sets
$U\subseteq M$ such that the kernel of 
$\dd\rho$ 
contains $\Gamma_{c}(U,\fK)$.
\end{Definition}


Note that the support is a closed subset of $M$. If the representation $\oline\rho$ extends to the semidirect product
$\Gamma_c(M,\cK) \rtimes_{\alpha} P$, then the support of $\ol{\rho}$ is invariant
under the action of $P$ on~$M$. This leads to severe restrictions
for positive energy representations.

\begin{Theorem}[Equivariant Localization Theorem]\label{Thm:equivarloc}
Let $(\overline{\rho},\cH)$ be
a smooth, projective, positive energy representation of
$\Gamma_c(M,\cK)_{0} \rtimes_{\alpha} P$,
and suppose that $\cone$ has no fixed points.
Then there exists a 1-dimensional, $P$-equivariantly embedded
submanifold $S \subseteq M$ such that
$\ol{\rho}$ factors through the restriction homomorphism
$r_{S} \colon \Gamma_{c}(M,\cK)_{0} \rightarrow \Gamma_{c}(S,\cK)$.
\end{Theorem}
\begin{Remark}{\rm{(Equivariant localization for the simply connected cover)}}
Since the $P$-action on $\Gamma_{c}(M,\cK)$ preserves
the identity component $\Gamma_{c}(M,\cK)_{0}$, it lifts to the simply connected cover
$\tilde\Gamma_{c}(M,\cK)_{0}$.
In this context the same result remains valid: every
smooth, projective, positive energy representation $\ol{\rho}$ of
${\tilde\Gamma_c(M,\cK) \rtimes_{\alpha} P}$ factors through the homomorphism
$\tilde r_{S} \colon \tilde\Gamma_{c}(M,\cK)_{0} \rightarrow \Gamma_{c}(S,\cK)$
obtained by composing the restriction $r_S$ with the covering map.
\end{Remark}
\begin{proof}
For every $p\in \cone$, let $U_{p} \subseteq M$ be
the open
set of points in $M$ where $\bv_{M}^{p}$ is nonvanishing.
Applying Lemma~\ref{lem:kleinbewijs} to the manifold $U_{p}$, with
the gauge group $\Gamma_{c}(U_{p},\cK)$  and
the $\R$-action $\alpha_{p}(t):= \alpha(\exp(tp))$, one finds
an embedded, 1-dimensional submanifold $S_{p} \subseteq U_{p}$ such that
the projective Lie algebra representation $\dd\rho$ factors through
the restriction map $r^{\fk}_{S_{p}} \colon \Gamma_{c}(U_{p},\fK) \rightarrow \Gamma_{c}(S_{p},\fK)$.
The support of $\dd\rho|_{\Gamma_{c}(U_{p},\fK)}$
is thus contained in $S_{p}$. It actually equals
$S_{p}$ because the cocycle on $\Gamma_{c}(U_{p},\fK)$ is given
by a measure with support $S_{p}$. The sets
$S_{p}$ and $S_{p'}$ therefore coincide on $U_{p}\cap U_{p'}$, so
the union $S = \bigcup_{p\in \cone} S_{p}$
is a 1-dimensional, closed embedded submanifold of $M$. Here
we use that the $U_p$ cover $M$ because $\cone$ has no common fixed point.
Since $g S_{p} = S_{\mathrm{Ad}_{g}(p)}$ for every $g\in P$,
the union $S$ is $P$-invariant.


Let $I_{S} := \{\xi \in \Gamma_{c}(M,\fK)\,;\, \xi|_{S} = 0\}$ be the vanishing ideal
of $S$ in $\Gamma_{c}(M,\fK)$.
Since any $\xi \in I_{S}$ can be written as a finite sum of $\xi_{p}\in I_{S_{p}}\subseteq \Gamma_{c}(U_p,\fK)$,
and since the restriction of $\dd\rho$ to $\Gamma_{c}(U_{p},\fK)$ vanishes on
$I_{S_{p}}$,
we conclude that $\dd\rho$ vanishes on $I_{S}$.
From Lemma~\ref{lem:7.11} and Lemma~\ref{lem:1-connectedfibers},
we then find (as in the proof of Theorem~\ref{thm:7.11}
on p.~\pageref{proof:Prrofmainresult}) that \ol{\rho} factors through the restriction
$\Gamma_{c}(M,\cK)_{0} \rightarrow \Gamma_{c}(S,\cK)$
and that the corresponding assertion holds for representations
of the covering group $\tilde\Gamma_c(M,\cK)_{0}$.
\end{proof}

The building blocks
for positive energy representations therefore come from
actions of $P$ on one-dimensional manifolds on which $\cone$ has no fixed point.
According to the classification of hyperplane subalgebras of
finite dimensional Lie algebras (\cite{Ho65, Ho90}), an effective action of
a connected finite-dimensional Lie group $P$
on a simply connected one-dimensional manifold is of one of the following 3
types:
\begin{itemize}
\item the action of $P = \R$ on the line $\R$,
\item the action of the affine group $P = \Aff(\R)$ on the real line $\R$,
\item the action of $P = \widetilde{\SL}(2, \R)$ on the real
line $\R$, considered as
the simply connected cover of $\bP_1(\R)\cong \bS^1$.
\end{itemize}
In the infinite dimensional context, the action of the simply connected
covering group $P = \tilde\Diff_+(\bS^1)$
on $\R \cong \tilde\bS^1$ is a natural example.

\section{The classification for \texorpdfstring{$M$}{M} compact}\label{sec:8}


If the flow $\gamma_M$ on $M$
has no fixed points,
the Localization Theorem~\ref{thm:7.11} reduces the classification of
projective positive energy representations of the identity component $\Gamma_{c}(M,\cK)_{0}$
of the compactly supported gauge group
to the situation where the base manifold is a closed, embedded, flow-invariant
submanifold ${S\subseteq M}$ of dimension one.

The connected components of $S$
are either diffeomorphic
to $\R$ (for a non-periodic orbit), or to $\bS^1 \cong \R/\Z$
(for a periodic orbit).
Since a gauge group on $\R$ is
equivariantly isomorphic to
$C^\infty_c(\R,K)$ (with $\R$ acting by translation), and a gauge group on
$\bS^1$ is equivariantly isomorphic to a twisted
loop group (with $\R$ acting by rotation), the gauge group on $S$
is a product of twisted loop groups and groups of the form $C^{\infty}_{c}(\R,K)$.

In this section, we describe the complete classification of positive
energy representations for twisted loop groups.
This
leads to a classification of the positive
energy representations of $\Gamma_{c}(M,\cK)_{0}$ for which the one-dimensional submanifold
$S$ is compact.
Since this is automatically the case if $M$ is compact, we arrive at a complete classification in this setting.

\subsection{Positive energy representation of twisted loop groups} \label{subsec:7.4}

We now describe the complete classification of projective positive
energy representations for twisted  loop groups.

In this subsection $K$ denotes a $1$-connected compact (hence semisimple) Lie group,
$\Phi \in \Aut(K)$ is an automorphism of \emph{finite order} $\Phi^N = \id_K$,
and $\phi = \Lie(\Phi) \in \Aut(\fk)$ is the corresponding automorphism of $\fk$.
We further assume that the invariant form $\kappa$ on $\fk$
is normalized in such a way that
$\kappa(i\alpha^\vee, i\alpha^\vee) = 2$ for all long roots~$\alpha$.
We denote the (twisted) loop groups and algebras by
$\cL_\Phi(K)$ and $\cL_\phi(\fk)$ respectively, as in \eqref{eq:twistloop1} and \eqref{eq:twistloop1b}.
The (double) extensions with $c=1$ are
denoted%
\index{loop group!30@extended \scheiding $\cL^{\sharp}_\phi(K)$}
\index{loop group!40@Kac-Moody \scheiding $\widehat{\cL}_\phi(K)$}
by $\cL^{\sharp}_{\Phi}(K)$ and $\widehat{\cL}_{\Phi}(K)$, cf.~Remark~\ref{rem:7.13}.



\begin{Definition}  We call a positive energy representation $(\rho, \cH)$ of
$\hat\cL_\Phi(K)$
\begin{itemize}
\item[\rm(i)] {\it basic}  \index{representation!basic \vulop}
if $U_t := \rho(\exp tD)\subeq \rho(\cL^{\sharp}_\Phi(K))''$
for every $t \in \R$,
\item[\rm(ii)] {\it periodic} \index{representation!periodic \vulop}
if $U_T  = \one$ for some $T > 0$.
\end{itemize}
\end{Definition}
Note that if $\rho$ is minimal (Definition~\ref{def:mini}), then it is in particular
basic.


\begin{Remark} \label{rem:x.1} If $(\rho, \cH)$ is periodic
with $U_T =  \one$,
then \cite[Lemma~5.1]{Ne14a} implies that the space
$\cH^\infty$ of smooth vectors is invariant under the operators
\[ p_n(v) := \frac{1}{T} \int_0^T e^{- 2\pi i nt/T} U_tv\, dt\,.\]
These are orthogonal projections onto the eigenvectors of
$H = i\dd\rho(D)$ for the eigenvalues $-2 \pi n/T$, $n \in \Z$.
%
\end{Remark}
Recall from \S\ref{subsec:8.1}
that with the identification $\bS^1 \simeq \R/\Z$ and with $\Phi^{N} = \mathrm{id}_{K}$, we have $D(\xi) = \frac{1}{T}(\xi' + [A,\xi])$.
It will be convenient to introduce the derivative
\begin{equation}\label{eq:defvanbd}
\bd(\xi) = \frac{d}{dx}\xi\,,
\quad\mbox{so that}\quad D = \frac{1}{T}(\bd + \mathrm{ad}_{A}).
\end{equation}

\begin{Remark} \label{rem:indep-posen}
(Independence of positive energy condition from lift of $\R$-action) 
From Proposition~\ref{Prop:propsembo},  
applied to $M = \R/\Z$, it follows
that
a smooth representation of $\cL_\Phi^\sharp(K)$
is of positive energy with respect to
the derivation $D$
if and only if it is of positive energy with respect to the derivation $\bd$.
Then the representation is semibounded
in the sense of Definition~\ref{def:semibounded}.
As this holds for $D = \frac{1}{T}(\bd+ \mathrm{ad}_{L})$ with any $T>0$ and $L \in \cL_{\phi}(\fk)$,
the positive energy condition does not depend on the choice
of the vector field $\bv$ on $\cK_\Phi = \R \times_\Phi K$ lifting the
vector field $\bv_{M} = \frac{1}{T}\frac{d}{dt}$ on $\bS^1\cong \R/\Z$.
\end{Remark}


From $\Phi^N= \id_K$, we immediately derive that $\phi^N = \id_\fk$.
For 
$\hat\g = \hat\cL_\phi(\fk)$, we define the canonical triangular decomposition by
\[ \hat\g_\C = \hat\g_\C^+ \oplus \hat\g_\C^0 \oplus \hat \g_\C^- \quad \mbox{ with }\quad
\hat\g_\C^\pm := \oline{\sum_{\pm n > 0} \hat\g_\C^n}\,,\]
where
\[\hat\g_\C^n := \ker\left(\mathbf{d} + \frac{2\pi i n}{N}\one\right) \quad \mbox{ for } \quad n \in \Z\]
(see \eqref{eq:roots} in the appendix).
For $\fg = \cL_{\phi}(\fk)$, 
we have the analogous decomposition
$\fg_{\C} = \fg^+_{\C}\oplus \fg^{0}_{\C} \oplus \fg^{-}_{\C}$
with $\fg^+_{\C} = \hat \fg^+_{\C}$
and $\fg^-_{\C} = \hat \fg^-_{\C}$.

For a smooth unitary representation of $\hat\cL_\Phi(K)$, we define
its \emph{minimal energy subspace} \index{minimal energy subspace \scheiding $\cE$}
with respect to $i\dd\rho(\bd)$ by
\begin{equation}
  \label{eq:cale}
 \cE := \oline{(\cH^\infty)^{\fg_\C^-}} \quad \mbox{ for } \quad
 (\cH^\infty)^{\g_\C^-} := \{ \psi \in \cH^\infty \:
(\forall x \in \g_\C^-)\; \dd\rho(x)\psi = 0\}.
\end{equation}

\begin{Lemma} \label{lem:1.9} For every smooth 
positive energy representation $(\rho,\cH)$ of
$\hat\cL_\Phi(K)$, the subspace $\cE$ is generating
for $\cL_\Phi^\sharp(K)$.
\end{Lemma}

\begin{proof}
Note that $\cE$ is defined in terms
of $\rho\res_{\cL^\sharp_\Phi(K)}$.
In view of Corollary~\ref{cor:borch}
and the fact that $\alpha_N = \id_{\cL_\Phi(K)}$, we may therefore assume w.l.o.g.\
that $\rho$ is periodic.

Let $\cH' \subeq \cH$ denote the smallest closed
$\cL_\Phi^\sharp(K)$-invariant subspace containing~$\cE$. Then $\cH'$ is
$U$-invariant, and the representation of $\hat\cL_\Phi(K)$ on
$(\cH')^\bot$ is also a positive energy representation.
If $(\cH')^\bot\not=\{0\}$, then its minimal energy subspace $\cF$ is
non-zero by Remark~\ref{rem:x.1}, and since it contains smooth vectors,
we obtain a contradiction to $\cF \bot \cE$.
Therefore $(\cH')^\bot = \{0\}$ and the subspace $\cE$ is $\cL_\Phi^\sharp(K)$-generating.
\end{proof}

We now abbreviate
\begin{equation}
  \label{eq:gdach}
G := \cL_\Phi(K), \quad
\hat G := \hat\cL_\Phi(K)\quad \mbox{ and }  \quad
G^\sharp := \cL^\sharp_\Phi(K)
\end{equation}
and denote the corresponding groups of fixed points by
\begin{equation}
  \label{eq:defhdach}
 L = K^\Phi, \quad
\hat L := \Fix_\alpha(\hat G)\cong \T \times K^\Phi \times \R, \quad
L^\sharp := \hat L \cap G^\sharp \cong \T \times L.
\end{equation}
From the discussion in \cite[\S5.2 and App.~C]{Ne14a},
it follows that the homogeneous space
$G/L \cong \hat G/\hat L \cong G^\sharp/L^\sharp $
carries the structure of a complex
Fr\'echet manifold on which $\hat G$ acts analytically, and the tangent space in
the base point is isomorphic to the quotient space $\hat\g_\C/(\hat\g_\C^0 + \g_\C^+)$.
For any bounded unitary representation $(\rho^L,E)$ of $\hat L$, we then
obtain a holomorphic vector bundle $\bE := \hat G \times_{\hat L} E$ over
$\hat G/\hat L$.
We write $\Gamma_{\rm hol}(G/L, \bE)$ for the space of holomorphic sections of~$\bE$.
\begin{Definition}[Holomorphically induced representations] \label{def:holind}
A unitary representation $(\rho, \cH)$ of $\hat G$ is said to be
{\it holomorphically induced} \index{representation!holomorphically induced \vulop}
from $(\rho^L,E)$
if there exists a $G$-equivariant linear injection
$\Psi \: \cH \to \Gamma_{\rm hol}(G/L,\bE)$ such that the
adjoint of the evaluation map $\ev_{\one \hat L} \: \cH \to E = \bE_{\one \hat L}$
defines an isometry $\ev_{\one \hat L}^* \: E \into \cH$.
If there exists a unitary representation $(\rho, \cH)$ holomorphically induced
from $(\rho^L,E)$, then it is uniquely determined
(\cite[Def.~3.10]{Ne13}). We then call the representation
$(\rho^L,E)$ of $\hat L$ {\it (holomorphically)
inducible}. 
The same statements apply to $G^{\sharp}$ and~$L^\sharp$.
\end{Definition}

Let $\ft^\circ \subeq \fk^\phi$ be maximal abelian\index{maximal!20@abelian subalgebra \scheiding $\ft$}, so that
\[\ft = \R C \oplus \ft^\circ \oplus \R \bd\] is maximal abelian
in $\hat\cL_\phi(\fk)$.
We write $T^{\sharp}=\T \times T^{\circ}$ for the torus group\index{maximal!10@torus \scheiding $T$, $T^{\sharp}$} with
Lie algebra $\ft^{\sharp} = \R C \oplus \ft^\circ$.
Let $\Delta^+$ be a positive system\index{roots!20@positive system \scheiding $\Delta^{+}$}for the affine Kac--Moody Lie algebra
$\hat\cL_\phi(\fk_\C)$
with respect to the Cartan subalgebra $\ft_\C$ such that,
for all $\alpha \in \Delta$, the relation
$\alpha(i\bd) >0$ implies $\alpha \in \Delta^+$ (cf.~Appendix~\ref{app:1}).


\begin{Proposition} \label{prop:7.22}
A bounded representation $(\rho^L,E)$ of
\[ L^\sharp = \exp(\R C) \times L \cong \T \times L \]
is holomorphically inducible if and only if
\begin{equation}
  \label{eq:holindcon}
\dd\rho^L([z^*,z]) \geq 0 \quad \mbox{ for all} \quad z \in \g_\C^n,\, n > 0.
\end{equation}
In particular, the irreducible, holomorphically inducible representations of $L^{\sharp}$
are parametrized by the anti-dominant, integral
weights $\lambda$ of the form
\begin{equation}
  \label{eq:loweight}
\lambda = (\lambda(C),\lambda_0,0) \in i\ft^*
\end{equation}
of the affine Kac--Moody Lie algebra
$\hat\cL_\phi(\fk_\C)$ with respect to the Cartan subalgebra $\ft_\C$ and the
positive system $\Delta^+$.
For every central charge $c \in \N_{0}$,
there are only finitely many such representations
with $\lambda(C) = ic$
\end{Proposition}

\begin{proof} Since the representation $\rho^L$ of the compact group $L^{\sharp}$ is a direct sum of irreducible representations,
we may assume that
it is  a representation with lowest weight $\lambda$ with respect to the positive
system of roots $\Delta_0^+$ of $(\fk_\C^\phi,\ft^\circ_\C)$.

The necessity of \eqref{eq:holindcon} follows from  \cite[Prop.~5.6]{Ne14a}.
To show that
$\lambda$ is anti-dominant for $(\hat\cL_\phi(\fk_\C), \ft_\C, \Delta^+)$,
we need that $\lambda((\alpha,n)^\vee) \leq 0$ for $(\alpha,n) \in \Delta_+$.
We distinguish the cases $n>0$ and $n=0$.
If $n>0$, we use
\eqref{eq:coroot} in Appendix~\ref{app:1}, to see that
 \eqref{eq:holindcon} implies
$\lambda((\alpha,n)^\vee) \leq 0$ for
$0 \not=\alpha \in \Delta_0$.
For $n=0$, the assertion follows from $\lambda(\beta^\vee) \leq 0$ for $\beta \in \Delta_0^+$.

Next, we prove the integrality of $\lambda$.
For $\alpha \not=0$, the relation
\begin{equation}
  \label{eq:2}
\exp(2\pi i (\alpha,n)^\vee) = \one
\end{equation}
in $T^\sharp$
follows from the fact that
\[ \fk(\alpha,n) := \Spann_\R \{ x \otimes e_n- x^* \otimes e_{-n},
i(x \otimes e_n + x^* \otimes e_{-n}), i(\alpha,n)^\vee  \} \cong \su(2,\C). \]
Since $\lambda$ corresponds to a character of $T^\sharp$, the relation
\eqref{eq:2}
implies that \[\exp(2\pi i \lambda((\alpha,n)^\vee)) = 1,\]
so that
$\lambda((\alpha,n)^\vee)\in \Z$. We conclude that $\lambda$ is anti-dominant integral.

We now argue that every integral, anti-dominant weight $\lambda$
as in \eqref{eq:loweight} specifies a holomorphically
inducible representation $(\rho^L,E_{\lambda})$ of $L^{\sharp}$.
In fact, the unitarity of the corresponding lowest weight module
$L(\lambda,-\Delta^+)$ of $\hat\cL_\phi(\fk_\C)$ (\cite[Thm.~11.7]{Ka85})
can be used as in the proof of \cite[Thm.~5.10]{Ne14a} to see with
\cite[Thm.~C.6]{Ne14a} that $(\rho^L,E_{\lambda})$ is holomorphically inducible.
\end{proof}

The following theorem is well-known for untwisted loop groups $\cL(K)$, but
we did not find an appropriate
statement in the literature for the twisted case.
It requires some refined methods based on holomorphic induction
which we draw from \cite{Ne14a}.

\begin{Theorem} \label{thm:7.16} 
If $K$ is $1$-connected and
$(\rho, \cH)$ is a positive energy representation of $\hat\cL_\Phi(K)$,
then its restriction to $\cL^\sharp_\Phi(K)$
is a finite direct sum of factor representations of
type $I$, hence in particular a direct sum of irreducible representations.
\end{Theorem}

\begin{proof} Since the assertion only refers to
the restriction $\rho\res_{\cL^\sharp_\Phi(K)}$, we may assume w.l.o.g.\ that
$\rho = \rho_0$ is minimal (Definition~\ref{def:mini},
Theorem~\ref{thm:BAthm}).
Then ${\alpha_N = \id_{\cL_\Phi(K)}}$ implies that $\rho$ is periodic
and that every subrepresentation is generated by the fixed points of
$U_t= \rho(\one,t) = e^{-itH}$.

In view of Remark~\ref{rem:x.1}, the space $\cH^\infty$ of smooth vectors for
$\hat G = \hat\cL_\Phi(K)$ (see \eqref{eq:gdach})
is invariant under the projections $p_n\: \cH \to \cH_n$ onto the
eigenspaces of $H = i \dd\rho(D)$. Since $\rho = \rho_0$ is minimal,
we have $\cH_n = \{0\}$ for $n < 0$ and $\cH_0$ is generating.
Now $\cH_0\cap \cH^\infty$ is contained in $\cE$, the closure
of $(\cH^\infty)^{\fg_\C^-}$ from \eqref{eq:cale}.
As the intersection $\cH_0\cap \cH^\infty$ is dense in $\cH_0$, we have
$\cH_0 \subeq \cE.$

Recall that
\begin{equation}
  \label{eq:hath}
 \hat L = \Fix_\alpha(\hat G)\cong L \times \R \cong \T \times K^\Phi \times \R.
\end{equation}
As $K^\Phi$ is compact and $U_N = \one$ follows from $\phi^N = \id_\fk$,
$\rho(\hat L)$ is a compact subgroup of $\U(\cH)$.
Hence the $\hat L$-invariant subspace $\cH_0 \subeq \cE$ is a direct sum of finite dimensional
subrepresentations. In particular, it decomposes into isotypic components
$\cE_j := E_j \otimes \cM_j$, $j \in J$,
where $\cM_j \cong B(E_j,\cH_0)^{\hat L}$ is the multiplicity
space of the (finite dimensional) irreducible
representation $(\rho^L_j, E_j)$.
It also follows that  the representation of $\hat L$ on each $\cE_j$ is
semisimple in the algebraic sense and that the irreducible subrepresentations
are of the form $E_j \otimes \psi$, $\psi \in \cM_j$. As a consequence,
every $\hat L$-invariant subspace of $\cE_j$ is of the form
$E_j \otimes \cM_j'$ for a linear subspace $\cM_j' \subeq \cM_j$.


The dense subspace $(\cH^\infty)^{\g_\C^-}$ of $\cE$ is invariant under the
projections onto the isotypic components because they are given by integration
over a compact group.\begin{footnote}
{This follows  by differentiation under the integral sign,
see \cite[Prop.~1.3.25]{GN}.}
\end{footnote} This implies that
$\cE_j \cap (\cH^\infty)^{\g_\C^-}$ is dense in $\cE_j$.
In view of the preceding discussion, we thus obtain
\[ \cE_j \cap (\cH^\infty)^{\g_\C^-}
\cong E_j \otimes \cM_j^\infty \]
for a dense linear subspace $\cM_j^\infty \subeq \cM_j$.
In view of Lemma~\ref{lem:1.9},
we now have to show that, for every $\psi \in \cM_j^\infty$, the subspace
$E_j \otimes \psi \subeq \cE$ generates an irreducible
subrepresentation of $G^\sharp = \cL^\sharp_\Phi(K)$.

For the untwisted case, i.e., $\Phi = \id_K$, this follows from
\cite[Prop.~VII.1]{Ne01}.
For the twisted case we have to invoke the machinery
of holomorphic induction described in Definition~\ref{def:holind}.
For the following argument, observe that $\cL^{\sharp}_{\Phi}(K)$
is connected by Lemma~\ref{lem:1-connectedfibers}.
On the finite dimensional subspace
$E := E_j \otimes \psi \subeq (\cH^\infty)^{\g_\C^-}$, the representation
of $\hat L$ is bounded. Hence \cite[Thm.~C.3]{Ne14a} implies that
the $\hat G$-subrepresentation $(\rho',\cH')$ of $(\rho,\cH)$ generated by
$E$ is holomorphically induced from the $\hat L$-representation
$(\rho^L,E)$. In view of \cite[Thm.~C.2]{Ne14a}, the irreducibility
of $(\rho^L,E)$ implies the irreducibility of $(\rho',\cH')$.

We have seen in the proof of Proposition~\ref{prop:7.22}
that the holomorphically inducible
irreducible representations $\rho^L$ of $\hat L$
are parametrized by a set of anti-dominant integral
weights of an affine Kac--Moody algebra $\hat\cL_\psi(\fk_\C)$
with a fixed central charge.
This implies the finiteness of the possible types.
\end{proof}

The following corollary can be used to deal with gauge groups
if the structure group $K$ is not $1$-connected.
It covers in particular the case $K =\Aut(\fk)$ that arises from
structure groups of Lie algebra bundles $\fK \to \bS^1$.


\begin{Corollary} \label{cor:7.16} {\rm(Non-connected fibers)} 
If $K$ is a compact Lie group with simple Lie algebra and
$(\rho, \cH)$ a positive energy representation of $\hat\cL_\Phi(K)$,
then its restriction to $\cL^\sharp_\Phi(K)$
is a finite direct sum of factor representations of
type $I$, hence in particular a direct sum of irreducible representations.
\end{Corollary}

\begin{proof} Since $K$ is compact with simple Lie algebra, the
groups $\pi_0(K)$ and $\pi_1(K)$ are finite.
Therefore the exact sequence
\[ {\bf 1} \to \pi_1(K)/\im(\pi_1(\Phi) -\id)
\into \pi_0(\cL_\Phi(K)) \onto \pi_0(K)^\Phi \to {\bf 1}  \]
from \cite[Rem.~2.6]{NW09}(a)
implies that $\pi_0(\cL_\Phi(K))$
is finite. The identity component $\cL_\Phi(K)_0$ is isomorphic to
$\cL_\Phi(\tilde K_0)$, where $\tilde K_0$ is the simply connected covering
of the identity component $K_0$ of~$K$.
Now the assertion follows by combining Theorem~\ref{thm:7.16}
with Theorem~\ref{thm:typeicrit}.
\end{proof}

\begin{Remark} \label{rem:8.10} (Explicit aspects of the Borchers--Arveson Theorem)

(a) Let $(\rho, \cH)$ be a positive energy representation of
$\hat\cL_\Phi(K)$ for which the restriction $\rho^{\sharp}$ to
$\cL^\sharp_\Phi(K)$ is isotypic.
Then the proof of Theorem~\ref{thm:7.16} shows that
$\rho^{\sharp}$  is holomorphically induced from
$(\rho^L,\cE)$, where $\cE \cong E \otimes \cM$ and $(\rho^L,E)$ is an irreducible
representation of $L^\sharp$, and hence of $K^\Phi$.

That the representation is basic,
$U_\R \subeq \rho(G^\sharp)''$,  is equivalent to
$U_\R$ commuting with the commutant $\rho(G^\sharp)'$. Since the restriction to $\cE$
yields an isomorphism $\rho(G^\sharp)' \to
\rho^L(L^\sharp)' = \rho^L(L)'$ (\cite[Thm.~C.2]{Ne14a})
and $\cE$ is invariant under
$U_\R$, the inclusion  
$U_\R \subeq (\rho(G^\sharp)')'$ is equivalent to
$U_\R\res_{\cE} \subeq (\rho^L(L)')' = B(E) \otimes \one$.
Since $\hat L = \T \times K^{\Phi} \times \R$, where $K^{\Phi}$
is considered as a subgroup of constant sections, we have
${U_{\R}}\res_{\cE} \subseteq \rho(\hat L)'$.
The representation is therefore basic if and only if
${U_{\R}}\res_{\cE} \subseteq \rho(\hat L)' \cap \rho^L(L)'' = \C\one$, that is,
if and only if $U$ acts on $\cE$ by a
character.

(b) We construct an example which is not basic,
but which is factorial on~$G^{\sharp}$.
 Let $(\rho, \cH)$ be an irreducible positive energy representation of $\hat G
= \hat\cL_\Phi(K)$. For any nontrivial character
$\chi \: \R \to \T$, the representation
$\rho \oplus (\hat\chi \otimes \rho)$ with
${\hat\chi(g,t) := \chi(t)}$ is factorial on $G^\sharp$,
but not on~$\hat  G$.
\end{Remark}

\subsection{The classification theorem for compact base manifolds}
\label{subsec:8.3}

Let $M$ be a manifold on which the flow $\gamma_{M}$
has no fixed points, and let $K$ be a compact, connected, simple Lie group.
We now obtain a full classification of the projective
positive energy representations of $\Gamma_c(M,\cK)_0$
in the case where $M$ is compact,
by combining  the Localization Theorem~\ref{thm:7.11}
with the results on twisted loop groups from
\S\ref{subsec:7.4}.

\subsubsection{One-dimensional manifolds with compact components}\label{sec:1dcompact}

By Theorem~\ref{thm:7.11} and Corollary~\ref{cor:locnonsimply}, every projective positive energy
representation of $\Gamma_{c}(M,\cK)_{0}$ factors through
the gauge group $\Gamma_c(S,\tilde{\cK})$ of a $1$-dimensional,
$\R$-equivariantly closed embedded submanifold $S \subseteq M$.
If $S$ is compact, then  it is the disjoint union of
finitely many circles $S_j$ on which $\R$ acts with period $T_j$.

In this subsection we assume that $S$ is a (not necessarily finite) union of circles.
The restricted gauge group $G := \Gamma_c(S,\tilde{\cK})$ is
then a restricted direct product of loop groups
$\cL_{\Phi_j}(\tilde{K}_j)$, where $\tilde{K}_j$ is the 1-connected cover
of the structure group $K_j$ of $\cK\res_{S_j}$.
On the Lie algebra level, we have a direct sum of Lie algebras
\begin{equation}
\fg \cong \bigoplus_{j\in J} \cL_{\phi_j}(\fk_{j}).
\end{equation}
As in \eqref{eq:defvanbd}, the infinitesimal generator $D$ of the $\R$-action
acts on $\xi \in \fg$ by
\begin{equation} D(\xi) = \bigoplus_{j\in J} \frac{1}{T_j}\left(\bd_{j} \xi_j + [A_j, \xi_j]\right),
\end{equation}
where
$ A_j \in \cL_{\phi_{j}}(\fk)$ is determined by the $\R$-action
according to \eqref{eq:tilde-D}.

Let $(\dd\rho,\cH)$ be a positive energy representation
of $\g^{\sharp} = \R C \oplus_{\omega} \fg$
with cocycle
\[ \omega(\xi,\eta) = \sum_{j \in J} \frac{c_j}{2\pi} \int_{0}^{1}
\kappa(\xi_{j}',\eta_{j})\, dt \qquad \mbox{ with } \quad c_j \in \N_0. \]
For each $j$, $\dd\rho$ restricts to a positive energy representation
of the centrally extended twisted loop algebra $\cL^{\sharp}_{\phi_j}(\fk_{j})$
with central charge $c_{j}$.
By Proposition~\ref{prop:7.22} and Remark~\ref{rem:indep-posen}
(cf.\ \cite[Ch.~9]{PS86} for the untwisted case),
the irreducible positive energy representations
of $\cL_{\phi_j}^{\sharp}(\fk_{j})$ with central charge $c_j$ are precisely the irreducible unitary
lowest weight representations $(\dd\rho_{\lambda}, \cH_{\lambda})$ with
integral anti-dominant weight $\lambda$ satisfying $\lambda(C) = ic_j$.
Since there are finitely many of these, the representation can be written
as a finite sum
\begin{equation}
  \label{eq:isotypes}
\cH = \bigoplus_\lambda \cH_{\lambda}^j\otimes \cM_{\lambda}^j
\end{equation}
where the sum runs over the integral anti-dominant weights of
$\cL_{\phi_j}^\sharp(\fk_j)$ with central charge $c_j$ (cf.~\eqref{eq:loweight}) and
$\cL_{\phi_j}^\sharp(\fk_j)$ acts trivially on the multiplicity
space $\cM_{\lambda}^j$
(Theorem~\ref{thm:7.16}).


Now suppose that $(\rho,\cH)$ is a positive energy factor representation
of $G^\sharp$.
Then the restriction to a normal subgroup $G_j := \cL_{\Phi_j}(\tilde{K})$
decomposes discretely with finitely many isotypes
(Theorem~\ref{thm:7.16}).
For a subset $F \subeq J$, we denote
the corresponding normal subgroup of $G$ by
\[G_F := \bigoplus_{j \in F}\cL_{\Phi_j}(\tilde{K}_{j})\,.\]
Since $G_{j}$ commutes with $G_{J \setminus \{j\}}$,
the factoriality of $\rho$ on $G^{\sharp}$ implies that the restriction of
$\rho$ to
$G_j^\sharp$ is factorial as well. Hence
there is only one summand in \eqref{eq:isotypes}, and we have
\[
\cH = \cH_{\lambda}^{j} \otimes \cH'
\]
for some multiplicity space $\cH'$.
Although a priori we only have a single operator $H$ for all components $S_j$,
we now obtain an operator $\dd\rho(\bd_j)$  satisfying
\[ [\dd\rho(\bd_j),\dd\rho(\xi_i)] = \delta_{ij}\dd\rho(\xi'_j)  \]
from the minimal implementation\footnote{One could use the Segal--Sugawara
construction
(\cite[\S3]{KR87}, \cite{GW84}), but this leads to a non-zero minimal eigenvalue;
see \S\ref{subsec:9.3} for more details.}
in
Corollary~\ref{cor:borch}.

Since $H':= H - \frac{i}{T_j}\dd\rho(\bd_j + A_j)$
commutes with $\hat\cL_{\phi_j}(\fk)$, we obtain
a positive energy representation on $\cH'$ with Hamiltonian $H'$, but now
for the group $G^{\sharp}_{J \setminus\{j\}}$.
Continuing this way, we obtain for each
$j \in J$ an integral anti-dominant weight $\lambda_j$ of central charge $c_j$,
and for each finite subset $F \subeq J$ a tensor product decomposition
\begin{equation}\label{eq:eenlabel}
\rho = \rho_F \otimes \rho_F', \qquad \cH = \cH_F \otimes \cH_F' \quad \mbox{ with }\quad
\cH_F := \bigotimes_{j \in F} \cH_{\lambda_{j}}
\end{equation}
into positive energy representations for the gauge groups $G_F^\sharp$ and
$G_{J \setminus F}^\sharp$.

%
%

\subsubsection{Compact base manifolds}

For gauge groups over a compact base manifold $M$, we thus obtain
the following classification result.
It contains in particular Torresani's classification
for linear flows on a torus; see \cite{To87} and \cite[\S5.4]{A-T93}.

\begin{Theorem}\label{Thm:ClassificationCompactM}
Let $M$ be a compact manifold with a fixed point free $\R$-action~$\gamma_{M}$, and let $\cK \rightarrow M$ be a bundle of Lie groups with compact, simple, connected fibers.
Let $\overline{\rho} \colon \Gamma(M,\cK)_{0} \rtimes \R \rightarrow \PU(\cH)$
be a minimal projective positive energy representation with respect to a lift $\gamma$ of the $\R$-action to $\cK$.
Then there exist finitely many
$\R$-orbits $S_j \subseteq M$, $j \in J$, with central charge $c_j\in \N_{0}$ 
such that $\overline{\rho}$ arises by factorization from an isotypic positive energy
representation $\rho_{S}$ of
\[
\hat{G} = G^{\sharp}\rtimes \R, \quad\text{where}\quad
G := \Gamma(S,\tilde \cK)
\simeq  \prod_{j\in J}\cL_{\Phi_j}(\tilde{K}_{j})\,.
\]
If $\rho_{S}$ is irreducible, then
\[\cH = \bigotimes_{j\in J} \cH_{\lambda_j}\]
is a tensor product of lowest weight
representations $(\rho_{\lambda_j},\cH_{\lambda_{j}})$
of the corresponding Kac--Moody group $\hat{\cL}_{\Phi_j}(\tilde{K}_j)$,
where $\lambda_j$ is an integral anti-dominant weight
of central charge $c_j$.
On the level of the Lie algebra
\[\hat{\fg} = \R C \times_{\omega} \Big( \bigoplus_{j\in J} \cL_{\phi_j}(\fk) \rtimes
\R D\Big),\]
the central element acts by
$\dd\rho(C) = i\one$, and the generator $D$ acts by
\[
\dd\rho(D) = \sum_{j\in J}\frac{1}{T_j}\dd\rho_{\lambda_j}(\bd_{j} + A_{j})\,,
\]
where $A_j \in \cL_{\phi_{j}}(\fk)$ is specified by the
$\R$-action on $G$.
 \end{Theorem}

\begin{proof}
Since $M$ is compact, the $\R$-invariant, embedded,
one-dimensional submanifold $S$ is a union of finitely many periodic orbits.
By Corollary~\ref{cor:locnonsimply}, the projective positive
energy representation $\overline{\rho}$ of $\Gamma(S,\cK)_0$ thus
arises by factorization from
a projective positive energy representation of
$G =  \Gamma(S,\tilde \cK)$, which is trivial on the image $Z$ of the diagonal
group $Z_{[M]}$ in $Z_{[S]}$ (cf.\ Remark~\ref{rem:7.6}).
It then follows from \eqref{eq:eenlabel} and the discussion in
\S\ref{sec:1dcompact}
that every factorial positive energy
representation is a multiple of a product of
lowest weight representations as described above.
The only thing left to check is that this representation restricts to a character
on the image $Z$ of $Z_{[M]}$ in~$G$.
Since $Z$ is a subgroup of the central group
$Z_{[S]} = \prod_{j\in J} \pi_1(K_{j})^{\Phi_j}$,
it is in particular contained in the
group $\prod_{j\in J} (\tilde{K}_j)^{\Phi_{j}}$
of constant sections, which is connected by \cite[Thm.~12.4.26]{HN12}.
Its Lie algebra
$\bigoplus_{j\in J}\fk^{\phi_j}$ is contained in the
radical of the cocycle~$\omega$.
Since $\cL(\tilde{K}_j)^{\Phi_j}$ is 1-connected (Lemma~\ref{lem:1-connectedfibers}), this implies that $Z$ is not only central in $G$, but also in $\hat{G}$.
In particular, every factor representation restricts to
a character on~$Z$.
\end{proof}

\begin{Remark} {\rm(Semisimple groups)}
In Theorem~\ref{Thm:ClassificationCompactM}, the restriction to simple fibers
is by no means essential.
For Lie group bundles $\cK \rightarrow M$ with compact semisimple, 
$1$-connected fibers, the representation still localizes
to an embedded {$1$-dimensional}
submanifold $S \subseteq M$ by Theorem~\ref{thm:7.11}. As $M$ is compact,
$S$ consists of finitely many circles $S_{j}$.
Since the fibers of $\cK \rightarrow M$ are $1$-connected,
the passage from $M$ to the finite cover $\hat{M}$ (Theorem~\ref{reductienaarsimpel}) yields
not only a Lie algebra bundle $\hat{\fK} \rightarrow \hat{M}$,
but also a Lie group bundle $\hat{\cK} \rightarrow \hat{M}$ with simple, compact fibers.
By the same argument as in Remark~\ref{remark:overdederivatie},
the $\R$-action on $\cK\rightarrow M$ lifts to $\hat{\cK} \rightarrow \hat{M}$.
Applying Theorem~\ref{Thm:ClassificationCompactM} to
$\hat{\cK} \rightarrow \hat{M}$, we find that the minimal factorial positive
energy representations are again multiples of the irreducible ones.
The latter are now parametrized
by embedded circles $\hat{S}_{j,r} \subseteq \widehat{M}$, together
with an integral anti-dominant weight $\lambda_{j,r}$ with central
charge $\lambda_{j,r}(C) = ic_{j,r}$.
Here, the circle $\hat{S}_{j,r} \subseteq \hat {M}$ is a finite cover of the circle
$S_{j} \subseteq M$. The weight
$\lambda_{j,r}$ is associated to the Kac--Moody algebra
$\hat{\cL}_{\Phi_{j,r}}(\fk_{j,r})$, where $\fk_{j,r}$ is a simple ideal
in the semisimple Lie algebra $\fk_{j}$, and $\Phi_{j,r}$ is the smallest power of the holonomy
around $S_{j}$ that maps $\fk_{j,r}$ to itself.
\end{Remark}

\subsection{Extensions to non-connected groups}\label{sec:nonconnected}

In this subsection we discuss several phenomena related to
non-connected variants of the group $G$. Dealing with
non-connected groups is typically more complicated because
they may not have a simply connected covering group,
nor do central extensions or representations of the
identity component always extend to the whole group.

This suggests the following classification scheme to deal with projective
positive energy representations of $G \rtimes_\alpha \R$ if
$G$ is not connected:
\begin{itemize}
\item determine which central extensions of $G_0$
extend to the non-connected groups $G$.
\item determine which of these do this in an $\R$-equivariant fashion.
This leads us to central extensions of the non-connected group
$G \rtimes_\alpha \R$.
\item determine the irreducible positive energy representations
of the non-con\-nected
groups $\hat G$ in terms of the representations of $\hat G_0$
(this may be carried out with Mackey's method of unitary induction, as in \cite{TL99}).
\end{itemize}

The following factorization theorem reduces
the classification of the irreducible representations to the corresponding
problem for the identity component $G_0$ and the group
$\pi_0(G)$ of connected components. It shows in particular that
no additional difficulties arise if $K$ is a $1$-connected simple group.
We shall use the notation
$G \to \pi_0(G), g \mapsto [g]$ for the quotient map.
\begin{Theorem} \label{thm:factorize-conncomp}
{\rm(Factorization Theorem for non-connected gauge groups)}
Suppose that $K$ is a $1$-connected simple compact Lie group, that $M$ is compact
and that $\bv_M$ has no zeros.
Then 
every positive energy representation
$(\rho,\cH)$ of $G^\sharp = \Gamma(M,\cK)^\sharp$
can be written as $\rho(g) = \rho'(g) \zeta([g])$, where 
$\rho'$ factors through a 1-dimensional, closed, $\R$-equivariantly embedded
submanifold $S\subseteq M$, and
$\zeta \: \pi_0(G) \to \U(\cH)$ is a representation that commutes with
$\rho'(G^\sharp_0) = \rho(G^\sharp_0)$. In particular, every irreducible
positive energy representation of $G^{\sharp}$ is of the form $\rho' \otimes \zeta$ where both
$\rho'$ and $\zeta$ are irreducible, and, conversely, any such tensor product is irreducible.
\end{Theorem}

\begin{proof} Let $(\rho,\cH)$ be a positive energy representation of $G^\sharp$.
From Theorem~\ref{Thm:ClassificationCompactM} we know
that the restriction of $\rho$ to $G_0^\sharp$
factors through an evaluation homomorphism
\[ \ev \: G \to G_S := \Gamma(S,\cK)
\cong \prod_{j \in J} \cL_{\Phi_j}(K),\]
that is, there exists a positive energy representation $\rho_1$ of~$G_S^\sharp$ such that
\[\rho\res_{G_0^\sharp} = \rho_1 \circ \ev\res_{G_0^\sharp}.\]
Since $K$ is $1$-connected, the groups $\cL_{\Phi_j}(K)$ are connected and therefore
$G_S$ is connected.
{\color{red}}
Then $\rho' := \rho_1 \circ \ev$ is a positive energy representation
of $G^\sharp$ that coincides with $\rho$ on $G_0^\sharp$.

This construction shows in particular that
$\pi_0(G)$ acts trivially on the set of equivalence classes of irreducible
positive energy representation
of $G_0^\sharp$.  Indeed, for every irreducible representation
$\rho_1$ of $G_S^\sharp$, the representation
$\rho'$ extends the representation $\rho_1 \circ \ev\res_{G_0^\sharp}$ to a representation of $G^\sharp$
on the same space.

As every positive energy representation of $G^\sharp$ decomposes on $G_0^\sharp$ into
irreducible ones (Theorem~\ref{Thm:ClassificationCompactM}), it follows that
$\rho(G^\sharp)$ preserves all $G^\sharp_0$ isotypic subspaces
$\cH_j \cong \cF_j \otimes \cM_j$, $j \in J$, and
on these the representation
of $G^\sharp_0$ has the form $\rho_j \otimes \one$.
Extending $\rho_j$ to a representation $\tilde\rho_j$ of $G^\sharp$,
the restriction of $\rho$ from $\cH$ to $\cH_{j}$
takes the form
$\tilde\rho_j \otimes \zeta_j$, where
$\zeta_j \: \pi_0(G^\sharp) \cong \pi_0(G) \to \U(\cM_j)$ is a unitary
representation on the multiplicity space.
Putting everything together, we obtain a factorization
$\rho = \rho' \otimes \zeta$, where $\zeta$ is a representation of
$\pi_0(G)$  that commutes with $\rho(G_0^\sharp)$.

In view of Schur's Lemma, our construction shows in particular that the
representation $\rho$ is irreducible if and only if it is
isotypical on $G_0^\sharp$, that is, $\cH = \cF \otimes \cM$,
and the representation $\zeta$ of $\pi_0(G)$ on $\cM$ is irreducible.
\end{proof}

\begin{Remark} (a) If $K$ is connected but not simply connected
and $\fk$ is a compact simple Lie algebra, then
the classification in \cite{TL99} shows that not all central extensions of
$\cL(K)_0$ extend to the whole group $\cL(K)$, so that the situation
becomes more complicated. Likewise, irreducible
projective positive energy representations of $\cL(K)_0$ do not in general
extend to the whole group $\cL(K)$. In \cite{TL99} one finds a classification
of the irreducible projective positive energy representations
of the groups $\cL(K)$ for connected simple groups~$K$.
Here the new difficulty is that the group
$\pi_0(\cL(K))\cong \pi_1(K)$ acts non-trivially on the alcove
whose intersection with the weight lattice classifies the irreducible
projective positive energy representations of the connected group
$\cL(K)_0 \cong \cL(\tilde K)$ for a fixed central charge.

(b) If we
start with a projective representation of the non-connected group
$\Gamma_c(M,\cK)$, we get a representation of the image of $\Gamma_c(M,\cK)$
in $\Gamma_c(S,\cK)$, which is a restricted direct product
of twisted loop groups.
It maps $\Gamma_c(M,\cK)_0$ onto the identity component, but
additional information is contained in the images of the other
connected components.
We then get  a projective
representation of a Lie group whose Lie algebra
is $\Gamma_c(S,\fK)$ and whose group of connected components
is an image of $\pi_0(\Gamma_c(M,\cK))$. Its action on the Lie algebra does not
permute the ideals of the type $\cL_\phi(\fk)$, so it acts on each twisted
loop algebra separately by the adjoint action of some
element of $\cL_\Phi(K)$. This suggests that one
needs a description of those Lie algebra
cocycles $\omega$ on $\Gamma_c(M,\fK)$ that actually correspond to central
Lie group extensions of the full group $\Gamma_c(M,\cK)$. Here the obstructions
lie in $H^3(\pi_0(\Gamma_c(M,\cK)), \T)$.
We refer to \cite{Ne07} for further details
on such obstructions and for methods of their computation.

(c) For a bundle of Lie groups $\cK \to M$, passing to the simply connected covering of the structure group~$K$ may
not always be possible. For this, an obstruction class in
$H^3(M,\pi_1(K))$ has to vanish (see \cite{NWW13}). Since $\pi_1(K)$ is finite
for semisimple compact groups $K$, this is a torsion class.
So for a discrete central subgroup $D \subeq K$,
every bundle with structure group $K$ factorizes to a bundle with structure group
$K/D$, but in general, not all bundles with structure group $K/D$ are of this form.
\end{Remark}

\section{The classification for \texorpdfstring{$M$}{M} noncompact}
\label{sec:9}

Even in the noncompact case, the techniques developed so far open up a number of
new perspectives.
The Localization Theorem~\ref{thm:7.11} allows us to restrict
attention to a 1-dimensional invariant 
submanifold ${S\subseteq M}$. 
If $M$ is noncompact, then $S$ can have infinitely many connected components $S_j$,
each of which is diffeomorphic to either $\R$ or $\bS^1$.
We consider these two cases separately.

In \S\ref{sec:PosenNoncptComponent} we consider the case where $S$ consists of infinitely many lines.
In order to arrive at a (partial) classification, we
impose the additional condition that the positive energy representation $(\ol{\rho}, \cH)$ admits
a cyclic ground state vector $\Omega \in \cH$ that is unique up to scalar.
In Theorem~\ref{Thm:NoncompactClassification} we show that
these \emph{vacuum representations}
are classified up to unitary equivalence by a nonzero central charge
$c_j \in \N$ for every connected component $S_j\simeq \R$.
The proof proceeds by reducing to the (important) special case $M=\R$, where the classification is
essentially due to Tanimoto~\cite{Ta11}.


In \S\ref{subsec:9.2} we consider the case where $S$ consists of infinitely many circles.
Here we impose the much less restrictive
condition that $\cH$ is a \emph{ground state representation}. This means that $\cH$ is generated by the space of ground states, but we do not require these ground states to be unique.
We show that this condition is automatically satisfied if the periods \eqref{eq:Period} of the $\R$-action are
uniformly bounded.
In Theorem~\ref{prop:9.5} we classify this type of representations in terms of $C^{*}$-algebraic data, using
techniques similar to those used in \cite{JN15}
for norm-continuous representations. The possibility of an infinite dimensional space of ground states
gives rise to interesting phenomena, such as factor representations of type II and III.

Finally, in \S\ref{sec:fixedpoint}, we briefly explore a simple situation where the $\R$-action has a fixed point.
The main thing we wish to point out is that the lift of the $\R$-action \emph{at the fixed point} has a qualitative influence
on the type of representation theory that one encounters.
In Part II of this series we develop the necessary tools to resolve the positive energy representation theory in more detail.

%
%
%
%
%

\subsection{Infinitely many lines}\label{sec:PosenNoncptComponent}
In contrast to the case of (twisted) loop groups,
the classification of projective positive energy representations of
$C^{\infty}_{c}(\R,K)$, for $K$ a compact $1$-connected
simple Lie group, is an open problem---closely
related to the
classification problem for representations of loop group nets
(cf.~\cite{Wa98,TL97} and Remark~\ref{rem:10.11}).

A large class of examples can be obtained by restricting projective positive energy representations
of the loop group\index{loop group!05@of $K$ \scheiding $\cL(K)$} $G := \cL(K)$ to
$G_{\mathrm{cs}} := C^{\infty}_{c}(\R,K)$\index{loop group!60@flat \scheiding $G_{\mathrm{cs}}$}, where the latter is considered as a subgroup
by identifying the circle
with the real projective line
$\bP^1(\R) = \R \cup \{\infty\}$.
The restriction of an irreducible projective positive energy representation
of $\cL(K)$ remains irreducible, essentially by
\cite[Cor.~IV.1.3.3]{TL97}.
In \S\ref{subsec:9.4} we show that the restriction remains of positive energy as well.
This is not a priori clear,
since the positive energy is defined in terms of rotations of
the circle for $G$ and in terms of translations on the real line for~$G_{\mathrm{cs}}$.

It is
not true that \emph{all} projective unitary positive energy representations of $G_{\mathrm{cs}}$ arise by restriction in this way,
and the classification remains an open problem.
We can, however, classify the projective positive energy representations under the additional assumption
that they admit a cyclic ground state vector which is unique up to scalar.
These  \emph{vacuum representations}
were classified by Tanimoto for the Lie algebra
of $\fk$-valued Schwartz functions \cite{Ta11}, and in \S\ref{sec:groundstates}
we use Theorem~\ref{Shylock}
to push these results to the compactly supported setting.

Finally, in \S\ref{sec:groundstatesncpt}, we classify the vacuum state representations for a noncompact
manifold $M$ with a free $\R$-action. The proof proceeds by identifying  the restricted gauge group
$\Gamma_{c}(S,\cK)$ with the weak product $\prod{}^{'}_{j}C^{\infty}_{c}(S_j,K)$, where $j$ labels the connected components $S_j \simeq \R$.
We then use the results from Appendix~\ref{subsec:a.5}, where we show that
the classification of vacuum
representations for a weak product of Lie groups reduces to the same problem for each of its factors.

\subsubsection{Restriction from $\cL(K)$ to $C^{\infty}_{c}(\R,K)$}\label{subsec:9.4}

By identifying the circle $\bS^1$
with the real projective line
$\bP^1(\R) = \R \cup \{\infty\}$, we can
consider $G_{\cs} := C^{\infty}_{c}(\R,K)$ as a subgroup of the loop group
$G := \cL(K)$. \index{flat loops \scheiding $G_{\cs}$}

Note that the natural $\R$-action by translations on $G_{\cs}$ does not agree with
the $\R$-action by rigid rotations on $G$.
In terms of the real projective line,
the rotation action of $\R/\Z$ is given by the fractional linear maps
\[ R_t(x) = \frac{\cos 2\pi t \cdot x + \sin 2\pi t}
{-\sin 2\pi t \cdot x + \cos {2\pi}t}, \quad
x \in \R\cup \{\infty\}, \;{[t] \in \R/\Z},\]
whereas the translation action of is given by $T_{t}(x) = x+t$.

\begin{Proposition}{\rm{(Restriction of positive energy representations)}}\label{prop:restrictionPEtoline}
Let $(\rho,\cH)$ be an irreducible positive energy representation
of $\cL^\sharp(K)$ with respect to the $\R$-action by rotations.
Then the restriction of $\rho$ to $C^{\infty}_{c}(\R,K)^{\sharp}$ is an irreducible positive energy with respect to
the $\R$-action by translations.
\end{Proposition}

We first prove that the restriction remains irreducible, and then continue with the positive energy condition.

\begin{proof}[Proof of irreducibility]
Let
$ G_* :=  \{ \xi \in G \: \xi(\infty) = \one\} $
be the subgroup of based loops.
\index{loop group!80@pointed \scheiding $G_*$}
Since
$\rho(G^\sharp_*)'' = \C \one$ by\begin{footnote}{Alternatively, one can use \cite[Thm.~6.4]{Ta11},
which uses \cite[Cor.~1.2.3]{Ba95}.}\end{footnote} \cite[Cor.~IV.1.3.3]{TL97},
it suffices to show that $\rho(G^\sharp_{\cs})$ is dense in $\rho(G^\sharp_*)$
for the strong operator topology.
By \cite[App.~A]{Ne14a}, the representation of $G^\sharp$ extends
 to a smooth representation
of the Banach--Lie group $H^1(\bS^1,K)$ of $H^1$-loops, whose Lie algebra
is the space $H^1(\bS^1, \fk)$ of $H^1$-functions $\xi \: \bS^1 \to \fk$.
Since these are the absolutely continuous functions whose derivatives are~$L^2$,
the derivative $\xi \mapsto \xi'$ maps the
subspace $H^1_*(\bS^1, \fk)$ of $H^1$-functions
that vanish in the base point homeomorphically to
\[ L^2_*(\bS^1,\fk) = \Big\{ \xi \in L^2(\bS^1, \fk) \,\: \int_{\bS^1} \xi(t)\, dt = 0
\Big\}.\]
In this space the subspace $\{ \eta' \: \eta \in C^\infty_c(\R,\fk) \}$
is easily seen to be dense.
Since $G^\sharp_*$ is connected, this implies that
$\rho(G^\sharp_{\cs})$ is dense in~$\rho(G^\sharp_*)$.
\end{proof}

To prove the positive energy condition for the restriction, we need to compare the generator $\bd_0$
of rigid rotations with the generator $\bd_1$ of translations. In $\fsl(2,\R)$, these are given by
\begin{equation}\label{eq:d1andd2}
\bd_0 = \frac{1}{2}\pmat{0 & 1 \\ -1 & 0}, \quad  \bd_1 = \pmat{0 & 1 \\ 0 & 0}.
\end{equation}
The fact that $\bd_0$ and $\bd_1$ generate the same
$\Ad$-invariant closed convex cone in the Lie algebra $\fsl(2,\R)$
leads to the following characterization
(cf.\ \cite[\S 1.3]{Lo08}).
\begin{Lemma} \label{lem:sl2}
For a unitary representation $(\rho,\cH)$ of
$\tilde\SL(2,\R)$, the generator $i\dd\rho(\bd_0)$ is bounded from below
if and only if $i\dd\rho(\bd_1)$ is bounded from below.
Moreover, if this is the case, then $i\dd\rho(\bd_0)\geq 0$ and $i\dd\rho(\bd_1)\geq 0$.
\end{Lemma}
In particular, an $\tilde\SL(2,\R)$-representation is of positive energy for $\bd_0$ if and only
if it is of positive energy for $\bd_1$.
To prove that the restriction from $\cL(K)^{\sharp}$ to $C^{\infty}_{c}(\R,K)$
is of positive energy with respect to $\bd_1$, it therefore suffices to extend the action by rigid rotations
to an action of $\tilde\SL(2,\R)$. This is done using the
Segal--Sugawara construction.

\begin{proof}[Proof of positive energy]
Recall from \cite[\S7]{GW84} and \cite{GW85} that every
irreducible projective positive energy representation
$(\oline\rho_\lambda, \cH_\lambda)$ of
$\cL(K)$ with lowest weight $\lambda$
extends to a projective representation of the semidirect
product $\cL(K) \rtimes \Diff_+(\bS^1)$,
where $\Diff_+(\bS^1)$ acts on $\cL(K)$ by
$\alpha_\phi(\xi) := \xi \circ \phi^{-1}$. The cocycle
\begin{equation}
  \label{eq:cocyc-onemore}
\omega(\xi,\eta) = \frac{c}{2\pi} \int_{\bS^1} \kappa(\xi'(t),\eta(t))\, dt
\end{equation}
is easily seen to be invariant under the action of $\Diff_+(\bS^1)$, but it
is much harder to verify the covariance of the representations $\oline\rho_\lambda$.
In \cite[\S7.2]{GW84}, the representation of the Virasoro algebra
obtained from the Segal--Sugawara construction is
integrated
to a group representation.
Since this respects the semidirect product structure of $\cL(K) \rtimes \Diff_+(\bS^1)$,
it follows in particular that
\begin{equation}
  \label{eq:twistbydiff}
\oline\rho_\lambda \circ \phi \cong \oline\rho_\lambda
\quad \mbox{ for every } \quad \phi \in \Diff_+(\bS_1).
\end{equation}

By Schur's Lemma and the irreducibility of
$\oline\rho_\lambda$, the projective representation $\oline\rho^P$ of
$\Diff_+(\bS^1)$ on $\cH_\lambda$ is uniquely determined by
the intertwining property
\[ \oline\rho^P(\phi) \oline\rho_\lambda(\xi) \oline\rho^P(\phi)^{-1}
= \oline\rho(\alpha_\phi \xi)
\quad \mbox{ for } \quad
\xi \in \cL(K), \phi \in \Diff_+(\bS^1).\]
Since $\Diff_+(\bS^1)$ contains the group of rigid rotations
with respect to which $\ol{\rho}_{\lambda}$ is a positive energy representation,
the Hamiltonian
$H = i \dd\rho^P(\bd_0)$ associated to $\bd_0$
is bounded below.
Since $\rho^P$ is a positive energy representation of the
Virasoro group,
it restricts to a positive energy representation of
its subgroup $\tilde\SL(2,\R)$, the
simply connected cover of the group $\PSL(2,\R) \subeq \Diff_+(\bS^1)$ of
fractional linear transformations of $\bS^1 \cong \bP_1(\R)$. By Lemma~\ref{lem:sl2},
the generator $i\dd\rho^P(\bd_1)$ then has nonnegative spectrum.
\end{proof}
\begin{Remark}
Since
the cocycle  \eqref{eq:cocyc-onemore}
is invariant under the action of $\Diff_+(\bS^1)$, twisting $\ol{\rho}_{\lambda}$with
$\phi \in \Diff_+(\bS^1)$ leads to an irreducible
projective unitary representation $\oline\rho_\lambda \circ \phi$ with the same
central charge~$c$.
By Proposition~\ref{prop:7.22},
there are only finitely many types of such representations
satisfying the positive energy condition.
If we knew a priori that this twist preserves the positive
energy condition (which is presently not the case), then we could bypass the integration
procedure in \cite{GW84}, and construct the projective representation of $\Diff_+(\bS^1)$ as follows.

By the Epstein--Hermann--Thurston Theorem,
$\Diff(M)_0$ is a simple group for every compact connected smooth manifold~$M$
(\cite{Ep70}). In particular, $\Diff_+(\bS^1)$ is a simple group. Since it has no normal subgroup
of finite index, it acts trivially on any finite set.
This implies that $\oline\rho_\lambda \circ \phi\cong \oline\rho_\lambda$
for every $\phi \in \Diff_+(\bS^1)$. The unitaries that implement this equivalence constitute
a projective unitary representation of $\Diff_+(\bS^1)$.
\end{Remark}

\begin{Remark} \label{rem:10.11}
The class of positive energy representations is by no means exhausted by the
representations of Proposition~\ref{prop:restrictionPEtoline}.
We briefly sketch the construction of
a class of type~III$_1$ factor representations, following \cite{We06,FH05}.

Recall from \cite[\S7.2]{GW84} that an irreducible positive energy vacuum representation $\rho$
of $G^\sharp = \cL^\sharp(K)$ gives rise to a vacuum representation $\rho^P$ of $\Diff_+(\bS^1)^{\sharp}$.
If we lift the $\R$-action by translations
along the $2$-fold covering
$q \: \bS^1 \to \R \cup \{\infty\} \cong \bS^1$, we obtain a flow on $\bS^1$
with exactly two fixed points. Its generator $\bv$
is obtained from the vector field $\bd_0$ generating rigid
rotations by multiplication with a non-negative function.
This implies that 
the operator $i \dd \rho^P(\bv)$ is bounded from below.

Let $I \subeq \bS^1$ be one of the two connected components
of $q^{-1}(\R)$ and identify
$G_{\cs} = C^\infty_c(\R,K)$ with $C^\infty_c(I, K)$.
Then the restriction of $\rho$ to $G^\sharp_{\cs}$
is a factor representation of type III$_1$. Combining this with
the one-parameter group generated by the vector field $\bv$, we
obtain a projective positive energy representation
of $G^\sharp_{\cs} \rtimes \R$ with respect to the translation
action on $\R$ (\cite[Prop.~3.2]{We06}). We refer to \cite{We06}
and \cite{dVIT20} for further
details (see also the Remark after \cite[Thm.~IV.2.2.1]{TL97}).

More generally, we may consider smooth vector fields
$\bv \in \cV(\bS^1)$ which are non-negative multiples
$f \bd_0$, $f \geq 0$, of the generator $\bd_0$ of rigid rotations.
For vacuum representations of $\cL(K)\rtimes \Diff_{+}(\bS^1)$, the corresponding selfadjoint operator
$i\dd \rho^P(\bv)$ is bounded from below
(cf.~\cite{FH05}).
If $I \subeq \bS^1$ is an open interval on which $\bv$ has no zeros
but for which $\bv$ vanishes in the boundary $\partial I$, then
we obtain an embedding
\[ C^\infty_c(\R, \fk) \rtimes \R \into \cL(\fk)  \rtimes \R \bv \]
that integrates to the group level, where we obtain a
projective positive energy representation of $C^\infty_c(\R,K)$.
\end{Remark}


\subsubsection{Vacuum representations of $C^{\infty}_{c}(\R,K)$}\label{sec:groundstates}
\label{subsec:9.3}

Although the classification of projective positive energy representations $(\ol{\rho},\cH)$ of $C^{\infty}_{c}(\R,K)$ is an open problem in general, it can be resolved under the additional assumption that $\cH$ admits a unique, cyclic ground state.

\begin{Definition} \label{def:groundstate}
Let $(\rho,\cH)$ be a positive energy representation
of $\hat G$.
\begin{itemize}
\item[a)] A {\it ground state vector} \index{ground state vector \scheiding $\Omega$}
is a vector  $\Omega \in \cD(H) \subseteq \cH$ such that
$H\Omega = E_0\Omega$ for $E_0 := \inf(\mathrm{spec}(H))$. We denote the space of ground state vectors by $\cE$.
\item[b)] A \emph{ground state representation} is a positive energy representation $(\rho,\cH)$  \index{representation!ground state \vulop}
that is generated by its space of ground states, in the sense that the linear span of $\rho(\widehat{G})\cE$ is dense in $\cH$.
\item[c)] A \emph{vacuum representation} is a ground state representation where the ground state is   \index{representation!vacuum \vulop}
unique up to scalar, $\cE = \C \Omega$.
\end{itemize}
\end{Definition}

At the Lie algebra level, we obtain analogous definitions if we replace
the requirement that $\rho(\widehat{G})\cE$ is dense in $\cH$ by the requirement
that
$\mathcal{U}(\fg)\Omega = \mathcal{U}(\fg^\sharp)\Omega$ is dense in $\cH$.
Although the translation between these two concepts requires some caution,
the two notions turn out to be compatible for positive energy representations.

\begin{Proposition}\label{Prop:analyticPEvacuum}
Let $(\rho,\cH)$ be a positive energy representation of $\hat G$ with ground state vector $\Omega$.
Then
$\cU(\g^\sharp)\Omega$ is dense in $\cH$
if and only if $\Omega$ is cyclic under $\rho(G^\sharp)$.
\end{Proposition}
\begin{proof}
For a closed interval $I \subeq \R$, let
\[  G_I := \{ \xi \in G = C^\infty_c(\R,K) \:
\xi(\R \setminus I) = \{e\} \} \]
denote
the Fr\'echet--Lie subgroup of maps supported by $I$.
We claim that the Lie group $G_I^\sharp$ is  BCH, i.e.,
it is locally exponential
and its Lie algebra $\g_I^\sharp$ is BCH, which means that
the Baker--Campbell--Hausdorff
series defines an analytic local multiplication
on a $0$-neighborhood of $\g$ (\cite[Thm.~14.7.1]{GN}).
For $G_I$ this follows from \cite[Ex.~6.1.4(c)]{GN}
because the BCH property is inherited
from the target group $K$. Further \cite[Thm.~14.4.19]{GN} implies
that the centrally extended Lie algebra $\g_I^\sharp$
is also locally exponential
and the proof of this theorem shows that the analyticity of the local
multiplication is inherited by the central extension.

Lemma~\ref{lem:6.26} implies that $\Omega$ is an analytic
vector for each element in $\g^\sharp_I$, so that
\cite[Prop.~4.10]{Ne11} further entails that  $\Omega$ is an analytic
vector for $G^\sharp_I$. Hence
the closure of $\cU(\g^\sharp_I)\Omega$ is
$G^\sharp_I$-invariant. As the interval $I$ was arbitrary,
the closure of $\cU(\g^\sharp)\Omega$ is
invariant under $G^\sharp$,
hence also under $\hat G$, because $\Omega$ is an $H$-eigenvector.
This shows that $\cU(\g^\sharp)\Omega$ is dense in $\cH$
if and only if $\Omega$ is cyclic under $\rho(G^\sharp)$.
\begin{footnote}
{For the concept of an analytic map to make sense,
we need the group to be analytic. Since the
$\R$-action on $G$ need not be analytic,
the semidirect product $G \rtimes_\alpha \R$ is in general not an analytic Lie group.
In particular, $\hat G$ need not be an analytic Lie group.}
\end{footnote}
\end{proof}


The vacuum representations
for the Lie algebra $\fg_{\cS} = \cS(\R,\fk)$ of $\fk$-valued Schwartz
functions have been classified by Yoh Tanimoto.
\begin{Theorem}{\rm (\cite[Cor.~5.8]{Ta11})} \label{Thm:Tanimoto}
Let $(\pi, \cH^{\infty})$ be a vacuum representation of $\widehat{\fg}_{\cS}$
with respect to the $\R$-action by translations. Suppose that
for all $\psi, \chi \in \cH^{\infty}$,
the functional
$\xi \mapsto \lra{\psi, \pi(\xi) \chi}$ is a tempered distribution.
Then $(\pi,\cH^{\infty})$ is characterized up to unitary equivalence by its central charge $c \in \N_0$.
\end{Theorem}
Using the continuity results from \S\ref{sec:6}, we show that Tanimoto's theorem
remains true for the smaller Lie algebra $\fg_{\cs} := C^{\infty}_{c}(\R,\fk)$
of compactly supported smooth $\fk$-valued
functions.
This is an important improvement because the relevant Lie algebra
for the classification of ground states of loop group nets
is not $\cS(\R,\fk)$, but $C^{\infty}_{c}(\R,\fk)$ (cf. \cite[\S6]{Ta11}).

As usual, we denote
\[
\hat{\fg}_{\cs} := \R C \oplus_{\omega} \fg_{\cs} \rtimes \R D, \;\text{ and }\;
\hat{\fg}_{\cS} := \R C \oplus_{\omega} \fg_{\cS} \rtimes \R D,
\]
where $D$ acts by infinitesimal translations.
Since the inclusion of $\cS(\R,\fk)$ in $H^1_{\partial}(\R,\fk)$
is continuous, the following is an immediate consequence of Theorem~\ref{Shylock}.
\begin{Proposition}\label{cor:verlegenslot}
Let $(\rho,\cH)$ be a
positive energy representation of $\widehat{G}_{\cs} := (C^{\infty}_{c}(\R,K)\rtimes \R)^{\sharp}$. 
Then the derived representation $\dd\rho$ of $\widehat{\fg}_{\cs}$
extends uniquely to a positive energy representation $r$ of $\widehat{\fg}_{\cS}$
such that,
for all $\psi, \chi \in \cH^{\infty}$,
the functional
$\xi \mapsto \lra{\psi, r(\xi) \chi}$ is a tempered distribution.
\end{Proposition}

Combined with Theorem~\ref{Thm:Tanimoto}, this immediately yields the classification of vacuum representations
in the compactly supported setting.

%

\begin{Theorem} {\rm(Vacuum representations of $C^{\infty}_{c}(\R,K)$)}
\label{thm:10.9}
Let $K$ be a $1$-connected, compact, simple Lie group and
$G_{\cs} = C^{\infty}_{c}(\R,K)$.
Then a vacuum representation $(\rho,\cH)$ of $\hat{G}_{\cs}$
is characterized up to unitary equivalence
by its central charge $c\in \N_{0}$.
\end{Theorem}

\begin{proof}
By Proposition~\ref{Prop:analyticPEvacuum}, the derived representation
$\dd\rho$ of $\widehat{\fg}_{\cs}$ is a vacuum representation
which by Proposition~\ref{cor:verlegenslot}  extends
to a continuous representation of the Lie algebra $\widehat{\fg}_{\cS}$.
By \cite[Cor.~5.8]{Ta11}, the latter is determined up to isomorphism by its central charge
$c\in \N_{0}$. Since $G_{\cs}$ is connected (Lemma~\ref{lem:1-connectedfibers}),
the representation $\rho$ of $\hat{G}_{\cs}$ is uniquely determined by
its derived Lie algebra representation (Theorem~\ref{ThmProjRepLARep}), and the result follows.
\end{proof}

In \S\ref{subsec:9.4}, we saw that the restriction of an irreducible positive energy representation of $\cL^{\sharp}(K)$
(with respect to rotations) yields an irreducible positive energy representation of
$C^{\infty}_{c}(\R,K)^{\sharp}$ (with respect to translations).
We now show that the unique vacuum representation of $C^{\infty}_{c}(\R,K)^{\sharp}$ with central charge $c$
arises by restriction of the irreducible positive energy representation
of $\cL^{\sharp}(K)$ with lowest weight $\lambda = (ic, 0, 0)$.

\begin{Proposition}\label{prop:groundstatezeroweight}
The irreducible lowest weight representation of $\cL^{\sharp}(K)$ with lowest weight $\lambda$
restricts to a vacuum representation of $C^{\infty}_{c}(\R,K)^{\sharp}$ if and only if
the restriction $\lambda_0$ of $\lambda$ to $i\ft$ is zero.
\end{Proposition}

\begin{proof}
Recall from \S\ref{subsec:9.4} that
every irreducible projective positive energy representation of $\cL(K)\rtimes \R \dd_0$ with lowest weight $\lambda$
extends to $\cL(K) \rtimes \mathrm{Diff}_{+}(\bS^1)$.
By Lemma~\ref{lem:sl2}, this induces a unitary representation of
\[\widetilde{\SL}(2,\R) \subseteq \mathrm{Diff}_{+}(\bS^1)^{\sharp},\]
which is of positive energy not only with respect to $\bd_0\in \fsl(2,\R)$, but also with respect to~$\bd_1\in \fsl(2,\R)$ (cf.~\eqref{eq:d1andd2}).


Recall from Remark~\ref{rem:8.10} that the space $\cE_0$ of ground states for $i\dd\rho(\bd_0)$ is an irreducible unitary
$K$-representation. Its lowest weight $\lambda_0$ is the restriction of $\lambda$ to $i\ft$.
By the formula in \cite[Thm.~3.5(iii)]{GW84},
the minimal eigenvalue of $H_0 = i\dd\rho(\bd_0)$ is a positive
multiple of the Casimir eigenvalue for $K$ on $\cE_0$.
In particular
it vanishes if and only if $\lambda_0 = 0$, which is the case if and only if $\dim\cE_0 = 1$,
\begin{equation}
  \label{eq:infspec-Suga}
\inf\Spec(H_0) = 0 \quad
\Longleftrightarrow \quad \lambda_0 = 0 \quad
\Longleftrightarrow \quad \dim \cE_0 = 1.
\end{equation}
By a result of Mautner and Moore \cite{Mautner1957,Moore1966},
\begin{equation}\label{eq:kerniskern}
\ker(\dd\rho(\bd_0)) = \ker(\dd\rho(\bd_1))
\end{equation}
coincides with the subspace of vectors that are fixed under
$\tilde\SL(2,\R)$.\begin{footnote}
{See Appendix~\ref{Appendix:sl2} for a simplified direct proof.}
\end{footnote}If $\lambda_0 = 0$, the ground state for $H_1 := i\dd\rho(\bd_1)$ is therefore unique up to a scalar.

Conversely, suppose that the space $\cE_{1}$ of ground states for $H_1$ is nontrivial.
Since the adjoint orbit through $\bd_1$
contains $\R^+\bd_1$, the spectrum of $H_1 = i\dd\rho(\bd_1)$ is scale invariant.
Any ground state $H_1\Omega = E \Omega$ then has $E=0$,
and will satisfy $H_0 \Omega = 0$ by \eqref{eq:kerniskern}.
Since $H_0$ is nonnegative, it has minimal eigenvalue zero, the space $\cE_0$ of ground states for $i\dd\rho(\bd_0)$
is one-dimensional. We conclude that $\lambda_0 = 0$, and that $\cE_{1}\subseteq \cE_{0}$ is one-dimensional as well.
\end{proof}

\subsubsection{Vacuum representations for noncompact manifolds}\label{sec:groundstatesncpt}

Let $\cK \rightarrow M$ be a bundle of 1-connected simple compact
Lie groups over a $2^{\mathrm{nd}}$ countable
manifold $M$, equipped with a smooth $\R$-action by automorphisms.

\begin{Theorem}\label{Thm:NoncompactClassification}
If the action of $\R$ on $M$ is free, then
up to unitary equivalence, there is a bijective correspondence between:
\begin{itemize}
\item[{\rm a)}] Smooth projective unitary representations $\ol{\rho} \colon \Gamma_{c}(M,\cK)_0 \rightarrow \PU(\cH)$ that extend to a vacuum representation of $\Gamma_{c}(M,\cK)_0^{\sharp}\rtimes_{\alpha} \R$
with smooth ground state vector $\Omega$.
\item[{\rm b)}] Closed, embedded, 1-dimensional flow-invariant submanifolds
$S$, together with a nonzero central charge $c_j \in \N$ for every connected
component ${S_{j}\simeq\R}$ of~$S$.
\end{itemize}
Under this correspondence we have
\[
(\cH,\Omega) = \bigotimes_{j\in J}(\cH_j,\Omega_j)\quad\text{and}\quad
\rho(g) = \bigotimes_{j \in J} \rho_{j}(g|_{S_{j}}),\]
where $(\rho_j, \cH_j, \Omega_j)$
is the restriction to $C^{\infty}_{c}(\R,K) \simeq \Gamma_{c}(S_j,\cK)$
of the lowest weight representation of $\cL^{\sharp}(K)$
with lowest weight $\lambda = (c_j,0,0)$ and
$J$ is the countable set of connected components of $S$.
\end{Theorem}

\begin{proof}
By the Localization Theorem \ref{thm:7.11}, every projective positive energy representation $\ol{\rho}$
factors through $\Gamma_{c}(S,\cK)$ for a closed, embedded, 1-dimensional submanifold $S \subseteq M$.
It follows that $\rho$ factors through $\Gamma_{c}(S,\cK)^{\sharp}$.
Since $M$ is $2^{\mathrm{nd}}$ countable, $S$ has at most
countably many connected components $S_{j}$, $j \in J$,
and the freeness of the action implies that
each of these is $\R$-equivariantly isomorphic to $\R$.
By Lemma~\ref{eq:LemmaWeakProductsGauge}, the Lie group $\Gamma_{c}(S,\cK)$ is isomorphic
to the weak product
\[
	G := \textstyle{\prod}'_{j\in J} \,G_{j}, \quad\text{with}\quad
G_j = \Gamma_{c}(S_{j},\cK).\]
The cocycle $\psi \colon \fg \times \fg \rightarrow \R$
on $\fg := \Lie(G) = \bigoplus_{j\in J}\fg_j$ vanishes on
$\fg_i \times \fg_j$ for $i\neq j$. Since every $G_j$ is connected, this implies that
\[
 G^{\sharp} \cong \big(\textstyle{\prod}'_{j\in J} \,G^{\sharp}_{j}\big)/N,
\]
where $N \subseteq \textstyle{\prod}'_{j\in J} \T_j$ is the kernel of
the smooth character
\[
\chi \colon \textstyle{\prod}'_{j\in J}\;\T_j \rightarrow \T, \quad (z_j)_{j \in J}
 \mapsto \prod_{j\in J} z_j.
\]
The vacuum representations of $G^{\sharp}$ therefore correspond
to vacuum representations of
the weak product
$\prod'_{j\in J}G_j^{\sharp}$ such that the central subgroup
$\prod'_{j\in J}\T_{j}$ acts by $\chi$.
We may assume without loss of generality that
the ground state energy is zero, $H\Omega = 0$.

By Theorem~\ref{thm:a.11}, every  vacuum representation $(\rho,\cH,\Omega)$ of the weak product
$\prod'_{j\in J}G_j^{\sharp}$ is a product of vacuum representations
$(\rho_j,\cH_j,\Omega_j)$
of $G_j^{\sharp}$, and by Proposition~\ref{prop:a.11},
$\rho$ is smooth with smooth ground state vector $\Omega$
if and only if all the $\rho_i$ are smooth with smooth ground state $\Omega_j$.

Since $\rho_j$ is irreducible by Proposition~\ref{prop:a.9},
its restriction to the central subgroup
$\T_j \subseteq G^{\sharp}_j$ is a character
$\chi_j \colon \T_j \rightarrow \T$. The product $\rho = \bigotimes_{j\in J}\rho_j$
acts by $\chi$ on the center
$\prod'_{j\in J}\T_j$
if and only if $\chi_j(z) = z\one$ for all $j\in J$.

Using the free $\R$-action to identify $\cK|_{S_j}$ with $\R \times K$,
we obtain an $\R$-equivariant
isomorphism between $G_j = \Gamma_{c}(S_j,\cK)$ and $C^{\infty}_{c}(\R,K)$
(cf.\ \S\ref{subsec:8.1}).
By Theorem~\ref{thm:10.9}, the  vacuum representations of $G_j^{\sharp}$ are characterized up to unitary equivalence by their central charge $c_j \in \N$, and
by Proposition~\ref{prop:groundstatezeroweight},
$(\rho_j, \cH_j, \Omega_j)$ is unitarily equivalent to
the restriction to $C^{\infty}_{c}(\R,K)$
of the lowest weight representation of $\cL^{\sharp}(K)$ with $\lambda = (c_j,0,0)$.
If $c_j=0$, then the corresponding representation is trivial,
so we can omit both $S_j$ and $c_j$ from the description.
\end{proof}

\subsection{Infinitely many circles}
\label{subsec:9.2}

We continue with the case where all connected components $S_j$ of $S$ are
circles. In marked contrast with the case of infinitely many lines, the projective positive energy
representations associated to a single connected component $S_j$ are well understood, allowing
us to classify the projective positive energy representations of $\Gamma_c(S,\cK)$
under the much weaker condition that the Hilbert space $\cH$ is generated by the space $\cE$ of ground states.
These are the \emph{ground state representations} of Definition~\ref{def:groundstate}.

As before, we assume that $K$ is a $1$-connected compact Lie group,  which is not a serious restriction as long as $K$ is connected (cf.~Remark~\ref{rem:7.6}).
In \S\ref{sec:spectralgap} we describe the \emph{spectral gap condition}, an essentially geometric
sufficient condition for all positive energy representations to be generated by the space of ground states.
The main result of this section is Theorem~\ref{prop:9.5} in \S\ref{sec:classificationnoncompactcircles},
where we describe the ground state representations
in terms of the representation theory of UHF $C^{*}$-algebras.

\subsubsection{The spectral gap condition}\label{sec:spectralgap}

Following the line of reasoning in \S\ref{sec:8}, we associate to every compact connected component $S_j$
a `local' Hamiltonian $H_j$. If these local hamiltonians have a uniform spectral gap, we say that
$(\rho,\cH)$ satisfies the \emph{spectral gap condition}. We show that this (essentially geometric) condition
guarantees that the
positive energy representations are generated by their space of ground states.

We continue with the notation
\[ G = \Gamma_c(S,\cK) \cong \prod_{j \in J}{}^{'} \cL_{\Phi_j}( K_j),
\]
where $\prod_{j\in J}'$ denotes the weak direct product\index{weak product!10@group \scheiding $\prod^{'}_{n \in \N} G_n$} as in \S\ref{appendix:weakproducts}.
As in \S\ref{subsec:8.1}, we identify $S_j$ with $\R/\Z$, where the time translation
$\gamma_{S,t}$ acts on $[x_{j}] \in S_{j}$ by
\[\gamma_{S, t}([x_{j}]) = \Big[x_{j}+\frac{t}{T_{j}}\Big].\]
The derivation acts on $\xi_{j} \in \cL_{\Phi_{j}}(\fk_j)$
by
\[
D\xi_{j} = \frac{1}{T_{j}} (\bd_j \xi_{j} + [A_{j},\xi_j])\,.
\]



By choosing a suitable parametrization of $\cK|_{S_{j}}$,
we may assume that $A_{j}$ is constant (see \cite[Prop.~2.14]{Ne14a} or \cite[\S5.2]{MW04})
and lies in the maximal abelian
subalgebra $\ft^{\circ}$ of $\fk^{\phi_{j}}$ (Theorem~\ref{thm:a.3}).
By acting with the $\phi_j$-twisted Weyl group $\cW$, i.e.,
the Weyl group of the underlying Kac--Moody Lie algebra, we may
also assume that $\bd_j + A_{j}$ lies in the positive Weyl chamber,
i.e., ${(\alpha,n)(i(\bd_j + A_j)) \geq 0}$ for all positive roots
$(\alpha, n) \in \Delta^+$ (\cite[\S3]{MW04} and Appendix~\ref{app:1}).

In the following $(\rho_{\lambda_j}, \cH_{\lambda_j})$ denotes the
irreducible positive energy representation of
\[ G_j^\sharp \cong  \cL_{\Phi_j}^\sharp(K_j) \cong \Gamma(S_{j},\cK)^\sharp\]
with lowest weight $\lambda_j$ (cf.~\S\ref{subsec:8.3}).
Then the minimal eigenspace $V^{0}_{j}$ of $\bd_{j}$ in $\cH_{\lambda_{j}}$
is an irreducible $K^{\Phi}$-representation.
Since $A_{j}$ is anti-dominant,
the minimal eigenspace $W^{0}_{j}$ of $H_{j}$
(which is also finite dimensional by Kac--Moody theory)
contains all weight vectors $v_{\mu}$ in $V_{j}^0$
with $\mu(A_{j}) = 0$.
Note that $W_{j}^{0}$ is $1$-dimensional for generic $A_{j}$
and increases in dimension as $\bd_j + A_{j}$ is contained in a smaller face
of the Weyl chamber (or, equivalently, as $A_j$ is contained in a smaller
face of the Weyl alcove),
and that $W_{j}^0 = V_j^0 = V_{\lambda^{0}_{j}}$
if $A_{j} = 0$.
We denote the orthogonal projection $\cH_{\lambda_j} \rightarrow V^{0}_{j}$
by $P_{j}$, and for a finite subset $F \subeq J$, we
set $P_{F} := \prod_{j\in F}P_j$.

Let $(\rho, \cH)$ be a factorial projective positive energy
representation of $G^\sharp$.
Recall from \S\ref{sec:1dcompact} that, for every finite subset $F \subseteq J$,
we have a tensor product decomposition $\cH = \cH_{F} \otimes \cH'_{F}$.
Here $\cH_{F} = \bigotimes_{j\in F} \cH_{\lambda_{j}}$
is a positive energy representation of
$G^{\sharp}_{F}$ with Hamiltonian
\begin{equation}
  \label{eq:HFsplit}
H_{F} = \sum_{j\in F} H_{\lambda_{j}}, \quad\text{where}\quad
H_{\lambda_j} = i\dd\rho_{\lambda_{j}}\Big(\frac{1}{T_{j}} (\bd_j  + A_{j})\Big)
-  \frac{i}{T_j}\lambda_j(\bd_j + A_j)\one
\end{equation}
is the minimal non-negative Hamiltonian on $\cH_{\lambda_j}$
from \S\ref{sec:minimal}.
The other factor $\cH_F'$ is a minimal positive energy representation
of $G^{\sharp}_{J\setminus F}$ with Hamiltonian $H'$, and we have
\[H = H_{F}\otimes \one + \one \otimes H'\,. \]


The ground states for a $\widehat{G}$-representation (in the sense of Definition~\ref{def:groundstate})
can be characterized in terms of the `local' Hamiltonians as follows.
\begin{Lemma} \label{lem:9.3}
For a factorial minimal positive energy representation $(\rho,\cH)$ of~$\hat G$,
a vector $\Omega \in \cD(H) \subeq \cH$ is a
ground state vector if and only if $H_{\lambda_j} \Omega = 0$ for every $j \in J$.
\end{Lemma}

\begin{proof} ``$\Rarrow$'': Suppose first that $\Omega$ is a ground state
vector. Then $0 \leq H_{\lambda_j} \leq H$ implies that $H_{\lambda_j} \Omega = 0$.

``$\Leftarrow$'': Conversely, suppose that $H_{\lambda_j} \Omega = 0$
holds for all $j \in J$. By minimality, the cyclic subspace generated by
$\Omega$ under $G^\sharp$ is $\hat G$-invariant and the corresponding representation
on this subspace is minimal. We may therefore assume that $\Omega$ is cyclic.

For every finite subset $F \subeq J$, $\Omega$ is fixed by the  operators
${V^F_t := e^{-it H_F}}$, $t \in \R$. These operators satisfy
\[ V^F_t \rho(g) V^F_{-t} = \rho(\alpha_t(g)) \quad \mbox{ for }\quad g \in G_F^\sharp,
t \in \R.\]
For any finite superset $F' \supeq F$ we then have
\[V^{F'}_t \rho(g) \Omega
= \rho(\alpha_t(g)) \Omega
=V^{F}_t \rho(g) \Omega \quad \mbox{ for } \quad g \in G_F^\sharp.\]
This means that $V^F_t$ and $V^{F'}_t$ coincide on the closed subspace $\cH^F$ generated
by $\rho(G_F^\sharp)\Omega$. Since the union of these subspaces is dense in $\cH$,
we obtain a unitary one-parameter group $(V_t)_{t \in \R}$ on $\cH$
whose restriction to $\cH^F$  coincides with $(V^F_t)_{t \in \R}$.
This implies that $V_t \rho(g) V_{-t} = \rho(\alpha_t(g))$ for $g \in G^\sharp, t \in \R$.

Write $V_t = e^{-it \tilde H}$ for a positive selfadjoint operator $\tilde H$.
Then our construction shows that
$\tilde H$ coincides with $H_F = \sum_{j \in F} H_{\lambda_j}$ on $\cH^F$,
and thus $\tilde H \geq 0$.
By minimality of $H$, we have
$0 \leq H \leq \tilde H$, so that $\tilde H\Omega = 0$ leads to
$H\Omega = 0$.
\end{proof}

\begin{Definition}{\rm (Spectral gap)} \label{def:spectralgap}
We say that the family
$(\lambda_j, A_j, T_j)_{j \in J}$ satisfies the \index{spectral gap condition \vulop}
{\it spectral gap condition} if
there exists a positive real number $\Delta E$ such that, for every $j \in J$,
\[ \Spec(H_{\lambda_j}) \subeq \{0\} \cup [\Delta E,\infty).\]
\end{Definition}

The spectral gap condition is essentially geometric in nature.
Recall that for $m\in M$, the $\R$-action $\gamma_t$ yields a group automomorphism
$\gamma_{t}(m) \colon \cK_{m} \rightarrow \cK_{\gamma_{M}(m)}$.
The spectral gap condition is
automatically satisfied if the period
\begin{equation}\label{eq:Period}
T(m) := \inf \Big\{t > 0\,;\, \gamma_{M,t}(m) = m, \; \gamma_{t}  = \mathrm{Id} \in \mathrm{Aut}(\cK_{m})\Big\}
\end{equation}
is uniformly bounded on $M$.
Indeed, the $\R$-action on $\Gamma(S_j,\cK)$ then has period $T_j \leq\sup_{m\in M}T(m)$,
so the spectrum of $\bd_j + A_j$ in every minimal unitary positive energy representation
will be contained in $(2\pi i /T_j) \Z$.

\begin{Proposition} {\rm(Spectral gaps yield ground state vectors)}\label{Prop:spectralGAP}
Let $(\rho,\cH)$ be a factorial
minimal positive energy representation of $\hat G$
such that the corresponding family $(\lambda_j, A_j, T_j)_{j \in J}$
satisfies the spectral gap condition with some $\Delta E > 0$. Then $\cH$ is generated
under $G^\sharp$ by the subspace $\ker H$ of ground state vectors.
\end{Proposition}

\begin{proof} The minimality implies that $0$ is the infimum of the spectrum of $H$,
so that the spectral projection $P := P([0,\Delta E/2])$ is non-zero. First we show that
$P\cH$ is contained in the kernel of every $H_{\lambda_j}$. In fact, the operator
$H - H_{\lambda_j}$ is non-negative.
Since the minimal non-zero spectral value of $H_{\lambda_j}$ is $\geq \Delta E$, it follows that
$P\cH \subeq \ker H_{\lambda_j}$.
Lemma~\ref{lem:9.3} now shows that $H\Omega = 0$.
Therefore $\cF := P\cH$ coincides with the subspace $\ker H$ of ground state vectors.

Next we show that $\cF$ is generating under $G^\sharp$.
Let $\cH^1 \subeq \cH$ be the closed subspace generated by $\cF$ under $G^\sharp$.
Then we obtain a $G^\sharp$-invariant decomposition
$\cH = \cH^1 \oplus \cH^2$. Minimality of $\rho$ now
implies that it is also $\hat G$-invariant, so that
the Hamiltonian $H$ decomposes accordingly as $H = H^1 \oplus H^2$.
Since $\cF \cap \cH^2 = \{0\}$, we obtain $\cH^2 = \{0\}$ by minimality of $H^2$
and the first part of the proof.
This shows that $\cH = \cH^1$ is generated by $\cF$ under $G^\sharp$.
\end{proof}

\subsubsection{Classification in terms of UHF $C^{*}$-algebras}\label{sec:classificationnoncompactcircles}

As in \cite{JN15}, where we dealt with norm continuous representations
of gauge groups,
we aim at a description of the factor representations
of positive energy in terms of $C^*$-algebras. As semiboundedness is crucial
to obtain corresponding $C^*$-algebras (\cite{NSZ17}), we first
observe that positive energy representations are semibounded (cf.~Definition~\ref{def:semibounded}).

Applying Corollary~\ref{cor:6.34} with $M = S$, we immediately obtain:

\begin{Theorem} \label{thm:nc1} If all connected components of $S$ are compact,
then every projective positive energy representation $(\rho,\cH)$ of
$\Gamma_c(S,\cK) \cong \prod'_{j \in J} \cL_{\Phi_j}(K_j)$
is semibounded with the affine hyperplane
$\Gamma_c(S,\fK)^\sharp - D$ contained in the open cone $W_\rho$, so that
$W_\rho$ is an open half space.
In particular, it is a positive energy representation for all derivations
$D_A := D - \ad A$, $A \in \Gamma_c(S,\fK)$.
\end{Theorem}


Let $(\rho,\cH)$ be a factorial
minimal positive energy representation of $\hat G$
and let $(\lambda_j, A_j, T_j)_{j \in J}$ be as above.
Since
the projection ${P_j \: \cH_{\lambda_j} \to V_j^0}$ onto the minimal
energy space for $H_{\lambda_j}$ in $\cH_{\lambda_j}$ is finite dimensional,
$P_j$ is a compact operator.
We may therefore consider the direct limit\index{direct limit $C^*$-algebra \scheiding $\cB$}
\begin{equation}
  \label{eq:defcB}
  \cB := \bigotimes_{j \in J} (K(\cH_{\lambda_j}), P_j)
\end{equation}
of the $C^*$-algebras
\[   \cB_F := \bigotimes_{j \in F} (K(\cH_{\lambda_j}), P_j), \qquad
F \subeq J \ \mbox{ finite}, \]
where 
the tensor product of the non-unital algebras $K(\cH_{\lambda_j})$
is constructed  as in \cite{GrNe09} with the inclusions
\[ \cB_{F_1} \into \cB_{F_2}, \quad
A \mapsto A \otimes \bigotimes_{j \in F_2\setminus F_1} P_j\]
for finite subsets $F_1 \subeq F_2$ of $J$.
We write $B \otimes \bigotimes_{j \in J \setminus F} P_j$ for the image of
$B \in \cB_F$ in $\cB$ and
\begin{equation}
  \label{eq:pinf}
P_\infty := \bigotimes_{j \in J} P_j.
\end{equation}
If $J$ is finite, then $\cB \cong \cB_J$ and the above tensor product is finite.
The $C^*$-algebra $\cB$ carries a natural one-parameter group
of automorphisms $(\alpha^\cB_t)_{t \in \R}$ specified by
\[ \alpha^\cB_t(B) = e^{-it H_F} B e^{it H_F} \quad \mbox{ for } \quad
t \in \R, B \in \cB_F,\]
which fixes the projection $P_\infty$.

Since every ground state representation can be written as a direct sum of cyclic ones, we may
assume w.l.o.g.\ that $\cH$ has a cyclic ground state $\Omega \in \cH$. This
defines a state of $\cB$ by
\[ \omega(B)
:= \lra{\Omega, B \Omega}
\quad \mbox{ for } \quad B \in \cB_F\]  because
$P_j$ projects onto the kernel of $H_{\lambda_j}$ which contains~$\Omega$.
Conversely, if $(\pi, \cH)$ is a representation of the $C^*$-algebra $\cB$ that is generated by a
vector $\Omega$ with
\[  \pi(P_\infty)\Omega = \Omega,\]
then we obtain commuting representations of the multiplier algebras $B(\cH_{\lambda_j})$ of $\cK(\cH_{\lambda_{j}})$.
In particular we recover a
 unitary representation of the restricted product
$\Pi'_{j \in J} \U(\cH_{\lambda_j})$, and hence a unitary representation
of $G^\sharp$. This representation extends canonically to a minimal
positive energy representation of $\hat G$, where the Hamiltonian $H$
is determined uniquely by
\begin{equation}
  \label{eq:bf2}
e^{-it H}\pi(B)\Omega
= \pi(\alpha^\cB_t(B))\Omega \quad \mbox{ for } \quad
B \in \cB.
\end{equation}
The representations constructed above are now positive energy
representations for the $C^*$-dynamical system $(\cB, \R,\alpha^\cB)$
generated by ground states (cf.~\cite{BR02}).

From this correspondence, we derive
 the following non-compact analogue of Theorem~\ref{Thm:ClassificationCompactM}.
\begin{Theorem} \label{prop:9.5}
Let $\cB$ be the $C^*$-algebra constructed for 
$(\lambda_j, A_j, T_j)_{j \in J}$ with a possibly infinite index set $J$
as above. Then the above construction
yields a one-to-one correspondence between:
\begin{itemize}
\item[{\rm(a)}] Isomorphism classes of
minimal factorial positive energy representations of $\hat G$ corresponding to the
family $(\lambda_j, A_j, T_j)_{j \in J}$. 
\item[{\rm(b)}] Isomorphism classes of
factorial representations of $\cB$ that are generated by fixed points of the projection
$P_\infty$.
\end{itemize}
\end{Theorem}

\begin{proof} ``${\rm(a)} \Rarrow {\rm(b)}$'':
Let $(\rho,\cH)$ be a factorial minimal positive energy
representation of $\hat G$ corresponding to the
family $(\lambda_j, A_j, T_j)_{j \in J}$.
As $J$ is at most countably infinite, we may w.l.o.g.\ assume that
$J = \N$ (the case of finite $J$ is proved along the same lines)
and put $\cB_n := \cB_{F_n}$ for $F_n = \{ 1,\ldots, n\}$.
Then we inductively choose factorizations of
$(\rho,\cH)$ as $(\rho_{F_n}\otimes \rho_{F_n}', \cH_{F_n} \otimes
\cH_{F_n}')$ with
\[ \cH_{F_n} = \cH_{\lambda_1} \otimes \cdots \otimes \cH_{\lambda_n},
\quad \rho_{F_n} \cong \rho_{\lambda_1} \otimes \cdots \otimes
\rho_{\lambda_n}.\]
We then obtain a consistent sequence of representations of the
$C^*$-algebras $\cB_n$ on the subspaces
$\cH_n := \cH_{F_n} \otimes \cE_n'$, where
$\cE_n' \subeq \cH_{F_n}'$ is the minimal eigenspace of $H'_F$ on $\cH_{F_n}'$, by
\[ \pi_n \: \cB_n \to B(\cH_{F_n} \otimes \cE_n'), \quad
\pi_n(B) := B \otimes \one \quad \mbox{ for } \quad
B \in \cB_n.\]
As the union of the subspaces $(\cH_n)_{n \in \N}$ is dense in $\cH$,
we thus obtain a non-degenerate representation $(\pi, \cH)$ of $\cB$
satisfying
\begin{equation}
  \label{eq:mult-rel}
 \rho_{F_n}(g) \pi_n(B) = \pi_n(\rho_{F_n}(g)B) \quad \mbox{ for } \quad
g\in G_{F_n}^\sharp, B \in \cB_n,
 \end{equation}
and $\pi(P_\infty)$ is the projection onto the minimal eigenspace of~$H$.
Note that \eqref{eq:mult-rel} determines the representation $\rho$ uniquely
in terms of the representation $(\pi,\cH)$ of~$\cB$.

``${\rm(b)} \Rarrow {\rm(a)}$'': Suppose, conversely, that
$(\pi, \cH)$ is a factorial representations of $\cB$
generated by the subspace $\cE := P_\infty \cH$.
Then the union of the closed subspaces
\[ \cH_n := \pi(\cB_n)\pi(P_\infty)\cH \]
is dense in $\cH$. Since the representation of $\cB_n \cong \cK(\cH_{F_n})$
on $\cH_n$ is non-degenerate, we obtain consistent factorizations
\[ \cH_n \cong \cH_{\lambda_1}  \otimes \cdots \otimes
\cH_{\lambda_n} \otimes \cE_n' \cong \cH_{F_n} \otimes \cE_n'
\quad \mbox{ with } \quad \pi(B) = B \otimes \one_{\cE_n'}, \quad B \in \cB_n.\]
This implies the existence of smooth unitary representations
$\rho_n$ of the groups $G_{F_n}^\sharp$ on $\cH_n$ which are uniquely determined by
\begin{equation}
  \label{eq:mult-rel2}
 \rho_{n}(g) \pi(B) \pi(P_\infty)= \pi(\rho_{F_n}(g)B) \pi(P_\infty)
\quad \mbox{ for } \quad
g\in G_{F_n}^\sharp, B \in \cB_n.
 \end{equation}
The uniqueness implies that
$\rho_{n+1}(g)\res_{\cH_n} =\rho_{n}(g)$ for $g \in G_{F_n}^\sharp$,
so that we obtain a unitary representation of $G^\sharp = \bigcup_F G^\sharp_F$
on~$\cH$ which naturally extends to $\hat G = G^\sharp \rtimes \R$.
Its continuity follows from \cite[Lemma~4.4]{Gl07}.

For the smoothness we use \cite[Thm.~2.9]{Ze17}: The Lie algebra
$\hat\g$ is the union of the subalgebras $\hat \g_{F_n}$, and the
representation is smooth on the corresponding subgroup $\hat G_{F_n}$.
Further, the element $D \in \bigcap_n \hat\g_{F_n}$ lies in the
interior of the open cones
\[W_{\rho\res_{\hat G_{F_n}}}
\supeq \g^\sharp_{F_n} + (0,\infty) D,\]
which are open half spaces
(Theorem~\ref{thm:nc1}).
To apply Zellner's Theorem,
we have to show that the groups $\hat G_{F_n}$ have the
Trotter property, i.e., for any two elements $x,y$ in the Lie algebra, we
have
\[ \exp(t(x+y))=\lim_{n\to\infty}
\Big(\exp\Big(\frac{t}{n}x\Big)
\exp\Big(\frac{t}{n}y\Big)
\Big)^n \]
in the sense of uniform convergence on compact subsets of $\R$.
We first use \cite[Thm.~4.11]{NS13} to see that
$G_{F_n} \rtimes \R$ has the Trotter property; as these groups
are $C^0$-regular (\cite[Thm.~J]{Gl16}),
\cite[Thm.~4.15]{NS13} implies that the central extension $\hat G_{F_n}$
also has the Trotter property.
As any two elements $x,y \in \hat\g$ are contained in some $\hat \g_{F_n}$,
the group $\hat G$ also has the Trotter property.
Therefore \cite[Thm.~2.9(a)]{Ze17} implies that
the dense subspace $\cD^\infty(\dd \rho(D))$ of
smooth vectors of the Hamiltonian coincides with $\cD^\infty_c(\hat \g)$,
the set of all vectors $\xi$ in the common domain of all
finite products of elements in $\hat\g$, for which all maps
\[ \hat\g^n \to \cH, \quad (x_1, \ldots, x_n) \mapsto
\dd\rho(x_1) \cdots \dd \rho(x_n) \xi \]
are continuous and $n$-linear.
As the subgroup $G^\sharp$ is locally exponential,
\cite[Lemma~4.3]{Ne10b} now implies that $\xi$ is a smooth vector
for $G^\sharp$, and since it is also smooth for $H = i \oline{\dd \rho(D)}$,
\cite[Thm.~7.2]{Ne10b} further entails that it is smooth for $\hat G$.
This proves the smoothness of $\rho$.

Clearly, the two constructions are mutually inverse, up to unitary equivalence.
\end{proof}

\begin{Remark} (a) By Lemma~\ref{lem:9.3}, the preceding theorem
covers all minimal factorial representations for which
$(\lambda_j, A_j, T_j)_{j \in J}$
satisfies the spectral gap condition.

(b) The projection $P_\infty \in \cB$ defines the hereditary subalgebra
$\cA := P_\infty \cB P_\infty$ onto which
\[ \eps \: \cB \to \cA, \quad B \mapsto P_\infty B P_\infty \]
defines a conditional expectation, so in particular a completely positive map.
From this perspective, the representations specified in
Theorem~\ref{prop:9.5} are precisely those obtained by
Stinespring dilation from the completely positive maps
of the form $\omega = \pi \circ \eps$, where $(\pi,\cF)$ is a
non-degenerate representation of $\cA$.

For $n_j := \tr P_j$, we have
\[ \cA \cong \bigotimes_{j \in J} M_{n_j}(\C),\]
showing that $\cA$ is a UHF algebra (\cite{Po67}).
The representation theory of these algebras also appears naturally
in the context of norm continuous representations
of gauge groups (cf.~\cite{JN15}).
If infinitely many of the $n_j$ are $> 1$, this
leads to factor representations of type II and III. So the situation depends on
the size of the minimal energy spaces in $\cH_{\lambda_j}$.
In particular, we obtain factorial representations as infinite tensor products
corresponding to factorial product states on $\cA$ because they correspond
to product states on $\cB$. We refer to \cite{JN15} for details
on the connection between norm-continuous representations of the restricted
product $\prod'_{j \in J} K_j$ of the compact groups $K_j$ and
representations of infinite tensor products of matrix algebras.
\end{Remark}

\subsection{A simple example with fixed points}\label{sec:fixedpoint}

In Part II of this series, we will focus on the type of phenomena one encounters when the $\R$-action on $M$
is \emph{not} fixed point free. To give a preview of the problems one encounters there,
we briefly revisit the simple example of the circle action on $\bS^2$, lifted to an $\R$-action
on the trivial bundle $\cK = \bS^2 \times K$ (cf.\ Example~\ref{ex:R2circles}).
The fixed points are then the `north pole' $n = (0,0,1)$ and the `south pole' $(0,0,-1)$.


Since every projective positive
energy representation of $G = C^{\infty}(\bS^2,K)$ restricts to a projective
positive energy representation of the normal subgroup \[G^{\times} = C^{\infty}_{c}(\bS^2\setminus \{n,s\}, K),\]
we can apply the techniques developed so far to $G^{\times}$.
The two problems that remain are then
to determine if a representation extends from $G^{\times}$ to $G$,
and, if so, to classify the possible extensions. We will pursue these problems elsewhere,
and for the moment content ourselves with describing the representation theory
of~$G^{\times}$. Although the Lie algebra bundle $\fK\rightarrow \bS^2$ is trivial,
the $\R$-action \eqref{eq:liftsphere} on $\fK$ that covers
the circle action on $\bS^2$ will in general not be trivializable.
It turns out that
the lift of the circle action at the fixed points $n, s\in \bS^2$ has a
qualitative effect on the positive energy representation theory of $G^{\times}$.


By Theorem~\ref{thm:7.11}, every projective positive energy representation of $G^{\times}$
factors through a projective positive energy representation of $C^{\infty}_{c}(S,K)$, where
\[
S = \big\{ (x,y,z) \in \bS^2 \: z \in J\big\}
\]
is a union of circles labeled by a discrete subset $J \subset (-1,1)$ that has at most two accumulation
points $\pm 1$, corresponding to the fixed points $n$ and $s$.
Recall from Example~\ref{ex:R2circles} that the fundamental vector field for the $\R$-action is of the form
\[
\mathbf{v}(x,y,z) = (y\partial_x - x \partial_y) + A(x,y,z),
\]
with $A(x,y,z)\in \fk$. For simplicity, consider first the case where
$A\in \fk$ is \emph{independent} of $(x,y,z)$.
If we identify the loop algebras of the various circles in the obvious manner, then
the infinitesimal $\R$-action is represented by the \emph{same}
element $\bd_j + A_j = \bd + A$ for every circle $S_j$.
It follows that $(\lambda_j, A_j, T_j) = (\lambda_j, A, 2\pi)$, so the $C^*$-algebra $\cA = P_{\infty}\cB P_{\infty}$
that governs the ground state representations is essentially determined
by 
a sequence $\lambda_j$ of anti-dominant integral weights for the affine Kac-Moody algebra $\widehat{\cL}(\fk)$.

The operators $H_j = i\pi_{\lambda_{j}}(\bd + A)$ are readily seen to satisfy the spectral gap property~\ref{def:spectralgap}.
Indeed, the operators $i\pi_{\lambda_j}(\bd)$ on $\cH_{\lambda_j}$ have a uniform
spectral gap because the $\R$-action on $\bS^2$ is periodic.
Since $A\in \fk$ has a uniform spectral gap in all finite dimensional lowest weight representations,
it also has a uniform spectral gap in
the minimal eigenspaces $W_{j}^{0}$ of the operators $i\pi_{\lambda_j}(\bd)$.
The spectral gap for $H_j$ then follows from the fact that $\bd$ commutes with $A$ in $\widehat{\cL}(\fk)$.

By Proposition~\ref{Prop:spectralGAP}, every
factorial positive energy representation of $\widehat{G}^{\times}$ is a ground state representation, so
the factorial projective positive energy representations are
completely classified by Theorem~\ref{prop:9.5}.

\begin{itemize}
\item[a)] If $A$ is an inner point of the Weyl chamber, then the minimal eigenspace of $H_j$ in an irreducible $\fk$-representation
is always 1-dimensional, $W_j^{0} = \C\Omega_j$. In this case every
projective irreducible positive energy representation is a
\emph{vacuum representation},
and it is of the form
\begin{equation}\label{eq:productsphere}
	(\cH,\Omega) = \bigotimes_{j\in j}(\cH_{\lambda_j}, \Omega_j)
\end{equation}
by the results in \S\ref{appendix:vacuum}.
Moreover, every factorial positive energy representation
of $\widehat{G}^{\times}$ is of type I, i.e., a direct sum of irreducible
representations. This follows from the fact that, if in the construction
of Subsection~\ref{subsec:9.2} all projections $P_j$ are
of rank~$1$, then the projection $P_\infty \in \cB$ has the property
that the subalgebra $P_\infty \cB P_\infty$ is one-dimensional. In particular
$P_\infty a P_\infty = \phi(a) P_\infty$ defines a state of $\cB$ and
every representation $(\pi,\cH)$ of $\cB$ generated by the range of
$\pi(P_\infty)$ is a multiple of the GNS representation
$(\pi_\phi, \cH_\phi)$. For unit vectors $\Omega_j \in \im(P_j)$,
\eqref{eq:defcB} implies that
\[ (\cH_\phi, \Omega_\phi) \cong
\bigotimes_{j \in J} (\cH_{\lambda_j}, \Omega_j), \]
that the representation $\pi_\phi$ is faithful, and that
$\pi_\phi(\cB) \cong K(\cH_\phi).$ \\

\item[\rm b)] If $A$ lies in at least one face of the Weyl
alcove, then the space $W^{0}_j$ of ground states
need not be 1-dimensional. The projective positive energy factor representations of $\fg^{\times}$ are then classified by the lowest weights
$\lambda_{j}$, together with a representation of the UHF $C^*$-algebra
\[
	\cA = \bigotimes_{j\in J} B(W^0_{j}).
\]
If $W^{0}_j$ is of dimension $>1$ for infinitely many $j$, then this is an
infinite tensor product of matrix algebras.  By \cite{Po67} it follows that $G^{\times}$ admits factor representations
of type II and III.
\end{itemize}

If $A(x,y,z)$ is not constant and $A(n)$ and $A(s)$ are inner points of the
Weyl
alcove,
then the situation remains qualitatively the same as in a).
Indeed, the
holonomy with respect to $A$ on $S_j$ will approach $\exp(A(n))$ or $\exp(A(s))$
as $z_j \rightarrow \pm 1$, so the spectral gap condition holds for all but finitely many circles.
In this case one finds a tensor product decomposition analogous to \eqref{eq:productsphere},
where all but finitely terms are vacuum representations.
In particular the space of ground states
is finite dimensional.

However, if either $A(n)$ or $A(s)$ is not an inner point of the Weyl
alcove, then the spectral gap condition
need no longer be satisfied. The ground state representations can still be
classified in the manner outlined above, but these can no longer be expected to
exhaust the positive energy representations.

\appendix

\section{Appendix}
\subsection{Twisted loop algebras and groups}
\label{app:1}

Let $K$ be a simple compact Lie group, $\Phi\in \Aut(K)$ and $\phi = \Lie(\Phi) \in \Aut(\fk)$. We assume that $\phi^N = \id_K$ and let
\[  \cL_\Phi^T(K)
:= \big\{ \xi \in C^\infty(\R,K) \: (\forall t \in \R)\
\xi(t +T) = \Phi^{-1}(\xi(t))\big\} \]
be the corresponding {\it twisted loop group}\index{loop group!20@$T$-periodic \scheiding $\cL_\Phi^{T}(K)$}.
The rotation
action
$(\alpha_t f)(s) := f(s+t)$ satisfies $\alpha_{NT} = \id_{\cL_\Phi(K)}$.
The Lie algebra of $\cL^{T}_\phi(K)$ is the twisted loop algebra\index{loop algebra!20@$T$-periodic \scheiding $\cL^{T}_\phi(\fk)$}
\[ \cL^{T}_\phi(\fk)
:= \big\{ \xi \in C^\infty(\R,\fk) \: (\forall t \in \R)\
\xi(t +T) = \phi^{-1}(\xi(t))\big\}. \]
Accordingly, we obtain
\[  \hat\cL^{T}_\phi(\fk) :=
(\R \oplus_{\omega} \cL^{T}_\phi(\fk)) \rtimes_{D} \R,\quad D\xi = \xi', \]
where
\[ \omega(\xi,\eta)
:=  \frac{c}{2\pi T}\int_0^T \kappa(\xi'(t),\eta(t))\, dt \]
for some $c \in \Z$ (the {\it central charge}).  \index{central charge \scheiding $c$}
Here $\kappa$ is the Killing form of $\fk$,
normalized as in \eqref{eq:kappa-normal} by
$\kappa(i\alpha^\vee, i\alpha^\vee)=2$ for the coroots
corresponding to long roots.

We write
\[ \cL^{T,\sharp}_\phi(\fk) = \R \oplus_{\omega} \cL^{T}_\phi(\fk).\]


Let  $\ft^\circ \subeq \fk^\phi$ be a maximal abelian
subalgebra, so that $\fz_\fk(\ft^{\circ})$ 
is maximal abelian in $\fk$
by \cite[Lemma~D.2]{Ne14a} (see also \cite{Ka85}).
Then $\ft = \R \oplus \ft^\circ \oplus \R$
is maximal abelian\index{maximal!20@abelian subalgebra \scheiding $\ft$} in $\hat\cL^{T}_\phi(\fk)$ and the corresponding set of roots\index{roots!10@set of \scheiding $\Delta$}
$\Delta$ can be identified with the set of pairs
$(\alpha,n)$, where
\begin{equation}
  \label{eq:roots}
(\alpha,n)(z,h,s) := (0,\alpha,n)(z,h,s) = \alpha(h) + i s \frac{2\pi n}{NT}, \quad
n \in \Z, \alpha \in \Delta_n.
\end{equation}
Here $\Delta_n \subeq i(\ft^\circ)^*$ is the set of
$\ft^\circ$-weights in
\[ \fk_\C^n  = \{ x \in \fk_\C \: \phi^{-1}(x) = e^{2\pi in/N} x\}.\]
For $(\alpha,n) \not=(0,0)$, the corresponding root space is
\[  \cL^{T,\sharp}_\phi(\fk_\C)^{(\alpha,n)}
=  \fk_\C^{(\alpha,n)} \otimes e_n
= (\fk_\C^\alpha \cap \fk_\C^{n}) \otimes e_n,
\quad \mbox{ where } \quad
e_n(t) = e^{\frac{2\pi i nt}{NT}}. \]
The set
\[ \Delta^\times = \{ (\alpha,n) \: 0 \not= \alpha \in \Delta_n, n \in \Z \} \]
has an $N$-fold layer structure
\[ \Delta^\times  = \bigcup_{n = 0}^{N-1} \Delta_n^\times \times (n+ N \Z),
\quad \mbox{ where } \quad
\Delta_n^\times := \Delta_n \setminus \{0\}.\]
For $n \in \Z$ and $x \in \fk_\C^{(\alpha,n)}$ with
$[x,x^*] = \alpha^\vee$, the element
$e_n \otimes x \in \cL^{T,\sharp}_\phi(\fk_\C)^{(\alpha,n)}$ satisfies
$(e_n \otimes x)^* = e_{-n} \otimes x^*$,
which leads to the coroot
\begin{align}
  \label{eq:coroot}
[e_n \otimes x, (e_n \otimes x)^*] 
&= (\alpha,n)^\vee
= \Big( -i\frac{cn}{NT} \frac{\|\alpha^\vee\|^2}{2}, \alpha^\vee,0\Big)\notag\\
&= \alpha^\vee - \frac{i cn}{NT}\frac{\|\alpha^\vee\|^2}{2} C,
\end{align}
where
$C = (1,0,0)$. Here we have used that
\[ \omega(e_n \otimes x, e_{-n} \otimes x^*)
=  \frac{i c n}{NT} \kappa(x,x^*) \]
and
\[ \kappa(x,x^*)
= \shalf \kappa([\alpha^\vee, x], x^*)
= \shalf \kappa(\alpha^\vee, [x, x^*])
= \shalf \kappa(\alpha^\vee, \alpha^\vee)
= -\shalf \|\alpha^\vee\|^2.\] 

Since $\fk$ is simple, $\Delta^\times$ does not decompose
into two mutually orthogonal proper subsets (\cite[Lemma~D.3]{Ne14a}), so that
\[  \hat\cL^{T}_\phi(\fk)^{\rm alg}_\C :=
\ft_\C + \sum_{(\alpha,n)\in \Delta} \cL^{T,\sharp}_\phi(\fk_\C)^{(\alpha,n)} \]
is an affine Kac--Moody--Lie algebra (\cite[Thm.~8.5]{Ka85}).
In this context the root $(\alpha,n)$ is real if and only if $\alpha \not=0$.
Choosing a positive system\index{roots!20@positive system \scheiding $\Delta^{+}$} $\Delta^+ \subeq \Delta$ such that
the roots $(\alpha, n)$, $n > 0$, are positive, the lowest weights of
unitary lowest weight representations of $\hat\cL^{T}_\phi(\fk)$ are the
anti-dominant
integral weights:
\[ 
\cP(\ft, \Delta^+) := \{ \lambda \in i\ft^* \:(\forall (\alpha,n))\
0 \not= \alpha, (\alpha,n) \in \Delta^+ \ \ \Rarrow \ \
\lambda((\alpha,n)^\vee\,) \in \N_0\}.\] 
Note that, for $n > 0$, we have
\[ \lambda((\alpha,n)^\vee\,)
= \lambda(\alpha^\vee) + \frac{cn}{NT}\frac{\|\alpha^\vee\|^2}{2},\]
so that we obtain $c > 0$ as a necessary condition for the existence
of non-trivial unitary lowest weight modules.

\subsection{Twisted conjugacy classes in compact groups}
\label{app:a.3} 

In this appendix we collect some more details concerning twisted
conjugacy classes in compact groups.

A~{\it Cartan subgroup} \index{Cartan subgroup \vulop}
of a compact Lie group $K$ is an abelian subgroup $S$
topologically generated by a single element $s$ ($s^\Z$ is dense in $S$)
which has finite index in its
normalizer $N_K(S) = \{ k \in K \: kSk^{-1} = S\}$.

\begin{Remark} \label{rem:a.3.1}
(a) For any Cartan subgroup $S$, the identity component $S_0$ is an abelian compact Lie group,
hence a torus, and since tori are divisible, the short exact sequence
$S_0 \into S \onto \pi_0(S)$ splits, so that $S \cong S_0 \times \pi_0(S)$.
By construction, $\pi_0(S)$ is a finite cyclic group.
If $s_0 \in S_0$ is a topological generator,  then, for every $N \in \Z$, the closure of
$s_0^{N\Z}$ is a closed subgroup of finite index in $S_0$, hence equal to $S_0$.
This implies that the topological generators of $S$ are the elements of the
form $s = (s_0, s_1) \in S_0 \times \pi_0(S)$, where
$s_0$ is a topological generator of $S_0$ and $s_1$ is a generator of the cyclic
group $\pi_0(S)$.

(b) By \cite[Prop.~IV.4.2]{BtD85}, every element $k \in K$ is contained
in a Cartan subgroup $S$ such that the connected component $k S_0$ generates
$\pi_0(S)$. The preceding discussion now shows that there exists an element
$s_0 \in S_0$ such that $z := k s_0$ is a topological generator of $S$.
Now \cite[Prop.~IV.4.3]{BtD85} implies that every element $g \in k K_0 = zK_0$
is conjugate to an element of $k S_0$.
\end{Remark}

\begin{Theorem} \label{thm:a.3}
Let $K$ be a compact connected Lie group and $\Phi \in \Aut(K)$ be an automorphism
of finite order $N \in \N$.
We consider the {\it twisted conjugation action} of $K$
on itself given by \index{twisted conjugation \scheiding $g * k$}
\[ g * k := gk \Phi(g)^{-1} \quad \mbox{ for }\quad g,k \in K.\]
Then the orbit of every element in $K$ under this action
intersects a maximal torus $T^\Phi$ of the subgroup $K^\Phi$ of $\Phi$-fixed points.
\end{Theorem}

\begin{proof} We consider the compact Lie group
$K_1 := K \rtimes \Phi^\Z$, where $\Phi^\Z \subeq \Aut(K)$ is the finite subgroup
generated by~$\Phi$. For $g,k \in K$, we then have
\[ (g,\one)(k,\Phi)(g,\one)^{-1} = (gk \Phi(g)^{-1}, \Phi),\]
so that the conjugacy classes in the coset $K \times \{\Phi\} \subeq K_1$ correspond to the
$\Phi$-twisted conjugacy classes in~$K$.

According to Remark~\ref{rem:a.3.1}(b), the element $(\one,\Phi) \in K_1$ is contained in a
Cartan subgroup $S$ which is generated by an element of the form
$z = (s_0, \Phi)$. As $S_0$ is abelian and commutes with $(\one,\Phi)$,
it is contained in $K^\Phi$. Let $T^\Phi \subeq K^\Phi$ be a maximal torus
containing $S_0$. Then $T^\Phi$ commutes with $S$, so that the finiteness
of $N_K(S)/S$ shows that $T^\Phi \subeq S_0$. We conclude that
\[ S = T^\Phi \times \Phi^\Z \]
is a Cartan subgroup of $K_1$. Therefore Remark~\ref{rem:a.3.1}(b) implies that
every $\Phi$-twisted conjugacy class in $K$ intersects $S_0 = T^\Phi\subeq K^\Phi$.
\end{proof}

We refer to \cite{MW04} for more details on twisted conjugacy classes in compact
groups, representatives, and stabilizer groups.

\begin{Remark} If $\Phi$ is not of finite order, then the situation is more complicated.
If, however, $K$ is a compact Lie group with semisimple Lie algebra, then
$\Aut(K)$ is a compact group with the same Lie algebra and one can apply the theory
of Cartan subgroups of compact Lie groups to $\Aut(K)$.
\end{Remark}

 \subsection{Restricting representations to normal subgroups}
%
%
 \begin{Theorem} \label{thm:typeicrit}
 Let $G$ be a group, and let $N \trile G$ be a normal subgroup of finite
 index. Suppose that $(\pi, \cH)$ is a unitary representation of $G$
 whose restriction $\pi\res_N$ decomposes discretely with finitely many isotypic
 components. Then the same holds for~$\pi$.
 \end{Theorem}

 \begin{proof} We consider the two von Neumann algebras
 \[ \cN := \pi(N)'' \subeq  \cM := \pi(G)''.\]
 Let
 \[ \cH = \bigoplus_{j =1}^m \cH_j, \quad \text{with} \quad \cH_j = \cF_j \otimes \cC_j \]
 be the isotypic decomposition
 for $N$, where the representations $(\rho_j, \cF_j)$ of $N$ are irreducible and
 $N$ acts on $\cH_j$ by $\pi_j := \rho_j\otimes \one$. Then
 \[ \cN' = \pi(N)' \cong \bigoplus_{j = 1}^m B(\cC_j). \]
 The conjugation action of $G$ on $\cN'$ factors through an action of the finite
 group $G/N$. We have to show that $\cM' = (\cN')^{G/N}$
 also is a finite direct sum of full operator algebras.

 Let $F := \{[\rho_j] \: j =1,\ldots, m\} \subeq \hat N$ be the support of the
 restriction $\rho\res_N$. This set decomposes
 under the natural action of $G/N$ on the unitary dual $\hat N$ into finitely
 many orbits $F_1, \ldots, F_k$.
 The group $G$ permutes the isotypic subspaces $\cH_j$ of $N$ and, accordingly,
 \[ \pi_{k} \cong \pi_j \circ c_{g}^{-1}\res_N \quad \mbox{ if and only if } \quad
 \pi(g)\cH_j = \cH_{k}.\]
 This follows from the relation $\pi_k(n)\pi(g) = \pi(g) \pi_j(g^{-1}ng)$
 for $g \in G$, $n \in N$. We conclude that
 \[ P_j := \{ g\in G \: \pi(g) \cH_j = \cH_j \}
 = \{ g \in G \: \rho_j \circ c_g \cong \rho_j\}.\]
 For every $g \in P_{j}$, we thus obtain a unitary operator $U_g \colon \cF_j \rightarrow \cF_j$ such that $\rho \circ c_g = U_g \rho U^{-1}_{g}$.
 Since $\cF_j$ is irreducible, $U_g$ is well defined modulo $\T$.
 The projective unitary representation $\overline{\rho}_j(g) = [U_g]$ of $P_j$
 yields a central extension $q_j \: P_j^\sharp \to P_j$, a homomorphic lift
 $N \into P_j^\sharp$
 and an extension $\rho_j^\sharp \: P_j^\sharp \to \U(\cF_j)$ of the unitary representation
 $\rho_j$ of $N$ to $P_j^\sharp$. Accordingly, the representation
 of $P_j^\sharp$ on $\cH_j$ takes the form
 \[ \pi_j(q_j(p)) = \rho_j^\sharp(p) \otimes \beta_j(p), \]
 where $\beta_j \: P_j^\sharp \to \U(\cC_j)$ a unitary representation with
 $\ker\beta_j \supeq N$.

 Let $\cH = \cH^1 \oplus \cdots \oplus \cH^k$ denote the decomposition of $\cH$ under
 $G$, corresponding to the decomposition of $F$ under $G/N$, so that each
 subspace $\cH^j$ is a sum of certain subspaces $\cH_\ell$.
 We may  assume w.l.o.g.\ that $G/N$ acts transitively on $F$, i.e., that
 \[ \cH = \Spann\big(\pi(G) \cH_1\big)
 = \bigoplus_{[g] \in G/P_1} \pi(g) \cH_1.\]
 This means that $(\pi, \cH)$ is  induced from the
 representation $\rho_1^\sharp \otimes \beta_1$ of $P_1$ on~$\cH_1$.

 The subspace $\cH_1$ is generating for $G$, and hence separating for the commutant
 $\cM'$. As $\cH_1$ is isotypic for $N$, the commutant $\cM' \subeq \cN'$ leaves
 $\cH_1$ invariant; likewise all subspaces $\cH_j$ are $\cM'$-invariant.
 Since an operator $A \in B(\cH_1)$ extends to an element
 of $\cM'$ if and only if it commutes with $P_1$, we have
 \[ \cM' \cong (\rho_1^\sharp \otimes \beta_1)(P_1)'
 = \one \otimes \beta(P_1^\sharp)'.\]
 Since $\beta(P_1^\sharp)$ is a finite group, the assertion follows from the
 fact that every unitary representation of a finite group decomposes discretely
 with finitely many isotypes.
 \end{proof}

\subsection{Vacuum representations}
\label{subsec:a.5}

In this appendix, we show that vacuum representations of weak products of topological groups arise as products of vacuum representations.

\subsubsection{Weak products and $\R$-actions}\label{appendix:weakproducts}

The \emph{weak product} \index{weak product!10@group \scheiding $\prod^{'}_{n \in \N} G_n$}
of a sequence $(G_{n})_{n\in \N}$ of topological groups is defined~as
\[ G := \textstyle{\prod^{'}_{n \in \N}} G_n = \bigcup_{N = 1}^\infty G^N, \quad
G^N = G_1 \times \cdots \times G_N,\]
where the group structure is inherited from the product group
$\prod_{n \in \N} G_n$.  However, we will need a topology
that is finer than the product topology.
We equip $G$ with the {\it box topology}, \index{weak product!20@box topology \vulop}
for which a basis
of $e$-neighborhoods consists of the sets
$G \cap \prod_{n = 1}^\infty U_n$, where $U_n \subeq G_n$ is an
$e$-neighborhood in $G_n$.  By \cite[Lemma~4.4]{Gl07}, this turns
$G$ into a topological group, and $G$ is the direct limit in
the category of topological groups
of the increasing sequence of subgroups $G^N$, endowed with the
product topology.

To study vacuum representations of weak products, consider a sequence $(G_n, \R, \alpha_n)_{n \in \N}$
of topological groups, with for each $n\in \N$ a homomorphism
$\alpha_n \: \R \to \Aut(G_n)$ that defines a continuous action
of $\R$ on $G_n$.
The
homomorphisms $\alpha_n \: \R \to \Aut(G_n)$ combine to a homomorphism
$\alpha \: \R \to \Aut(G)$ by
\[ \alpha_t(g_1, \ldots, g_N, e,\ldots)
:= (\alpha_{1,t}(g_1), \ldots, \alpha_{N,t}(g_N), e,\ldots), \]
where $\Aut(G)$ denotes the group of topological automorphisms.

\begin{Proposition} \label{prop:a.10}
The above map $\alpha$ is a continuous action of $\R$ on $G$.
\end{Proposition}
\begin{proof}
To see this, 
we first note that all orbit maps are continuous because the subgroups
$G^N$ carry the product topology. Since all automorphisms
$\alpha_t$ are continuous by \cite[Lemma~4.4]{Gl07}, it suffices to verify continuity
of the action in all pairs $(0,g) \in \R \times G^N$.
So we have to find for every
sequence $(U_n)_{n \in \N}$ of $e$-neighborhoods in $G_n$
an $\eps > 0$ and a sequence of $e$-neighborhoods ${V_n \subeq G_n}$
such that
\[ \alpha_{n,t}(g_n V_n) \subeq g_n U_n \quad \mbox{ for } \quad
|t| < \eps, n \in \N.\]
As $[-1,1] \subeq \R$ is compact, we find for every $n \in \N$
an identity neighborhood ${V_n\subeq  W_n \subeq G_n}$ such that
$W_n W_n \subeq U_n$ and
$\alpha_{n,t}(V_n) \subeq W_n$ for $|t| \leq 1$.
For $n \leq N$ we now choose $\eps > 0$ in such a way that
$\alpha_{n,t}(g_n) \in g_n W_n$ holds for $|t| \leq \eps$. Then
\[ \alpha_{n,t}(g_n V_n)
= \alpha_{n,t}(g_n) \alpha_{n,t}(V_n)
\subeq g_n W_n W_n \subeq g_n U_n \]
holds for $|t| < \eps$ and $n \leq N$. For $n > N$, we have
$g_n = e$ and
\[ \alpha_{n,t}(g_n V_n)
= \alpha_{n,t}(V_n) \subeq W_n \subeq U_n \quad \mbox{ for } \quad
|t| \leq \eps.\]
Therefore $\alpha$ defines a continuous action on the
weak direct product~$G$.
\end{proof}

If, in addition, the groups $G_n$ are Lie groups, then
the box topology on $G$ is compatible with a Lie group
structure on $G$ (\cite[Rem.~4.3]{Gl07}).

\begin{Lemma} \label{lem:a.11} If all groups $G_n$ are locally exponential, then
$\alpha$ defines a smooth action on $G$.
\end{Lemma}
\begin{proof} By \cite[Rem.~4.3]{Gl07}, the group $G$ is locally exponential
as well. Therefore it suffices to show that the $\R$-action on
the Lie algebra $\g \cong \oplus_{n \in \N} \g_n$
(the locally convex direct sum), is smooth.
Let $D_n \in \der(\g_n)$ denote the infinitesimal generator of the smooth
actions $\alpha^n$ on $\g_n$. Then
\[ \alpha(t,x) = (e^{t D_n} x_n)_{n \in \N} = e^{tD} x
\quad \mbox{ for } \quad
D(x_n) = (D_n x_n)\]
and  the tangent map of $\alpha$ is given by
\[ \dd \alpha(t,x)(s,y) = s D(\alpha(t,x)) + \alpha(t,y).\]
As $D \: \g \to \g$ is a continuous linear operator, we inductively obtain
from the continuity of $\alpha$
(Proposition~\ref{prop:a.10}) that $\alpha$ is $C^k$ for each $k \in \N$,
and hence that $\alpha$ is smooth.
\end{proof}

The weak products encountered in this paper are mostly of the following form.
\begin{Lemma}\label{eq:LemmaWeakProductsGauge} Suppose that the smooth manifold
$S$ has countably many connected components and that $\cK \to S$
is a Lie group bundle. Then the Lie group
$\Gamma_c(\cK)$ is isomorphic to the restricted Lie group
product $\prod'_{n \in \N} \Gamma_c(\cK\res_{S_n})$.
\end{Lemma}

\begin{proof} Since the groups $G = \Gamma_c(\cK)$ and $G_n = \Gamma_c(\cK\res_{S_n})$ are locally exponential,
it suffices to verify that the Lie algebra
$\g = \Gamma_c(\fK)$ is the locally convex direct sum of the ideals
$\g_n = \Gamma_c(\fK\res_{S_n})$. That the summation map
\[ \Phi \: \bigoplus_{n \in \N} \Gamma_c(\fK\res_{S_n}) \to\g \]
is continuous follows  from the universal property of the locally convex
direct sum. That its inverse $\Phi^{-1}$ is also continuous, follows
from its continuity on the Fr\'echet subspaces
$\Gamma_D(\cK)$, where $D \subeq S$ is compact, because any compact
subset intersects at most finitely many connected components.
\end{proof}

\subsubsection{Vacuum representations}\label{appendix:vacuum}
Let $G$ be a topological group,
and let $\alpha \: \R \to \Aut(G)$ be a homomorphism that defines
a continuous action of $\R$ on $G$.

\begin{Definition} A triple $(\rho,\cH, \Omega)$
is called a {\it vacuum representation} of \index{representation!vacuum \vulop}
$(G,\R,\alpha)$, if
$\rho \: G \rtimes_\alpha \R \to \U(\cH)$ is a continuous unitary
representation, $\Omega \in \cH$ is a $G$-cyclic unit vector,
and the selfadjoint operator $H$, defined by
$U_t := \rho(e,t) = e^{-itH}$ for $t \in \R$, satisfies
$\ker(H- E_0\one) = \C \Omega$ for $E_0 = \inf(\mathrm{spec}(H))$.
\end{Definition}

The following is an immediate consequence of \cite[Prop.~5.4]{BGN20}.
\begin{Proposition} \label{prop:a.9} For a vacuum representation $(\rho,\cH, \Omega)$ of $(G,\R,\alpha)$, the
following assertions hold:
\begin{itemize}
\item[\rm(a)] $U_\R \subeq \rho(G)''$.
\item[\rm(b)] The representation $\rho\res_G$ of $G$ on $\cH$ is irreducible.
\end{itemize}
\end{Proposition}

\begin{proof} (a) The one-parameter group $(U_t^0)_{t \in \R}$
defined by $U_t^0 := e^{it E_0} U_t$ is minimal
for the
von Neumann algebra $\rho(G)''$ {\rm(cf.\ Definition~\ref{def:mini})}
by  \cite[Prop.~5.4]{BGN20}, hence contained in $\rho(G)''$, and this
implies (a).

\nin (b) From (a) it follows that the closed subspace $\C\Omega =
\ker (H -E_0\one) \subeq \cH$ is invariant under the commutant
$\cM' := \rho(G)'$ of $\cM := \rho(G)''$.
As $\Omega$ is generating for $\cM$, it is separating for $\cM'$,
so that $\dim \ker (H_0-E_0\one) = 1$ leads to $\cM' = \C \one$.
Now the assertion follows  from Schur's Lemma.
\end{proof}

Let $(G_n, \R, \alpha_n)_{n \in \N}$ be a sequence
of topological groups, with for each $n\in \N$ a homomorphism
$\alpha_n \: \R \to \Aut(G_n)$ that defines a continuous action
of $\R$ on $G_n$.
The following theorem identifies the vacuum representations
of the weak product $(G,\R,\alpha)$ in terms of  vacuum representations of
the triples $(G_n, \R, \alpha_n)$.

\begin{Theorem} \label{thm:a.11}
For any sequence $(\rho_n, \cH_n, \Omega_n)$
of  vacuum representations of $(G_n, \R, \alpha_n)$
with minimal energy $E_0 =0$,
the infinite tensor product
\begin{equation}\label{eq:inftensorHilbert}
(\cH,\Omega) :=  \bigotimes_{n = 1}^\infty (\cH_n, \Omega_n)
\end{equation}
carries a continuous vacuum representation of
$(G,\R,\alpha)$, defined by
\begin{equation}\label{eq:tensorrepA5}
 \rho(g_1, \ldots, g_n,e,\ldots)
:= \rho_1(g_1) \otimes \cdots \otimes \rho_n(g_n)
\otimes \one_{n+1} \otimes \cdots.
\end{equation}
Conversely, every  vacuum representation of $(G,\R,\alpha)$
with $E_0 = 0$ is equivalent to such a representation.
\end{Theorem}

\begin{proof}
First we prove that if all $(\rho_n, \cH_{n}, \Omega_n)$ are vacuum representations, then so is
their infinite tensor product.
Since the $\Omega_n$ are unit vectors, the infinite
tensor product Hilbert space $\cH$ is defined. It contains the
subspaces
\[ \cH^N :=
\cH_1 \otimes \cdots \otimes \cH_N
\otimes \Omega_{N+1} \otimes \cdots
\cong \cH_1 \otimes \cdots \otimes \cH_N,  \]
whose union is dense in $\cH$. On $\cH^N$,
the representation
$\rho^N$ of $G^N \rtimes \R$,  defined by
\[ \rho^N((g_1, \ldots, g_N),t)
:= \rho_1(g_1, t) \otimes \cdots \otimes \rho_N(g_N,t),\]
 is continuous with cyclic vector
$\Omega = \otimes_{n = 1}^{\infty} \Omega_n.$
The representation $(\rho,\cH)$ of $G$ now  is a direct limit of
the representations $(\rho^N, \cH^N)$ of the subgroups $G^N$, hence
a continuous unitary representation. Further,
the invariance of $\Omega_n$ under the one-parameter group
$U^n_t := \rho_n(e,t)$ implies that
\begin{equation}
  \label{eq:u-def}
U_t(v_1 \otimes \cdots \otimes v_N \otimes \Omega_{N+1}
\otimes \cdots)
:= U^1_t v_1 \otimes \cdots \otimes U^N_t v_N \otimes \Omega_{N+1}
\otimes \cdots
\end{equation}
defines a continuous unitary one-parameter group on $\cH$ satisfying
\[ U_t \rho(g) U_t^* = \rho(\alpha_t(g)) \quad \mbox{ for } \quad
g \in G, t \in \R.\]
By $\rho(g,t) := \rho(g) U_t$, we thus obtain a continuous
unitary representation of $G$ on $\cH$ for which
$\Omega$  is a
$G$-cyclic unit vector fixed by the one-parameter group $(U_t)_{t \in \R}$.
Writing $U_t = e^{-itH}$ and
$U^n_t = e^{-it H_n}$ for selfadjoint operators $H_n \geq 0$,
\eqref{eq:u-def} implies that $H \geq 0$.
To see that $\ker H = \C \Omega$, we decompose
\[ \cH = \cH^N \otimes \cK_N \quad \mbox{ for }\quad N \in \N.\]
Accordingly,
\[ U_t = V_t \otimes W_t \quad \mbox{ with } \quad
V_t = U^1_t \otimes \cdots \otimes U^N_t,\]
and both one-parameter groups $(V_t)_{t \in \R}$ and
$(W_t)_{t \in \R}$ have positive generators
$H_V$ and $H_W$. From \cite[Lemma~A.3]{BGN20} we thus infer that
\[ H = (H_V \otimes \one_{\cK^N}) + (\one_{\cH^N} \otimes H_W)\]
in the sense of unbounded operators,
hence in particular that
\[  \cD(H) = (\cD(H_V) \otimes \cK^N)
\cap (\cH^N \otimes \cD(H_W)).\]
We conclude that, for every $N \in \N$,
\[ \ker H \subeq \ker H_V \otimes \cK^N
=  \Omega_1 \otimes \cdots \otimes \Omega_N \otimes \cK^N,\]
and this shows that
\[ \ker H \subeq \bigcap_N
 \Omega_1 \otimes \cdots \otimes \Omega_N \otimes \cK^N
= \C \Omega.\]
Therefore $(\rho,\cH,\Omega)$ is a  vacuum representation
of $(G,\R,\alpha)$.\\

Now we assume, conversely, that
$(\rho,\cH,\Omega)$ is a  vacuum representation
of $(G,\R,\alpha)$. Then the subspace
\[ \cH^N := \oline{\Spann{\rho(G^N) \Omega}} \]
carries a  vacuum representation
of $(G^N,\R,\alpha^N)$. In particular, this representation
is irreducible by
Proposition~\ref{prop:a.9}. The group $G$ is a topological product
\[ G = G^N \times G^{>N},
\quad \mbox{ where } \quad
G^{>N} := \textstyle{\prod_{n > N}'} G_n,\]
and the representation $\rho$ is irreducible by
Proposition~\ref{prop:a.9}. Since its restriction to
$G^N$ carries an irreducible subrepresentation, the
restriction to $G^N$ is factorial of type I, hence of the form
\[ \rho\res_{G^N} = \rho^N  \otimes \one \]
with respect to some factorization
$\cH = \cH^N \otimes \cK^N$.
Starting with $N = 1$ and proceeding inductively, we see that
\[ \rho^N \cong \rho_1 \otimes \cdots \otimes \rho_N \]
for  vacuum representations
$(\rho_n, \cH_n, \Omega_n)$ of $(G_n,\R, \alpha_n)$.
In particular, we obtain factorizations
\[ \Omega = \Omega^N \otimes \tilde\Omega_N
= \Omega_1  \otimes \cdots \otimes \Omega_N \otimes \tilde\Omega_N,\]
so that we may identify $\cH^N$ with the subspace
\[ \cH^N \otimes \tilde\Omega_N \subeq \cH.\]
As $\Omega$ is $G$-cyclic, the union of these $G^N$-invariant subspaces
is dense in $\cH$. This implies that the  vacuum representation
$(\rho, \cH,\Omega)$ is equivalent to the
infinite tensor product
$\otimes_{n \in \N} (\rho_n, \cH_n, \Omega_n)$ of the ground
state representations $(\rho_n, \cH_n, \Omega_n)$.
This completes the proof.
\end{proof}

The following allows us to reduce the classification of smooth
 vacuum representations to the local case,
under the assumption that the ground state is smooth.

\begin{Proposition} \label{prop:a.11}
Suppose that the $G_n$ are
Lie groups and that the $\R$-actions on $G_n$ are smooth.
Then the  vacuum representation $(\rho,\cH,\Omega)$ is smooth with smooth vector $\Omega$
if and only if
the  vacuum representations $(\rho_n,\cH_n,\Omega_n)$ are smooth with smooth vector $\Omega_n$.
\end{Proposition}

\begin{proof} If $(\rho,\cH,\Omega)$ is a smooth representation with $\Omega \in \cH^{\infty}$,
then $\Omega$ will be a smooth vector for every $(\cH_{n},\rho_{n},\Omega)$ as well. Since $\Omega$ is cyclic
in $\cH_n$, the latter will be a smooth representation.\\
Suppose, conversely,
that the  vacuum representations $(\rho_n,\cH_n,\Omega_n)$ are smooth,
and that $\Omega_n \in \cH^{\infty}_n$ for all $n\in \N$.
From Theorem~\ref{thm:a.11}, we know that the tensor product
representation $(\rho,\cH,\Omega)$ is continuous and cyclic.
To show that the  vacuum representation $(\rho,\cH,\Omega)
= \bigotimes_{n=1}^{\infty}(\rho_n,\cH_n,\Omega_n)$
is smooth with smooth vector $\Omega \in \cH^{\infty}$, it suffices
by \cite[Thm.~7.2]{Ne10b} to show that
$\phi(g) := \langle\Omega, \rho(g) \Omega\rangle$ is a smooth function from $G$ to $\C$.

Note that $\phi$ is the infinite product $\prod_{n = 1}^\infty \phi_n(g_n)$ of the smooth,
positive definite functions $\phi_n \colon G_n \rightarrow \C$ defined by
$\phi_{n}(g) :=  \langle\Omega_n, \rho_n(g)\Omega_n\rangle$.
To see that $\phi \colon G \rightarrow \C$ is smooth, note that
it can be decomposed into the smooth maps
\[ G = {\prod}'_{n \in \N} G_n
\ssmapright{\Phi_1}  \one + {\prod}'_{n \in \N} \C
\ssmapright{\Phi_2} \one + \ell^1(\N) \ssmapright{\Phi_3} \C, \]
where $\one = (1)_{n \in \N}$ and
\[ \Phi_1((g_n)) = (\phi_n(g_n)), \quad
\Phi_2((z_n)) = (z_n), \quad
\Phi_3((z_n)) = \prod_{n \in \N} z_n.\]
Here the smoothness of $\Phi_1$ follows from the compatibility
with the box manifold structure, $\Phi_2$ is continuous affine,
and $\Phi_3$ is holomorphic. It follows that
\[ \phi = \Phi_3 \circ \Phi_2 \circ \Phi_1\] is smooth, and hence that $(\rho,\cH,\Omega)$
is a smooth  vacuum representation with smooth vector $\Omega$.
\end{proof}

\subsection{Ergodic property of 1-parameter subgroups of
\texorpdfstring{$\widetilde{\SL}(2,\R)$}{SL(2,R)}}\label{Appendix:sl2}

We give a simplified proof for the following characterization of the ergodic property for 1-parameter subgroups of
$\widetilde{\SL}(2,\R)$ due to Mautner and Moore.
Define the 1-parameter groups $x(t)$, $y(t)$ and $h(t)$ in $\SL(2,\R)$ by
\[
x(t) = \begin{pmatrix} 1 & t\\ 0 & 1\end{pmatrix}, \quad y(t) = \begin{pmatrix} 1 & 0\\ t & 1\end{pmatrix},
\quad\text{and}\quad h(t) = \begin{pmatrix} e^{t} & 0\\ 0 & e^{-t}\end{pmatrix},
\]
and let $\widetilde{x}(t)$, $\widetilde{y}(t)$ and $\widetilde{h}(t)$ be their lift to $\widetilde{\SL}(2,\R)$.

\begin{Lemma}\label{Lemma:sl2}
Let $(\pi, \cH)$ be a continuous unitary representation of $\widetilde{\mathrm{SL}}(2,\R)$, and let $\Omega\in \cH$ be a unit vector. Then the following are equivalent:
\begin{itemize}
\item[{\rm (a)}] $\pi(\widetilde{x}(t))\Omega = \Omega$ for all $t\in \R$
\item[{\rm (b)}] $\pi(\widetilde{h}(t))\Omega = \Omega$ for all $t\in \R$
\item[{\rm (c)}] $\pi(g)\Omega = \Omega$ for all $g\in \widetilde{\mathrm{SL}}(2,\R)$.
\end{itemize}
\end{Lemma}
This well-known result plays an important role in ergodic theory.
It is due to Calvin Moore \cite{Moore1966, Moore1970}, and in the proof below
we almost literally follow his argument for the implication ${\rm(a)} \Rightarrow {\rm(b)}$
from \cite[p.7]{Moore1970}. The implication ${\rm(b)} \Rightarrow {\rm(c)}$
is due to Mautner \cite{Mautner1957}, and this is implicitly used by Moore in \cite{Moore1970}, and by
Howe and Moore in their seminal paper \cite{HoweMoore1979}.
In his proof, Mautner uses the classification of irreducible unitary representations of $\widetilde{\mathrm{SL}}(2,\R)$.
We bypass this with a simple argument.

\begin{proof}
For $(a) \Rightarrow (b)$, 
let $w(t) = x(t)y(-t^{-1})x(t)$, and note that we have 
${h(t) = w(e^{t})w(1)^{-1}}$ for all $t\in \R$.
If we define $\widetilde{w}(t) := \widetilde{x}(t)\widetilde{y}(-t^{-1})\widetilde{x}(t)$, then the curve $\widetilde{w}(e^t)\widetilde{w}(1)^{-1}$ covers $h(t)$.
Since it is the identity for $t=0$, we have $\widetilde{w}(e^t)\widetilde{w}(1)^{-1} = \widetilde{h}(t)$.
Since $\|\pi(\widetilde{w}(t))\Omega\| = 1$ for all $t\neq 0$, it follows from
\[
\lim_{|t|\rightarrow \infty}
\langle \pi(\widetilde{w}(t)) \Omega, \Omega \rangle 
= \lim_{|t|\rightarrow \infty} \langle \pi(\widetilde{y}(-t^{-1})) \Omega, \Omega \rangle = 1
\]
that $\lim_{|t|\rightarrow \infty} \pi(\widetilde{w}(t))\Omega = \Omega$. So for $\psi = \pi(\widetilde{w}(1))\Omega$, we find
${\lim_{t\rightarrow \infty} \pi(\widetilde{h}(t)) \psi = \Omega}$.
For every $s\in \R$ we thus have
\[
\Omega = \lim_{t\rightarrow \infty} \pi(\widetilde{h}(s + t)) \psi = \pi(\widetilde{h}(s)) \lim_{t\rightarrow \infty} \pi(\widetilde{h}(t)) \psi = \pi(\widetilde{h}(s))\Omega,
\]
so $\Omega$ is fixed by $\widetilde{h}(s)$ for all $s\in \R$.

For $(b) \Rightarrow (a)$, note that since
\[x(te^{-2s}) = h(-s)x(t)h(s)\quad \mbox{ for all } \quad
s, t \in \R,\] the same equation
$\widetilde{x}(te^{-2s}) = \widetilde{h}(-s)\widetilde{x}(t)\widetilde{h}(s)$ holds in $\widetilde{\mathrm{SL}}(2,\R)$
(both sides are the identity for $s=t=0$).
The invariance of $\Omega$ under the 1-parameter group $\widetilde{h}$ then implies
\[
	\langle \pi(\widetilde{x}(te^{-2s})) \Omega, \Omega \rangle 
	=
	\langle \pi(\widetilde{x}(t)) \Omega,\Omega\rangle.
\]
Since $\lim_{s\rightarrow \infty} \widetilde{x}(te^{-2s})$ is the identity, we have $\langle \pi(\widetilde{x}(t)) \Omega,\Omega\rangle = 1$, and it
follows that
$\pi(\widetilde{x}(t))\Omega = \Omega$ for all $t\in \R$.

Since $h(s)y(t)h(-s) = y(te^{-2s})$,
a similar argument shows that if $\Omega$ is fixed by $\widetilde{h}$, then it is fixed by $\widetilde{y}$.
It follows that if either (a) or (b) hold, then $\Omega$ is fixed by $\widetilde{x}(t)$, $\widetilde{y}(t)$ and $\widetilde{h}(t)$ alike, and hence by
the group $\widetilde{\mathrm{SL}}(2,\R)$ that they generate.
\end{proof}

\bibliographystyle{alpha}

\printindex 

\end{document}